\documentclass[acmsmall]{acmart}
\AtBeginDocument{%
  }

\setcopyright{rightsretained}
\acmDOI{10.1145/3632849}
\acmYear{2024}
\copyrightyear{2024}
\acmSubmissionID{popl24main-p43-p}
\acmJournal{PACMPL}
\acmVolume{8}
\acmNumber{POPL}
\acmArticle{7}
\acmMonth{1}
\received{2023-07-11}
\received[accepted]{2023-11-07}
\newcommand{\qef}{\hfill \ensuremath{\blacksquare}}

\AtBeginDocument{\let\oldendexample\endexample\def\endexample{\qef\oldendexample}}
\AtEndPreamble{\newtheorem{remark}[theorem]{Remark}}%
\AtBeginDocument{\let\oldRemark\remark\let\oldendremark\endremark\def\remark{\oldRemark\normalfont}\def\endremark{\qef\oldendremark}}
\AtEndPreamble{}%
\AtEndPreamble{}%
\hypersetup{hypertexnames=false}
\usepackage{enumitem}
\usepackage[appendix=append]{apxproof}

\newif\ifshort
\makeatletter
\ifthenelse{\equal{\axp@appendix}{append}}{\shorttrue}{\shortfalse}
\makeatother
\ifshort\else\fi
\newcounter{refappendix}
\ifshort
\makeatletter
\newcommand{\proofinapx}{%
\addtocounter{refappendix}{1}%
\immediate\write\axp@proofsfile{\noexpand\hypertarget\string{apx\arabic{refappendix}\string}\string{\string~\string~\string}}%
\hyperlink{apx\arabic{refappendix}}{\rlap{\hskip0.8mm\raisebox{1pt}[0pt][0pt]{\textbf{\scriptsize A}}}$\bigcirc$}}%
\makeatother
\else
\newcommand{\proofinapx}{\relax}
\fi
\usepackage{MnSymbol}
\usepackage{pifont}

\usepackage{tikz}
\usetikzlibrary{shapes}
\usetikzlibrary{arrows}
\newcommand{\ltuple}[1]{\langle\nobreak#1,\allowbreak}
\newcommand{\mtuple}[1]{\:#1\nobreak,\allowbreak}
\newcommand{\rtuple}[1]{\:#1\nobreak\rangle}
\newcommand{\singleton}[1]{\langle#1\rangle}
\newcommand{\pair}[2]{\ltuple{#1}\allowbreak\rtuple{#2}}
\newcommand{\triple}[3]{\ltuple{#1}\mtuple{#2}\rtuple{#3}}
\newcommand{\quadruple}[4]{\ltuple{#1}\mtuple{#2}\mtuple{#3}\rtuple{#4}}

\newcommand{\sextuple}[6]{\ltuple{#1}\mtuple{#2}\mtuple{#3}\mtuple{#4}\mtuple{#5}\rtuple{#6}}
\newcommand{\septuple}[7]{\ltuple{#1}\mtuple{#2}\mtuple{#3}\mtuple{#4}\mtuple{#5}\mtuple{#6}\rtuple{#7}}

\newcommand{\decatuple}[9]
{\ltuple{#1}\mtuple{#2}\mtuple{#3}\mtuple{#4}\mtuple{#5}\mtuple{#6}\mtuple{#7}\mtuple{#8}\mtuple{#9}\rtuple}
\makeatletter
\newcommand{\undecatuple}[9]
{\ltuple{#1}\mtuple{#2}\mtuple{#3}\mtuple{#4}\mtuple{#5}\mtuple{#6}\mtuple{#7}\mtuple{#8}\mtuple{#9}\@undecatuple}
\newcommand{\@undecatuple}[2]{\mtuple{#1}\rtuple{#2}}
\newcommand{\dodecatuple}[9]
{\ltuple{#1}\mtuple{#2}\mtuple{#3}\mtuple{#4}\mtuple{#5}\mtuple{#6}\mtuple{#7}\mtuple{#8}\mtuple{#9}\@dodecatuple}
\newcommand{\@dodecatuple}[3]{\mtuple{#1}\mtuple{#2}\rtuple{#3}}
\newcommand{\tridecatuple}[9]
{\ltuple{#1}\mtuple{#2}\mtuple{#3}\mtuple{#4}\mtuple{#5}\mtuple{#6}\mtuple{#7}\mtuple{#8}\mtuple{#9}\@tridecatuple}
\newcommand{\@dtridecatuple}[4]{\mtuple{#1}\mtuple{#2}\mtuple{#3}\rtuple{#4}}
\newcommand{\tetradecatuple}[9]
{\ltuple{#1}\mtuple{#2}\mtuple{#3}\mtuple{#4}\mtuple{#5}\mtuple{#6}\mtuple{#7}\mtuple{#8}\mtuple{#9}\@tetradecatuple}
\newcommand{\@tetradecatuple}[5]{\mtuple{#1}\mtuple{#2}\mtuple{#3}\mtuple{#4}\rtuple{#5}}
\newcommand{\pentadecatuple}[9]
{\ltuple{#1}\mtuple{#2}\mtuple{#3}\mtuple{#4}\mtuple{#5}\mtuple{#6}\mtuple{#7}\mtuple{#8}\mtuple{#9}\@pentadecatuple}
\newcommand{\@pentadecatuple}[6]{\mtuple{#1}\mtuple{#2}\mtuple{#3}\mtuple{#4}\mtuple{#5}\rtuple{#6}}
\newcommand{\hexadecatuple}[9]
{\ltuple{#1}\mtuple{#2}\mtuple{#3}\mtuple{#4}\mtuple{#5}\mtuple{#6}\mtuple{#7}\mtuple{#8}\mtuple{#9}\@hexadecatuple}
\newcommand{\@hexadecatuple}[7]{\mtuple{#1}\mtuple{#2}\mtuple{#3}\mtuple{#4}\mtuple{#5}\mtuple{#6}\rtuple{#7}}
\makeatletter
\newcommand{\sqb}[1]{\ensuremath{\def\@paramsqb{#1}\ifx\@paramsqb\@empty\else\llbracket\@paramsqb\rrbracket\fi}}
\newcommand{\crb}[1]{\ensuremath{\def\@paramsqb{#1}\ifx\@paramsqb\@empty\else\mathopen{\{\mskip-4mu|}\@paramsqb\mathclose{|\mskip-4mu\}}\fi}}
\makeatother 
\def\displayrelsize{0}
\def\scriptrelsize{-1}
\def\scriptscriptrelsize{-2}
\makeatletter
\newcommand{\sisymbol}{{\ensuremath{\llparenthesis}}}
\newcommand{\alorssymbol}{\ensuremath{\mathbb{?}}}%
\newcommand{\sinonsisymbol}{{\ensuremath{\lceil\mskip-6.5mu\rfloor}}}
\newcommand{\sinonsymbol}{\ensuremath{\mathbb{:}}}%
\newcommand{\fsinonsi}{{\ensuremath{\rrparenthesis}}}
\newcommand{\si}{\ensuremath{\mathopen{\mathchoice{\mbox{\relsize{\displayrelsize}\textup{{\sisymbol}}}}{\mbox{\relsize{\displayrelsize}\textup{{\sisymbol}}}}{\mbox{\relsize{\scriptrelsize}\textup{{\sisymbol}}}}{\mbox{\relsize{\scriptscriptrelsize}\textup{{\sisymbol}}}}}\mskip3mu}}
\newcommand{\alors}{\ensuremath{\mathrel{\mathchoice{\mbox{\relsize{\displayrelsize}\textup{{\alorssymbol}}}}{\mbox{\relsize{\displayrelsize}\textup{{\alorssymbol}}}}{\mbox{\relsize{\scriptrelsize}\textup{{\alorssymbol}}}}{\mbox{\relsize{\scriptscriptrelsize}\textup{{\alorssymbol}}}}}}}
\newcommand{\sinon}{\ensuremath{\mathrel{\mathchoice{\mbox{\relsize{\displayrelsize}\textup{{\sinonsymbol}}}}{\mbox{\relsize{\displayrelsize}\textup{{\sinonsymbol}}}}{\mbox{\relsize{\scriptrelsize}\textup{{\sinonsymbol}}}}{\mbox{\relsize{\scriptscriptrelsize}\textup{{\sinonsymbol}}}}}}}
\newcommand{\sinonsi}{\ensuremath{\mathrel{\mathchoice{\mbox{\relsize{\displayrelsize}\textup{{\sinonsisymbol}}}}{\mbox{\relsize{\displayrelsize}\textup{{\sinonsisymbol}}}}{\mbox{\relsize{\scriptrelsize}\textup{{\sinonsisymbol}}}}{\mbox{\relsize{\scriptscriptrelsize}\textup{{\sinonsisymbol}}}}}}}
\newcommand{\fsi}{\ensuremath{\mskip3mu\mathclose{\mathchoice{\mbox{\relsize{\displayrelsize}\textup{{\fsinonsi}}}}{\mbox{\relsize{\displayrelsize}\textup{{\fsinonsi}}}}{\mbox{\relsize{\scriptrelsize}\textup{{\fsinonsi}}}}{\mbox{\relsize{\scriptscriptrelsize}\textup{{\fsinonsi}}}}}}}
\makeatother
\newcommand{\Lfp}[1]{\ensuremath{\textsf{\upshape lfp}^{\scriptscriptstyle\mskip2mu #1}\,}}
\newcommand{\Gfp}[1]{\ensuremath{\textsf{\upshape gfp}^{\scriptscriptstyle\mskip2mu #1}\,}}

\usepackage{relsize}
\makeatletter
\def\@LAMBDAoperator{\text{\boldmath$\lambda$}}%
\def\@LAMBDApoint{%
\mathchoice%
{\,\mbox{\relsize{2}\bf\raisebox{0.3ex}{.}}\,}%
{\,\mbox{\relsize{2}\bf\raisebox{0.3ex}{.}}\,}%
{\,\mbox{\relsize{1}\bf\raisebox{0.3ex}{.}}\,}%
{\,\mbox{\bf\raisebox{0.3ex}{.}}\,}%
}
\def\LAMBDA#1{\@ifnextchar[{{\@@LAMBDA@IN{#1}}}{{\@@LAMBDA{#1}}}}
\def\@@LAMBDA#1{\@LAMBDAoperator{#1}{\@LAMBDApoint}}
\def\@@LAMBDA@IN#1[#2]{\@LAMBDAoperator{#1}\,{\in}\,{#2}{\@LAMBDApoint}}
\def\LAMBDAoperator{{\ensuremath{\@LAMBDAoperator}}}
\makeatother 
\usepackage{accents}
\usepackage{rotating}
\makeatletter
\DeclareRobustCommand{\cev}[1]{%
  \mathpalette\do@cev{#1}%
}
\newcommand{\do@cev}[2]{%
  \fix@cev{#1}{+}%
  \reflectbox{$\m@th#1\vec{\reflectbox{$\fix@cev{#1}{-}\m@th#1#2\fix@cev{#1}{+}$}}$}%
  \fix@cev{#1}{-}%
}
\newcommand{\fix@cev}[2]{%
  \ifx#1\displaystyle
    \mkern#23mu
  \else
    \ifx#1\textstyle
      \mkern#23mu
    \else
      \ifx#1\scriptstyle
        \mkern#22mu
      \else
        \mkern#22mu
      \fi
    \fi
  \fi
}
\newbox\mystrutbox
\newcommand{\@ustrut}[1]{\setbox\mystrutbox\hbox{#1\strut}\hbox{\vrule 
     height\ht\mystrutbox
     depth\z@
     width\z@}}
\newcommand{\@lstrut}[1]{\setbox\mystrutbox\hbox{#1\strut}\hbox{\vrule 
     height\z@
     depth\dp\mystrutbox
     width\z@}}
\newcommand{\ustrut}{\@ustrut{\normalfont}}
\newcommand{\lstrut}{\@lstrut{\normalfont}}
\newcommand{\ulstrut}{\@ustrut{\normalfont\relsize{1}}}
\newcommand{\uLstrut}{\@ustrut{\normalfont\relsize{2}}}
\newcommand{\uhstrut}{\@ustrut{\normalfont\relsize{3}}}
\newcommand{\uHstrut}{\@ustrut{\normalfont\relsize{4}}}
\newcommand{\llstrut}{\@lstrut{\normalfont\relsize{1}}}
\newcommand{\lLstrut}{\@lstrut{\normalfont\relsize{2}}}
\newcommand{\lhstrut}{\@lstrut{\normalfont\relsize{3}}}
\newcommand{\lHstrut}{\@lstrut{\normalfont\relsize{4}}}

\newcommand{\usstrut}{\@ustrut{\normalfont\relsize{-1}}}
\newcommand{\lsstrut}{\@lstrut{\normalfont\relsize{-1}}}
\newcommand{\ussstrut}{\@ustrut{\normalfont\relsize{-2}}}
\newcommand{\lssstrut}{\@lstrut{\normalfont\relsize{-2}}}
\newcommand{\usssstrut}{\@ustrut{\normalfont\relsize{-3}}}
\newcommand{\lsssstrut}{\@lstrut{\normalfont\relsize{-3}}}
\makeatother
\usepackage {interval}
\usepackage {galois}
\usepackage {calculus,eqntabular,renumber}
\newcommand{\angelic}{\not\bot}%
\newcommand{\demoniac}{\bot}%
\newsavebox{\largecirclebox}
\savebox{\largecirclebox}{\normalfont\Large$\largecircle$}
\newsavebox{\numberbox}
\newcommand{\circled}[1]{\savebox{\numberbox}{\scriptsize\upshape#1}%
\rlap{\raisebox{1pt}[0pt][0pt]{\hskip0.5\wd\largecirclebox\hskip-0.5\wd\numberbox\usebox{\numberbox}}}\usebox{\largecirclebox}}
\newcommand{\circledgray}[1]{\savebox{\numberbox}{\color{gray}\scriptsize\upshape#1}%
\rlap{\raisebox{1pt}[0pt][0pt]{\hskip0.5\wd\largecirclebox\hskip-0.5\wd\numberbox\usebox{\numberbox}}}{\normalfont\Large$\color{gray}\largecircle$}}
\usepackage{wrapfig}
\makeatletter\AtBeginDocument{\makeatletter
\immediate\write\axp@proofsfile{\string\noindent\string\textbf\string{\string\LARGE\string{Appendix of ``Calculational Design of [In]Correctness Transformational Program Logics by Abstract Interpretation''\string\ulstrut\string}\string}\string\\[1em]}
\immediate\write\axp@proofsfile{\string\textsf\string{PATRICK COUSOT\string}, Courant Institute of Mathematical Studies, New York University, USA\string\\[-0.75em]}
\immediate\write\axp@proofsfile{\string\let\string\my@title\string\@title}
\immediate\write\axp@proofsfile{\string\def\string\@title\string{Appendix of  ``\string\my@title''\string}}
\immediate\write\axp@proofsfile{\string\if@ACM@printacmref\string\@mkbibcitation\string\fi}
\makeatother}\makeatother
\begin{document}
\title{Calculational Design of [In]Correctness Transformational Program Logics by Abstract Interpretation}

\author{Patrick Cousot}
\email{pcousot@cims.nyu.edu}
\orcid{0000-0003-0101-9953}
\affiliation{%
  \institution{New York University}
  \city{}
  \country{USA}
}

\renewcommand{\shortauthors}{Patrick Cousot}

\begin{abstract}
We study transformational program logics for correctness and incorrectness that we extend to explicitly handle both termination and nontermination. We show that the logics are abstract interpretations of the right image transformer for a natural relational semantics covering both finite and infinite executions. This understanding of logics as abstractions of a semantics facilitates their comparisons through their respective abstractions of the semantics  (rather that the much more difficult comparison through their formal proof systems). More importantly, the formalization provides a calculational method for constructively designing the sound and complete formal proof system by abstraction of the semantics. As an example, we extend Hoare logic to cover all possible behaviors of nondeterministic programs and design a new precondition (in)correctness logic.
\end{abstract}

\keywords{program logic, transformer, semantics, correctness, incorrectness, termination, nontermination, abstract interpretation\setcounter{TotPages}{-32}
}

%

\maketitle
\section{Introduction}
\vphantom{\circled{1}}
In verification, the focus is on which program properties can be expressed and proved. We discuss
transformational (or Hoare's style) logics characterized by formulas expressing program properties that relate initial/input values of variables to their final/output values, nontermination, or runtime errors (or inversely final to initial) and a  Hilbert-style proof system \cite[\S10]{HilbertAckermann-Logic28} to prove that a program has a property expressed by a formula of the logic (but not that a given program does not have a property expressed by a formula of the logic or that no program can have this property \cite{DBLP:journals/pacmpl/KimDR23}). Examples are Hoare's logic \cite{DBLP:journals/cacm/Hoare69} and the reverse Hoare logic \cite{DBLP:conf/sefm/VriesK11} aka incorrectness logic \cite{DBLP:journals/pacmpl/OHearn20}. 
 
\subsection{The Classic Proof-Theoretic Approach}
The ``classic approach'' to the design of a Hoare style logic follows the proof-theoretic semantics in logic originated by Hilbert, Gentzen,
Prawitz, and others \cite{DBLP:journals/sLogica/PiechaS19}. The true program properties are the provable ones, which is also the idea of ``axiomatic semantics'' \cite{DBLP:books/daglib/0070910}, that is, Floyd's idea that a program proof method is ``Assigning Meaning to Programs'' \cite{Floyd67-1}. First the syntax of program properties is defined (e.g\@. $P\{C\}Q$, $\{P\}C\{Q\}$, $[P]C[Q]$). Then proof rules are postulated  (e.g\@. ``If $\vdash P\{Q\}R$ and $\vdash R\supset  S$ then $\vdash P\{Q\}S$'' \cite[page 578]{DBLP:journals/cacm/Hoare69}). Finally, soundness and completeness theorems are proved to relate the logic properties to a more concrete/refined semantics (e.g\@. years after its design, Hoare logic \cite{DBLP:journals/cacm/Hoare69} was proved sound  by Donahue \cite{DBLP:books/sp/Donahue76} (with respect to a denotational semantics) and sound and relatively complete by Pratt \cite{DBLP:conf/focs/Pratt76} (with respect to a relational semantics excluding nontermination) and Cook \cite{DBLP:journals/siamcomp/Cook78,DBLP:journals/siamcomp/Cook81} (with respect to an operational trace semantics)). This design method has perdured over time, even if, nowadays, soundness and completeness proofs are often published together with the logic (e.g\@. 
\cite{DBLP:journals/jacm/BruniGGR23,%
DBLP:journals/afp/Dardinier23a,%
DBLP:conf/sefm/VriesK11,%
DBLP:journals/entcs/GotsmanBC11,%
DBLP:conf/RelMiCS/MollerOH21,%
DBLP:journals/pacmpl/ZhangAG22,%
DBLP:journals/pacmpl/OHearn20,%
DBLP:conf/sas/Vanegue22,%
DBLP:journals/pacmpl/ZhangK22,%
DBLP:journals/pacmpl/ZilbersteinDS23} a.o\@.). 
Therefore, in this ``classic approach'' the program properties of interest (partial correctness, total correctness, incorrectness, etc) are the one provable by the proof system, while soundness and completeness theorems aims at connecting the provable properties to the program semantics.

\subsection{The Model-Theoretic Semantic Abstraction Approach}
In this paper, we consider an alternative ``semantic abstraction approach'' which is based on Tarski’s truth paradigm 
\cite{Tarski-truth-44} in model theory and the abstract interpretation of the semantics of languages  \cite{DBLP:conf/popl/CousotC77,Cousot-PAI-2021}. First, a formal semantics is specified for the language (preferable using structural fixpoints or deductive proof systems). This induces a collecting semantics defining the strongest (hyper) property of programs. Then the program properties of interest for the logic are specified by a Galois connection abstracting the collecting (hyper) properties. The abstraction is usually the composition of several primitive ones, in the spirit of \cite{DBLP:conf/popl/CousotC14}. Varying the primitives and their composition yields different logics.
At this point, the logic is precisely and fully determined since all expressible properties of all programs have been formally specified. For example, the logic can be compared and combined with other logics (see e.g\@. Figs\@. \ref{fig:Forward-semantics-logics}, \ref{fig:taxonomy}, \ref{fig:taxonomy-assertional} and the taxonomy in Sect\@. \ref{sec:subhierarchy-assertional-logics}). Finally the rules of the proof system are designed by calculus using fixpoint abstraction (Sect\@. \ref{sect:FixpointAbstraction}), fixpoint induction principles  (Sect\@. \ref{sect:FixpointInduction}), and Peter Aczel \cite{Aczel:1977:inductive-definitions} construction of deductive rule-based systems from fixpoints, or conversely (Sect\@. \ref{sect:SemanticsDeductiveSystems}).

The advantage is that reasoning on abstractions of program properties is much more concise and easy than reasoning on proof systems. This clearly appears e.g\@. in Fig\@.  \ref{fig:taxonomy} comparing 40 logics by combining only 8 abstractions (plus one, common to all logics defining ``transformational''). Fig\@.  \ref{fig:taxonomy} is itself part of the lattice of abstract interpretations of \cite[section 8]{DBLP:conf/popl/CousotC77} including many logics whose abstraction is given in this paper. Another advantage is that the proof system is derived by calculus so sound and complete by construction.

\subsection{The Structure of the Paper}

The paper has two main parts. In the first part, we characterize the semantics of a transformational logics, i.e\@. the true formulas (a theory in logic), as an abstract interpretation of the program (collecting) semantics. This allows us to provide a taxonomy of transformational semantics by comparing their abstractions, without referring to their proof systems.

After showing that theories of logics are set abstractions of the program (collecting) semantics in the first part, we have to design the corresponding proof systems in the second part.

Aczel has shown that deductive rule-based systems and set-theoretic fixpoint definitions are equivalent \cite{Aczel:1977:inductive-definitions}. Therefore we first define the program semantics
in fixpoint form, then abstract this semantics to get a fixpoint definition of the theory of the logic, and finally apply Aczel's method to derive the equivalent
proof system. The proof system is then sound and complete by construction.

\ifshort We use the clickable symbol \proofinapx\ to refer to complements and proofs found in the appendix.\fi

\vskip-2ex
\part{}{{\LARGE\bfseries Part I: Design of the Theory of Logics by Abstraction of the \phantom{\LARGE\bfseries Part I:\ $\,$}Program Semantics}}\label{part:SemanticsLogics}
\let\originalthesection\thesection
\setcounter{section}{0}%
\def\thesection{I.\originalthesection}%
\vskip2ex

In part I, we show that the theory (or semantics) of transformational logics are abstractions of the relational 
semantics of programs, which leads to a taxonomy of transformational logics, as well as, to
their combinations. The meaning or semantics of a logic is the set of true formulas of that logic which is also
called the theory of the logic. So we use ``theory'' for the meaning of a logic and ``semantics'' for the meaning of a program or a programming language.

\ifshort\vskip-2em\fi

\section{Relational Semantics}\label{sect:RelationalSemantics}
``Relational'' means that the semantics defines a relation between initial states of executions and final states or $\bot$ to denote nontermination (as conventional in denotational semantics \cite{ScottStrachey71-PRG6}). 
\ifshort Our notations on relations are classic and defined in the appendix \proofinapx.\vskip-1em\fi
\begin{toappendix}
\subsection{Relations}\label{sect:Relations}
The Cartesian product of sets $\mathcal{X}$ and $\mathcal{Y}$ is $\mathcal{X}\times\mathcal{Y}\triangleq\{\pair{x}{y}\mid x\in \mathcal{X} \wedge y\in \mathcal{Y}\}$. The Cartesian power is 
$\mathcal{X}^n\triangleq\{\triple{x_1}{\ldots}{x_n}\mid \forall i\in\interval{1}{n}\mathrel{,}x_i\in\mathcal{X}\}$, $n\in\mathbb{N}$, with $\mathcal{X}^0=\emptyset$ and $\mathcal{X}^1=X$.
The powerset $\wp(\mathcal{X})$ of a set $\mathcal{X}$ is the set of all subsets of $\mathcal{X}$. A relation $r$ on sets $\mathcal{X}$ and $\mathcal{Y}$ is a set of pairs in their Cartesian product so
$r\in\wp(\mathcal{X}\times\mathcal{Y})$.  Its domain is
$\textsf{dom}(r)\triangleq\{x\mid\exists y\mathrel{.}\pair{x}{y}\in r\}$ and its codomain is $\textsf{codom}(r)\triangleq\{y\mid\exists x\mathrel{.}\pair{x}{y}\in r\}$. The inverse of a relation $r$ is $r^{-1}\triangleq\{\pair{y}{x}\mid \pair{x}{y}\in r\}$. The left composition of relations is $r_1\fatsemi r2\triangleq\{\pair{x}{z}\mid\exists y\mathrel{.}\pair{x}{y}\in r_1\wedge \pair{y}{z} \in r_2\}$. The composition of functions is $g\comp f\triangleq\LAMBDA{x}g(f(x))$ which, for their graphs, is the  functional relation $\{\pair{x}{f(x)}\mid x\in\textsf{dom}(f)\}\fatsemi\{\pair{y}{g(y)}\mid y\in\textsf{dom}(g)\}$. A relation is functional if it is the graph of a total function.  Therefore, the set of functional relations between $\mathcal{X}$ and $\mathcal{Y}$ is defined as the set of relations such that any element of their domain has a unique image in the codomain, that is, $\wp_{\textsf{fun}}(\mathcal{X}\times\mathcal{Y})\triangleq
\{r\in\wp(\mathcal{X}\times\mathcal{Y})\mid\forall x\in\mathcal{X}\mathrel{.}\exists y\in\mathcal{Y}\mathrel{.}\pair{x}{y}\in r\wedge\forall y,z\in\mathcal{Y}\mathrel{.}\pair{x}{y}\in r\wedge\pair{x}{z}\in r\Rightarrow y=z\}$.  We extend the definition of the left relation composition $\fatsemi$ to nontermination $\bot$ by $r\fatsemi r'\triangleq\{\pair{x}{\bot}\mid \pair{x}{\bot}\in r\}\cup\{\pair{x}{y}\mid (\exists z\in \mathcal{X}\setminus\{\bot\}\mathrel{.}\pair{x}{z}\in r\wedge \pair{z}{y}\in r'\}$. 
The conditional with several alternatives, à la \texttt{C}, is $\si \ldots \alors \ldots\sinonsi \ldots\alors \ldots\sinonsi \ldots\sinon\ldots\fsi$ where $\sinonsi$ is the optional ``else if''.
\end{toappendix}
\ifshort\vskip-1em\fi
\subsection{Structural Deductive Definition of the Natural Relational Semantics}\label{sec:natural-relational-semantics-deductive}
We consider an imperative language $\mathbb{S}$ with assignments, sequential composition, conditionals, and conditional iteration with breaks. The syntax is $\texttt{\small S}\in\mathbb{S}\mathbin{{:}{:}{=}}\texttt{\small x = A}
\mid\texttt{\small x = [$a$,$b$]}
\mid\texttt{\small skip}
\mid\texttt{\small S;S}
\mid\texttt{\small if (B) S else S}
\mid\texttt{\small while (B) S}
\mid\texttt{\small break}
$. The nondeterministic assignment \texttt{\small x = [$a$, $b$]} with $a\in\mathbb{Z}\cup\{-\infty\}$ and $b\in\mathbb{Z}\cup\{\infty\}$, $-\infty-1=-\infty$, $\infty+1=\infty$ may be unbounded. \texttt{\small break} is a simple form of exception (to answer a question on exceptions by \href{https://felleisen.org/matthias/}{Matthias Felleisen} at POPL 2014 \cite{DBLP:conf/popl/CousotC14}).

States $\sigma\in\Sigma\triangleq\mathbb{X}\rightarrow \mathbb{V}$ (also called environments) map variables  $\texttt{\small x}\in\mathbb{X}$ to their values $\sigma(\texttt{\small x})$ in $\mathbb{V}$ including integers, $\mathbb{Z}\subseteq \mathbb{V}$. We let $\bot\not\in\Sigma$ denote nontermination with $\Sigma_{\bot}\triangleq\Sigma\cup\{\bot\}$. 

We deliberately leave unspecified the syntax and semantics of arithmetic expressions $\mathcal{A}\sqb{\texttt{\small A}}\in\Sigma\rightarrow \mathbb{V}$ and Boolean expressions  $\mathcal{B}\sqb{\texttt{\small B}}\in\wp(\Sigma)\simeq\Sigma\rightarrow \{\textsf{\upshape true},\textsf{\upshape false}\}$. The only assumption on expressions is the absence of side effects.

The relational semantics $\sqb{\texttt{\small S}}_{\bot}$ of a command $\texttt{\small S}\in\mathbb{S}$ is an element of $\wp(\Sigma\times\Sigma_{\bot})$. Formally, $\pair{\sigma}{\sigma'}\in\sqb{\texttt{\small S}}_{\bot}$ means that an execution of the nondeterministic command $\texttt{\small S}$ from initial state $\sigma\in\Sigma$ may terminate in final state $\sigma'\in\Sigma$ or may not terminate when $\sigma'=\bot$. (The relational semantics could have been proven to be the abstraction of a finite and infinite trace semantics \cite{Cousot-PAI-2021}.) The right-image $\LAMBDA{\sigma}\{\sigma'\in\Sigma_\bot\mid\pair{\sigma}{\sigma'}\in\sqb{\texttt{\small S}}_{\bot}\}$ of the natural relational semantics $\sqb{\texttt{\small S}}_{\bot}$ is isomorphic to Plotkin's natural denotational semantics \cite{DBLP:journals/siamcomp/Plotkin76}. Such natural relational semantics have been originated by Park \cite{DBLP:conf/ac/Park79}.

We partition the relational natural semantics into the semantics $\sqb{\texttt{\small S}}^e\in\wp(\Sigma\times\Sigma)$ of statement \texttt{\small S} terminating/\underline{e}nding normally, the semantics $\sqb{\texttt{\small S}}^b\in\wp(\Sigma\times\Sigma)$ of statement \texttt{\small S} terminating by a \texttt{\small \underline{b}reak} statement, and the semantics $\sqb{\texttt{\small S}}^\bot\in\wp(\Sigma\times\{\bot\})$ of nontermination \underline{$\bot$}. Therefore $\sqb{\texttt{\small S}}_{\bot}\triangleq\sqb{\texttt{\small S}}^e\cup\sqb{\texttt{\small S}}^b\cup\sqb{\texttt{\small S}}^\bot$. The angelic semantics 
\bgroup\abovedisplayskip0.5\abovedisplayskip\belowdisplayskip0.5\belowdisplayskip\begin{eqntabular}{rclcl}
\sqb{\texttt{\small S}}&\triangleq&\sqb{\texttt{\small S}}^e\cup\sqb{\texttt{\small S}}^b&=& 
\sqb{\texttt{\small S}}_{\bot}\cap(\Sigma\times\Sigma)
\label{eq:angelic-semantics}
\end{eqntabular}\egroup
ignores non termination. 

We follow the tradition established by Plotkin \cite{DBLP:journals/jlp/Plotkin04,DBLP:journals/jlp/Plotkin04a} to define the program semantics by structural induction (i.e\@. by induction on the program syntax) using a deductive system of rules. We extend the semantics of the deductive system using bi-induction combining induction for terminating executions and co-induction for nonterminating ones \cite{DBLP:conf/popl/CousotC92,DBLP:conf/cav/CousotC95,DBLP:journals/iandc/CousotC09}. 

Let us write judgements $\sigma\vdash\texttt{\small S}\stackrel{e}{\Rightarrow}\sigma'$ for 
$\pair{\sigma}{\sigma'}\in\sqb{\texttt{\small S}}^e$, $\sigma\vdash\texttt{\small S}\stackrel{b}{\Rightarrow}\sigma'$ for 
$\pair{\sigma}{\sigma'}\in\sqb{\texttt{\small S}}^b$, and $\sigma\vdash\texttt{\small S}\stackrel{\infty}{\Rightarrow}$ for 
$\pair{\sigma}{\bot}\in\sqb{\texttt{\small S}}^\bot$. Moreover, for the conditional iteration statement \texttt{\small W} $\triangleq$ \texttt{\small while (B) S}, we write $\sigma\vdash\texttt{\small W}\stackrel{i}{\Rightarrow}\sigma'$ to mean that if $\sigma$ is a state
before executing \texttt{\small W}, then $\sigma'$ is reachable after 0 or more iterations of the loop body (so  $\sigma=\sigma'$ for 0 iterations, before entering the loop in case (\ref{eq:while:invariant}.a)). We have the axiom and inductive rule for iterations \texttt{\small W} 
\bgroup\abovedisplayskip0.3\abovedisplayskip
\belowdisplayskip0.3\belowdisplayskip
\begin{eqntabular}{c@{\qquad\qquad}c}
\textup{(a)}\quad\sigma\vdash\texttt{\small W}\stackrel{i}{\Rightarrow}\sigma
&
\textup{(b)}\quad\frac{\mathcal{B}\sqb{\texttt{\small B}}\sigma,\quad \sigma\vdash\texttt{\small S}\stackrel{e}{\Rightarrow}\sigma',\quad \sigma'\vdash\texttt{\small W}\stackrel{i}{\Rightarrow}\sigma''}{\sigma\vdash\texttt{\small W}\stackrel{i}{\Rightarrow}\sigma''}
\label{eq:while:invariant}
\end{eqntabular}\egroup
The following axioms define termination (these are axioms since the precondition has been previously established either by $\stackrel{i}{\Rightarrow}$ or by structural induction). (\ref{eq:W:e}.b) is for termination by a \texttt{\small break}.
\bgroup\abovedisplayskip0.33\abovedisplayskip\begin{eqntabular}{c@{\qquad\qquad}c}
\textup{(a)}\quad\frac{\sigma\vdash\texttt{\small W}\stackrel{i}{\Rightarrow}\sigma',\quad\mathcal{B}\sqb{\neg\texttt{\small B}}\sigma'}{\sigma\vdash\texttt{\small W}\stackrel{e}{\Rightarrow}\sigma'}
&
\textup{(b)}\quad\frac{\sigma\vdash\texttt{\small W}\stackrel{i}{\Rightarrow}\sigma',\quad\mathcal{B}\sqb{\texttt{\small B}}\sigma',\quad\sigma'\vdash\texttt{\small S}\stackrel{b}{\Rightarrow}\sigma''}{\sigma\vdash\texttt{\small W}\stackrel{e}{\Rightarrow}\sigma''}
\label{eq:W:e}
\end{eqntabular}\egroup
The following axiom and co-inductive rule define nontermination (the left rule is an axiom since the precondition has already been defined either by $\stackrel{i}{\Rightarrow}$ or by structural induction). Rule (\ref{eq:W:infty}.b) right-marked $\infty$ is co-inductive.
\bgroup\abovedisplayskip0.33\abovedisplayskip\begin{eqntabular}{c@{\qquad}c}
\textup{(a)}\quad\frac{\sigma\vdash\texttt{\small W}\stackrel{i}{\Rightarrow}\sigma',\quad\mathcal{B}\sqb{\texttt{\small B}}\sigma',\quad \sigma'\vdash\texttt{\small S}\stackrel{\infty}{\Rightarrow}}{\sigma\vdash\texttt{\small W}\stackrel{\infty}{\Rightarrow}}
&
\textup{(b)}\quad\frac{\mathcal{B}\sqb{\texttt{\small B}}\sigma,\quad \sigma\vdash\texttt{\small S}\stackrel{e}{\Rightarrow}\sigma', \quad
\sigma'\vdash\texttt{\small W}\stackrel{\infty}{\Rightarrow}}{\sigma\vdash\texttt{\small W}\stackrel{\infty}{\Rightarrow}}\infty
\label{eq:W:infty}
\end{eqntabular}\egroup

\subsection{State Properties, Semantics Properties, and Collecting Semantics}
We define properties in extension as the set of elements of a universe $\mathbb{U}$ that have this property. So false is $\emptyset$, true is $\mathbb{U}$, logical implication is $\subseteq$, disjunction is $\cup$, conjunction is $\cap$, negation is $\neg P\triangleq\mathbb{U}\setminus P$ and $\sextuple{\wp(\mathbb{U})}{\emptyset}{\mathbb{U}}{\cup}{\cap}{\neg}$ is a complete Boolean lattice \cite{DBLP:books/sp/Gratzer98}.

For example, properties of states $\sigma\in\Sigma_\bot$ (considered to be the universe) belong to
$\wp(\Sigma_\bot)$. The singleton $\{\bot\}$ is the property ``not to terminate'', $\emptyset$ is ``false'', $\{\sigma_1,\ldots,\sigma_n\}\subseteq\Sigma$ is ``to terminate with any one of the states $\sigma_1,\ldots,\sigma_n\in\Sigma$'',  $\{\sigma_1,\ldots,\sigma_n,\bot\}$ is `` ``to terminate with any one of the states $\sigma_1,\ldots,\sigma_n\in\Sigma$ or not to terminate'', $\Sigma$ is to terminate, $\Sigma_{\bot}$ is ``true'' i.e\@. ``to terminate with any state in $\Sigma$ or not to terminate'' (the common alternative to terminate with an error is assumed to be encoded with some specific values in the set $\Sigma$ of states).

Let $\sqb{\texttt{\small S}}_\bot\in\wp(\Sigma\times\Sigma_\bot)$ be the natural relational semantics of programs $\texttt{\small S}\in\mathbb{S}$ in Sect\@. \ref{sec:natural-relational-semantics-deductive}. When defined in extension, semantic properties belong to $\wp(\wp(\Sigma\times\Sigma_\bot))$. The program collecting semantics 
$\crb{\texttt{\small S}_\bot}\triangleq \{\sqb{\texttt{\small S}_\bot}\}\in\wp(\wp(\Sigma\times\Sigma_\bot))$ is the strongest (hyper) property of program $\texttt{\small S}$.

\section{Galois Connections}\label{sect:GaloisConnections}
Galois connections \cite[Ch\@. 11]{Cousot-PAI-2021} are used throughout the paper. They formalize correspondences between program properties which preserve implication and one is less precise/expressive than the other. The interest is that proofs in the abstract are valid in the concrete (or equivalent in case of Galois isomorphisms). Moreover, there is a most precise way to abstract any concrete property or logic, which provides a guideline for calculational design of logics from a program semantics.
\ifshort The definition and  properties of Galois connections are recalled in the appendix \proofinapx.
\fi
\begin{toappendix}
Formally,
a Galois connection $\pair{C}{\sqsubseteq}\galois{\alpha}{\gamma}\pair{A}{\preceq}$ is a pair $\pair{\alpha}{\gamma}$ of functions between posets 
$\pair{C}{\sqsubseteq}$ and $\pair{A}{\preceq}$ satisfying $\forall x\in C\mathrel{.}\forall y\in A\mathrel{.}\alpha(x)\preceq y\Leftrightarrow x\sqsubseteq\gamma(y)$. We use a double headed arrow \rlap{$\mskip9.5mu{\rightarrow}$}$\longrightarrow$ to indicate surjection and $\GaloiS{}{}$ for bijections. We use classic properties of Galois connections for which proofs are found in \cite[Ch\@. 11]{Cousot-PAI-2021} and \cite{DeneckeErneWismath-GC-03}. In particular, the abstraction $\alpha$ preserves existing arbitrary joins (so in strict and increasing) and dually the concretization $\gamma$ preserves existing arbitrary meets.

\begin{proposition}\label{prop:GC-composition} The composition of Galois connections with Galois connections and isomorphisms is a Galois connection. 
\end{proposition}
\begin{proof}[Proof of (\ref{prop:GC-composition})]
\hyphen{5}\quad Assume that $\pair{\mathcal{X}}{\sqsubseteq}\galois{\alpha_1}{\gamma_1}\pair{\mathcal{Y}}{\preceq}$ and
$\pair{\mathcal{Y}}{\preceq}\galois{\alpha_2}{\gamma_2}\pair{\mathcal{Z}}{\leq}$. Then
\begin{calculus}[$\Leftrightarrow$\ ]
\formula{\alpha_2\comp\alpha_1 (x)\leq y}\\
$\Leftrightarrow$
\formulaexplanation{\alpha_1 (x)\preceq \gamma_2(y)}{by $\pair{\mathcal{Y}}{\preceq}\galois{\alpha_2}{\gamma_2}\pair{\mathcal{Z}}{\leq}$}\\[2pt]
$\Leftrightarrow$
\formulaexplanation{x\sqsubseteq \gamma_1\comp\gamma_2(y)}{by $\pair{\mathcal{X}}{\sqsubseteq}\galois{\alpha_1}{\gamma_1}\pair{\mathcal{Y}}{\preceq}$}
\end{calculus}
proving $\pair{\mathcal{X}}{\sqsubseteq}\galois{\alpha_2\,\comp\,\alpha_1}{\gamma_1\,\comp\,\gamma_2}\pair{\mathcal{Z}}{\leq}$.

\hyphen{5}\quad Assume that $\pair{\mathcal{X}}{=}\GaloiS{\alpha}{\alpha^{-1}}\pair{\mathcal{Y}}{=}$ and
$\pair{\mathcal{Y}}{\preceq}\galois{\alpha_2}{\gamma_2}\pair{\mathcal{Z}}{\leq}$. Define the partial order $\pair{\mathcal{X}}{\sqsubseteq}$
by $x\sqsubseteq y$ $\Leftrightarrow$ $\alpha(x)\preceq\alpha(y)$.
\begin{calculus}[$\Leftrightarrow$\ ]
\formula{\alpha_2\comp\alpha (x)\leq y}\\
$\Leftrightarrow$
\formulaexplanation{\alpha (x)\preceq \gamma_2(y)}{by $\pair{\mathcal{Y}}{\preceq}\galois{\alpha_2}{\gamma_2}\pair{\mathcal{Z}}{\leq}$}\\
$\Leftrightarrow$
\formulaexplanation{\alpha (x)\preceq \alpha\comp\alpha^{-1}\comp\gamma_2(y)}{by $\pair{\mathcal{X}}{=}\GaloiS{\alpha}{\alpha^{-1}}\pair{\mathcal{Y}}{=}$ }\\
$\Leftrightarrow$
\formulaexplanation{x\sqsubseteq \alpha^{-1}\comp\gamma_2(y)}{def\@. $\sqsubseteq$}
\end{calculus}
proving $\pair{\mathcal{X}}{\sqsubseteq}\galois{\alpha_2\,\comp\,\alpha}{\alpha^{-1}\,\comp\,\gamma_2}\pair{\mathcal{Z}}{\leq}$.

\hyphen{5}\quad The proof that $\pair{\mathcal{X}}{\sqsubseteq}\galois{\alpha_1}{\gamma_1}\pair{\mathcal{Y}}{\preceq}$ and
$\pair{\mathcal{Y}}{=}\GaloiS{\alpha}{\alpha^{-1}}\pair{\mathcal{Z}}{=}$ implies $\pair{\mathcal{X}}{\sqsubseteq}\galois{\alpha\,\comp\,\alpha_1}{\gamma_1\,\comp\,\alpha^{-1}}\pair{\mathcal{Z}}{\leq}$ where $x\leq y$ $\Leftrightarrow$ $\alpha^{-1}(x) \preceq\alpha^{-1}(y)$ is similar.
\end{proof}

While a Galois connection $\pair{C}{\sqsubseteq}\galois{\alpha}{\gamma}\pair{A}{\preceq}$ is appropriate for over approximation, a semidual Galois correspondence $\pair{C}{\sqsubseteq}\galois{\alpha}{\gamma}\pair{A}{\succeq}$ (as originally defined by \'Evariste Galois for $\mathord{\sqsubseteq}\triangleq\mathord{\preceq}\triangleq\mathord{\subseteq}$) is convenient for under approximation (as discussed e.g\@. by \cite{DBLP:conf/fossacs/AscariBG22}). The dual $\pair{A}{\succeq}\galois{\gamma}{\alpha}\pair{C}{\sqsupseteq}$ may also be useful (e.g\@. to approximate greatest rather that least fixpoints).

By Prop\@. \ref{prop:GC-composition}, the composition of Galois connections (with corresponding partial orders) is a Galois connection. We will compose  Galois connections to show that all known transformational logics are abstractions of the natural relation semantics.
\end{toappendix}

\section{The Design of a Natural Transformational Logic Theory by Composing Abstractions of the Natural Relational Semantics}
\label{sec:DesignTheory}
A program logic consists of formal statements some of which are true and constitute the theory of the logic. Our objective in this section
is to characterize the theory of transformational logics by abstraction of the natural relational collecting semantics. This abstraction is obtained by composition of basic
Galois connections and functors introduced in this section.

\begin{example}\label{ex:fact:spec}The body \texttt{\small fact $\equiv$  while (n!=0) \{ f = f * n; n = n - 1;\}} of the factorial 
\texttt{\small f = 1; fact} can be specified as
$
\{ n=\underline{n}\wedge f=1\}\,\texttt{\small fact}\,\{(\underline{n}\geqslant0\wedge f=!\underline{n})\vee (\underline{n}<0\wedge n=f=\bot)\}
$
where, following Manna \cite{DBLP:journals/jcss/Manna71},  $\underline{\texttt{\small x}}$ or $\texttt{\small x}_0$ denotes the initial value of variable $\texttt{\small x}$ in the postcondition. When later incorporating \texttt{\small break} statements, the specification will be $\{ n=\underline{n}\wedge f=1\}\,\texttt{\small fact}\,\{ok:(\underline{n}\geqslant0\wedge f=!\underline{n})\vee (\underline{n}<0\wedge n=f=\bot),\ br:\textsf{false}\}$.

Assuming $P\neq\emptyset$ (false) and $\bot\not\in Q$, $\{P\}\texttt{\small S}\{Q\}$ specifies total correctness, as does Manna and Pnueli logic \cite{DBLP:journals/acta/MannaP74}. $\{P\}\texttt{\small S}\{\bot\}$ specifies definite non termination. Otherwise, when $\bot\in Q$, $\{P\}\texttt{\small S}\{Q\}$ expresses partial correctness, as does Hoare logic \cite{DBLP:journals/cacm/Hoare69}. 

By adding auxiliary variables (see Sect\@. \ref{sec:auxiliary:variables}\ifshort\ in the appendix\fi), this specification can also be partially formulated by two Hoare triples $\{ n=\underline{n}\geqslant0\wedge f=1\}\,\texttt{\small fact}\,\{f=!\underline{n}\}$ (although not ensuring termination) and $\{n<0\wedge f=1\}\,\texttt{\small fact}\,\{\textsf{false}\}$ (ensuring nontermination) but the conjunction of Hoare triples is not a Hoare triple and anyway the partial specification cannot preclude nontermination when $\underline{n}\geqslant0$. 

This specification cannot be expressed by Manna and Pnueli \cite{DBLP:journals/acta/MannaP74} logic since the program is not totally correct. 

The theory of the adequate logic (that we call the natural transformational over approximation logic) will be formally specified in (\ref{eq:natural-upper-transformational-logic-theory}) as $\{P\}\texttt{\small S}\{Q\}$, $P\in\wp(\Sigma\times\Sigma)$, $Q\in \wp(\Sigma\times\Sigma_\bot)$ if and only if $\forall \pair{\underline{\sigma}}{\sigma}\in P\mathrel{.}\forall \sigma'\mathrel{.}
\pair{\sigma}{\sigma'}\in\sqb{\texttt{\small S}}_\bot \Rightarrow\pair{\underline{\sigma}}{\sigma'}\in{Q}$. 
The proof system of this logic is designed in Sect\@. \ref{Calculational-Design-of-the-Extended-Hoare-Logic}.
\end{example}

\subsection{Collecting Semantics to Semantics Abstraction}\label{sec:Collecting-semantics-semantics-abstraction}
The collecting semantics of a program component is its strongest property, so transformational logic statements are weaker abstract properties that we specify by composition of Galois connections. The first abstraction $\alpha_C$ abstracts hyper properties into properties.

Let $\mathcal{D}$ be a set (e.g\@. $\mathcal{D}=\Sigma\times\Sigma_\bot$ for the  natural relational semantics of Sect\@. \ref{sec:natural-relational-semantics-deductive}).
There is a Galois connection \bgroup\abovedisplayskip-3.5pt\belowdisplayskip 0.5\belowdisplayskip\begin{eqntabular}{c}
\pair{\wp(\wp(\mathcal{D}))}{\subseteq}\galoiS{\alpha_C}{\gamma_C}\pair{\wp(\mathcal{D})}{\subseteq}\label{eq:hyper-tp-trace-abs}
\end{eqntabular}\egroup
where
$\alpha_C(P)\triangleq\bigcup P$ is surjective and $\gamma_C(S)\triangleq\wp(S)$ is injective (since 
$\alpha_C(P)\subseteq S$
$\Leftrightarrow$
$\bigcup P\subseteq S$
$\Leftrightarrow$
$P\subseteq \wp(S)$
$\Leftrightarrow$
$P\subseteq \gamma_C(S)$).

\begin{example} If $\mathcal{D}$ is a set of finite or infinite traces, $\sqb{\texttt{\small S}}_\bot$ defines the finite or infinite execution traces of \texttt{\small S},  $\crb{\texttt{\small S}}_{\!\bot}$ is the strongest hyper property of program \texttt{\small S} \cite{DBLP:journals/jcs/ClarksonS10}, and
$\alpha_C(\crb{\texttt{\small S}}_{\!\bot})=\sqb{\texttt{\small S}}_\bot$ is the strongest semantic  property of \texttt{\small S} (called a
trace property in \cite{DBLP:journals/jcs/ClarksonS10}).
\end{example}

Our first abstraction is therefore $\alpha_C(\mskip-1mu\crb{\texttt{\small S}}_{\!\bot}\mskip-1mu)$
=
$\alpha_C(\mskip-1mu\{\sqb{\texttt{\small S}}_{\bot}\}\mskip-1mu)$
=
$\sqb{\texttt{\small S}}_\bot$ where this natural relational semantics defines in Sect\@. \ref{sec:natural-relational-semantics-deductive} specifies the program properties of interest.

\subsection{Semantics to  Relational Postcondition Transformer $\normalfont\textsf{\upshape Post}$ Abstraction}\label{eq:post}
While the natural relational semantics establishes a relation between initial and final states or nontermination, the postcondition transformers establish a relation between properties of initial states and properties of final states or nontermination. The postcondition may be an assertion on final states only (as in Hoare partial correctness logic \cite{DBLP:journals/cacm/Hoare69}) or a relation between initial and final states (as in Manna partial correctness \cite{DBLP:journals/jcss/Manna71}). The postcondition may also include nontermination. Although Hoare logic is assertional, the initial values of variables can be recorded into auxiliary variables (see Sect\@. \ref{sec:auxiliary:variables}\ifshort\ in the appendix\fi). We start with the relational case since assertional property transformers are abstractions of relational ones (as shown in Sect\@. \ref{sec:Relational-assertional-abstraction}).

The relational postcondition transformer $\normalfont\textsf{\upshape Post}$ is also called the relational forward/right-image/\allowbreak post-image/strongest consequent/strongest post condition. 

\begin{eqntabular}{rcl}
\textsf{\upshape Post}&\in&\wp(\mathcal{X}\times\mathcal{Y})\rightarrow\wp(\mathcal{Z}\times\mathcal{X})\rightarrow\wp(\mathcal{Z}\times\mathcal{X})\nonumber\\
\textsf{\upshape Post}(r)P&\triangleq&\{\pair{\sigma_0}{\sigma'}\mid\exists\sigma\mathrel{.}\pair{\sigma_0}{\sigma}\in P\wedge\pair{\sigma}{\sigma'}\in r\}
\label{eq:def:Post}
\end{eqntabular}

$\textsf{\upshape Post}(\alpha_C(\crb{\texttt{\small S}}_{\!\bot}))P$
=
$\textsf{\upshape Post}\sqb{\texttt{\small S}}_\bot P$ is a relation between initial states $\sigma$ related to $\sigma_0$ satisfying the precondition $P$ and final states $\sigma'$ related to $\sigma_0$ upon termination of \texttt{\small S} or  $\sigma'=\bot$ in case of nontermination. This is the basis for the natural relational transformational logic (as in example \ref{ex:fact:spec} and Sect\@. \ref{Calculational-Design-of-the-Extended-Hoare-Logic}), except for the use of a transformer instead of logic triples. 
We will later prove in (\ref{eq:Post-Pret:Pre-Postt-GC}) that $\textsf{\upshape Post}$ is the lower adjoint of a Galois connection.
\begin{example}\label{ex-Floyd-assignment}Incrementation is characterized by \textsf{\upshape Post}(\texttt{\small x = x+1})(\texttt{\small x = x$_0$}) = (\texttt{\small x = x$_0+1$}) which, representing the semantics and relational properties as sets, is $\textsf{\upshape Post}(\{\pair{x}{x+1}\mid x\in\mathbb{Z}\})
\{\pair{x_0}{x}\mid x=x_0\in\mathbb{Z}\} = \{\pair{x_0}{x'}\mid\exists x\mathrel{.}\pair{x_0}{x}\in \{\pair{x_0}{x}\mid x=x_0\in\mathbb{Z}\}\wedge\pair{x}{x'}\in \{\pair{x}{x+1}\mid x\in\mathbb{Z}\}\} = \{\pair{x_0}{x_0+1}\mid x_0\in\mathbb{Z}\}$. Here, $\sigma_0$ is the initial value of the variables before the assignment  but, in general, this initial relation can be arbitrary.
More generally, Floyd's strongest postcondition for assignment \texttt{\small x := A} \cite{Floyd67-1} is $\textsf{\upshape Post}\sqb{\texttt{\small x := A}}P$ = $\{\pair{\sigma_0}{\sigma'}\mid\exists\sigma\mathrel{.}\pair{\sigma_0}{\sigma}\in P\wedge \sigma'=\sigma[\texttt{\small x}\leftarrow\mathcal{A}\sqb{\texttt{\small A}}\sigma]\}$.
\end{example}

\subsection{Relational Postcondition Transformer to Antecedent/Consequent Pairs}
Transformational logic triples $\{P\}\texttt{\small S}\{Q\}$ associate pairs $\pair{P}{Q}$ of predicates to each program command $\texttt{\small S}\in\mathbb{S}$. So the theory of the logic is the set $\{\pair{P}{Q}\mid \{P\}\texttt{\small S}\{Q\}\}$ for each statement \texttt{\small S}. For the natural transformational logic, this theory contains the graph of the $\textsf{\upshape Post}(\sqb{\texttt{\small S}}_\bot)$ function. Conversely, from this graph, we can recover the strongest valid triples $\{P\}\texttt{\small S}\{Q\}$. (Notice that we say ``contains'' not ``is'' and ``strongest'' since, in absence of a consequence rule, the graph does not contain all valid triples, only the $\{P\}\texttt{\small S}\{\textsf{\upshape Post}(\sqb{\texttt{\small S}}_\bot)P\}$ ones. Consequence rules will be introduced thanks to another abstraction discussed in next Sect\@. \ref{sec:Weakening-strengthening-abstractions}.)

More generally, a function $f\in\mathcal{X}\rightarrow\mathcal{Y}$ is isomorphic to its graph $\alpha_{\textup{G}}(f)=\{\pair{x}{f(x)}\mid x\in \mathcal{X}\}$. This graph $\alpha_{\textup{G}}(f)$ is a functional relation.
We have the Galois isomorphism \proofinapx
\bgroup\abovedisplayskip0.75\abovedisplayskip\belowdisplayskip0.75\belowdisplayskip
\begin{eqntabular}{c}
\pair{\mathcal{X}\rightarrow\mathcal{Y}}{=}
\GaloiS{\alpha_{\textup{G}}}{\gamma_{\textup{G}}}
\pair{\wp_{\textsf{fun}}(\mathcal{X}\times\mathcal{Y})}{=}
\label{eq:function-graph-GC}
\end{eqntabular}\egroup
where 
$\gamma_{\textup{G}}(r)\triangleq\LAMBDA{x}(y\mbox{\itshape\ such\ that\ }\pair{x}{y}\in r)$ is  uniquely well-defined since $r$ is a functional relation. 
\begin{toappendix}
\begin{proof}[Proof of (\ref{eq:function-graph-GC})]
\begin{calculus}[$\Leftrightarrow$\ ]
\formulaexplanation{\alpha_{\textup{G}}(f)=r}{assuming $r\in{\wp_{\textsf{fun}}({X}\times{Y})}$}\\
$\Leftrightarrow$
\formulaexplanation{\{\pair{x}{f(x)}\mid x\in{X}\}= r}{def\@. $\alpha_{\textup{G}}(f)$}\\
$\Leftrightarrow$
\formulaexplanation{\forall x\in{X}\mathrel{.}f(x)=(y\mbox{\itshape\ such\ that\ }\pair{x}{y}\in r)}{$r$ is a functional relation}\\
$\Leftrightarrow$
\formulaexplanation{f=\LAMBDA{x}(y\mbox{\itshape\ such\ that\ }\pair{x}{y}\in r)}{def\@. function equality}\\
$\Leftrightarrow$
\lastformulaexplanation{f=\gamma_{\textup{G}}(r)}{def\@. $\gamma_{\textup{G}}(r)$}{\mbox{\qed}}
\end{calculus}
\let\qed\relax
\end{proof}
\end{toappendix}
We have \proofinapx
\bgroup\arraycolsep0.5\arraycolsep
\begin{eqntabular}{rcl}
\alpha_{\textup{G}}(\textsf{\upshape Post}(\alpha_C(\crb{\texttt{\small S}}_{\!\bot})))
&=&
\{\pair{P}{\{\pair{\sigma_0}{\sigma'}\mid\exists\sigma\mathrel{.}\pair{\sigma_0}{\sigma}\in P\wedge\pair{\sigma}{\sigma'}\in \sqb{\texttt{\small S}}_\bot\}}\mid P\in \wp(\Sigma\times \Sigma)\}
\label{eq:alphaG-Post-natural-semantics}
\end{eqntabular}\egroup
\begin{toappendix}
\begin{proof}[Proof of (\ref{eq:alphaG-Post-natural-semantics})]
\begin{calculus}
\formula{\alpha_{\textup{G}}(\textsf{\upshape Post}(\alpha_C(\crb{\texttt{\small S}}_{\!\bot})))}\\
=
\formulaexplanation{\alpha_{\textup{G}}(\textsf{\upshape Post}(\sqb{\texttt{\small S}}_\bot))}{from Sect\@. \ref{sec:Collecting-semantics-semantics-abstraction}}\\
=
\formulaexplanation{\{\pair{P}{\textsf{\upshape Post}(\sqb{\texttt{\small S}}_\bot))P}\mid P\in \wp(\Sigma)\}}{def\@. $\alpha_{\textup{G}}$}\\
=
\lastformulaexplanation{\{\pair{P}{\{\pair{\sigma_0}{\sigma'}\mid\exists\sigma\mathrel{.}\pair{\sigma_0}{\sigma}\in P\wedge\pair{\sigma}{\sigma'}\in\sqb{\texttt{\small S}}_\bot\}}\mid P\in \wp(\Sigma)\}}{def\@. (\ref{eq:def:Post}) of $\textsf{\upshape Post}$}{\mbox{\qed}}
\end{calculus}\let\qed\relax
\end{proof}
\end{toappendix}
So $\alpha_{\textup{G}}(\textsf{\upshape Post}(\alpha_C(\crb{\texttt{\small S}}_{\!\bot})))$ is the set of pairs $\pair{P}{Q}$ such that
$Q$ is the strongest relational postcondition of $P$ for the natural relational semantics $\sqb{\texttt{\small S}}_\bot$. It is not a program logic since, as was the case for transformers, it is missing a consequence rule.
\begin{example}Floyd/Hoare logic rules  \cite{DBLP:journals/jacm/Hoare78} provide the strongest assertional post-condi\-tion except for the iteration and consequence rule, e.g\@.,
$\{P\}\texttt{\small skip}\{P\}$ is $\{P\}\texttt{\small skip}\{\textsf{\textup{post}}(\sqb{\texttt{\small skip}})P\}$ (see (\ref{eq:def:post}) below for the classic definition of \textsf{\textup{post}}). But excluding the consequence rule and using the following iteration rule (for bounded nondeterminism)
\bgroup\abovedisplayskip0.5\abovedisplayskip\belowdisplayskip0.5\belowdisplayskip\begin{eqntabular}{c}
\frac{I^0=P,\quad\forall n\in\mathbb{N}\mathrel{.}\{I^n\wedge\texttt{\small B}\}\,\texttt{\small S}\,\{I^{n+1}\}}{\{P\}\,\texttt{\small while (B) S}\,\{\exists n\in\mathbb{N}\mathrel{.}I^n\wedge \neg B\}}
\label{eq:hoare-while-strongest}
\end{eqntabular}\egroup
would yield the strongest post condition in all cases.
\end{example}

\subsection{Weakening and Strengthening Abstractions}\label{sec:Weakening-strengthening-abstractions}

Following \cite{Burstall-MI5-69} to make program proofs using the natural relational semantics proof rules \ref{eq:while:invariant}--\ref{eq:W:infty}, or, by (\ref{eq:alphaG-Post-natural-semantics}), the transformer $\textsf{\upshape Post}(\sqb{\texttt{\small S}}_\bot)$ or, isomorphically by (\ref{eq:function-graph-GC}), its graph $\{\pair{P}{\{\pair{\sigma_0}{\sigma'}\mid\exists\sigma\mathrel{.}\pair{\sigma_0}{\sigma}\in P\wedge\pair{\sigma}{\sigma'}\in \sqb{\texttt{\small S}}_\bot\}}\mid P\in \wp(\Sigma\times \Sigma)\}$ $\in$
$\wp(\wp(\Sigma\times\Sigma)\times\wp(\Sigma\times\Sigma_\bot))$ is inadequate  since the semantics  describes executions exactly, without any possibility of approximation.

In contrast, as first shown by Turing \cite{Turing49-program-proof}, using executions properties is the basis for elegant and concise program correctness proofs since it allows for approximations.  

This is even implicitly acknowledged by the most enthusiastic supporter of transformers. Edsger W\@. D\@. Dijkstra in \cite{DBLP:books/ph/Dijkstra76} has chapters 0 to 4 defining predicate transformers until chapter 5 introducing properties weakening by implication (i.e. one form of approximation) as well as the ``Fundamental Invariance Theorem for Loops'' (i.e\@. fixpoint induction Th\@. \ref{th:Fixpoint-Overapproximation} replacing the strongest loop invariant by weaker ones). 
Moreover, in chapter 6, it is explained how ``to choose an appropriate proof for termination'' (for bounded nondeterminism). Iterative program design and proofs are only considered after over approximation (invariance)  and under approximation (for termination) have been introduced, from chapter 7 on.  We have to do the same, but for any transformer (including $\textsf{\upshape Post}\sqb{\texttt{\small S}}_\bot$).  

For that purpose, we introduce weakening and  strengthening abstractions. Consequence rules, understood as an abstraction losing precision on program properties, will be a specific instance for a specific transformer. We also need compatible general induction principles to handle loops (of which invariance and (non)termination will be specific instances). Such induction principles are not relative to expressivity but to proofs, and so will be considered in part 2 of the paper.

\subsubsection{The Over Approximation Abstraction}
Pairs of properties $\pair{P}{Q}\in R\in\wp(\wp(\mathcal{X})\times\wp(\mathcal{Y}))$ can be approximated by weakening or strengthening $P$ and/or
$Q$. For Hoare logic \cite{DBLP:journals/cacm/Hoare69}, we can strengthen $P$ by $P'\subseteq P$ and weaken $Q$ by $Q'$ such that $Q\subseteq Q'$. This is the over approximation abstraction $\textsf{\textup{post}}({\supseteq},{\subseteq})R$ = $\{\pair{P'}{Q'}\mid
\exists \pair{P}{Q}\in R\mathrel{.}\pair{P}{Q}\mathrel{{\supseteq},{\subseteq}}\pair{P'}{Q'}\}$ = $\{\pair{P'}{Q'}\mid
\exists \pair{P}{Q}\in R\mathrel{.}{P}\supseteq {P'}\wedge {Q}\subseteq{Q'}\}$  = $\{\pair{P'}{Q'}\mid
\exists \pair{P}{Q}\in R\mathrel{.}{P'}\subseteq {P}\wedge {Q}\subseteq{Q'}\}$ by defining the classic assertional right image transformer (denoted $Xr$ in \cite{DBLP:conf/focs/Pratt76})
\bgroup\abovedisplayskip0.55\abovedisplayskip\belowdisplayskip0.5\belowdisplayskip
\begin{eqntabular}{rcl}
\textsf{\upshape post}(r)X&\triangleq&\{{y}\mid\exists x\in X\mathrel{.}\pair{x}{y}\in r\}
\label{eq:def:post}
\end{eqntabular}\egroup
and the component wise ordering $\mathrel{{\sqsubseteq},{\preceq}}$ on pairs
\bgroup\abovedisplayskip0.75\abovedisplayskip\belowdisplayskip0.75\belowdisplayskip
\begin{eqntabular}{rcl}
\pair{x}{y}\mathrel{{\sqsubseteq},{\preceq}}\pair{x'}{y'}&\triangleq& x\sqsubseteq x'\wedge y\preceq y'
\end{eqntabular}\egroup
If $r\in\wp(\mathcal{X}\times\mathcal{Y})$, we have the classic Galois connection
\bgroup\abovedisplayskip0.75\abovedisplayskip\belowdisplayskip0.75\belowdisplayskip
\begin{eqntabular}{c}
\pair{\wp(\mathcal{X})}{\subseteq}\galois{\textsf{\upshape post}(r)}{\widetilde{\textsf{\upshape pre}}(r)}\pair{\wp(\mathcal{Y})}{\subseteq}\label{eq:def:post:GC}
\end{eqntabular}\egroup
where ${\widetilde{\textsf{\upshape pre}}(r)}Q=\{x\mid\forall y\mathrel{.}\pair{x}{y}\in r\Rightarrow y\in Q\}$\ifshort\ (see example \ref{ex-Hoare-assignment} in the appendix)\fi.
\begin{toappendix}
\begin{example}\label{ex-Hoare-assignment}Continuing example \ref{ex-Floyd-assignment}, instead of Floyd's verification condition $\{P\}\texttt{\small S}\{\textsf{\textup{post}}\sqb{\texttt{\small S}}P\}$ \cite{Floyd67-1},
Hoare logic \cite{DBLP:journals/cacm/Hoare69} uses $\{\widetilde{\textsf{pre}}\sqb{\texttt{\small S}}Q\}\texttt{\small S}\{Q\}$ that is, $\{\{\sigma\mid\sigma[\texttt{\small x}\leftarrow\mathcal{A}\sqb{\texttt{\small A}}\sigma\in Q\}\}\texttt{\small S}\{Q\}$,
in syntactic form $\{Q[\texttt{\small A}/\texttt{\small x}]\}\texttt{\small S}\{Q\}$ where $Q[\texttt{\small A}/\texttt{\small x}]$ is $Q$ where \texttt{\small A} is substituted for \texttt{\small x}. By (\ref{eq:def:post:GC}), the two verification conditions for assignment are equivalent.
\end{example}
\end{toappendix}
The theory of the \emph{natural transformational over approximation logic} is
therefore  \proofinapx
\bgroup\abovedisplayskip0.9\abovedisplayskip\belowdisplayskip0.75\belowdisplayskip
\begin{eqntabular}[fl]{rcl@{\qquad}}
\textsf{\textup{post}}({\supseteq},{\subseteq})(\alpha_{\textup{G}}(\textsf{\upshape Post}\sqb{\texttt{\small S}}_\bot))
&=&\{\pair{P}{Q}\mid\textsf{\upshape Post}\sqb{\texttt{\small S}}_\bot P\subseteq{Q}\}\label{eq:natural-upper-transformational-logic-theory}\\
&=&
\{\pair{P}{Q}\mid\forall \pair{\sigma_0}{\sigma}\in P\mathrel{.}\forall \sigma'\mathrel{.}
\pair{\sigma}{\sigma'}\in\sqb{\texttt{\small S}}_\bot\Rightarrow\pair{\sigma_0}{\sigma'}\in{Q}\}
\nonumber
\end{eqntabular}\egroup
that is, for any initial state $\sigma$ related to $\sigma_0$ by the precondition $P$ and any final state $\sigma'$ of \texttt{\small S}, possibly $\bot$,
the pair $\pair{\sigma_0}{\sigma'}$ satisfies the postcondition $Q$, as considered in example \ref{ex:fact:spec}. The difference with the interpretation of Manna and Pnueli total correctness logic \cite{DBLP:journals/acta/MannaP74} is that we may have $\pair{\sigma_0}{\bot}\in Q$ thus allowing possible nontermination for some initial pair of
states $\pair{\sigma_0}{\sigma}$ of $P$. Therefore we can both express both total and partial correctness plus nontermination when $Q=\Sigma\times\{\bot\}$. With this convention, only one of Dijkstra's weakest preconditions  transformers \cite{DBLP:journals/cacm/Dijkstra75,DBLP:books/ph/Dijkstra76,DBLP:books/daglib/0067387} is needed since $\textsf{\upshape wlp}(\texttt{\small S},Q)=\textsf{\upshape wp}(\texttt{\small S},Q\cup\{\bot\})$.
\begin{toappendix}
\begin{proof}[Proof of (\ref{eq:natural-upper-transformational-logic-theory})]
\begin{calculus}
\formulaexplanation{\textsf{\textup{post}}({\supseteq},{\subseteq})(\alpha_{\textup{G}}(\textsf{\upshape Post}(\sqb{\texttt{\small S}}_\bot)))}{def\@. (\ref{eq:natural-upper-transformational-logic-theory}) of consequence}\\ 
= 
\formulaexplanation{\textsf{\textup{post}}({\supseteq},{\subseteq})(\{\pair{P}{\textsf{\upshape Post}(\sqb{\texttt{\small S}}_\bot)P}\mid P\in\wp(\Sigma\times\Sigma)\})}{def\@. $\alpha_{\textup{G}}$}\\ 
= 
\formulaexplanation{\textsf{\textup{post}}({\supseteq},{\subseteq})(\{\pair{P}{\{\pair{\sigma_0}{\sigma'}\mid\exists\sigma\mathrel{.}\pair{\sigma_0}{\sigma}\in P\wedge\pair{\sigma}{\sigma'}\in\sqb{\texttt{\small S}}_\bot\}}\mid P\in\wp(\Sigma\times\Sigma)\})}{def\@. (\ref{eq:def:Post}) of $\textsf{\upshape Post}$}\\
= 
\formula{\{{y}\mid\exists x\in \{\pair{P}{\{\pair{\sigma_0}{\sigma'}\mid\exists\sigma\mathrel{.}\pair{\sigma_0}{\sigma}\in P\wedge\pair{\sigma}{\sigma'}\in\sqb{\texttt{\small S}}_\bot\}}\mid P\in\wp(\Sigma\times\Sigma)\}\mathrel{.}\pair{x}{y}\in ({\supseteq},{\subseteq})\}}\\[-0.5ex]
\rightexplanation{def\@. (\ref{eq:def:post}) of $\textsf{\textup{post}}$}\\
= 
\formula{\{\pair{P'}{Q'}\mid\exists \pair{P}{Q}\in \{\pair{P}{\{\pair{\sigma_0}{\sigma'}\mid\exists\sigma\mathrel{.}\pair{\sigma_0}{\sigma}\in P\wedge\pair{\sigma}{\sigma'}\in\sqb{\texttt{\small S}}_\bot\}}\mid P\in\wp(\Sigma\times\Sigma)\}\mathrel{.}\pair{\pair{P}{Q}}{\pair{P'}{Q'}}\in ({\supseteq},{\subseteq})\}}\\
\rightexplanation{$y=\pair{P'}{Q'}\in\wp(\Sigma\times\Sigma)\times\wp(\Sigma\times\Sigma_\bot)$ and $x=\pair{P}{Q}\in\wp(\Sigma\times\Sigma)\times\wp(\Sigma\times\Sigma_\bot)$ }\\[0.75ex]
= 
\formulaexplanation{\{\pair{P'}{Q'}\mid
\exists \pair{P}{Q}\in \{\pair{P}{\{\pair{\sigma_0}{\sigma'}\mid\exists\sigma\mathrel{.}\pair{\sigma_0}{\sigma}\in P\wedge\pair{\sigma}{\sigma'}\in\sqb{\texttt{\small S}}_\bot\}}\mid P\in\wp(\Sigma\times\Sigma)\}\mathrel{.}{P}\supseteq {P'}\wedge {Q}\subseteq{Q'}\}}{def\@. ${\supseteq},{\subseteq}$}\\[0.5ex]
 = 
\formulaexplanation{\{\pair{P'}{Q'}\mid
\exists P\in\wp(\Sigma\times\Sigma) \mathrel{.}{P'}\subseteq {P}\wedge
\{\pair{\sigma_0}{\sigma'}\mid\exists\sigma\mathrel{.}\pair{\sigma_0}{\sigma}\in P\wedge\pair{\sigma}{\sigma'}\in\sqb{\texttt{\small S}}_\bot\}\subseteq{Q'}\}}{def\@. $\in$}\\
= 
\formula{\{\pair{P'}{Q'}\mid
\exists P\in\wp(\Sigma\times\Sigma) \mathrel{.}{P'}\subseteq {P}\wedge\forall \sigma_0,\sigma,\sigma'\mathrel{.}
(\pair{\sigma_0}{\sigma}\in P\wedge\pair{\sigma}{\sigma'}\in\sqb{\texttt{\small S}}_\bot)\Rightarrow(\pair{\sigma_0}{\sigma'}\in{Q'})\}}\\[-0.5ex]
\rightexplanation{def\@. $\subseteq$}\\[-1.5ex]
= 
\formula{\{\pair{P'}{Q'}\mid\forall \sigma_0,\sigma,\sigma'\mathrel{.}
(\pair{\sigma_0}{\sigma}\in P'\wedge\pair{\sigma}{\sigma'}\in\sqb{\texttt{\small S}}_\bot)\Rightarrow(\pair{\sigma_0}{\sigma'}\in{Q'})}\\
\explanation{($\subseteq$)\quad If $\pair{\sigma_0}{\sigma}\in P'$, then $\pair{\sigma_0}{\sigma}\in P$ since ${P'}\subseteq {P}$ which implies $\pair{\sigma_0}{\sigma'}\in{Q'}$ by hypothesis. Otherwise, $\pair{\sigma_0}{\sigma}\not\in P'$, and $\textsf{false}\Rightarrow(\pair{\sigma_0}{\sigma'}\in{Q'})$.\\
($\supseteq$)\quad Take $P=P'$ and reflexivity}\\[0.5ex]
= 
\formulaexplanation{\{\pair{P}{Q}\mid\forall\pair{\sigma_0}{\sigma'}\mathrel{.}(\exists\sigma\mathrel{.}\pair{\sigma_0}{\sigma}\in P\wedge\pair{\sigma}{\sigma'}\in \sqb{\texttt{\small S}}_\bot)\Rightarrow(\pair{\sigma_0}{\sigma'}\in Q)\}}{def\@. $\Rightarrow$, renaming}\\
=
\formulaexplanation{\{\pair{P}{Q}\mid\{\pair{\sigma_0}{\sigma'}\mid\exists\sigma\mathrel{.}\pair{\sigma_0}{\sigma}\in P\wedge\pair{\sigma}{\sigma'}\in \sqb{\texttt{\small S}}_\bot\}\subseteq{Q}\}}{def\@. $\subseteq$}\\
=
\lastformulaexplanation{\{\pair{P}{Q}\mid\textsf{\upshape Post}(\sqb{\texttt{\small S}}_\bot)P\subseteq{Q}\}}{def\@. (\ref{eq:def:Post}) of $\textsf{\upshape Post}$}{\mbox{\qed}}
\end{calculus}\let\qed\relax
\end{proof}
\end{toappendix}
This is similar to the classic characterization of Hoare logic by a forward transformer, $\{P\}\texttt{\small S}\{Q\}$ if and only if $\textsf{\textup{post}}\sqb{\texttt{\small S}}P\Rightarrow Q$ given by \cite[equation (S), p\@. 110]{DBLP:conf/focs/Pratt76}  or, equivalently by (\ref{eq:def:post:GC}), $P\Rightarrow \widetilde{\textsf{\upshape pre}}\sqb{\texttt{\small S}}Q$ \cite[equation (w), p\@. 110]{DBLP:conf/focs/Pratt76} (except that in (\ref{eq:natural-upper-transformational-logic-theory}), $P$  and $Q$ are relational and take nontermination into account). By (\ref{eq:def:post:GC}), the abstraction $\textsf{\textup{post}}({\supseteq},{\subseteq})$ is the lower adjoint of a Galois connection.

\subsubsection{The Under Approximation Abstraction}
For the \emph{natural transformational under approximation logic}, as well as reverse Hoare logic \cite{DBLP:conf/sefm/VriesK11} aka incorrectness logic \cite{DBLP:journals/pacmpl/OHearn20}, we can weaken $P$ by $P'\supseteq P$ and strengthen $Q$ by $Q'$ such that $Q\supseteq Q'$. This is the under approximation abstraction $\textsf{\textup{post}}({\subseteq},{\supseteq})R$ = $\{\pair{P'}{Q'}\mid
\exists \pair{P}{Q}\in R\mathrel{.}\pair{P}{Q}\mathrel{{\subseteq},{\supseteq}}\pair{P'}{Q'}\}$ = $\{\pair{P'}{Q'}\mid
\exists \pair{P}{Q}\in R\mathrel{.}{P}\subseteq {P'}\wedge {Q}\supseteq{Q'}\}$  = $\{\pair{P'}{Q'}\mid
\exists \pair{P}{Q}\in R\mathrel{.}{P}\subseteq {P'}\wedge {Q'}\subseteq{Q}\}$ which is the consequence rule called \emph{Symmetry} in 
\cite[Fig\@. 1]{DBLP:journals/pacmpl/OHearn20} and \emph{Consequence}  in 
\cite[Fig\@. 2]{DBLP:journals/pacmpl/OHearn20}.

The theory of the natural transformational under approximation logic is
therefore  \proofinapx
\bgroup\arraycolsep0.8\arraycolsep
\begin{eqntabular}{rcl@{\qquad}}
\textsf{\textup{post}}({\subseteq},{\supseteq})(\alpha_{\textup{G}}(\textsf{\upshape Post}\sqb{\texttt{\small S}}_\bot))
&=&\{\pair{P}{Q}\mid Q\subseteq\textsf{\upshape Post}\sqb{\texttt{\small S}}_\bot P\}\label{eq:natural-lower-transformational-logic-theory}\\
&=&
\{\pair{P}{Q}\mid\forall \pair{\sigma_0}{\sigma}\in P\mathrel{.}\forall \sigma'\mathrel{.}\pair{\sigma_0}{\sigma'}\in{Q}\Rightarrow
\pair{\sigma}{\sigma'}\in\sqb{\texttt{\small S}}_\bot\} 
\nonumber
\end{eqntabular}
\egroup
that is, for any initial state $\sigma$ related to $\sigma_0$ satisfying the precondition $P$ and any final state $\sigma'$ related to $\sigma_0$, possibly $\bot$,
if the pair $\pair{\sigma_0}{\sigma'}$ satisfies the postcondition $Q$ then there exists an execution of \texttt{\small S} from
$\sigma$ to $\sigma'$ (possibly non termination). The difference with reverse Hoare logic \cite{DBLP:conf/sefm/VriesK11} aka incorrectness logic \cite{DBLP:journals/pacmpl/OHearn20} is that we may have $\pair{\sigma}{\bot}\in Q$ thus allowing possible nontermination for some initial
states $\pair{\sigma_0}{\sigma}$ of $P$ so we can both express total and partial correctness plus nontermination when $Q=\Sigma\times\{\bot\}$.
\begin{toappendix}
\begin{proof}[Proof of (\ref{eq:natural-lower-transformational-logic-theory})]
\begin{calculus}
\formulaexplanation{\textsf{\textup{post}}({\subseteq},{\supseteq})(\alpha_{\textup{G}}(\textsf{\upshape Post}(\sqb{\texttt{\small S}}_\bot)))}{def\@. (\ref{eq:natural-upper-transformational-logic-theory}) of consequence}\\ 
= 
\formulaexplanation{\textsf{\textup{post}}({\subseteq},{\supseteq})(\{\pair{P}{\textsf{\upshape Post}(\sqb{\texttt{\small S}}_\bot)P}\mid P\in\wp(\Sigma\times\Sigma)\})}{def\@. $\alpha_{\textup{G}}$}\\ 
= 
\formulaexplanation{\textsf{\textup{post}}({\subseteq},{\supseteq})(\{\pair{P}{\{\pair{\sigma_0}{\sigma'}\mid\exists\sigma\mathrel{.}\pair{\sigma_0}{\sigma}\in P\wedge\pair{\sigma}{\sigma'}\in \sqb{\texttt{\small S}}_\bot\}}\mid P\in\wp(\Sigma\times\Sigma)\})}{def\@. (\ref{eq:def:Post}) of $\textsf{\upshape Post}$}\\
= 
\formula{\{{y}\mid\exists x\in \{\pair{P}{\{\pair{\sigma_0}{\sigma'}\mid\exists\sigma\mathrel{.}\pair{\sigma_0}{\sigma}\in P\wedge\pair{\sigma}{\sigma'}\in\sqb{\texttt{\small S}}_\bot\}}\mid P\in\wp(\Sigma\times\Sigma)\}\mathrel{.}\pair{x}{y}\in ({\subseteq},{\supseteq})\}}\\[-0.5ex]\rightexplanation{def\@. (\ref{eq:def:post}) of $\textsf{\textup{post}}$}\\[-0.5ex]
= 
\formula{\{\pair{P'}{Q'}\mid\exists \pair{P}{Q}\in \{\pair{P}{\{\pair{\sigma_0}{\sigma'}\mid\exists\sigma\mathrel{.}\pair{\sigma_0}{\sigma}\in P\wedge\pair{\sigma}{\sigma'}\in\sqb{\texttt{\small S}}_\bot\}}\mid P\in\wp(\Sigma\times\Sigma)\}\mathrel{.}\pair{\pair{P}{Q}}{\pair{P'}{Q'}}\in ({\subseteq},{\supseteq})\}}\\\rightexplanation{$y=\pair{P'}{Q'}\in\wp(\Sigma\times\Sigma)\times\wp(\Sigma\times\Sigma_\bot)$ and $x=\pair{P}{Q}\in\wp(\Sigma\times\Sigma)\times\wp(\Sigma\times\Sigma_\bot)$}\\
= 
\formulaexplanation{\{\pair{P'}{Q'}\mid
\exists \pair{P}{Q}\in \{\pair{P}{\{\pair{\sigma_0}{\sigma'}\mid\exists\sigma\mathrel{.}\pair{\sigma_0}{\sigma}\in P\wedge\pair{\sigma}{\sigma'}\in\sqb{\texttt{\small S}}_\bot\}}\mid P\in\wp(\Sigma\times\Sigma)\}\mathrel{.}{P}\subseteq {P'}\wedge {Q}\supseteq{Q'}\}}{def\@. ${\supseteq},{\subseteq}$}\\
= 
\formulaexplanation{\{\pair{P'}{Q'}\mid
\exists P\in\wp(\Sigma\times\Sigma) \mathrel{.}{P}\subseteq {P'}\wedge
\{\pair{\sigma_0}{\sigma'}\mid\exists\sigma\mathrel{.}\pair{\sigma_0}{\sigma}\in P\wedge\pair{\sigma}{\sigma'}\in\sqb{\texttt{\small S}}_\bot\}\supseteq{Q'}\}}{def\@. $\in$}\\
 = 
\formulaexplanation{\{\pair{P'}{Q'}\mid
\exists P\in\wp(\Sigma\times\Sigma) \mathrel{.}{P}\subseteq {P'}\wedge\forall\pair{\sigma_0}{\sigma'}\mathrel{.}
(\pair{\sigma_0}{\sigma'}\in{Q'})\Rightarrow(\exists\sigma\mathrel{.}\pair{\sigma_0}{\sigma}\in P\wedge\pair{\sigma}{\sigma'}\in\sqb{\texttt{\small S}}_\bot)\}}{def\@. $\subseteq$}\\
=
\formula{\{\pair{P'}{Q'}\mid \forall\pair{\sigma_0}{\sigma'}\mathrel{.}(\pair{\sigma_0}{\sigma'}\in Q')\Rightarrow
(\exists\sigma\mathrel{.}\pair{\sigma_0}{\sigma}\in P'\wedge\pair{\sigma}{\sigma'}\in \sqb{\texttt{\small S}}_\bot)\}}\\
\explanation{($\subseteq$) $\pair{\sigma_0}{\sigma}\in P$ and ${P}\subseteq {P'}$ implies $\pair{\sigma_0}{\sigma}\in P'$ \\
($\supseteq$) Take $P'=P$ and reflexivity}\\
=
\formulaexplanation{\{\pair{P}{Q}\mid Q\subseteq\{\pair{\sigma_0}{\sigma'}\mid\exists\sigma\mathrel{.}\pair{\sigma_0}{\sigma}\in P\wedge\pair{\sigma}{\sigma'}\in \sqb{\texttt{\small S}}_\bot\}\}}{def\@. $\subseteq$ and renaming}\\
=
\lastformulaexplanation{\{\pair{P}{Q}\mid Q\subseteq\textsf{\upshape Post}(\sqb{\texttt{\small S}}_\bot)P\}}{def\@. (\ref{eq:def:Post}) of $\textsf{\upshape Post}$}{\mbox{\qed}}
\end{calculus}\let\qed\relax
\end{proof}
\end{toappendix}

Up to the use of relations instead of assertions and the consideration of nontermination $\bot$, this is similar to the classic characterization of reverse Hoare logic aka incorrectness logic by a forward transformer, $\{P\}\texttt{\small S}\{Q\}$ if and only if $Q\Rightarrow\textsf{\textup{post}}(\sqb{\texttt{\small S}})P$ given by
\cite[section 5]{DBLP:conf/sefm/VriesK11} and \cite[Lemma 3.(2)]{DBLP:journals/pacmpl/OHearn20}, showing that both logics have the same semantics/theory (again up to nontermination and relational postconditions). By (\ref{eq:def:post:GC}), the abstraction $\textsf{\textup{post}}({\subseteq},{\supseteq})$ is the lower adjoint of a Galois connection.

\subsubsection{The Incorrectness Logic is Insufficient to Prove That All Alarms in Static Analysis Are True or False Alarms}
Incorrectness logic \cite{DBLP:journals/pacmpl/OHearn20} 
``was motivated in large part by the aim of providing a logical foundation for bug-catching program analyses'' \cite{DBLP:journals/pacmpl/LeRVBDO22}. In particular incorrectness logic is useful to prove that alarms in static analyzers are true alarms. This consists in showing that the alarm is definitely reachable from some input. However, not all alarms are reachable from initial states since static analyses are over approximating reachable states so that unreachable code under the precondition may produce false alarms. 

\begin{example}\label{ex:fact:pre:spec}Consider the factorial of example \ref{ex:fact:spec} specified by $\{ f=1\}\,\texttt{\small fact}\,\{f>0\}$. This contract is obviously satisfied since on exit $f=!\underline{n}>0$. However, an interval analysis of this program with initially $\underline{\texttt{\small n}}\in\mathbb{Z}$ is totally imprecise and will produce an alarm on program exit with postcondition $Q=\texttt{\small f}\leqslant0$. This is a false alarm since the loop exit is unreachable. This unreachability is not provable by incorrectness logic. This is provable by Hoare logic as $\{ \underline{n}<0 \wedge f=1\}\,\texttt{\small fact}\,\{\textsf{false}\}$ but then we don't want to use two different logics to prove incorrectness, the main motivation for recent work on combining logics (e.g\@. \cite{DBLP:journals/jacm/BruniGGR23,DBLP:conf/sas/MilaneseR22,DBLP:journals/pacmpl/ZilbersteinDS23,DBLP:conf/ecoop/MaksimovicCLSG23}, etc). This is also provable by the natural transformational under approximation logic which extends incorrectness logic to nontermination, that is, in the assertional form of Sect\@. \ref{sec:Relational-assertional-abstraction}, $\{\bot\}\subseteq\textsf{\upshape Post}\sqb{\texttt{\small fact}}_\bot\{ \underline{n}<0 \wedge f=1\}$, see example \ref{ex:fact:pre}.
\end{example}

\subsection{To Terminate or Not to Terminate Abstraction for Properties}\label{sec:terminate-not-terminate-properties}
Total correctness excludes nontermination while partial correctness allows it. This corresponds to different abstractions of the natural relational semantics.
\subsubsection{The Termination Exclusion Abstraction}
We can  exclude the possibility of
nontermination by the abstraction
\bgroup\abovedisplayskip0.0\abovedisplayskip\belowdisplayskip0.75\belowdisplayskip
\begin{eqntabular}{rcl}
\alpha^2_{\not\bot}(R)&\triangleq&\{\pair{P}{Q}\mid \pair{P}{Q}\in R\wedge Q\cap(\Sigma\times\{\bot\})=\emptyset\}
\label{eq:def:alpha-not-bot}
\end{eqntabular}\egroup
excluding $\bot$ from the postcondition.
This is an abstraction by the Galois connection 
\bgroup\belowdisplayskip0.75\belowdisplayskip\begin{eqntabular}{c}
\pair{\wp(\wp(\Sigma\times\Sigma)\times\wp(\Sigma\times\Sigma_\bot))}{\subseteq}\galoiS{\alpha^2_{\not\bot}}{\gamma^2_{\not\bot}}\pair{\wp(\wp(\Sigma\times\Sigma)\times\wp(\Sigma\times\Sigma))}{\subseteq}
\label{eq:alpha:2:bot:GC-1}
\end{eqntabular}\egroup
with $\gamma^2_{\not\bot}(R')\triangleq R'\cup \{\pair{P}{Q}\mid Q\cap(\Sigma\times\{\bot\})\neq\emptyset\}$ \proofinapx.
\begin{toappendix}
\begin{proof}[Proof of (\ref{eq:alpha:2:bot:GC-1})]
\begin{calculus}[$\Leftrightarrow$\ ]
\formulaexplanation{\alpha^2_{\not\bot}(R)\subseteq R'}{$R\in\wp(\wp(\Sigma\times\Sigma)\times\wp(\Sigma\times\Sigma_\bot))$ and $R'\in\wp(\wp(\Sigma\times\Sigma)\times\wp(\Sigma\times\Sigma))$ by hypothesis}\\
$\Leftrightarrow$
\formulaexplanation{\{\pair{P}{Q}\mid \pair{P}{Q}\in R\wedge Q\cap(\Sigma\times\{\bot\})=\emptyset\}\subseteq R'}{def\@. (\ref{eq:def:alpha-not-bot}) of $\alpha^2_{\not\bot}$}\\
$\Leftrightarrow$
\formulaexplanation{\forall P,Q\mathrel{.}(\pair{P}{Q}\in R\wedge Q\cap(\Sigma\times\{\bot\})=\emptyset)\Rightarrow(\pair{P}{Q}\in R')}{def\@. subset $\subseteq$}\\
$\Leftrightarrow$
\formulaexplanation{\forall P,Q\mathrel{.}(\pair{P}{Q}\in R)\Rightarrow(( Q\cap(\Sigma\times\{\bot\})=\emptyset)\Rightarrow(\pair{P}{Q}\in R'))}{def\@. implication $\Rightarrow$}\\[-0.25ex]
$\Leftrightarrow$
\formulaexplanation{R\subseteq\{\pair{P}{Q}\mid( Q\cap(\Sigma\times\{\bot\})=\emptyset)\Rightarrow(\pair{P}{Q}\in R')\}}{def\@. subset $\subseteq$}\\
$\Leftrightarrow$
\formula{R\subseteq R'\cup \{\pair{P}{Q}\mid Q\cap(\Sigma\times\{\bot\})\neq\emptyset\}}\\
\explanation{($\Rightarrow$)\quad Assume $\pair{P}{Q}\in R$. If $Q\subseteq \Sigma\times\Sigma$ then $Q\cap(\Sigma\times\{\bot\})=\emptyset$ so $\pair{P}{Q}\in R'$ by hypothesis. Otherwise $Q\cap(\Sigma\times
\{\bot\})\neq\emptyset$ so $(Q\cap(\Sigma\times\{\bot\})=\emptyset)\Rightarrow(\pair{P}{Q}\in R')$ is true, in which case
$R\subseteq \{\pair{P}{Q}\mid Q\cap(\Sigma\times\{\bot\})\neq\emptyset\}$;
\\
($\Leftarrow$)\quad If $\pair{P}{Q}\in R$ and $Q\cap(\Sigma\times\{\bot\})=\emptyset$ then $\pair{P}{Q}\in R'$
}\\
$\Leftrightarrow$
\lastformulaexplanation{R\subseteq \gamma^2_{\not\bot}(R')}{def\@. $ \gamma^2_{\not\bot}$}{\mbox{\qed}}
\end{calculus}
\let\qed\relax
\end{proof}
\end{toappendix}

\begin{example}[Manna and Pnueli total correctness logic]
By eliminating the  nontermination possibility from the postcondition of the natural transformational over approximation logic (\ref{eq:natural-upper-transformational-logic-theory}), we get
 Manna and Pnueli  logic \cite{DBLP:journals/acta/MannaP74} with theory \proofinapx
\begin{eqntabular}{rcl@{\qquad}}
\alpha^2_{\not\bot}(\textsf{\textup{post}}({\supseteq},{\subseteq})(\alpha_{\textup{G}}(\textsf{\upshape Post}\sqb{\texttt{\small S}}_\bot)))
&=&
\{\pair{P}{Q\setminus(\Sigma\times\{\bot\})}\mid\textsf{\upshape Post}\sqb{\texttt{\small S}}_\bot P\subseteq Q\setminus(\Sigma\times\{\bot\})\}
\label{eq:Manna-Pnueli-Logic-Theory}
\end{eqntabular}
that is, for any initial state $\pair{\sigma_0}{\sigma}$ satisfying the precondition $P$, execution terminates in a final state $\sigma'$ such that the pair $\pair{\sigma_0}{\sigma'}$ satisfies the postcondition $Q\cap\Sigma\times\Sigma)$. This is relational total correctness since nontermination is excluded.
\end{example}
\begin{toappendix}
\begin{proof}[Proof of (\ref{eq:Manna-Pnueli-Logic-Theory})]
\begin{calculus}
\formula{\alpha^2_{\not\bot}(\textsf{\textup{post}}({\supseteq},{\subseteq})(\alpha_{\textup{G}}(\textsf{\upshape Post}(\sqb{\texttt{\small S}}_\bot))))}\\

=\formulaexplanation{\alpha^2_{\not\bot}(\{\pair{P}{Q}\mid\textsf{\upshape Post}(\sqb{\texttt{\small S}}_\bot)P\subseteq{Q}\})}{by (\ref{eq:natural-upper-transformational-logic-theory})}\\

=\formulaexplanation{\{\pair{P}{Q}\mid \pair{P}{Q}\in \{\pair{P}{Q}\mid\textsf{\upshape Post}(\sqb{\texttt{\small S}}_\bot)P\subseteq{Q}\wedge Q\cap(\Sigma\times\{\bot\})=\emptyset\}\}}{def\@. (\ref{eq:def:alpha-not-bot}) of $\alpha^2_{\not\bot}$}\\
=\formulaexplanation{\{\pair{P}{Q}\mid\textsf{\upshape Post}(\sqb{\texttt{\small S}}_\bot)P\subseteq{Q}\wedge Q\cap(\Sigma\times\{\bot\})=\emptyset\}}{def\@. $\in$}\\
=\lastformulaexplanation{\{\pair{P}{Q\setminus(\Sigma\times\{\bot\})}\mid\textsf{\upshape Post}(\sqb{\texttt{\small S}}_\bot)P\subseteq Q\setminus(\Sigma\times\{\bot\})\}}{since $Q\in\wp(\Sigma\times\Sigma\cup\{\bot\})$}{\mbox{\qed}}
\end{calculus}
\let\qed\relax
\end{proof}
\end{toappendix}
Another abstraction to specify total correctness  is to consider a transformer for a modified semantics $\sqb{\texttt{\small S}}\cup\{\pair{\sigma_0}{\sigma'}\mid \pair{\sigma}{\bot}\in\sqb{\texttt{\small S}}_\bot\wedge \sigma'\in\Sigma\}$ returning any possible result in case of nontermination \cite{DBLP:conf/ac/Plotkin79} using Smyth powerdomain \cite{DBLP:journals/jcss/Smyth78} so that it is impossible to make any conclusion on final values in case of possible nontermination for an initial state. However, this is an impractical basis for static analysis since it the abstraction introduces great imprecision.

\subsubsection{The Termination Inclusion Abstraction}

We can include the possibility of
nontermination by the abstraction
\bgroup\abovedisplayskip-0.25\abovedisplayskip\belowdisplayskip0.75\belowdisplayskip%
\begin{eqntabular}{rcl}
\alpha^2_{\bot}(R)&\triangleq&\{\pair{P}{Q\cup(\Sigma\times\{\bot\})}\mid \pair{P}{Q}\in R\}
\label{eq:def:alpha-bot}
\end{eqntabular}\egroup
allowing the possibility of nontermination for all input states by adding
 $\bot$ to the postcondition. This is an abstraction by a Galois connection \proofinapx
\bgroup\abovedisplayskip0.5\abovedisplayskip\belowdisplayskip0.25\belowdisplayskip\begin{eqntabular}{c}
\pair{\wp(\wp(\Sigma\times\Sigma)\times\wp(\Sigma\times\Sigma_\bot))}{\subseteq}\galoiS{\alpha^2_{\bot}}{\gamma^2_{\bot}}\pair{\wp(\wp(\Sigma\times\Sigma)\times\wp(\Sigma\times\Sigma_\bot))}{\subseteq}
\label{eq:alpha:2:bot:GC-2}
\end{eqntabular}\egroup
with $\gamma^2_{\bot}(R')\triangleq\{ \pair{P}{Q}\mid\pair{P}{Q\cup(\Sigma\times\{\bot\})}\in R'\}$.
\begin{toappendix}
\begin{proof}[Proof of (\ref{eq:alpha:2:bot:GC-2})]
\begin{calculus}[$\Leftrightarrow$\ ]
\formula{\alpha^2_{\bot}(R)\subseteq R'}\\
$\Leftrightarrow$\
\formulaexplanation{\{\pair{P}{Q\cup(\Sigma\times\{\bot\})}\mid \pair{P}{Q}\in R\}\subseteq R'}{def\@. (\ref{eq:def:alpha-bot}) of $\alpha^2_{\bot}$}\\
$\Leftrightarrow$\
\formulaexplanation{\forall \pair{P}{Q}\in R\mathrel{.}\pair{P}{Q\cup(\Sigma\times\{\bot\})}\in R'}{def\@. subset $\subseteq$}\\
$\Leftrightarrow$\
\formulaexplanation{R\subseteq\{ \pair{P}{Q}\mid\pair{P}{Q\cup(\Sigma\times\{\bot\})}\in R'\}}{def\@. subset $\subseteq$}\\
$\Leftrightarrow$\
\lastformulaexplanation{R\subseteq\gamma^2_{\bot}(R')}{def\@. concretization $\gamma^2_{\bot}$}{\mbox{\qed}}
\end{calculus}\let\qed\relax
\end{proof}
\end{toappendix}

\begin{example}[Manna relational partial correctness logic]\label{ex:Manna-relational-partial-correctness-logic-1}
Manna's relational partial correctness logic \cite{DBLP:journals/jcss/Manna71} includes the nontermination possibility for all input states. Its theory is  \proofinapx
\bgroup\abovedisplayskip0.75\abovedisplayskip\belowdisplayskip0.75\belowdisplayskip\begin{eqntabular}{rcl@{\qquad}}
\alpha^2_{\bot}(\textsf{\textup{post}}({\supseteq},{\subseteq})(\alpha_{\textup{G}}(\textsf{\upshape Post}\sqb{\texttt{\small S}}_\bot)))
&=&
\{\pair{P}{Q\cup(\Sigma\times\{\bot\})}\mid \textsf{\upshape Post}\sqb{\texttt{\small S}}_\bot P\subseteq{Q}\}
\label{eq:Manna-Logic-Theory}
\end{eqntabular}\egroup
which is $\{\pair{P}{Q}\in \wp(\Sigma\times\Sigma)\times\wp(\Sigma\times\Sigma)\mid\forall \pair{\sigma_0}{\sigma}\in P\mathrel{.}\forall \sigma'\mathrel{.}
\pair{\sigma}{\sigma'}\in\sqb{\texttt{\small S}}\Rightarrow\pair{\sigma_0}{\sigma'}\in{Q}\}$ when using the angelic semantics $\sqb{\texttt{\small S}}$ 
i.e\@. any terminating execution started within $P$ satisfies $Q$.
\end{example}
\begin{toappendix}
\begin{proof}[Proof of (\ref{eq:Manna-Logic-Theory})]
\begin{calculus}
\formula{\alpha^2_{\bot}(\textsf{\textup{post}}({\supseteq},{\subseteq})(\alpha_{\textup{G}}(\textsf{\upshape Post}(\sqb{\texttt{\small S}}_\bot))))}\\
=\formulaexplanation{\alpha^2_{\bot}(\{\pair{P}{Q}\mid\textsf{\upshape Post}(\sqb{\texttt{\small S}}_\bot)P\subseteq{Q}\})}{by (\ref{eq:natural-upper-transformational-logic-theory})}\\
=\formulaexplanation{\{\pair{P}{Q\cup(\Sigma\times\{\bot\})}\mid \pair{P}{Q}\in \{\pair{P}{Q}\mid\textsf{\upshape Post}(\sqb{\texttt{\small S}}_\bot)P\subseteq{Q}\}\}}{def\@. (\ref{eq:def:alpha-bot}) of $\alpha^2_{\bot}$}\\
=\lastformulaexplanation{\{\pair{P}{Q\cup(\Sigma\times\{\bot\})}\mid \textsf{\upshape Post}(\sqb{\texttt{\small S}}_\bot)P\subseteq{Q}\}}{def\@. $\in$}{\mbox{\qed}}
\end{calculus}
\let\qed\relax
\end{proof}
\end{toappendix}

So to prove partial correctness, we essentially add the possibility of nontermination to postconditions in $\wp(\Sigma\times\Sigma_\bot)$. However, for partial correctness, postconditions are 
traditionally chosen in $\wp(\Sigma\times\Sigma)$ not $\wp(\Sigma\times\Sigma_\bot)$. This equivalent alternative uses the Galois connection~\proofinapx
\bgroup\abovedisplayskip0pt\belowdisplayskip-5pt\begin{eqntabular}{c}
\pair{\wp(\wp(\Sigma\times\Sigma)\times\wp(\Sigma\times\Sigma_\bot))}{\subseteq}\galoiS{{\alpha^{2}_{\bot}}'}{{\gamma^{2}_{\bot}}'}\pair{\wp(\wp(\Sigma\times\Sigma)\times\wp(\Sigma\times\Sigma))}{\subseteq}
\label{eq:alpha:bot:GC:in:Sigma}
\end{eqntabular}\egroup
with
\bgroup\abovedisplayskip0.25\abovedisplayskip\belowdisplayskip0.75\belowdisplayskip\arraycolsep0.5\arraycolsep
\begin{eqntabular}{rcl@{\qquad}rcl}
{{\alpha^{2}_{\bot}}'}(R)&\triangleq&\{\pair{P}{Q\cap(\Sigma\times\Sigma)}\mid \pair{P}{Q}\in R\}
\label{eq:def:alpha-bot:in:Sigma}
&{\gamma^{2}_{\bot}}'(R')&\triangleq&\{\pair{P}{Q} \mid\pair{P}{Q\cap(\Sigma\times\Sigma)}\in R'\}\nonumber
\end{eqntabular} 
\egroup
\begin{toappendix}
\begin{proof}[Proof of (\ref{eq:alpha:bot:GC:in:Sigma})]
\begin{calculus}[$\Leftrightarrow$\ ]
\formula{{\alpha^{2}_{\bot}}'(R)\subseteq R'}\\
$\Leftrightarrow$
\formulaexplanation{\{\pair{P}{Q\cap(\Sigma\times\Sigma)}\mid \pair{P}{Q}\in R\}\subseteq R'}{def\@. (\ref{eq:def:alpha-bot:in:Sigma}) of ${\alpha^{2}_{\bot}}'$}\\
$\Leftrightarrow$
\formulaexplanation{\forall \pair{P}{Q}\in R\mathrel{.}\pair{P}{Q\cap(\Sigma\times\Sigma)}\in R'}{def\@. $\subseteq$}\\
$\Leftrightarrow$
\formulaexplanation{R\subseteq\{\pair{P}{Q} \mid\pair{P}{Q\cap(\Sigma\times\Sigma)}\in R'\}}{def\@. $\subseteq$}\\
$\Leftrightarrow$
\lastformulaexplanation{R\subseteq{\gamma^{2}_{\bot}}'(R')}{def\@. (\ref{eq:def:alpha-bot:in:Sigma}) of ${\gamma^{2}_{\bot}}'$}{\mbox{\qed}}
\end{calculus}\let\qed\relax
\end{proof}
\end{toappendix}
\begin{example}[Manna relational partial correctness logic, continuing example \ref{ex:Manna-relational-partial-correctness-logic-1}]\label{ex:Manna-relational-partial-correctness-logic-2}
In that case, the theory of Manna's logic is  \proofinapx
\begin{eqntabular}{rcl@{\qquad\qquad}}
{\alpha^{2}_{\bot}}'(\textsf{\textup{post}}({\supseteq},{\subseteq})(\alpha_{\textup{G}}(\textsf{\upshape Post}\sqb{\texttt{\small S}}_\bot)))
&=&
\{\pair{P}{Q\cap(\Sigma\times\Sigma)}\mid\textsf{\upshape Post}\sqb{\texttt{\small S}}_\bot P\subseteq Q\}
\label{eq:Manna-Logic-Theory-Sigma}
\stepcounter{equation}
\renumber{(\ref{eq:Manna-Logic-Theory-Sigma})\qquad\qef}
\end{eqntabular}\let\qef\relax
\end{example}
\begin{toappendix}
\begin{proof}[Proof of (\ref{eq:Manna-Logic-Theory-Sigma})]
\begin{calculus}
\formula{{\alpha^{2}_{\bot}}'(\textsf{\textup{post}}({\supseteq},{\subseteq})(\alpha_{\textup{G}}(\textsf{\upshape Post}(\sqb{\texttt{\small S}}_\bot))))}\\
=
\formulaexplanation{{\alpha^{2}_{\bot}}'(\{\pair{P}{Q}\mid\textsf{\upshape Post}(\sqb{\texttt{\small S}}_\bot)P\subseteq{Q}\})}{by (\ref{eq:natural-upper-transformational-logic-theory})}\\
=
\formulaexplanation{\{\pair{P}{Q\cap(\Sigma\times\Sigma)}\mid \pair{P}{Q}\in \{\pair{P}{Q}\mid\textsf{\upshape Post}(\sqb{\texttt{\small S}}_\bot)P\subseteq{Q}\}\}}{def\@. (\ref{eq:def:alpha-bot:in:Sigma}) of ${\alpha^{2}_{\bot}}'$}\\
=
\lastformulaexplanation{\{\pair{P}{Q\cap(\Sigma\times\Sigma)}\mid \textsf{\upshape Post}(\sqb{\texttt{\small S}}_\bot)P\subseteq{Q}\}}{def\@. $\in$}{\mbox{\qed}}
\end{calculus}\let\qed\relax
\end{proof}
\end{toappendix}

\subsection{Relational to Assertional Abstraction}\label{sec:Relational-assertional-abstraction}
Since they relate initial pairs $\pair{\sigma_0}{\sigma}$ to final pairs $\pair{\sigma_0}{\sigma'}$, $\sigma_0\in\mathcal{X}$, $\sigma\in\mathcal{Y}$, and $\sigma'\in\mathcal{Z}$, 
relational logics have their theory in a set
$\wp(\wp(\mathcal{X}\times\mathcal{Y})\times\wp(\mathcal{X}\times\mathcal{Z}))$ while assertional logic theories are in $\wp(\wp(\mathcal{Y})\times\wp(\mathcal{Z}))$ where e.g\@.  the postcondition is on final states and unrelated to the initial ones. This is an
abstraction by projection on the second component \proofinapx
\bgroup\abovedisplayskip0pt\belowdisplayskip-5pt\begin{eqntabular}{c@{\ \ }c@{\qquad}}
\mbox{}\hskip-1ex\pair{\wp(\mathcal{X}\times\mathcal{Y})}{\subseteq}\galois{\alpha_{\downarrow^2}}{\gamma_{\downarrow^2}}\pair{\wp(\mathcal{Y})}{\subseteq},\label{eq:alpha:relationbal-assertional:GC}
&
\pair{\wp(\wp(\mathcal{X}\times\mathcal{Y})\times\wp(\mathcal{X}\times\mathcal{Z}))}{\subseteq}\galois{\mathord{\stackrel{.}{\alpha}}_{\downarrow^2}}{\mathord{\stackrel{.}{\gamma}}_{\downarrow^2}}\pair{\wp(\wp(\mathcal{Y})\times\wp(\mathcal{Z}))}{\subseteq}
\end{eqntabular}\egroup
with
\bgroup\arraycolsep0.8\arraycolsep\abovedisplayskip5pt\belowdisplayskip0pt
\begin{eqntabular}{@{\hskip-1ex}rcl@{\quad}rcl}
\alpha_{\downarrow^2}(P)&\triangleq&\{\sigma\mid\exists \sigma_0\mathrel{.}\pair{\sigma_0}{\sigma}\in P\}
&
\gamma_{\downarrow^2}(Q)&\triangleq&\mathcal{X}\times Q
\label{eq:def:alpha-gamma-d2}
\\
\mathord{\stackrel{.}{\alpha}}_{\downarrow^2}(R)&\triangleq&\{\pair{\alpha_{\downarrow^2}(P)}{\alpha_{\downarrow^2}(Q)}\mid\pair{P}{Q}\in R\}
&\mathord{\stackrel{.}{\gamma}}_{\downarrow^2}(R')&\triangleq&\{ \pair{\gamma_{\downarrow^2}(P')}{\gamma_{\downarrow^2}(Q')}\mid \pair{P'}{Q'}\in P\}\nonumber
\end{eqntabular}\egroup
\begin{toappendix}
\begin{proof}[Proof of (\ref{eq:alpha:relationbal-assertional:GC})]
\begin{calculus}[$\Leftrightarrow$\ ]
\hyphen{5}\formula{\alpha_{\downarrow^2}(P)\subseteq Q}\\
$\Leftrightarrow$
\formulaexplanation{\{\sigma\mid\exists \sigma_0\mathrel{.}\pair{\sigma_0}{\sigma}\in P\}\subseteq Q}{def\@. (\ref{eq:def:alpha-gamma-d2}) of $\alpha_{\downarrow^2}$}\\
$\Leftrightarrow$
\formulaexplanation{\forall \sigma\mathrel{.}(\exists \sigma_0\mathrel{.}\pair{\sigma_0}{\sigma}\in P)\Rightarrow(\sigma\in Q)}{def\@. $\subseteq$}\\
$\Leftrightarrow$
\formulaexplanation{\forall \sigma\mathrel{.}\forall\sigma_0\mathrel{.}(\pair{\sigma_0}{\sigma}\in P)\Rightarrow(\sigma\in Q)}{def\@. $\Rightarrow$}\\
$\Leftrightarrow$
\formulaexplanation{\forall\sigma_0\mathrel{.}\forall \sigma\mathrel{.}(\pair{\sigma_0}{\sigma}\in P)\Rightarrow(\pair{\sigma_0}{\sigma}\in \mathcal{X}\times Q)}{def\@. $\forall$ and $\sigma_0\in\mathcal{X}$}\\
$\Leftrightarrow$
\formulaexplanation{P\subseteq(\mathcal{X}\times Q)}{def\@. $\subseteq$}\\
$\Leftrightarrow$
\formulaexplanation{P\subseteq\gamma_{\downarrow^2}(Q)}{def\@. (\ref{eq:def:alpha-gamma-d2}) of  $\gamma_{\downarrow^2}$}\\[1ex]

\hyphen{5}\formula{\mathord{\stackrel{.}{\alpha}}_{\downarrow^2}(R)\subseteq R'}\\
$\Leftrightarrow$
\formulaexplanation{\{\pair{\alpha_{\downarrow^2}(P)}{\alpha_{\downarrow^2}(Q)}\mid\pair{P}{Q}\in R\}
\subseteq R'}{def\@. $\mathord{\stackrel{.}{\alpha}}_{\downarrow^2}$}\\

$\Leftrightarrow$
\formulaexplanation{\{\pair{\{\sigma'\mid\exists \sigma\mathrel{.}\pair{\sigma}{\sigma'}\in p\}}{\{\sigma'\mid\exists \sigma\mathrel{.}\pair{\sigma}{\sigma'}\in Q\}}\mid\pair{P}{Q}\in R'\}\subseteq R'}{def\@. $\alpha_{\downarrow^2}$}\\
$\Leftrightarrow$
\formulaexplanation{\forall \pair{P}{Q}\in R\mathrel{.}\pair{\{\sigma'\mid\exists \sigma\mathrel{.}\pair{\sigma}{\sigma'}\in P\}}{\{\sigma'\mid\exists \sigma\mathrel{.}\pair{\sigma}{\sigma'}\in Q\}}\in R'}{def\@. $\subseteq$}\\
$\Leftrightarrow$
\formulaexplanation{R\subseteq\{ \pair{P}{Q}\mid\pair{\{\sigma'\mid\exists \sigma\mathrel{.}\pair{\sigma}{\sigma'}\in P\}\in R'}{\{\sigma'\mid\exists \sigma\mathrel{.}\pair{\sigma}{\sigma'}\in Q\}\in R'}}{def\@. $\subseteq$}\\
$\Leftrightarrow$
\formulaexplanation{R\subseteq\{ \pair{P}{Q}\mid\exists \pair{P'}{Q'}\in 'R\mathrel{.}P'=\{\sigma'\mid\exists \sigma\mathrel{.}\pair{\sigma}{\sigma'}\in P\}\wedge Q'=\{\sigma'\mid\exists \sigma\mathrel{.}\pair{\sigma}{\sigma'}\in Q\}\}}{def\@. $\in$}\\
$\Leftrightarrow$
\formulaexplanation{R\subseteq\{ \pair{\Sigma\times P'}{\Sigma\times Q'}\mid \pair{P'}{Q'}\in R'\}}{since $P=\Sigma\times P'$ by $\sigma\in\Sigma$, same for $Q$}\\
$\Leftrightarrow$
\formulaexplanation{R\subseteq\{ \pair{\gamma_{\downarrow^2}(P')}{\gamma_{\downarrow^2}(Q')}\mid \pair{P'}{Q'}\in R'\}}{def\@. $\gamma_{\downarrow^2}$}\\
$\Leftrightarrow$
\lastformulaexplanation{R\subseteq\mathord{\stackrel{.}{\gamma}}_{\downarrow^2}(R')}{def\@. $\mathord{\stackrel{.}{\gamma}}_{\downarrow^2}$}{\mbox{\qed}}
\end{calculus}
\let\qed\relax
\end{proof}
\end{toappendix}
\begin{example}At this point we have got the theory of Hoare logic as the abstraction
\bgroup\abovedisplayskip0.75\abovedisplayskip\belowdisplayskip0.75\belowdisplayskip\begin{eqntabular}{rcl}
\alpha_{\textsf{H}}(\crb{\texttt{\small S}}_{\!\bot})&\triangleq&
\footnotesize\begin{array}[b]{@{}c@{}}\mbox{\llap{assertio}\rlap{nal}}\\\downarrow\\\alpha_{\downarrow^2}\end{array}\comp\begin{array}[b]{@{}c@{}}\mbox{\llap{nonter}m\rlap{ination}}\\[-0.5ex]\raisebox{-1.25ex}[0pt][0pt]{$\vert$}\\\raisebox{-0.75ex}[0pt][0pt]{$\vert$}\\\downarrow\\\alpha^2_{\bot}\end{array}\comp\begin{array}[b]{@{}c@{}}\mbox{\llap{cons}equ\rlap{ence}}\\\downarrow\\\textsf{\textup{post}}({\supseteq},{\subseteq})\end{array}\comp\begin{array}[b]{@{}c@{}}\mbox{\llap{gr}a\rlap{ph}}\\[-0.5ex]\raisebox{-1.25ex}[0pt][0pt]{$\vert$}\\\raisebox{-0.75ex}[0pt][0pt]{$\vert$}\\\downarrow\\\alpha_{\textup{G}}\end{array}\comp\begin{array}[b]{@{}c@{}}\mbox{\llap{tr}a\rlap{ns-}}\\[-0.5ex]\mbox{\llap{fo}r\rlap{mer}}\\\downarrow\\\textsf{\upshape Post}\end{array}\comp\begin{array}[b]{@{}c@{}}\mbox{\llap{relati}o\rlap{nal}}\\[-0.5ex]\mbox{\llap{sema}n\rlap{tics}}\\[-0.5ex]\raisebox{-1.25ex}[0pt][0pt]{$\vert$}\\\raisebox{-0.75ex}[0pt][0pt]{$\vert$}\\\downarrow\\\alpha_{\textup{C}}\end{array}(\begin{array}[b]{@{}c@{}}\mbox{\llap{co}l\rlap{lecting}}\\[-0.5ex]\mbox{\llap{se}m\rlap{antics}}\\\downarrow\\\crb{\texttt{\small S}}_{\!\bot}\end{array}
)\label{eq:Hoare-logic-theory}\\
&=&\{\pair{P}{Q}\mid\forall \sigma\in P\mathrel{.}\forall \sigma'\mathrel{.}
\pair{\sigma}{\sigma'}\in\sqb{\texttt{\small S}}_\bot\Rightarrow{\sigma'}\in{Q}\cup\{\bot\}\}
\nonumber
\end{eqntabular}\egroup
The set of valid Hoare triples $\{P\}\texttt{\small S}\{Q\}$ is the set of pairs $\pair{P}{Q}$ in $\alpha_{\textsf{H}}(\sqb{\texttt{\small S}}_\bot)$  such that any execution started in a state $\sigma$ of $P$, that terminates, if ever, does terminate in a state $\sigma'$ of $Q$.
\end{example}
\begin{example}
Similarly the assertional abstraction $\alpha_{\downarrow^2}$ of Manna and Pnueli logic (\ref{eq:Manna-Pnueli-Logic-Theory}) yields  Apt and Plotkin generalization of Hoare logic to total correctness \cite[equation (6), page 749]{DBLP:journals/jacm/AptP86} (generalizing \cite{DBLP:books/sp/Harel79} using naturals to unbounded nondeterminism using ordinals, equivalently a variant function in well-founded sets, as first considered by Turing \cite{Turing49-program-proof} and Floyd \cite{Floyd67-1}). 
\end{example}
Similarly, we can define  an
abstraction by projection on the first component 
\bgroup\abovedisplayskip0.75\abovedisplayskip\belowdisplayskip0.75\belowdisplayskip
\begin{eqntabular}{rcl@{\qquad\quad}rcl@{\qquad\quad}rcl}
\alpha^{-1}(r)&\triangleq&r^{-1}
&
\alpha_{\downarrow^1}&\triangleq&\alpha_{\downarrow^2}\comp\alpha^{-1}
&
\gamma_{\downarrow^1}&\triangleq&\alpha^{-1}
\comp\gamma_{\downarrow^2}
\label{eq:def:alpha-gamma-d1}
\end{eqntabular}
so that by composition of Galois connections and isomorphisms (proposition \ref{prop:GC-composition}) and by the forthcoming (\ref{eq:def:alpha-1}), we have Galois connection similar to (\ref{eq:alpha:relationbal-assertional:GC}) for $\pair{\alpha_{\downarrow^1}}{\gamma_{\downarrow^1}}$.
\egroup

One may wonder why, for such a well-known result, we have considered so many successive abstractions (six when including the abstraction (\ref{eq:hyper-tp-trace-abs}) of the collecting semantics into the relational semantics). There are three main reasons.
\begin{enumerate}[leftmargin=*]
\item The composition of Galois connections and isomorphisms is a Galois connection (Prop\@. \ref{prop:GC-composition}\ifshort\ in the appendix\fi). Since abstractions preserves existing joins and concretizations preserve existing meets, we get ``healthiness conditions'' (such as \cite[(H2), page 469]{DBLP:journals/jacm/Hoare78}) as theorems, not hypotheses. In absence of a Galois connection, there would be no unique, most precise approximation, of the collecting semantics by a formula of the logic (e.g\@. \cite{DBLP:journals/entcs/GotsmanBC11});
\item By varying slightly the abstractions, we get a hierarchy of transformational logics (which extends the hierarchy of semantics in \cite{DBLP:journals/tcs/Cousot02}), that we can compare without even knowing their proof systems. This is the objective for the rest of this part I on the theories of logics;
\item Knowing the program semantics and its abstraction to the theory of a logic, we can constructively design, by calculus, a sound and complete proof system for this logic. This will be developed in part II.
\end{enumerate}
\subsection{The Forward Transformational Forward Logics Hierarchy}
We have built the theories of logics in Fig\@. \ref{fig:Forward-semantics-logics} by composition of abstractions. \noindent The relational and assertional logics are considered equivalent in practice by using an auxiliary program with phantom variables recording the values of the initial or final variables (see Sect\@. \ref{sec:auxiliary:variables}\ifshort\ in the appendix\fi). By allowing the explicit use of nontermination $\bot$ in the postcondition, the over/under approximating antecedent/consequent logics subsume their approximations by $\alpha^2_{\bot}$ or $\alpha^2_{\not\bot}$ and $\alpha_{\downarrow^2}$
(including the logics marked by circled numbers that do not look to have been considered in the literature). 
\vskip3pt
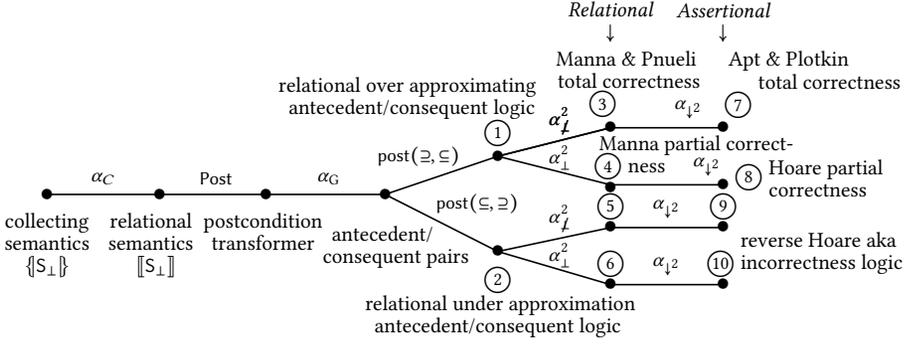
\begin{figure}[ht]
\begin{flushleft}\let\textsize\footnotesize
\vskip8mm
\begin{tikzpicture}[scale=0.75]
\draw (0,3) node[name=ColSem] {\textsize\raisebox{1pt}[0pt][0pt]{\begin{tabular}[t]{c}\large$\bullet$\\collecting\\[-0.5ex] semantics\\[-0.5ex] $\crb{\texttt{\textsize S}_\bot}$\end{tabular}}};
\draw (2,3) node[name=Sem] {\textsize\raisebox{1pt}[0pt][0pt]{\begin{tabular}[t]{c}\large$\bullet$\\\llap{re}lational\\[-0.5ex] \llap{se}mantics\\[-0.5ex]$\hskip-1mm\sqb{\texttt{\textsize S}_\bot}$\end{tabular}}};
\draw[semithick] (ColSem.north) --  (Sem.north) node[draw=none,fill=none,font=\scriptsize,midway,above] {$\alpha_C$};
\draw (4,3) node[name=Post] {\textsize\raisebox{1pt}[0pt][0pt]{\begin{tabular}[t]{@{\hskip-2pt}c}\hskip2pt\large$\bullet$\\postcondition \\[-0.5ex] transformer\end{tabular}}};
\draw[semithick] (Sem.north) --  (Post.north) node[draw=none,fill=none,font=\scriptsize,midway,above] {$\textsf{\upshape Post}$};
\draw (6,3) node[name=Pairs] {\textsize\raisebox{1pt}[0pt][0pt]{\begin{tabular}[t]{c}\large$\bullet$\\[1ex]\hskip-2pt antecedent/\\[-0.5ex] consequent pa\rlap{irs}\end{tabular}}};
\draw[semithick] (Post.north) --  (Pairs.north) node[draw=none,fill=none,font=\scriptsize,midway,above] {$\alpha_{\textup{G}}$};
\draw (8,4) node[name=gtlt] {\textsize\raisebox{-6pt}[0pt][0pt]{\begin{tabular}[b]{c}\llap{relational over appro}ximating\\[-0.5ex] \llap{antecedent/consequ}ent logic \\[-0.25ex]\circled{1}\\[-0.5ex]\large$\bullet$\end{tabular}}};
\draw[semithick] (Pairs.north) --  (gtlt.south) node[draw=none,fill=none,font=\scriptsize,midway,above] {\llap{$\textsf{\textup{post}}({\supseteq},{\subseteq})$\hskip-1em}};
\draw (8,2) node[name=ltgt] {\textsize\raisebox{1pt}[0pt][0pt]{\begin{tabular}[t]{c}\large$\bullet$\\ \circled{2}\\[-0.5ex]~relational under approximation \\[-0.5ex] antecedent/consequent logic\end{tabular}}};
\draw[semithick] (Pairs.north) --  (ltgt.north) node[draw=none,fill=none,font=\scriptsize,midway,above] {\rlap{\hskip-2pt$\textsf{\textup{post}}({\subseteq},{\supseteq})$}};
\draw (10,4.5) node[name=MP] {\textsize\raisebox{-6pt}[0pt][0pt]{\begin{tabular}[b]{c}\itshape Relational\\[-0.5ex]$\downarrow$\\Manna \& Pn\rlap{ueli}\\[-0.5ex]total correct\rlap{ness}\\[-0.5ex]\circled{3}~\ \ \ \\[-0.5ex]\large$\bullet$\end{tabular}}};
\draw[semithick] (gtlt.south) --  (MP.south) node[draw=none,fill=none,font=\scriptsize,midway,above] {\rlap{\hskip-2pt$\alpha^2_{\not\bot}$}};
\draw (10,3.5) node[name=M] {\textsize\raisebox{-7pt}[0pt][0pt]{\begin{tabular}[b]{c}~Ma\rlap{nna partial correct-}\\[-0.5ex]\raisebox{-1pt}[0pt][0pt]{\circled{4}} \rlap{ness}\\[-0.5ex]\large$\bullet$\end{tabular}}};
\draw[semithick] (gtlt.south) --  (M.south) node[draw=none,fill=none,font=\scriptsize,midway,above] {\rlap{\raisebox{-1pt}[0pt][0pt]{\hskip-2pt$\alpha^2_{\bot}$}}};
\draw (10,2.75) node[name=MPunder] {\textsize\raisebox{-6pt}[0pt][0pt]{\begin{tabular}[b]{c}\raisebox{-1pt}[0pt][0pt]{\circled{5}}\\[-0.5ex]\large$\bullet$\end{tabular}}};
\draw[semithick] (gtlt.south) --  (MP.south) node[draw=none,fill=none,font=\scriptsize,midway,above] {\rlap{\hskip-2pt$\alpha^2_{\not\bot}$}};
\draw (10,1.75) node[name=Munder] {\textsize\raisebox{-6pt}[0pt][0pt]{\begin{tabular}[b]{c}\raisebox{-1pt}[0pt][0pt]{\circled{6}}\\[-0.5ex]\large$\bullet$\end{tabular}}};
\draw[semithick] (ltgt.north) --  (MPunder.south) node[draw=none,fill=none,font=\scriptsize,midway,above] {\rlap{\hskip-2pt$\alpha^2_{\not\bot}$}};
\draw[semithick] (ltgt.north) --  (Munder.south) node[draw=none,fill=none,font=\scriptsize,midway,above] {\rlap{\raisebox{-1pt}[0pt][0pt]{\hskip-2pt$\alpha^2_{\bot}$}}};
\draw (12,4.5) node[name=MPassert] {\textsize\raisebox{-6pt}[0pt][0pt]{\begin{tabular}[b]{c}\itshape Assertional\\[-0.6ex]$\downarrow$\\\rlap{~Apt \& Plotkin}\\[-0.6ex]\rlap{\phantom{\circled{7}}~total correctness}\\[-0.6ex]\rlap{\circled{7}~}\\[-0.6ex]\large$\bullet$\end{tabular}}};
\draw[semithick] (MP.south) --  (MPassert.south) node[draw=none,fill=none,font=\scriptsize,midway,above] {\rlap{\hskip+3pt$\alpha_{\downarrow^2}$}};
\draw (12,3.5) node[name=H] {\textsize\raisebox{-16pt}[0pt][0pt]{\begin{tabular}[b]{c}\rlap{\ \ \phantom{\circled{8}} Hoare partial}\\[-0.5ex]\rlap{\ \ \raisebox{4pt}[0pt][0pt]{\circled{8}} correctness}\\[-0.5ex]\raisebox{10pt}[0pt][0pt]{\large$\bullet$}\end{tabular}}};
\draw[semithick] (M.south) --  (H.south) node[draw=none,fill=none,font=\scriptsize,midway,above] {\rlap{\hskip+10pt$\alpha_{\downarrow^2}$}};
\draw (12,2.75) node[name=APunder] {\textsize\raisebox{-6pt}[0pt][0pt]{\begin{tabular}[b]{c}\raisebox{-1pt}[0pt][0pt]{\circled{9}}\\[-0.5ex]\large$\bullet$\end{tabular}}};
\draw[semithick] (MPunder.south) --  (APunder.south) node[draw=none,fill=none,font=\scriptsize,midway,above] {$\alpha_{\downarrow^2}$};
\draw (12,1.75) node[name=APunder] {\textsize\raisebox{-6pt}[0pt][0pt]{\begin{tabular}[b]{c}\phantom{\circled{10}}~~~~\rlap{reverse Hoare aka}\\[-0.5ex]\raisebox{-1pt}[0pt][0pt]{\circled{10}}~~\rlap{incorrectness logic}\\[-0.5ex]\large$\bullet$\end{tabular}}};
\draw[semithick] (Munder.south) --  (APunder.south) node[draw=none,fill=none,font=\scriptsize,midway,above] {$\alpha_{\downarrow^2}$};
\end{tikzpicture}
\end{flushleft}
\vskip0.75em
\caption{Forward semantics and logics\label{fig:Forward-semantics-logics}}
\vskip-1.5em
\end{figure}

\subsection{Singularities of Logics}\label{sec:Singularities}
\subsubsection{Emptiness Versus Universality}\label{sec:UniversalityEmptiness}
The same way that \textsf{false} is satisfied by no element of the universe in logic,
some transformational logics have this  emptiness property, meaning that some programs satisfy no formula of the logic. This is the case of a nonterminating program for Manna and Pnueli total correctness logic \cite{DBLP:journals/acta/MannaP74}. Emptiness  may look awkward since using the deductive system to prove any specification will always fail.

The same way that \textsf{true} is satisfied by all elements of the universe in logic,
transformational  logics may have the universality property, meaning that  there exist programs for which any pair $\pair{P}{Q}$ for that program is in the logic (i.e.\ is satisfied in logical terms). For example, in Hoare logic, $\{P\}\,\texttt{\small while (true) skip}\,\{Q\}$ is satisfied for all $P$ and $Q$. $[P]\,\texttt{\small S}\,[\texttt{\small false}]$ 
is always true in incorrectness logic \cite{DBLP:journals/pacmpl/OHearn20}. Universality may look awkward since using the deductive system to prove this obvious fact may be very complicated.

These phenomena have been criticized (e.g.\ emptiness  for necessary preconditions \cite{DBLP:conf/vmcai/CousotCL11,DBLP:conf/vmcai/CousotCFL13} in \cite[section, page 10:28]{DBLP:journals/pacmpl/OHearn20}) but are inherent to semantic approximation.

\subsubsection{Correctness Versus Incorrectness}\label{sec:Correctness-versus-incorrectness}
The use of a logic to prove correctness or incorrectness is not intrinsic but depending upon the application domain. For example, termination is required for most programs so that Manna and Pnueli logic is a correctness logic \cite{DBLP:journals/acta/MannaP74}. However, operating systems should not terminate, and proving the contrary by Manna and Pnueli logic \cite{DBLP:journals/acta/MannaP74} would make it an incorrectness logic. Another example is the incorrectness logic \cite{DBLP:journals/pacmpl/OHearn20} which has the same theory as the reverse Hoare logic used by \cite{DBLP:conf/sefm/VriesK11} to prove correctness. The qualification of under or over approximation instead of correctness or incorrectness logics looks more independent of specific applications, as suggested by \cite{DBLP:conf/ecoop/MaksimovicCLSG23}.

\subsection{Backward Logics}\label{sec:Backward-logics}
Backward logics originates from the inversion abstraction (using the inverse program semantics $(\sqb{\texttt{\small S}}_\bot)^{-1}$) or the dual complement abstraction (stating the impossibility of the negation of a property, which is called the \emph{duality principle for programs} by Pratt \cite[p\@. 110]{DBLP:conf/focs/Pratt76}) and the \emph{conjugate} in \cite[equation (2) page 82]{DBLP:books/daglib/0067387}. 
They correspond to the commutative diagram of \cite[page 241]{DBLP:conf/popl/CousotC77}, also found on \cite[page 98]{CousotCousot82-TNPC} (where inversion is ${}^{-1}$ and complement is \~{}$\,$), diagrams which are extended to Fig\@.~\ref{fig:taxonomy}.

\subsubsection{The Inversion Abstraction}
As noticed by \cite[section 1.2]{DBLP:conf/focs/Pratt76}, the inversion isomorphism transforms forward antecedent-consequent logics into backward consequent-antecedent logics.
For that purpose, let us define the relation isomorphic abstraction $\alpha^{-1}$, its pointwise extension $\mathord{\stackrel{.}{\alpha}^{-1}}$, and the inverse transformer abstraction $\vec\alpha^{-1}$.
\bgroup\arraycolsep0.5\arraycolsep
\begin{eqntabular}{rcl@{\qquad}rcl@{\qquad}rcl}
\alpha^{-1}(r) &\triangleq&r^{-1} 
& 
\mathord{\stackrel{.}{\alpha}^{-1}}(f)  
&\triangleq&\alpha^{-1}\comp f\comp\alpha^{-1}
&  
\vec\alpha^{-1}(T)&\triangleq&\mathord{\stackrel{.}{\alpha}^{-1}}\comp T\comp\alpha^{-1}
\label{eq:def:alpha-1}
\end{eqntabular}
so that we have the following Galois isomorphisms \proofinapx
\egroup
\begin{eqntabular}{c@{\qquad}}
\pair{\wp(\mathcal{X}\times \mathcal{Y})}{\subseteq}\GaloiS{\alpha^{-1}}{\alpha^{-1}}\pair{\wp(\mathcal{Y}\times \mathcal{X})}{\subseteq}
\nonumber\\[-1ex]
\pair{\wp(\mathcal{Z}\times \mathcal{X})\rightarrow\wp(\mathcal{Z}\times \mathcal{Y})}{\stackrel{.}{\subseteq}}\GaloiS{\mathord{\stackrel{.}{\alpha}^{-1}}}{\mathord{\stackrel{.}{\alpha}^{-1}}}\pair{\wp(\mathcal{X}\times \mathcal{Z})\rightarrow\wp(\mathcal{Y}\times \mathcal{Z})}{\stackrel{.}{\subseteq}}
\label{eq:GC:alpha-1}\\[-1ex]
\pair{\wp(\mathcal{X}\times \mathcal{Y})\rightarrow \wp(\mathcal{Z}\times \mathcal{X})\rightarrow\wp(\mathcal{Z}\times \mathcal{Y})}{\stackrel{..}{\subseteq}}
\GaloiS{\vec{\alpha}^{-1}}{\vec{\alpha}^{-1}}\pair{\wp(\mathcal{Y}\times \mathcal{X})\rightarrow \wp(\mathcal{X}\times \mathcal{Z})\rightarrow\wp(\mathcal{Y}\times \mathcal{Z})}{\stackrel{..}{\subseteq}}\nonumber
\end{eqntabular}
\begin{toappendix}
\begin{proof}[Proof of (\ref{eq:GC:alpha-1})]
We write $\mathrel{{\subseteq}/{=}}$ to either stand for $\subseteq$ (in the Galois connection proof) or $=$ (in the isomorphism proof), everywhere in the proof.
\begin{calculus}[$\Leftrightarrow$\ ]
\hyphen{5} \formula{\alpha^{-1}(r)\mathrel{{\subseteq}/{=}} r'}\\
$\Leftrightarrow$
\formulaexplanation{r^{-1}\mathrel{{\subseteq}/{=}} r'}{def\@. (\ref{eq:def:alpha-1}) of $\alpha^{-1}$}\\
$\Leftrightarrow$
\formulaexplanation{(r^{-1})^{-1}\mathrel{{\subseteq}/{=}} r'^{-1}}{def\@. inverse ${}^{-1}$}\\
$\Leftrightarrow$
\formulaexplanation{r\mathrel{{\subseteq}/{=}} r'^{-1}}{def\@. inverse ${}^{-1}$}\\
$\Leftrightarrow$
\formulaexplanation{r\mathrel{{\subseteq}/{=}} \alpha^{-1}(r')}{def\@. (\ref{eq:def:alpha-1}) of $\alpha^{-1}$}\\[1ex]
\hyphen{5} \formula{\mathord{\mathord{\stackrel{.}{\alpha}^{-1}}(f)}\mathrel{{\stackrel{.}{\subseteq}}/{=}}g}\\
$\Leftrightarrow$
\formulaexplanation{\forall x\mathrel{.}\mathord{\mathord{\stackrel{.}{\alpha}^{-1}}(f)}(x)\mathrel{{\subseteq}/{=}} g(x)}{pointwise def\@. $\stackrel{.}{\subseteq}$}\\
$\Leftrightarrow$ 
\formulaexplanation{\forall x\mathrel{.}\alpha^{-1}(f(\alpha^{-1}(x)))\mathrel{{\subseteq}/{=}} g(x)}{def\@. (\ref{eq:def:alpha-1}) of $\mathord{\stackrel{.}{\alpha}^{-1}}$}\\
$\Leftrightarrow$ 
\formula{\forall y\mathrel{.}\alpha^{-1}(f(\alpha^{-1}(\alpha^{-1}(y))))\mathrel{{\subseteq}/{=}} g(\alpha^{-1}(y))}\\[-0.5ex]
\explanation{($\subseteq$)\quad letting $x=\alpha^{-1}(y)$\\
($\supseteq$)\quad $y=\alpha^{-1}(x)$, Galois isomorphism (\ref{eq:GC:alpha-1}) for $\alpha^{-1}$}\\[0.5ex]
$\Leftrightarrow$
\formulaexplanation{\forall y\mathrel{.}f(y)\mathrel{{\subseteq}/{=}}\alpha^{-1}(g(\alpha^{-1}(y)))}{Galois isomorphism (\ref{eq:GC:alpha-1}) for $\alpha^{-1}$}\\
$\Leftrightarrow$
\formulaexplanation{\forall x\mathrel{.}f(x)\mathrel{{\subseteq}/{=}}\mathord{\stackrel{.}{\alpha}^{-1}}(g)(x)}{pointwise def\@. (\ref{eq:def:alpha-1}) of $\mathord{\stackrel{.}{\alpha}^{-1}}$}\\
$\Leftrightarrow$
\formulaexplanation{f\mathrel{{\stackrel{.}{\subseteq}}/{=}}\mathord{\stackrel{.}{\alpha}}^{-1}(g)}{pointwise def\@. $\stackrel{.}{\subseteq}$}\\[1ex]
\hyphen{5} \formula{\vec\alpha^{-1}(T)\mathrel{{\stackrel{..}{\subseteq}}/{=}}T'}\\
$\Leftrightarrow$
\formulaexplanation{\mathord{\stackrel{.}{\alpha}^{-1}}\comp T\comp\alpha^{-1}\mathrel{{\stackrel{..}{\subseteq}}/{=}}T'}{def\@. (\ref{eq:GC:alpha-1}) of $\vec\alpha^{-1}$}\\ 
$\Leftrightarrow$
\formulaexplanation{\forall r\mathrel{.}\mathord{\stackrel{.}{\alpha}^{-1}}( T(\alpha^{-1}(r)))\mathrel{{\stackrel{.}{\subseteq}}/{=}}T'(r)}{pointwise def\@. ${\stackrel{..}{\subseteq}}$ and def\@. function composition $\comp$}\\
$\Leftrightarrow$
\formulaexplanation{\forall r'\mathrel{.}\mathord{\stackrel{.}{\alpha}^{-1}}( T(\alpha^{-1}(\alpha^{-1}(r'))))\mathrel{{\stackrel{.}{\subseteq}}/{=}}T'(\alpha^{-1}(r'))}{letting $r=\alpha^{-1}(r')$}\\
$\Leftrightarrow$
\formulaexplanation{\forall r'\mathrel{.}\mathord{\stackrel{.}{\alpha}^{-1}}( T(r'))\mathrel{{\stackrel{.}{\subseteq}}/{=}}T'(\alpha^{-1}(r'))}{Galois isomorphism (\ref{eq:GC:alpha-1}) for $\alpha^{-1}$}\\
$\Leftrightarrow$
\formulaexplanation{\forall r'\mathrel{.}\forall Q\mathrel{.}\mathord{\stackrel{.}{\alpha}^{-1}}( T(r'))Q\mathrel{{\subseteq}/{=}}T'(\alpha^{-1}(r'))Q}{pointwise def\@. ${{\stackrel{.}{\subseteq}}/{=}}$}\\
$\Leftrightarrow$
\formulaexplanation{\forall r'\mathrel{.}\forall Q\mathrel{.}\mathord{{\alpha}^{-1}}(T(r')Q)\mathrel{{\subseteq}/{=}}T'(\alpha^{-1}(r'))Q}{pointwise def\@. (\ref{eq:def:alpha-1}) of $\mathord{\stackrel{.}{\alpha}^{-1}}$}\\
$\Leftrightarrow$
\formulaexplanation{\forall r'\mathrel{.}\forall Q\mathrel{.}T(r')Q\mathrel{{\subseteq}/{=}}\mathord{{\alpha}^{-1}}(T'(\alpha^{-1}(r'))Q)}{Galois isomorphism (\ref{eq:def:alpha-1}) for $\mathord{{\alpha}^{-1}}$}\\
$\Leftrightarrow$
\formulaexplanation{\forall r'\mathrel{.}\forall Q\mathrel{.}T(r')Q\mathrel{{\subseteq}/{=}}\mathord{\stackrel{.}{\alpha}^{-1}}(T'(\alpha^{-1}(r')))Q}{def\@. (\ref{eq:def:alpha-1}) of $\mathord{\stackrel{.}{\alpha}^{-1}}$}\\
$\Leftrightarrow$
\formulaexplanation{\forall r'\mathrel{.}T(r')\mathrel{{\stackrel{.}{\subseteq}}/{=}}\mathord{\stackrel{.}{\alpha}^{-1}}(T'(\alpha^{-1}(r')))Q}{pointwise def\@. $\stackrel{.}{\subseteq}$}\\
$\Leftrightarrow$
\formulaexplanation{T\mathrel{{\stackrel{..}{\subseteq}}/{=}}\mathord{\stackrel{.}{\alpha}^{-1}}\comp T'\comp\alpha^{-1}}{def\@. fonction composition $\comp$ and pointwise def\@. $\stackrel{..}{\subseteq}$}\\
$\Leftrightarrow$
\lastformulaexplanation{T\mathrel{{\stackrel{..}{\subseteq}}/{=}} \vec\alpha^{-1}(T')}{def\@. (\ref{eq:GC:alpha-1}) of $\vec\alpha^{-1}$}{\mbox{\qed}}
\end{calculus}\let\qed\relax
\end{proof}
\end{toappendix}
\noindent Using these Galois isomorphisms (\ref{eq:GC:alpha-1}), we define the precondition transformer \proofinapx
\begin{eqntabular}{rclcl}
\textsf{\upshape Pre}&\triangleq&\vec{\alpha}^{-1}(\textsf{\upshape Post})
&=&\LAMBDA{r}\LAMBDA{Q}\{\pair{\sigma}{\sigma_{\mskip-2muf}}\mid\exists\sigma'\mathrel{.}\pair{\sigma}{\sigma'}\in r\wedge \pair{\sigma'}{\sigma_{\mskip-2muf}}\in Q)\}
\label{eq:def:Pre}
\end{eqntabular}
so that $\textsf{\upshape Pre}(r)Q$ is the set of initial states ${\sigma}$ related to ${\sigma_{\mskip-2muf}}$ from which it is possible to reach a final state ${\sigma'}$ related to ${\sigma_{\mskip-2muf}}$ satisfying the consequent $Q$  through a transition by $r$.
\begin{toappendix}
\begin{proof}[Proof of (\ref{eq:def:Pre})]
\begin{calculus}
\formula{\vec{\alpha}^{-1}(\textsf{\upshape Post})}\\
=
\formulaexplanation{\mathord{\stackrel{.}{\alpha}^{-1}}\comp\textsf{\upshape Post}\comp\alpha^{-1}}{def\@. (\ref{eq:GC:alpha-1}) of $\vec\alpha^{-1}$}\\
=
\formulaexplanation{\LAMBDA{r}\mathord{\stackrel{.}{\alpha}^{-1}}(\textsf{\upshape Post}(\alpha^{-1}(r)))}{def\@.  function composition $\comp$}\\
=
\formulaexplanation{\LAMBDA{r}\mathord{\stackrel{.}{\alpha}^{-1}}(\textsf{\upshape Post}(r^{-1}))}{def\@. (\ref{eq:def:alpha-1}) of $\alpha^{-1}$}\\
=
\formulaexplanation{\LAMBDA{r}\mathord{\stackrel{.}{\alpha}^{-1}}(\LAMBDA{P}\{\pair{\sigma_0}{\sigma'}\mid\exists\sigma\mathrel{.}\pair{\sigma_0}{\sigma}\in P\wedge\pair{\sigma}{\sigma'}\in r^{-1}\})}{def\@. (\ref{eq:def:Post}) of $\textsf{\upshape Post}$}\\
=
\formulaexplanation{\LAMBDA{r}{\LAMBDA{P}\alpha}^{-1}(\{\pair{\sigma_0}{\sigma'}\mid\exists\sigma\mathrel{.}\pair{\sigma}{\sigma_0}\in \alpha^{-1}(P)\wedge\pair{\sigma}{\sigma'}\in r^{-1}\})}{pointwise def\@. of $\mathord{\stackrel{.}{\alpha}^{-1}}$, def\@. $\alpha^{-1}$}\\
= 
\formulaexplanation{\LAMBDA{r}\LAMBDA{P}\alpha^{-1}(\{\pair{\sigma_0}{\sigma'}\mid\exists\sigma\mathrel{.}\pair{\sigma}{\sigma_0}\in \alpha^{-1}(P)\wedge\pair{\sigma'}{\sigma}\in r)\})}{def\@. inverse $r^{-1}$}\\
=
\formulaexplanation{\LAMBDA{r}\LAMBDA{P}\alpha^{-1}(\{\pair{\sigma_0}{\sigma'}\mid\exists\sigma\mathrel{.}\pair{\sigma'}{\sigma}\in r\wedge \pair{\sigma}{\sigma_0}\in \alpha^{-1}(P))\})}{$\wedge$ commutative}\\
=
\formulaexplanation{\LAMBDA{r}\LAMBDA{P}\alpha^{-1}(\{\pair{\sigma_f}{\sigma'}\mid\exists\sigma\mathrel{.}\pair{\sigma'}{\sigma}\in r\wedge \pair{\sigma}{\sigma_f}\in \alpha^{-1}(P))\})}{renaming $\sigma_0$ into $\sigma_f$}\\
=
\formula{\LAMBDA{r}\LAMBDA{P}\alpha^{-1}(\{\pair{\sigma_f}{\sigma}\mid\exists\sigma'\mathrel{.}\pair{\sigma}{\sigma'}\in r\wedge \pair{\sigma'}{\sigma_f}\in \alpha^{-1}(P))\})}\\[-0.5ex]\rightexplanation{renaming $\sigma$ into $\sigma'$ and vice-versa}\\[-0.5ex]
=
\formulaexplanation{\LAMBDA{r}\LAMBDA{Q}\{\pair{\sigma}{\sigma_{\mskip-2muf}}\mid\exists\sigma'\mathrel{.}\pair{\sigma}{\sigma'}\in r\wedge \pair{\sigma'}{\sigma_{\mskip-2muf}}\in Q)\}}{def\@. (\ref{eq:def:alpha-1}) of $\alpha^{-1}$}\\
=
\lastformulaexplanation{\textsf{\upshape Pre}}{def\@. (\ref{eq:def:Pre}) of \textsf{\upshape Pre}}{\mbox{\qed}}
\end{calculus}\let\qed\relax
\end{proof}
\end{toappendix}

\subsubsection{The Complement Abstraction}
The complement abstraction is useful to express that a program property does not hold (e.g\@. to contradict a Hoare triple).

Let $\mathcal{X}$ be a set and $X\in\wp(\mathcal{X})$. The complement abstraction is $\alpha^{\neg}(X)\triangleq\neg X$ (where $\neg X$ $\triangleq$ 
$\mathcal{X}\setminus X$ when $X\in\wp(\mathcal{X})$). We have the Galois isomorphisms
\bgroup\arraycolsep2pt%
\begin{eqntabular}{c@{\qquad and\qquad}c}
\pair{\wp(\mathcal{X})}{\subseteq}\GaloiS{\alpha^{\neg}}{\alpha^{\neg}}\pair{\wp(\mathcal{X})}{\supseteq}
&\pair{\wp(\mathcal{X})}{\supseteq}\GaloiS{\alpha^{\neg}}{\alpha^{\neg}}\pair{\wp(\mathcal{X})}{\subseteq}
\label{eq-complement-GC}
\end{eqntabular}
\egroup
(which follow from $X\subseteq Y\Leftrightarrow \neg X\supseteq \neg Y$ and  $\neg \neg X = X$ and implies De Morgan laws $\alpha^{\neg}(\bigcup X)=\bigcap \alpha^{\neg}(X)$ and $\alpha^{\neg}(\bigcap X)=\bigcup \alpha^{\neg}(X)$ since, in a Galois connection, $\alpha$ preserves existing joins and $\gamma$ preserves existing meets).

\subsubsection{The Emptiness and Non-Emptiness Abstraction}
Negation is sometimes equivalent to an emptiness or non-emptiness check. For example, $\neg(A\subseteq B)\Leftrightarrow A\cap \neg B\neq\emptyset$. These are abstractions.
\bgroup\arraycolsep0.28\arraycolsep\begin{eqntabular}[fl]{@{}rcl@{\quad}rcl@{}}
\rlap{Emptiness}&&&\rlap{Non-emptiness}\label{eq:non-emptiness-abstraction}
\\
\alpha^{\accentset{\rightarrow}{\emptyset}}(\tau) &\triangleq& \{\pair{P}{Q}\mid Q\cap\tau(P)=\emptyset\}
&
\alpha^{\accentset{\rightarrow}{\otimes}}(\tau)  &\triangleq& 
\alpha^{\neg}\comp\alpha^{\accentset{\rightarrow}{\emptyset}}(\tau) 
\colsep{=}
\{\pair{P}{Q}\mid Q\cap\tau(P)\neq\emptyset\}
\nonumber\\
\alpha^{\accentset{\leftarrow}{\emptyset}}(\tau)  &\triangleq&  \alpha^{-1}(\alpha^{\accentset{\rightarrow}{\emptyset}}(\tau)) \colsep{=}\{\pair{P}{Q}\mid P\cap \tau(Q)=\emptyset\}
&
\alpha^{\accentset{\leftarrow}{\otimes}}(\tau) &\triangleq& \alpha^{-1}(\alpha^{\accentset{\rightarrow}{\otimes}}(\tau)) \colsep{=}\{\pair{P}{Q}\mid P\cap\tau(Q)\neq\emptyset\}\nonumber
\end{eqntabular}
\egroup
We have \proofinapx
\bgroup\abovedisplayskip-10pt\belowdisplayskip0pt
\begin{eqntabular}{c}
\pair{\wp(\mathcal{X})\rightarrow\wp(\mathcal{Y})}{\stackrel{.}{\subseteq}}
\galois{\alpha^{\accentset{\rightarrow}{\emptyset}}}{\gamma^{\accentset{\rightarrow}{\emptyset}}}
\pair{\wp(\mathcal{X}\times\mathcal{Y})}{\supseteq}
\label{eq:GC:Emptiness}
\end{eqntabular}\egroup
and similarly Galois connections for the other cases $\alpha^{\accentset{\leftarrow}{\emptyset}}$,  $\alpha^{\accentset{\rightarrow}{\otimes}}$, and $\alpha^{\accentset{\leftarrow}{\otimes}}$.
\begin{toappendix}
\begin{proof}[Proof of (\ref{eq:GC:Emptiness})]
\begin{calculus}[$\Leftrightarrow$~]
\hyphen{5}\formula{\alpha^{\accentset{\rightarrow}{\emptyset}}(\mathop{\stackrel{.}{\bigcup}}\limits_{i\in\Delta}\tau_i)}\\
=
\formula{\{\pair{P}{Q}\mid (\mathop{\stackrel{.}{\bigcup}}\limits_{i\in\Delta}\tau_i)(P)\subseteq \neg Q\}}\\[-0.5ex]\rightexplanation{def\@. (\ref{eq:non-emptiness-abstraction}) of $\alpha^{\accentset{\rightarrow}{\emptyset}}(\tau)$ = $\{\pair{P}{Q}\mid Q\cap\tau(P)=\emptyset\}$ = $\{\pair{P}{Q}\mid\tau(P)\subseteq\neg Q\}$}\\
=
\formulaexplanation{\{\pair{P}{Q}\mid (\mathop{\stackrel{.}{\bigcup}}\limits_{i\in\Delta}\tau_i(P))\subseteq \neg Q\}}{pointwise def\@. $\stackrel{.}{\bigcup}$}\\
=
\formulaexplanation{\{\pair{P}{Q}\mid \forall i\in\Delta\mathrel{.}\tau_i(P)\subseteq \neg Q\}}{pointwise def\@. ${\bigcup}$ and $\subseteq$}\\
=
\formulaexplanation{\bigcap_{i\in\Delta}\{\pair{P}{Q}\mid \tau_i(P)\subseteq \neg Q\}}{def\@. $\bigcap$}\\
=
\formulaexplanation{\bigcap_{i\in\Delta}\alpha^{\accentset{\rightarrow}{\emptyset}}(\tau_i)}{def\@. (\ref{eq:non-emptiness-abstraction}) of $\alpha^{\accentset{\rightarrow}{\emptyset}}$}
\end{calculus}
\noindent\rlap{\hyphen{5}}\phantom{$\Leftrightarrow$~} The other $\alpha^{\accentset{\leftarrow}{\emptyset}}$, $\alpha^{\accentset{\leftarrow}{\emptyset}}$, and $\alpha^{\accentset{\leftarrow}{\otimes}}$ follow from proposition \ref{prop:GC-composition} by composition of the Galois connection (\ref{eq:GC:Emptiness}) for $\alpha^{\accentset{\rightarrow}{\emptyset}}$ and those
(\ref{eq-complement-GC}) for ${\alpha^{\neg}}$ and (\ref{eq:def:alpha-gamma-d1}) for $\alpha^{-1}$.
\end{proof}
\end{toappendix}

\subsubsection{The Complement Dual Abstractions}

Pratt's ``Duality Principle for Programs'' \cite[section 1.2]{DBLP:conf/focs/Pratt76}, is similar the complement duality in \href{https://en.wikipedia.org/wiki/Classical_logic}{classical logic} i.e\@. something not false is true. 

This
can stated for functions $f$ by defining the complement dual abstraction $\alpha^{\sim}$ of functions and its pointwise extension
$\mathord{\stackrel{.}{\alpha}}^{\sim}$ below, which yields the Galois connections as follows \proofinapx
\bgroup
\begin{eqntabular}[fl]{l@{\qquad\qquad\qquad\qquad}c}
&\alpha^{\sim}(f)\colsep{\triangleq}\neg\comp f\comp\neg\qquad\qquad\mathord{\stackrel{.}{\alpha}}^{\sim}(F)\colsep{\triangleq}\LAMBDA{x}\alpha^{\sim}(F(x))\label{eq:eq-function-complement-GC}\\
\rlap{with connections \proofinapx}&\pair{\wp(\mathcal{X})\rightarrow\wp(\mathcal{Y})}{\stackrel{.}{\subseteq}}\GaloiS{\alpha^{\sim}}{\alpha^{\sim}}\pair{\wp(\mathcal{X})\rightarrow\wp(\mathcal{Y})}{\stackrel{.}{\supseteq}}
\nonumber\\
&\pair{\wp(\mathcal{Z})\rightarrow\wp(\mathcal{X})\rightarrow\wp(\mathcal{Y})}{\stackrel{..}{\subseteq}}\GaloiS{\mathord{\stackrel{.}{\alpha}}^{\sim}}{\mathord{\stackrel{.}{\alpha}}^{\sim}}\pair{\wp(\mathcal{Z})\rightarrow\wp(\mathcal{X})\rightarrow\wp(\mathcal{Y})}{\stackrel{..}{\supseteq}}\nonumber
\end{eqntabular}\egroup
where $\stackrel{.}{\subseteq}$ is the pointwise extension of $\subseteq$, that is, $f\stackrel{.}{\subseteq}g\Leftrightarrow\forall X\in \mathcal{X}\mathrel{.}f(X)\subseteq g(X)$, $\stackrel{..}{\subseteq}$ is the pointwise extension of  $\stackrel{.}{\subseteq}$, etc.  
\begin{toappendix}
\begin{proof}[Proof of (\ref{eq:eq-function-complement-GC})]
\begin{calculus}[$\Leftrightarrow$\ ]
\hyphen{5}\formula{\alpha^{\sim}(f)\mathrel{{\stackrel{.}{\subseteq}}/{=}} g}\\
$\Leftrightarrow$
\formulaexplanation{\forall X\in\wp(\mathcal{X})\mathrel{.}\alpha^{\sim}(f)X \mathrel{{\subseteq}/{=}} g(X)}{pointwise def\@. ${\stackrel{.}{\subseteq}}$ or $=$}\\
$\Leftrightarrow$
\formulaexplanation{\forall X\in\wp(\mathcal{X})\mathrel{.}\neg\comp f\comp\neg(X) \mathrel{{\subseteq}/{=}} g(X)}{def\@. $\alpha^{\sim}$}\\
$\Leftrightarrow$
\formulaexplanation{\forall X\in\wp(\mathcal{X})\mathrel{.}f\comp\neg(X) \mathrel{{\supseteq}/{=}} \neg g(X)}{$A \subseteq B$ iff $\neg B \subseteq \neg A$ i.e\@. $\neg A \supseteq \neg B$ and $A=B$ $\Leftrightarrow$ $\neg A=\neg B$}\\
$\Leftrightarrow$
\formulaexplanation{\forall Y\in\wp(\mathcal{X})\mathrel{.}f(Y) \mathrel{{\supseteq}/{=}} \neg g(\neg Y)}{letting $X=\neg Y$ and $\neg\comp\neg$ is the identity}\\
$\Leftrightarrow$
\formulaexplanation{f\mathrel{{\stackrel{.}{\supseteq}}/{=}}  \neg \comp g \comp\neg}{pointwise def\@. $\stackrel{.}{\supseteq}$ and function composition $\comp$}\\
$\Leftrightarrow$
\lastformulaexplanation{f\mathrel{{\stackrel{.}{\supseteq}}/{=}}  \alpha^{\sim}(g)}{def\@. $\alpha^{\sim}$}{\mbox{\qed}}
\end{calculus}\let\qed\relax
\end{proof}
\end{toappendix}

\noindent Using this Galois connection (\ref{eq:eq-function-complement-GC}), we define the dual complement transformers \proofinapx
\begin{eqntabular}{rclcl}
\widetilde{\textsf{\upshape Post}}&\triangleq&\mathord{\stackrel{.}{\alpha}}^{\sim}(\textsf{\upshape Post})&=&\LAMBDA{r}\LAMBDA{P}\{\pair{\sigma_0}{\sigma'}\mid\forall \sigma \mathrel{.}\pair{\sigma}{\sigma'}\in r\Rightarrow\pair{\sigma_0}{\sigma}\in P\}
\label{eq:def:Postt-Pret}\\
\widetilde{\textsf{\upshape Pre}}&\triangleq&\mathord{\stackrel{.}{\alpha}}^{\sim}(\textsf{\upshape Pre})&=&\LAMBDA{r}\LAMBDA{Q}\{\pair{\sigma}{\sigma_{\mskip-2muf}}\mid\forall \sigma' \mathrel{.}\pair{\sigma}{\sigma'}\in  r \Rightarrow\pair{\sigma'}{\sigma_{\mskip-2muf}}\in Q\}
\nonumber
\end{eqntabular}
\begin{toappendix}
\begin{proof}[Proof of (\ref{eq:def:Postt-Pret})]
\begin{calculus}
\hyphen{6}\formula{\widetilde{\textsf{\upshape Post}}}\\
=
\formulaexplanation{\mathord{\stackrel{.}{\alpha}}^{\sim}(\textsf{\upshape Post})}{def\@. (\ref{eq:def:Postt-Pret}) of ${\widetilde{\textsf{\upshape Post}}}$}\\
=
\formulaexplanation{\LAMBDA{r}\alpha^{\sim}(\textsf{\upshape Post}(r))}{pointwise def\@. $\mathord{\stackrel{.}{\alpha}}^{\sim}$}\\
=
\formulaexplanation{\LAMBDA{r}\neg\comp \textsf{\upshape Post}(r)\comp\neg}{def\@. $\alpha^{\sim}$}\\
=
\formulaexplanation{\LAMBDA{r}\LAMBDA{P}\neg(\textsf{\upshape Post}(r)(\neg P))}{def\@. function composition $\comp$}\\
=
\formulaexplanation{\LAMBDA{r}\LAMBDA{P}\neg\{\pair{\sigma_0}{\sigma'}\mid\exists\sigma\mathrel{.}\pair{\sigma_0}{\sigma}\in (\neg P)\wedge\pair{\sigma}{\sigma'}\in r\}
}{def\@. (\ref{eq:def:Post}) of $\textsf{\upshape Post}$}\\
=
\formulaexplanation{\LAMBDA{r}\LAMBDA{P}\{\pair{\sigma_0}{\sigma'}\mid\forall \sigma \mathrel{.}\pair{\sigma_0}{\sigma}\not\in (\neg P)\vee\pair{\sigma}{\sigma'}\not\in r\}}{def\@.negation  $\neg$}\\
=
\formulaexplanation{\LAMBDA{r}\LAMBDA{P}\{\pair{\sigma_0}{\sigma'}\mid\forall \sigma \mathrel{.}\pair{\sigma}{\sigma'}\in r\Rightarrow\pair{\sigma_0}{\sigma}\in P\}}{def\@. implication $\Rightarrow$ and $\neg$}\\[1ex]
\hyphen{6}\formula{\widetilde{\textsf{\upshape Pre}}}\\
=
\formulaexplanation{\mathord{\stackrel{.}{\alpha}}^{\sim}(\textsf{\upshape Pre})}{def\@. (\ref{eq:def:Postt-Pret}) of ${\widetilde{\textsf{\upshape Pre}}}$}\\
=
\formulaexplanation{\LAMBDA{r}\alpha^{\sim}(\textsf{\upshape Pre}(r))}{pointwise def\@. $\mathord{\stackrel{.}{\alpha}}^{\sim}$}\\
=
\formulaexplanation{\LAMBDA{r}\neg\comp \textsf{\upshape Pre}(r)\comp\neg}{def\@. $\alpha^{\sim}$}\\
=
\formulaexplanation{\LAMBDA{r}\LAMBDA{Q}\neg(\textsf{\upshape Pre}(r)(\neg Q))}{def\@. function composition $\comp$}\\
=
\formulaexplanation{\LAMBDA{r}\LAMBDA{Q}\neg\{\pair{\sigma}{\sigma_{\mskip-2muf}}\mid\exists\sigma'\mathrel{.}\pair{\sigma}{\sigma'}\in r\wedge \pair{\sigma'}{\sigma_{\mskip-2muf}}\in (\neg Q))\}}{def\@. (\ref{eq:def:Pre}) of $\textsf{\upshape Pre}$}\\
=
\formulaexplanation{\LAMBDA{r}\LAMBDA{Q}\{\pair{\sigma}{\sigma_{\mskip-2muf}}\mid \forall \sigma'\mathrel{.}\pair{\sigma}{\sigma'}\not\in  r \vee\pair{\sigma'}{\sigma_{\mskip-2muf}}\in Q\}}{def\@. negation  $\neg$}\\
=
\lastformulaexplanation{\LAMBDA{r}\LAMBDA{Q}\{\pair{\sigma}{\sigma_{\mskip-2muf}}\mid\forall \sigma' \mathrel{.}\pair{\sigma}{\sigma'}\in  r \Rightarrow\pair{\sigma'}{\sigma_{\mskip-2muf}}\in Q\}}{def\@. implication $\Rightarrow$}{\mbox{\qed}}
\end{calculus}\let\qed\relax
\end{proof}
\end{toappendix}

\noindent
If $r\in\wp(\mathcal{X}\times\mathcal{Y})$ then \proofinapx
\begin{eqntabular}[fl]{c@{\ \ \ \ }c@{}}
\pair{\wp(\mathcal{Z}\times\mathcal{X})}{\subseteq}\galois{\textsf{\upshape Post}(r)}{\mathord{\stackrel{.}{\alpha}^{-1}}(\widetilde{\textsf{\upshape Pre}}(r))}\pair{\wp(\mathcal{Z}\times\mathcal{Y})}{{\subseteq}}
&
\pair{\wp(\mathcal{X}\times\mathcal{Z})}{\subseteq}\galois{\textsf{\upshape Pre}(r)}{\mathord{\stackrel{.}{\alpha}^{-1}}(\widetilde{\textsf{\upshape Post}}(r))}\pair{\wp(\mathcal{Y}\times\mathcal{Z})}{\subseteq}\label{eq:Post-Pret:Pre-Postt-GC}
\end{eqntabular}
\begin{toappendix}
\begin{proof}[Proof of (\ref{eq:Post-Pret:Pre-Postt-GC})]
\begin{calculus}[$\Leftrightarrow$~]
\hyphen{6}\formulaexplanation{\textsf{\upshape Post}(r)P\subseteq Q}{$P,Q\in\wp(\mathcal{X}\times\mathcal{X})$}\\
$\Leftrightarrow$
\formulaexplanation{\{\pair{\sigma_0}{\sigma'}\mid\exists\sigma\mathrel{.}\pair{\sigma_0}{\sigma}\in P\wedge\pair{\sigma}{\sigma'}\in r\}\subseteq Q}{def\@. (\ref{eq:def:Post}) of $\textsf{\upshape Post}$}\\
$\Leftrightarrow$
\formulaexplanation{\forall \pair{\sigma_0}{\sigma'}\mathrel{.}(\exists\sigma\mathrel{.}\pair{\sigma_0}{\sigma}\in P\wedge\pair{\sigma}{\sigma'}\in r)\Rightarrow \pair{\sigma_0}{\sigma'}\in Q}{def\@. $\subseteq$}\\
$\Leftrightarrow$
\formulaexplanation{\forall {\sigma_0},{\sigma'},\sigma\mathrel{.}(\pair{\sigma_0}{\sigma}\in P\wedge\pair{\sigma}{\sigma'}\in r)\Rightarrow \pair{\sigma_0}{\sigma'}\in Q}{def\@. $\Rightarrow$ and $\exists/\forall$}\\
$\Leftrightarrow$
\formulaexplanation{\forall {\sigma_0},\sigma\mathrel{.}(\pair{\sigma_0}{\sigma}\in P\Rightarrow(\forall{\sigma'}\mathrel{.}\pair{\sigma}{\sigma'}\in r\Rightarrow \pair{\sigma_0}{\sigma'}\in Q)}{def\@. $\Rightarrow$ and $\wedge$}\\
$\Leftrightarrow$
\formulaexplanation{P\subseteq\{ \pair{\sigma_0}{\sigma}\mid \forall \sigma' \mathrel{.}\pair{\sigma}{\sigma'}\in  r \Rightarrow\pair{\sigma_0}{\sigma'}\in Q\}}{def\@. $\subseteq$}\\
$\Leftrightarrow$
\formulaexplanation{P\subseteq\{\pair{\sigma_{\mskip-2muf}}{\sigma}\mid\forall \sigma' \mathrel{.}\pair{\sigma}{\sigma'}\in  r \Rightarrow\pair{\sigma_{\mskip-2muf}}{\sigma'}\in  Q\}}{renaming}\\
$\Leftrightarrow$
\formulaexplanation{P\subseteq\alpha^{-1} (\{\pair{\sigma}{\sigma_{\mskip-2muf}}\mid\forall \sigma' \mathrel{.}\pair{\sigma}{\sigma'}\in  r \Rightarrow\pair{\sigma'}{\sigma_{\mskip-2muf}}\in \alpha^{-1}(Q)\})}{def\@. (\ref{eq:def:alpha-1}) of $\alpha^{-1}$}\\
$\Leftrightarrow$
\formula{P\subseteq\alpha^{-1}\comp (\LAMBDA{Q}\{\pair{\sigma}{\sigma_{\mskip-2muf}}\mid\forall \sigma' \mathrel{.}\pair{\sigma}{\sigma'}\in  r \Rightarrow\pair{\sigma'}{\sigma_{\mskip-2muf}}\in Q\})\comp\alpha^{-1}(Q)}\\[-0.5ex]\rightexplanation{def\@. function application and composition $\comp$}\\
$\Leftrightarrow$
\formulaexplanation{P\subseteq{\mathord{\stackrel{.}{\alpha}^{-1}}(\LAMBDA{Q}\{\pair{\sigma}{\sigma_{\mskip-2muf}}\mid\forall \sigma' \mathrel{.}\pair{\sigma}{\sigma'}\in  r \Rightarrow\pair{\sigma'}{\sigma_{\mskip-2muf}}\in Q\})}Q}{def\@. (\ref{eq:def:alpha-1}) of $\mathord{\stackrel{.}{\alpha}^{-1}}$}\\
$\Leftrightarrow$
\formulaexplanation{P\subseteq{\mathord{\stackrel{.}{\alpha}^{-1}}(\widetilde{\textsf{\upshape Pre}}(r))}Q}{def\@. (\ref{eq:def:Postt-Pret}) of $\widetilde{\textsf{\upshape Pre}}$}\\[1ex]
\hyphen{6}\formula{\textsf{\upshape Pre}(r)Q \subseteq P}\\
$\Leftrightarrow$
\formulaexplanation{\{\pair{\sigma}{\sigma_{\mskip-2muf}}\mid\exists\sigma'\mathrel{.}\pair{\sigma}{\sigma'}\in r\wedge \pair{\sigma'}{\sigma_{\mskip-2muf}}\in Q)\} \subseteq P}{def\@. (\ref{eq:def:Pre}) of $\textsf{\upshape Pre}$}\\
$\Leftrightarrow$
\formulaexplanation{\forall{\sigma},{\sigma_{\mskip-2muf}}\mathrel{.}(\exists\sigma'\mathrel{.}\pair{\sigma}{\sigma'}\in r\wedge \pair{\sigma'}{\sigma_{\mskip-2muf}}\in Q)\}) \Rightarrow \pair{\sigma}{\sigma_{\mskip-2muf}}\in P}{def\@. $\subseteq$}\\
$\Leftrightarrow$
\formulaexplanation{\forall{\sigma},{\sigma_{\mskip-2muf}},\sigma'\mathrel{.}(\pair{\sigma}{\sigma'}\in r\wedge \pair{\sigma'}{\sigma_{\mskip-2muf}}\in Q)\}) \Rightarrow \pair{\sigma}{\sigma_{\mskip-2muf}}\in P}{def\@. $\Rightarrow$ and $\forall/\exists$}\\
$\Leftrightarrow$
\formulaexplanation{\forall{\sigma},{\sigma_{\mskip-2muf}},\sigma'\mathrel{.}(\pair{\sigma'}{\sigma_{\mskip-2muf}}\in Q) \Rightarrow(\pair{\sigma}{\sigma'}\in r\Rightarrow \pair{\sigma}{\sigma_{\mskip-2muf}}\in P)}{def\@. $\Rightarrow$ and $\wedge$}\\
$\Leftrightarrow$
\formulaexplanation{\forall{\sigma_{\mskip-2muf}},\sigma'\mathrel{.}(\pair{\sigma'}{\sigma_{\mskip-2muf}}\in Q) \Rightarrow(\forall{\sigma}\mathrel{.}\pair{\sigma}{\sigma'}\in r\Rightarrow \pair{\sigma}{\sigma_{\mskip-2muf}}\in P)}{def\@. $\Rightarrow$ and $\forall$}\\
$\Leftrightarrow$
\formulaexplanation{Q\subseteq\{\pair{\sigma'}{\sigma_{\mskip-2muf}}\mid \forall{\sigma}\mathrel{.}\pair{\sigma}{\sigma'}\in r\Rightarrow \pair{\sigma}{\sigma_{\mskip-2muf}}\in P\}}{def\@. $\subseteq$}\\
$\Leftrightarrow$
\formulaexplanation{Q\subseteq\alpha^{-1}(\{\pair{\sigma_{\mskip-2muf}}{\sigma'}\mid \forall{\sigma}\mathrel{.}\pair{\sigma}{\sigma'}\in r\Rightarrow \pair{\sigma_{\mskip-2muf}}{\sigma}\in \alpha^{-1}(P)\})}{def\@. (\ref{eq:def:alpha-1}) of $\alpha^{-1}$}\\
$\Leftrightarrow$
\formulaexplanation{Q\subseteq\alpha^{-1}(\{\pair{\sigma_0}{\sigma'}\mid \forall{\sigma}\mathrel{.}\pair{\sigma}{\sigma'}\in r\Rightarrow \pair{\sigma_0}{\sigma}\in \alpha^{-1}(P)\})}{renaming}\\
$\Leftrightarrow$
\formulaexplanation{Q\subseteq\alpha^{-1}(\widetilde{\textsf{\upshape Post}}(r)(\alpha^{-1}(P)))}{def\@. (\ref{eq:def:Postt-Pret}) of $\widetilde{\textsf{\upshape Post}}$}\\
$\Leftrightarrow$
\formulaexplanation{Q\subseteq{\alpha}^{-1}\comp\widetilde{\textsf{\upshape Post}}(r)\comp\alpha^{-1}(P)}{def\@. function composition $\comp$}\\
$\Leftrightarrow$
\lastformulaexplanation{Q\subseteq\mathord{\stackrel{.}{\alpha}^{-1}}(\widetilde{\textsf{\upshape Post}}(r))P}{def\@. (\ref{eq:def:alpha-1}) of $\mathord{\stackrel{.}{\alpha}^{-1}}$}{\mbox{\qed}}
\end{calculus}\let\qed\relax
\end{proof}
\end{toappendix}

\subsection{The  Hierarchical Taxonomy  of Forward and Backward Transformational Logics}
The composition of abstractions applied to $\textsf{\upshape Post}\sqb{\texttt{\small S}}_\bot$ of Fig\@.  \ref{fig:Forward-semantics-logics} 
can also be applied to $\widetilde{\textsf{\upshape Post}}\sqb{\texttt{\small S}}_\bot$, $\textsf{\upshape Pre}\sqb{\texttt{\small S}}_\bot$, and $\widetilde{\textsf{\upshape Pre}}\sqb{\texttt{\small S}}_\bot$ to get Fig\@.  \ref{fig:taxonomy}. Fig\@.  \ref{fig:Forward-semantics-logics} can be recognized at the bottom right of Fig\@.  \ref{fig:taxonomy}.
\begin{figure}[ht]
\includegraphics[width=0.925\textwidth]{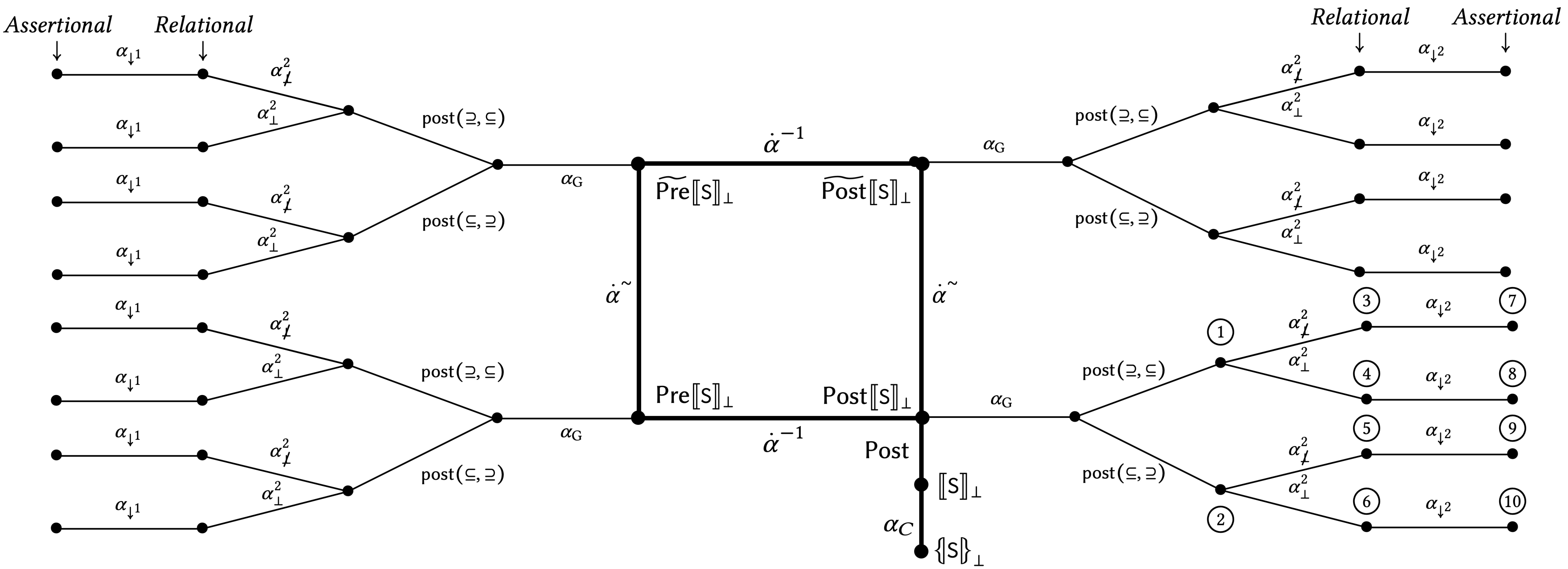}
\vskip-1em\caption{Hierarchical taxonomy of transformational logics\label{fig:taxonomy}}
\end{figure}
We get 40 transformational logics with 40 different proof systems which understanding is reduced to the composition of 9 abstractions (plus 2 to get $\textsf{\upshape Post}\sqb{\texttt{\small S}}_\bot$ by abstraction of the collecting semantics). Adding the negation abstraction (\ref{eq-complement-GC}), we obtain 40 more logics to disprove program properties (see sections \ref{PartialpossibleAccessibilityOfSomeFinalStateFromSomeInitialState} to 
\ref {PartialpossibleAccessibilityOfSomeFinalStateFromSomeInitialState} and 
\ref{PartialDefiniteInaccessibilityOfAllFinalStatesFromAllNoninitialStates} to \ref{TotalDefiniteAccessibilityOfSomeFinalStateFromSomeInitialState} for assertional logics), 160 logics with symbolic inversion in Sect\@. \ref{sec:SymbolicInversion}, etc.

\subsection{Abstraction for Assertional Logics}
Theories of forward assertional logics in Fig\@.  \ref{fig:Forward-semantics-logics} are abstractions of theories of relational logics by $\alpha_{\downarrow^2}$ (and backward ones by $\alpha_{\downarrow^1}$). The more classic view \cite{DBLP:conf/focs/Pratt76} and recent followers a.o\@. \cite{DBLP:conf/sefm/VriesK11,DBLP:journals/pacmpl/OHearn20,DBLP:journals/pacmpl/ZhangK22} directly abstract the program semantics by the assertional transformer
\textsf{\textup{post}} (\ref{eq:def:post}) which is the abstraction of the relational transformer \textsf{\textup{post}} (\ref{eq:def:Post}), as follows 
\bgroup\arraycolsep0.33\arraycolsep\abovedisplayskip0.5\abovedisplayskip\belowdisplayskip0.5\belowdisplayskip
\begin{eqntabular}[fl]{c}
\alpha^A_2(\theta)\colsep{\triangleq}\alpha_{\downarrow^2}\comp \theta\comp \gamma_{\downarrow^2}
\quad
\gamma^A_2\colsep{\triangleq}\LAMBDA{\theta}\gamma_{\downarrow^2}\comp \theta\comp \alpha_{\downarrow^2}(Q)
\quad
\mathord{\stackrel{.}{\alpha}^A_2}(\Theta) \colsep{\triangleq}\LAMBDA{r}\alpha^A_2(\Theta(r))
\quad
\mathord{\stackrel{.}{\gamma}}^A_2\colsep{\triangleq}\LAMBDA{r}\gamma^A_2(r)
\label{eq:abstraction:Post:post}
\\
\alpha^A_1(\theta)\colsep{\triangleq}\alpha_{\downarrow^1}\comp \theta\comp \gamma_{\downarrow^1}
\quad
\gamma^A_1\colsep{\triangleq}\LAMBDA{\theta}\gamma_{\downarrow^1}\comp \theta\comp \alpha_{\downarrow^1}(Q)
\quad
\mathord{\stackrel{.}{\alpha}^A_1}(\Theta) \colsep{\triangleq}\LAMBDA{r}\alpha^A_1(\Theta(r))
\quad
\mathord{\stackrel{.}{\gamma}}^A_1\colsep{\triangleq}\LAMBDA{r}\gamma^A_1(r)
\nonumber
\end{eqntabular}\egroup
we have Galois connections\footnote{We write $\mathrel{\smash{\stackrel{i}{\longrightarrow}}}$,  $\stackrel{c}{\longrightarrow}$,  $\stackrel{\sqcup}{\longrightarrow}$,  $\stackrel{\sqcupdot}{\longrightarrow}$,  $\stackrel{\sqcap}{\longrightarrow}$, and $\stackrel{\sqcapdot}{\longrightarrow}$ respectively for increasing, continuous, non-empty join, arbitrary join (including empty), non-empty meet, and arbitrary meet preserving functions.} \proofinapx
\bgroup\abovedisplayskip0.25\abovedisplayskip\belowdisplayskip0.5\belowdisplayskip\begin{eqntabular}{c}
\pair{\wp(\mathcal{Z}\times\mathcal{X})\mathrel{\smash{\stackrel{i}{\longrightarrow}}}\wp(\mathcal{Z}\times\mathcal{Y})}{\stackrel{.}{{\subseteq}}}
\galois{\alpha^A_2}{\gamma^A_2}
\pair{\wp(\mathcal{X})\mathrel{\smash{\stackrel{i}{\longrightarrow}}}\wp(\mathcal{Y})}{\stackrel{.}{{\subseteq}}}
\label{eq:GC:alpha-gamma-A}\\
\pair{\wp(\mathcal{X}\times\mathcal{Y})\rightarrow\wp(\mathcal{Z}\times\mathcal{X})\mathrel{\smash{\stackrel{i}{\longrightarrow}}}\wp(\mathcal{Z}\times\mathcal{Y})}{\stackrel{..}{{\subseteq}}}
\galois{\mathord{\stackrel{.}{\alpha}}^A_2}{\mathord{\stackrel{.}{\gamma}}^A_2}
\pair{\wp(\mathcal{X}\times\mathcal{Y})\rightarrow\wp(\mathcal{X})\mathrel{\smash{\stackrel{i}{\longrightarrow}}}\wp(\mathcal{Y})}{\stackrel{..}{{\subseteq}}}\nonumber
\end{eqntabular}\egroup
\begin{toappendix}
\begin{proof}[Proof of (\ref{eq:GC:alpha-gamma-A})]This is an application of \cite[Theorem 11.78]{Cousot-PAI-2021} and \cite[Exercise 11.21]{Cousot-PAI-2021}. We provide the proof for the appendix to be self-contained.
\begin{calculus}[$\Leftrightarrow$\ ]
\hyphen{5} \formulaexplanation{\alpha^A_2(\Theta)\stackrel{.}{{\subseteq}} \theta}{where $\Theta\in\wp(\mathcal{Z}\times\mathcal{X})\mathrel{\smash{\stackrel{i}{\longrightarrow}}}\wp(\mathcal{Z}\times\mathcal{Y})$
and $\theta\in\wp(\mathcal{X})\mathrel{\smash{\stackrel{i}{\longrightarrow}}}\wp(\mathcal{Y})$}\\
$\Leftrightarrow$
\formulaexplanation{\alpha^A_2(\Theta)P\subseteq\theta(P)}{pointwise def\@. $\stackrel{.}{{\subseteq}}$}\\
$\Leftrightarrow$
\formulaexplanation{\alpha_{\downarrow^2}\comp\Theta\comp \gamma_{\downarrow^2}(P)\subseteq\theta(P)}{def\@. (\ref{eq:def:alpha-gamma-d2}) of $\alpha^A_2$}\\
$\Leftrightarrow$
\formulaexplanation{ \Theta\comp \gamma_{\downarrow^2}(P)\subseteq\gamma_{\downarrow^2}\comp \theta(P)}{Galois connection (\ref{eq:alpha:relationbal-assertional:GC})}\\
$\Leftrightarrow$
\formula{ \Theta(Q)\subseteq\gamma_{\downarrow^2}\comp \theta\comp \alpha_{\downarrow^2}(Q)}\\

\explanation{($\Rightarrow$)\quad choosing $P=\alpha_{\downarrow^2}(Q)$, ${\gamma_{\downarrow^2}}\comp\alpha_{\downarrow^2}$ is extensive (by the Galois connection (\ref{eq:alpha:relationbal-assertional:GC}) for $\alpha_{\downarrow^2}$), and $\Theta$ is increasing by hypothesis;\\
($\Leftarrow$)\quad choosing $Q=\gamma_{\downarrow^2}(P)$, $\alpha_{\downarrow^2}\comp\gamma_{\downarrow^2}$ is reductive (by the Galois connection (\ref{eq:alpha:relationbal-assertional:GC}) for $\gamma_{\downarrow^2}$), and $\theta$ is increasing by hypothesis}\\

$\Leftrightarrow$
\formulaexplanation{ \Theta\stackrel{.}{\subseteq}\gamma_{\downarrow^2}\comp \theta\comp \alpha_{\downarrow^2}}{pointwise def\@. $\stackrel{.}{\subseteq}$}\\
$\Leftrightarrow$
\formulaexplanation{ \Theta\stackrel{.}{\subseteq}\gamma^A_2(\theta)}{def\@. $\gamma^A_2$}\\[1ex]

\hyphen{5}
\formulaexplanation{\mathord{\stackrel{.}{\alpha}}^A_2(\Theta)\stackrel{..}{{\subseteq}} \bar{\Theta}}{where $\Theta\in\wp(\mathcal{X}\times\mathcal{Y})\rightarrow\wp(\mathcal{Z}\times\mathcal{X})\mathrel{\smash{\stackrel{i}{\longrightarrow}}}\wp(\mathcal{Z}\times\mathcal{Y})$}\\
$\Leftrightarrow$
\formulaexplanation{{\mathord{\stackrel{.}{\alpha}}^A_2}(\Theta)(r)\stackrel{.}{{\subseteq}}\bar{\Theta}(r)}{pointwise def\@. $\stackrel{..}{{\subseteq}}$}\\
$\Leftrightarrow$
\formulaexplanation{\alpha^A_2(\Theta(r)) \stackrel{.}{{\subseteq}}\bar{\Theta}(r)}{def\@. (\ref{eq:abstraction:Post:post}) of $\mathord{\stackrel{.}{\alpha}^A_2}$}\\
$\Leftrightarrow$
\formulaexplanation{\Theta(r) \stackrel{.}{{\subseteq}}\gamma^A_2(\bar{\Theta}(r))}{Galois connection (\ref{eq:GC:alpha-gamma-A}) for $\alpha^A_2$}\\
$\Leftrightarrow$
\formulaexplanation{\Theta(r) \stackrel{.}{{\subseteq}}\mathord{\stackrel{.}{\gamma}}^A_2(\bar{\Theta})(r)}{def\@. $\mathord{\stackrel{.}{\gamma}}^A_2\triangleq\LAMBDA{r}\gamma^A_2(r)$}\\
$\Leftrightarrow$
\lastformulaexplanation{\Theta \stackrel{..}{{\subseteq}}{\mathord{\stackrel{.}{\gamma}}^A_2}(\bar{\Theta})}{pointwise def\@. $\stackrel{..}{{\subseteq}}$}{\mbox{\qed}}
\end{calculus}\let\qed\relax
\end{proof}
\end{toappendix}

These abstractions of the relational transformers yield the  following generalization of the classic predicate transformers  $\wp(\Sigma)\rightarrow\wp(\Sigma)$ \cite{DBLP:conf/focs/Pratt76}, by extension to nontermination $\bot$. \proofinapx
\begin{eqntabular}{rclclcl}
\textsf{\textup{post}}&=&\mathord{\stackrel{.}{\alpha}^A_2}(\textsf{\textup{post}})
&=&\rlap{$\LAMBDA{r}\LAMBDA{Q}\{\sigma' \mid\exists\sigma\mathrel{.}\sigma\in P\wedge\pair{\sigma}{\sigma'}\in r\}$,\qquad as in  (\ref{eq:def:post})}\nonumber
\\
\widetilde{\textsf{\upshape post}}&\triangleq&\mathord{\stackrel{.}{\alpha}^A_2}(\widetilde{\textsf{\upshape Post}})&=&\mathord{\stackrel{.}{\alpha}}^{\sim}\comp\textsf{\upshape post}&=&\LAMBDA{r}\LAMBDA{P}\{\sigma'\mid\forall\sigma\mathrel{.}\pair{\sigma}{\sigma'}\in r\Rightarrow \sigma\in P\}\phantom{\sqb{\texttt{\small S}}_{\bot}}\label{eq:def:psott-pre-pret}
\\
\textsf{\upshape pre}&\triangleq&\mathord{\stackrel{.}{\alpha}^A_1}(\textsf{\upshape Pre})&=&\mathord{\stackrel{.}{\alpha}}^{-1}\comp\textsf{\upshape post}&=&\LAMBDA{r}\LAMBDA{Q}\{\sigma \mid\exists\sigma'\mathrel{.}\pair{\sigma}{\sigma'}\in r\wedge\sigma'\in Q\}\nonumber\\
\widetilde{\textsf{\upshape pre}}&\triangleq&\mathord{\stackrel{.}{\alpha}^A_1}(\widetilde{\textsf{\upshape Pre}})&=&\mathord{\stackrel{.}{\alpha}}^{\sim}\comp\textsf{\upshape pre}&=&\LAMBDA{r}\LAMBDA{Q}\{\sigma \mid\forall\sigma'\mathrel{.}\pair{\sigma}{\sigma'}\in r\Rightarrow \sigma'\in Q\}\nonumber
\end{eqntabular}
\begin{toappendix}
\begin{proof}[Proof of (\ref{eq:def:psott-pre-pret})]
\begin{calculus}
\hyphen{6}\formula{\mathord{\stackrel{.}{\alpha}^A_2}(\textsf{\textup{post}})}\\
=
\formulaexplanation{\LAMBDA{r}\alpha^A_2(\textsf{\textup{post}}(r))}{def\@. (\ref{eq:abstraction:Post:post}) of $\mathord{\stackrel{.}{\alpha}^A_2}$}\\
=
\formulaexplanation{\LAMBDA{r}\alpha_{\downarrow^2}\comp \textsf{\textup{post}}(r)\comp \gamma_{\downarrow^2}}{def\@. (\ref{eq:abstraction:Post:post}) of $\alpha^A_2$}\\
=
\formulaexplanation{\LAMBDA{r}\LAMBDA{P}\alpha_{\downarrow^2}(\textsf{\textup{post}}(r)( \gamma_{\downarrow^2}(P)))}{def\@. function composition $\comp$}\\
=
\formulaexplanation{\LAMBDA{r}\LAMBDA{P}\alpha_{\downarrow^2}(\{\pair{\sigma_0}{\sigma'}\mid\exists\sigma\mathrel{.}\pair{\sigma_0}{\sigma}\in \gamma_{\downarrow^2}(P)\wedge\pair{\sigma}{\sigma'}\in r\})}{def\@. (\ref{eq:def:post}) of $\textsf{\textup{post}}$}\\
=
\formula{\LAMBDA{r}\LAMBDA{P}\{\sigma'\mid\exists \sigma_0\mathrel{.}\pair{\sigma_0}{\sigma'}\in \{\pair{\sigma_0}{\sigma'}\mid\exists\sigma\mathrel{.}\pair{\sigma_0}{\sigma}\in \Sigma\times P\wedge\pair{\sigma}{\sigma'}\in r\}\}}\\[-0.5ex]\rightexplanation{def\@. (\ref{eq:def:alpha-gamma-d2}) of $\alpha_{\downarrow^2}$ and $\gamma_{\downarrow^2}$}\\
=
\formulaexplanation{\LAMBDA{r}\LAMBDA{P}\{\sigma'\mid\exists \sigma_0\mathrel{.}\exists\sigma\mathrel{.}\pair{\sigma_0}{\sigma}\in \Sigma\times P\wedge\pair{\sigma}{\sigma'}\in r\}}{def\@. $\in$}\\
=
\formulaexplanation{\LAMBDA{r}\LAMBDA{P}\{\sigma'\mid\exists\sigma\mathrel{.}{\sigma}\in P\wedge\pair{\sigma}{\sigma'}\in r\}}{taking $ \sigma_0\in\Sigma$}\\
=
\formulaexplanation{\textsf{\textup{post}}}{def\@. (\ref{eq:def:post}) of $\textsf{\textup{post}}$}{\mbox{\qed}}\\[1em]

\hyphen{6}\formula{\mathord{\stackrel{.}{\alpha}^A_2}(\widetilde{\textsf{\upshape Post}})}\\
=
\formulaexplanation{\LAMBDA{r}\alpha^A_2(\widetilde{\textsf{\upshape Post}}(r))}{def\@. (\ref{eq:abstraction:Post:post}) of $\mathord{\stackrel{.}{\alpha}^A_2}$}\\
=
\formulaexplanation{\LAMBDA{r}\alpha_{\downarrow^2}\comp \widetilde{\textsf{\upshape Post}}(r)\comp \gamma_{\downarrow^2}}{def\@. (\ref{eq:abstraction:Post:post}) of $\alpha^A_2$}\\
=
\formulaexplanation{\LAMBDA{r}\alpha_{\downarrow^2}\comp \LAMBDA{P}\{\pair{\sigma_0}{\sigma'}\mid\forall \sigma \mathrel{.}\pair{\sigma}{\sigma'}\in r\Rightarrow\pair{\sigma_0}{\sigma}\in P\}\comp \gamma_{\downarrow^2}}{def\@. (\ref{eq:def:Postt-Pret}) of $\widetilde{\textsf{\upshape Post}}$}\\
=
\formulaexplanation{\LAMBDA{r}\LAMBDA{P}\alpha_{\downarrow^2}(\{\pair{\sigma_0}{\sigma'}\mid\forall \sigma \mathrel{.}\pair{\sigma}{\sigma'}\in r\Rightarrow\pair{\sigma_0}{\sigma}\in \gamma_{\downarrow^2}(P)\})}{def\@. function composition $\comp$}\\
=
\formulaexplanation{\LAMBDA{r}\LAMBDA{P}\{{\sigma'}\mid\forall \sigma \mathrel{.}\pair{\sigma}{\sigma'}\in r\Rightarrow\pair{\sigma_0}{\sigma}\in \mathcal{X}\times P\}}{def\@. (\ref{eq:def:alpha-gamma-d2}) of $\alpha_{\downarrow^2}$ and $\gamma_{\downarrow^2}$}\\
=
\formulaexplanation{\LAMBDA{r}\LAMBDA{P}\{{\sigma'}\mid\forall \sigma \mathrel{.}\pair{\sigma}{\sigma'}\in r\Rightarrow{\sigma}\in P\}}{def\@. pairs and $\in$, Q.E.D\@. for case $\widetilde{\textsf{\upshape post}}$ of (\ref{eq:def:psott-pre-pret})}\\
=
\formulaexplanation{\LAMBDA{r}\LAMBDA{P}\neg\{{\sigma'}\mid\exists \sigma \mathrel{.}\pair{\sigma}{\sigma'}\in r\wedge{\sigma}\not\in P\}}{def\@. $\neg$}\\
=
\formulaexplanation{\LAMBDA{r}\LAMBDA{P}\neg\textsf{\textup{post}}(r)(\neg P)}{def\@. (\ref{eq:def:post}) of $\textsf{\textup{post}}$ }\\
=
\formulaexplanation{\LAMBDA{r}\LAMBDA{P}\neg\comp \textsf{\textup{post}}(r)\comp \neg (P)}{def\@. function composition $\comp$ }\\
=
\formulaexplanation{\LAMBDA{r}\alpha^{\sim}(\textsf{\textup{post}}(r))}{def\@. $\alpha^{\sim}(f)\triangleq\neg\comp f\comp \neg$ in (\ref{eq:eq-function-complement-GC}) and {\text{\boldmath$\lambda$}}-calculus}\\
=
\formulaexplanation{\LAMBDA{r}\mathord{\stackrel{.}{\alpha}}^{\sim}(\textsf{\textup{post}})(r)}{pointwise def\@. $\mathord{\stackrel{.}{\alpha}}^{\sim}(\textsf{\textup{post}})\triangleq\LAMBDA{r}{\alpha}^{\sim}(\textsf{\textup{post}}(r))$ in (\ref{eq:eq-function-complement-GC})}\\
=
\formulaexplanation{\mathord{\stackrel{.}{\alpha}}^{\sim}\comp\textsf{\textup{post}}}{def\@. function composition $\comp$ and {\text{\boldmath$\lambda$}}-calculus, Q.E.D\@. for case $\widetilde{\textsf{\upshape post}}$ of (\ref{eq:def:psott-pre-pret})}\\[1em]

\hyphen{6}\formula{\mathord{\stackrel{.}{\alpha}^A_1}(\widetilde{\textsf{\upshape Pre}})}\\
=
\formulaexplanation{\LAMBDA{r}\alpha^A_2(\widetilde{\textsf{\upshape Pre}}(r))}{def\@. (\ref{eq:abstraction:Post:post}) of $\mathord{\stackrel{.}{\alpha}^A_1}$}\\
=
\formulaexplanation{\LAMBDA{r}\alpha_{\downarrow^1}\comp \widetilde{\textsf{\upshape Pre}}(r)\comp \gamma_{\downarrow^1}}{def\@. (\ref{eq:abstraction:Post:post}) of $\alpha^A_1$}\\
=
\formulaexplanation{\LAMBDA{r}\alpha_{\downarrow^1}\comp \LAMBDA{Q}\{\pair{\sigma}{\sigma_{\mskip-2muf}}\mid\forall \sigma' \mathrel{.}\pair{\sigma}{\sigma'}\in  r \Rightarrow\pair{\sigma'}{\sigma_{\mskip-2muf}}\in Q\}\comp \gamma_{\downarrow^1}}{def\@. (\ref{eq:def:Postt-Pret}) of $\widetilde{\textsf{\upshape Pre}}$}\\
=
\formulaexplanation{\LAMBDA{r}\LAMBDA{Q}\alpha_{\downarrow^1}(\{\pair{\sigma}{\sigma_{\mskip-2muf}}\mid\forall \sigma' \mathrel{.}\pair{\sigma}{\sigma'}\in  r \Rightarrow\pair{\sigma'}{\sigma_{\mskip-2muf}}\in \gamma_{\downarrow^1}(Q)\})}{def\@. function composition $\comp$}\\
=
\formulaexplanation{\LAMBDA{r}\LAMBDA{Q}\{{\sigma}\mid\forall \sigma' \mathrel{.}\pair{\sigma}{\sigma'}\in  r \Rightarrow{\sigma'}\in Q\}}{def\@. (\ref{eq:def:alpha-gamma-d1}) of $\alpha_{\downarrow^1}$ and $\gamma_{\downarrow^1}$ Q.E.D\@. for case $\widetilde{\textsf{\upshape pre}}$ of (\ref{eq:def:psott-pre-pret})}\\
=
\formulaexplanation{\LAMBDA{r}\LAMBDA{Q}\neg\{{\sigma}\mid\exists \sigma' \mathrel{.}\pair{\sigma}{\sigma'}\in  r \wedge{\sigma'}\in \neg(Q)\}}{def\@. $\neg$}\\
=
\formulaexplanation{\LAMBDA{r}\LAMBDA{P}\neg\textsf{pre}(r)(\neg P)}{def\@. (\ref{eq:def:psott-pre-pret}) of $\textsf{pre}$ }\\
=
\formulaexplanation{\LAMBDA{r}\LAMBDA{P}\neg\comp \textsf{pre}(r)\comp \neg (P)}{def\@. function composition $\comp$ }\\
=
\formulaexplanation{\LAMBDA{r}\alpha^{\sim}(\textsf{pre}(r))}{def\@. $\alpha^{\sim}(f)\triangleq\neg\comp f\comp \neg$ in (\ref{eq:eq-function-complement-GC}) and {\text{\boldmath$\lambda$}}-calculus}\\
=
\formulaexplanation{\LAMBDA{r}\mathord{\stackrel{.}{\alpha}}^{\sim}(\textsf{pre})(r)}{pointwise def\@. $\mathord{\stackrel{.}{\alpha}}^{\sim}(\textsf{pre})\triangleq\LAMBDA{r}{\alpha}^{\sim}(\textsf{pre}(r))$ in (\ref{eq:eq-function-complement-GC})}\\
=
\formulaexplanation{\mathord{\stackrel{.}{\alpha}}^{\sim}\comp\textsf{pre}}{def\@. function composition $\comp$ and {\text{\boldmath$\lambda$}}-calculus, Q.E.D\@. for case ${\textsf{\upshape pre}}$ of (\ref{eq:def:psott-pre-pret})}\\[1em]

\hyphen{6}\formula{\mathord{\stackrel{.}{\alpha}^A_1}({\textsf{\upshape Pre}})}\\
=
\formulaexplanation{\LAMBDA{r}\alpha^A_2(\widetilde{\textsf{\upshape Pre}}(r))}{def\@. (\ref{eq:abstraction:Post:post}) of $\mathord{\stackrel{.}{\alpha}^A_1}$}\\
=
\formulaexplanation{\LAMBDA{r}\alpha_{\downarrow^1}\comp {\textsf{\upshape Pre}}(r)\comp \gamma_{\downarrow^1}}{def\@. (\ref{eq:abstraction:Post:post}) of $\alpha^A_1$}\\
=
\formulaexplanation{\LAMBDA{r}\alpha_{\downarrow^1}\comp \LAMBDA{Q}\{\pair{\sigma}{\sigma_{\mskip-2muf}}\mid\exists \sigma' \mathrel{.}\pair{\sigma}{\sigma'}\in  r \wedge\pair{\sigma'}{\sigma_{\mskip-2muf}}\in Q\}\comp \gamma_{\downarrow^1}}{def\@. (\ref{eq:def:Postt-Pret}) of ${\textsf{\upshape Pre}}$}\\
=
\formulaexplanation{\LAMBDA{r}\LAMBDA{Q}\alpha_{\downarrow^1}(\{\pair{\sigma}{\sigma_{\mskip-2muf}}\mid\exists \sigma' \mathrel{.}\pair{\sigma}{\sigma'}\in  r \wedge\pair{\sigma'}{\sigma_{\mskip-2muf}}\in \gamma_{\downarrow^1}(Q)\})}{def\@. function composition $\comp$}\\
=
\formulaexplanation{\LAMBDA{r}\LAMBDA{Q}\{{\sigma}\mid\exists \sigma' \mathrel{.}\pair{\sigma}{\sigma'}\in  r \wedge{\sigma'}\in Q\}}{def\@. (\ref{eq:def:alpha-gamma-d1}) of $\alpha_{\downarrow^1}$ and $\gamma_{\downarrow^1}$ Q.E.D\@. for case ${\textsf{\upshape pre}}$ of (\ref{eq:def:psott-pre-pret})}\\
=
\formulaexplanation{\LAMBDA{r}\LAMBDA{Q}\neg\{{\sigma}\mid\forall \sigma' \mathrel{.}\pair{\sigma}{\sigma'}\in  r \Rightarrow{\sigma'}\in \neg(Q)\}}{def\@. $\neg$}\\
=
\formulaexplanation{\LAMBDA{r}\LAMBDA{P}\neg\widetilde{\textsf{pre}}(r)(\neg P)}{def\@. (\ref{eq:def:psott-pre-pret}) of $\widetilde{\textsf{pre}}$ }\\
=
\formulaexplanation{\LAMBDA{r}\LAMBDA{P}\neg\comp \widetilde{\textsf{pre}}(r)\comp \neg (P)}{def\@. function composition $\comp$ }\\
=
\formulaexplanation{\LAMBDA{r}\alpha^{\sim}(\widetilde{\textsf{pre}}(r))}{def\@. $\alpha^{\sim}(f)\triangleq\neg\comp f\comp \neg$ in (\ref{eq:eq-function-complement-GC}) and {\text{\boldmath$\lambda$}}-calculus}\\
=
\formulaexplanation{\LAMBDA{r}\mathord{\stackrel{.}{\alpha}}^{\sim}(\widetilde{\textsf{pre}})(r)}{pointwise def\@. $\mathord{\stackrel{.}{\alpha}}^{\sim}(\widetilde{\textsf{pre}})\triangleq\LAMBDA{r}{\alpha}^{\sim}(\widetilde{\textsf{pre}}(r))$ in (\ref{eq:eq-function-complement-GC})}\\
=
\lastformulaexplanation{\mathord{\stackrel{.}{\alpha}}^{\sim}\comp\widetilde{\textsf{pre}}}{def\@. function composition $\comp$ and {\text{\boldmath$\lambda$}}-calculus, Q.E.D\@. for case $\widetilde{\textsf{\upshape pre}}$ of (\ref{eq:def:psott-pre-pret})}{\mbox{\qed}}

\end{calculus}\let\qed\relax
\end{proof}
\end{toappendix}
\noindent The classic transformers (\ref{eq:def:psott-pre-pret}) are illustrated by Fig\@.  \ref{fig:Property-transformers}\ifshort\ in the appendix \proofinapx\fi.
\begin{toappendix}
\begin{figure}[ht]
\noindent\begin{minipage}[t]{0.28\textwidth}
\def\mstrut{\mbox{\rule{0pt}{0.0em}}}%
\def\reduce{-0.525ex}
\def\tick{\times}%
\def\none{}%
\begin{eqntabular*}[fl]{@{}|c|c|c|c|c|c|@{}}
\mbox{\rule{0pt}{3.em}}%
\hskip0.6em\turnbox{90}{\small$\textsf{\upshape wp}(\texttt{\small S}, Q_{\not\bot})$}\hskip0.2em
&\hskip0.6em\turnbox{90}{\small$\textsf{\upshape wlp}(\texttt{\small S}, Q_{\not\bot})$}\hskip0.2em
&\hskip0.6em\turnbox{90}{\small$\widetilde{\textsf{\upshape pre}}\sqb{\texttt{\small S}}Q$}\hskip0.2em
&\hskip0.6em\turnbox{90}{\small$\widetilde{\textsf{\upshape pre}}\sqb{\texttt{\small S}}_{\bot}Q$}\hskip0.2em
&\hskip0.6em\turnbox{90}{\small$\textsf{\upshape pre}\sqb{\texttt{\small S}}Q$}\hskip0.2em
&\hskip0.6em\turnbox{90}{\small$\textsf{\upshape pre}\sqb{\texttt{\small S}}_{\bot}Q$}\hskip0.2em\\\hline 
\mstrut\none&\none&\tick&\tick&\none&\none\\[\reduce]\hline 
\mstrut\none&\none&\tick&\none&\none&\tick\\[\reduce]\hline 
\mstrut\none&\tick&\tick&\tick&\tick&\tick\\[\reduce]\hline 
\mstrut\tick&\tick&\tick&\tick&\tick&\tick\\[\reduce]\hline 
\mstrut\none&\none&\none&\none&\tick&\tick\\[\reduce]\hline 
\mstrut\none&\none&\none&\none&\tick&\tick\\[\reduce]\hline 
\mstrut\none&\none&\none&\none&\none&\none\\[\reduce]\hline 
\mstrut\none&\none&\none&\none&\none&\tick\\[\reduce]\hline 
\mstrut\none&\none&\tick&\tick&\none&\none\\[\reduce]\hline 
\mstrut\none&\none&\tick&\tick&\none&\none\\[\reduce]\hline 
\end{eqntabular*}
\end{minipage}%
\begin{minipage}[t]{0.18\textwidth}
\raisebox{-15.25em}[0pt][0pt]{\includegraphics[scale=0.2]{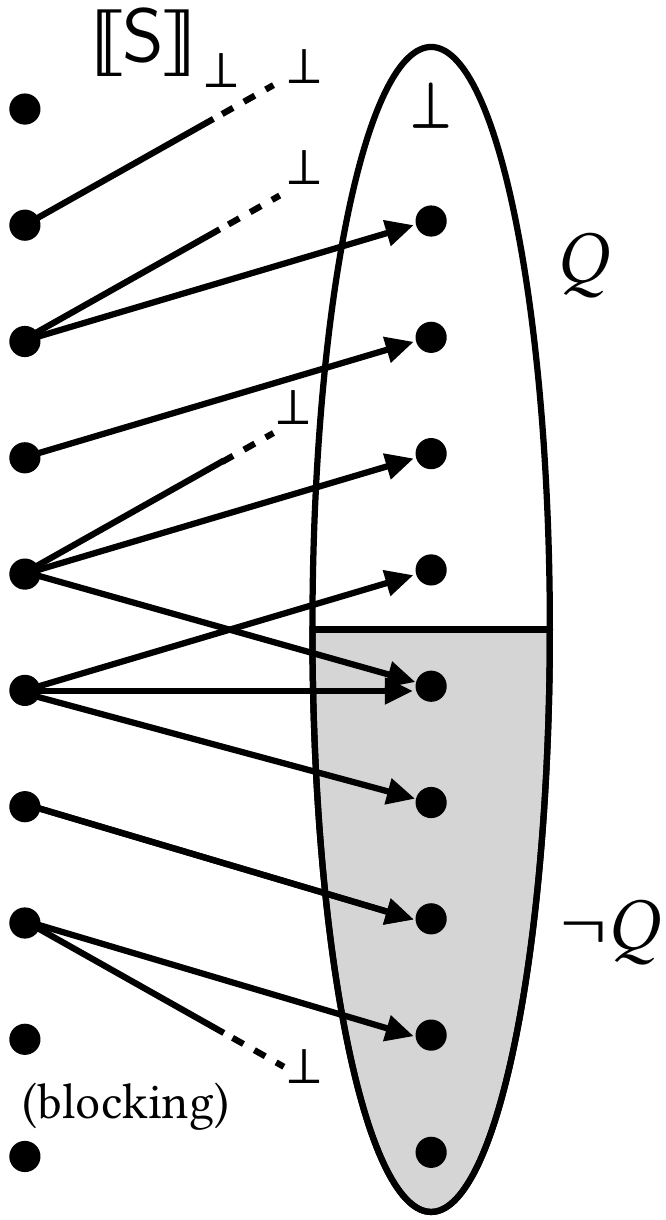}}
\end{minipage}%
\begin{minipage}[t]{0.1425\textwidth}
\raisebox{-15.25em}[0pt][0pt]{\includegraphics[scale=0.2]{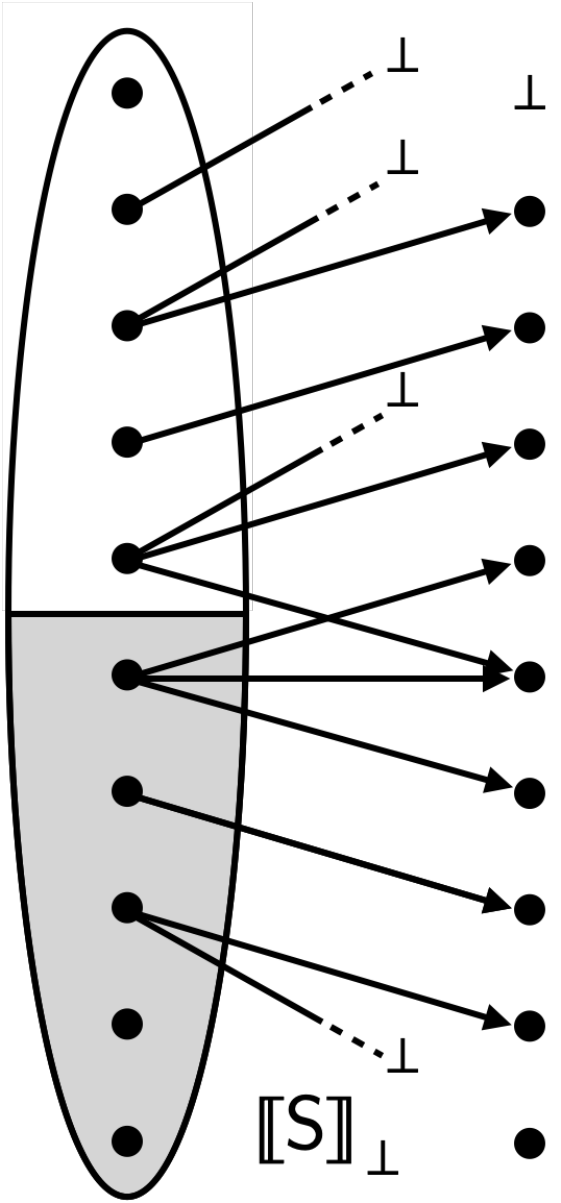}}
\end{minipage}%
\begin{minipage}[t]{0.28\textwidth}
\def\mstrut{\mbox{\rule{0pt}{0.0em}}}%
\def\reduce{-0.525ex}
\def\tick{\times}%
\def\none{}%
\begin{eqntabular*}[fl]{@{}|c|c|c|c|c|c|@{}}
\mbox{\rule{0pt}{3em}}%
\hskip0.6em\turnbox{90}{\small$\textsf{\upshape post}\sqb{\texttt{\small S}}_{\bot}Q$}\hskip0.2em\hskip0.2em
&\hskip0.6em\turnbox{90}{\small$\textsf{\upshape post}\sqb{\texttt{\small S}}Q$}\hskip0.2em
&\hskip0.6em\turnbox{90}{\small$\widetilde{\textsf{\upshape post}}\sqb{\texttt{\small S}}_{\bot}Q$}\hskip0.2em
&\hskip0.6em\turnbox{90}{\small$\widetilde{\textsf{\upshape post}}\sqb{\texttt{\small S}}Q$}\hskip0.2em
&\hskip0.6em\turnbox{90}{\small$\textsf{\upshape slp}(\texttt{\small S}, Q)$}\hskip0.2em
&\hskip0.6em\turnbox{90}{\small$\textsf{\upshape sp}(\texttt{\small S}, Q)$}\hskip0.2em\\\hline 
\mstrut\tick&\none&\none&\tick&\none&\none\\[\reduce]\hline 
\mstrut\tick&\tick&\tick&\tick&\tick&\none\\[\reduce]\hline 
\mstrut\tick&\tick&\tick&\tick&\tick&\tick\\[\reduce]\hline 
\mstrut\tick&\tick&\tick&\tick&\tick&\none\\[\reduce]\hline 
\mstrut\none&\none&\none&\none&\none&\none\\[\reduce]\hline 
\mstrut\tick&\tick&\none&\none&\tick&\none\\[\reduce]\hline 
\mstrut\none&\none&\none&\none&\none&\none\\[\reduce]\hline 
\mstrut\none&\none&\none&\none&\none&\none\\[\reduce]\hline 
\mstrut\none&\none&\none&\none&\none&\none\\[\reduce]\hline 
\mstrut\none&\none&\none&\none&\none&\none\\[\reduce]\hline 
\end{eqntabular*}
\end{minipage}%
\vskip-1ex%
\caption{Illustration of property transformers\label{fig:Property-transformers}}
\end{figure}

In Fig\@. \ref{fig:Property-transformers}, the points $\bullet$  represent states, the arrows between two states or a state and $\bot$ represent a pair in $\sqb{\texttt{\small S}}_{\bot}$. For the angelic semantics $\sqb{\texttt{\small S}}$, the states and arrows marked $\bot$ should be ignored. The consequent states on the left are partitioned into $Q$ and $\neg Q$. In the column for each transformer $\tau$ in the left table, tags $\times$ indicate that the antecedent state on the same line belongs to $\tau(Q)$. Similarly, the  antecedent states on the left are partitioned into $Q$ and $\neg Q$. In the column for each transformer $\tau$ in the right table, consequent states belonging to $\tau(Q)$ are tagged $\times$ on the same line. 
\end{toappendix}

Given a relation $r\in\wp(\mathcal{X}\times\mathcal{Y})$, in addition to (\ref{eq:def:post:GC}), these classic transformers are also connected as follows \cite[Chapter 12]{Cousot-PAI-2021}\ifshort,  (d) is proved in sect\@. \ref{sec:apx:eq:eq:post-t-pre-t} of the appendix \proofinapx\else.\fi
\begin{eqntabular}{L@{\hskip1ex}c@{\quad}L@{\hskip1ex}c@{\quad}}
(a)&\pair{\wp(\mathcal{X}\times \mathcal{Y})}{\subseteq}\galois{\textsf{\upshape post}}{\textsf{\upshape post}^{-1}}\pair{\wp(\mathcal{X})\stackrel{\cupdot}{\longrightarrow}\wp(\mathcal{Y})}{\stackrel{.}{\subseteq}}
&
(b)&\pair{\wp(\mathcal{X})}{\subseteq}\galois{\textsf{\upshape pre}(r)}{\widetilde{\textsf{\upshape post}}(r)}\pair{\wp(\mathcal{Y})}{\subseteq}
\label{eq:post-t-pre-t}\\
(c)&\pair{\wp(\mathcal{X}\times \mathcal{Y})}{\subseteq}\galois{\textsf{\upshape pre}}{\textsf{\upshape pre}^{-1}}\pair{\wp(\mathcal{X})\stackrel{\cupdot}{\longrightarrow}\wp(\mathcal{Y})}{\stackrel{.}{\subseteq}}
&
(d)&\textsf{\textup{post}}(R)P\cap Q\neq \emptyset\Leftrightarrow P\cap \mathsf{pre}(R)Q\neq\emptyset\nonumber
\end{eqntabular}
\begin{toappendix}
\label{sec:apx:eq:eq:post-t-pre-t}
\begin{proof}[proof of  (\ref{eq:post-t-pre-t}.d)]
Let us show that $\textsf{\textup{post}}(R)P\cap Q\neq \emptyset
\Leftrightarrow
\exists \sigma\in P\mathrel{.}\exists\sigma'\not\in Q\mathrel{.}\pair{\sigma}{\sigma'}\in R\nonumber\\[-0.5ex]
\Leftrightarrow$
\begin{calculus}[$\Leftrightarrow$~]
\formula{\textsf{\textup{post}}(R)P\cap Q\neq \emptyset}\\
$\Leftrightarrow$
\formulaexplanation{\{\sigma'\mid\exists\sigma\in P\mathrel{.}\pair{\sigma}{\sigma'}\in R\} \cap Q\neq \emptyset}{def\@. \textsf{\textup{post}}}\\
$\Leftrightarrow$
\formulaexplanation{\exists\sigma'\mathrel{.}\sigma'\in\{\sigma'\mid\exists\sigma\in P\mathrel{.}\pair{\sigma}{\sigma'}\in R\} \wedge\sigma'\in Q}{def\@. $\cap$ and $\emptyset$}\\
$\Leftrightarrow$
\formulaexplanation{\exists\sigma'\mathrel{.}\exists\sigma\in P\mathrel{.}\pair{\sigma}{\sigma'}\in R \wedge\sigma'\in Q}{def\@. $\in$}\\
$\Leftrightarrow$
\formulaexplanation{\exists\sigma\in P\mathrel{.}\exists\sigma'\in Q\mathrel{.}\pair{\sigma}{\sigma'}\in R}{commutativity}\\
$\Leftrightarrow$
\formulaexplanation{\exists\sigma\mathrel{.}\sigma\in P\wedge\sigma\in\{\sigma\mid\exists\sigma'\in Q\mathrel{.}\pair{\sigma}{\sigma'}\in R\}}{def\@. $\in$}\\
$\Leftrightarrow$
\formulaexplanation{P\cap\{\sigma\mid\exists\sigma'\in Q\mathrel{.}\pair{\sigma}{\sigma'}\in R\}\neq\emptyset}{def\@. $\cap$ and $\emptyset$}\\
$\Leftrightarrow$
\lastformulaexplanation{P\cap \textsf{pre}(R) Q\neq\emptyset}{\hyperlink{lem:slides:6}{\textsc{lem\@. 6}}}{\mbox{\qed}}
\end{calculus}\let\qed\relax
\end{proof}
\end{toappendix}
\begin{example}\label{ex-Hoare-incorrectness}Hoare incorrectness logic is 
$\neg(\{P\}\,\texttt{\small S}\{Q\})$
$\Leftrightarrow$
$\neg(\textsf{\textup{post}}\sqb{\texttt{\small S}}P\subseteq Q)$
$\Leftrightarrow$
$\textsf{\textup{post}}\sqb{\texttt{\small S}}P\cap\neg Q\neq\emptyset$
$\Leftrightarrow$
$\exists \sigma\in P\mathrel{.}\exists\sigma'\not\in Q\mathrel{.}\pair{\sigma}{\sigma'}\in\sqb{\texttt{\small S}}$
$\Leftrightarrow$
$P\cap \mathsf{pre}\sqb{\texttt{\small S}}\neg Q\neq\emptyset$ by. This is different from incorrectness logic \cite{DBLP:journals/pacmpl/OHearn20}, that is
$[P]\,\texttt{\small S}\,[Q]$
$\Leftrightarrow$
$\neg Q\subseteq \textsf{\textup{post}}\sqb{\texttt{\small S}}P$
$\Leftrightarrow$
$\forall \sigma'\not\in Q\mathrel{.}\exists\sigma\in P\mathrel{.}\pair{\sigma}{\sigma'}\in\sqb{\texttt{\small S}}$. 
The incorrectness Hoare logic is designed in Sect\@. \ref{apx:DesignHoareIncorrectnessLogic} in the appendix.
\end{example}
All transformers in (\ref{eq:Post-Pret:Pre-Postt-GC}),  (\ref{eq:def:post:GC}), and (\ref{eq:post-t-pre-t}) inherit the properties of Galois connections. For example, the lower adjoint preserves arbitrary joins and dually the upper adjoint preserves arbitrary meets. This implies, for example, the healthiness conditions postulated for transformers \cite{DBLP:books/daglib/0067387,DBLP:journals/jacm/Hoare78}.
\begin{remark}\label{rem:pre-does-not-preserve-meets}By (\ref{eq:def:post:GC}), \textsf{pre} preserves joins ($\cup$) but maybe   not meets ($\cap$). Same for \textsf{\textup{post}}. \proofinapx
\end{remark}
\begin{toappendix}
\label{sec:apx:rem:pre-does-not-preserve-meets}
\begin{proof}[proof of remark \ref{rem:pre-does-not-preserve-meets}]Define $t_i=\{\pair{n}{n+i}\mid n\in\mathbb{N}\}$. If $\pair{n}{m}\in\bigcap_{i\in\mathbb{N}}t_i$ then it belongs to all $t_i$ so
$\forall i\in \mathbb{N}\mathrel{.}m=n+i$, a contradiction. So $\bigcap_{i\in\mathbb{N}}t_i=\emptyset$
hence $\textsf{pre}(\bigcap_{i\in\mathbb{N}}t_i)\mathbb{Z}=\emptyset$. On the other hand, $\textsf{pre}(t_i)\mathbb{Z}$ = $\{n\mid\exists m\in\mathbb{Z}\mathrel{.}\pair{n}{m}\in t_i\}$ = $\{n\mid\exists m\in\mathbb{Z}\mathrel{.}m=n+i\}$ = $\mathbb{Z}$. Obviously $\textsf{pre}(\bigcap_{i\in\mathbb{N}}t_i)\mathbb{Z}=\emptyset\neq\mathbb{Z}=\bigcap_{i\in\mathbb{N}}\textsf{pre}(t_i)\mathbb{Z}$.

\medskip

Take $t=\{\pair{1}{m}\mid m\in\mathbb{N}\}$ and $X_i=\{n\in\mathbb{N}\mid n\geqslant i\}$. If $m\in\bigcap_{i\in\mathbb{N}}X_i$ then for all $i\in\mathbb{N}$, $m\in X_i$, in contradiction with $m\notin X_{m+1}$ proving $\bigcap_{i\in\mathbb{N}}X_i=\emptyset$ so that $\textsf{pre}(t)(\bigcap_{i\in\mathbb{N}}X_i)=\emptyset$. Now $\textsf{pre}(t)X_i$ = $\{n\mid\exists m\in X_i\mathrel{.}
\pair{n}{m}\in t_i\}$ = $\{1\}$ since no $X_i$ is empty. It follows that $\textsf{pre}(t)(\bigcap_{i\in\mathbb{N}}X_i)=\emptyset\neq\{1\}=\bigcap_{i\in\mathbb{N}}\textsf{pre}(t)X_i$.
\end{proof}
\end{toappendix}

\subsection{To Terminate or Not to Terminate Abstraction for Transformers}\label{sec:terminate-not-terminate-transformers}

We have shown in Sect\@. \ref{sec:terminate-not-terminate-properties} that we can abstract antecedant-consequence pairs by (\ref{eq:def:alpha-not-bot}) or (\ref{eq:def:alpha-bot}) to
take nontermination into account (e.g\@. total correctness) or not (partial correctness). An equivalent alternative uses the natural semantics 
$\sqb{\texttt{\large S}}_{\bot}$ or the angelic one $\sqb{\texttt{\large S}}$ in (\ref{eq:angelic-semantics}). We can also abstract transformers, which we do in the assertional case, by
\bgroup\arraycolsep0.5\arraycolsep\abovedisplayskip0.5\abovedisplayskip\belowdisplayskip0.5\belowdisplayskip
\begin{eqntabular}{rcl@{\qquad\quad}rcl@{\qquad\quad}rcl@{\qquad}}
{\alpha}_{\mskip-1mu\not\bot}(P)&\triangleq&P\setminus\{\bot\}&
\overrightarrow{\alpha}_{\mskip-7mu\not\bot}(\theta)&\triangleq&{\alpha}_{\not\bot}\comp\theta&
\overleftarrow{\alpha}_{\mskip-7mu\not\bot}(\theta)&\triangleq&\theta\comp{\gamma}_{\not\bot}\label{eq:assertional-not-bot}\\
{\gamma}_{\mskip-1mu\not\bot}(Q)&\triangleq& Q\cup\{\bot\}
&
\overrightarrow{\gamma}_{\mskip-7mu\not\bot}(\bar\theta)&\triangleq&{\gamma}_{\not\bot}\comp\bar{\theta}
&
\overleftarrow{\gamma}_{\mskip-7mu\not\bot}(\bar\theta)&\triangleq&\bar\theta\comp{\alpha}_{\not\bot}
\end{eqntabular}\egroup
\noindent
which yield Galois connections \proofinapx
\bgroup\abovedisplayskip0.0\abovedisplayskip\belowdisplayskip0pt\begin{eqntabular}{c@{\quad}}
\pair{\wp(\Sigma_\bot)}{\subseteq}\galois{{\alpha}_{\mskip-1mu\not\bot}}{{\gamma}_{\mskip-1mu\not\bot}}\pair{\wp(\Sigma)}{\subseteq}
\qquad
\pair{\mathcal{X}\rightarrow\wp(\Sigma_\bot)}{\stackrel{.}{\subseteq}}\galois{\overrightarrow{\alpha}_{\mskip-7mu\not\bot}}{\overrightarrow{\gamma}_{\mskip-7mu\not\bot}}\pair{\mathcal{X}\rightarrow\wp(\Sigma)}{\stackrel{.}{\subseteq}}
\label{eq:GC-assertional-not-bot}\\[-1ex]
\pair{\wp(\Sigma_\bot)\stackrel{i}{\longrightarrow}\wp(\Sigma)}{\stackrel{.}{\subseteq}}\galois{\overleftarrow{\alpha}_{\mskip-7mu\not\bot}}{\overleftarrow{\gamma}_{\mskip-7mu\not\bot}}\pair{\wp(\Sigma)\stackrel{i}{\longrightarrow}\wp(\Sigma)}{\stackrel{.}{\subseteq}}
\nonumber
\end{eqntabular}\egroup
\begin{toappendix}
\begin{proof}[Proof of (\ref{eq:GC-assertional-not-bot})]
\begin{calculus}[$\Leftrightarrow$\ ]
\hyphen{5}\formula{{\alpha}_{\mskip-1mu\not\bot}(P)\subseteq Q}\\
$\Leftrightarrow$
\formulaexplanation{P\setminus\{\bot\}\subseteq Q}{def\@. (\ref{eq:assertional-not-bot}) of $\alpha_{\mskip-1mu\not\bot}$}\\
$\Leftrightarrow$
\formulaexplanation{P\subseteq Q\cup\{\bot\}}{def\@. $\subseteq$}\\
$\Leftrightarrow$
\formulaexplanation{P\subseteq {\gamma}_{\mskip-1mu\not\bot}(Q)}{by defining ${\gamma}_{\mskip-1mu\not\bot}(Q)\triangleq Q\cup\{\bot\}$}\\[1ex]
\hyphen{5}\formula{\overrightarrow{\alpha}_{\mskip-7mu\not\bot}(\theta) \stackrel{.}{\subseteq}\bar{\theta}}\\
$\Leftrightarrow$
\formulaexplanation{\forall P\in\mathcal{X}\mathrel{.}\overrightarrow{\alpha}_{\mskip-7mu\not\bot}(\theta)P \subseteq\bar{\theta}(P)}{pointwise def\@. $\stackrel{.}{\subseteq}$}\\
$\Leftrightarrow$
\formulaexplanation{\forall P\in\mathcal{X}\mathrel{.}{\alpha}_{\not\bot}(\theta(P)) \subseteq\bar{\theta}(P)}{def\@. (\ref{eq:assertional-not-bot}) of $\overrightarrow{\alpha}_{\mskip-7mu\not\bot}$}\\
$\Leftrightarrow$
\formulaexplanation{\forall P\in\mathcal{X}\mathrel{.}\theta(P)\subseteq{\gamma}_{\not\bot}(\bar{\theta}(P))}{Galois connection $\pair{{\alpha}_{\mskip-1mu\not\bot}}{{\gamma}_{\mskip-1mu\not\bot}}$ (\ref{eq:GC-assertional-not-bot})}\\
$\Leftrightarrow$
\formulaexplanation{\theta\stackrel{.}{\subseteq}\LAMBDA{P}{\gamma}_{\not\bot}(\bar{\theta}(P))}{pointwise def\@. $\subseteq$}\\
$\Leftrightarrow$
\formulaexplanation{\theta\stackrel{.}{\subseteq}\overrightarrow{\gamma}_{\mskip-7mu\not\bot}(\bar\theta)}{by defining $\overrightarrow{\gamma}_{\mskip-7mu\not\bot}(\bar\theta)\triangleq\LAMBDA{P}{\gamma}_{\not\bot}(\bar{\theta}(P))$}\\[1ex]
\hyphen{5} The compositions $\overleftarrow{\alpha}_{\mskip-7mu\not\bot}$ and $\overleftarrow{\gamma}_{\mskip-7mu\not\bot}$ of increasing functions is increasing, Moreover,\\
\formula{\overleftarrow{\alpha}_{\mskip-7mu\not\bot}(\theta) \stackrel{.}{\subseteq}\bar{\theta} }\\
$\Leftrightarrow$
\formulaexplanation{\forall Q\in\wp(\Sigma_\bot)\mathrel{.}\overleftarrow{\alpha}_{\mskip-7mu\not\bot}(\theta)Q \subseteq\bar{\theta}(Q)}{pointwise def\@. $\stackrel{.}{\subseteq}$}\\
$\Leftrightarrow$
\formulaexplanation{\forall Q\in\wp(\Sigma_\bot)\mathrel{.}\theta({\gamma}_{\not\bot}(Q))\subseteq\bar{\theta}(Q)}{def\@. (\ref{eq:GC-assertional-not-bot}) of $\overleftarrow{\alpha}_{\mskip-7mu\not\bot}$}\\
$\Rightarrow$
\formulaexplanation{\forall P\in\wp(\Sigma)\mathrel{.}\theta({\gamma}_{\not\bot}({\alpha}_{\not\bot}(P)))\subseteq\bar{\theta}({\alpha}_{\not\bot}(P))}{for $Q={\alpha}_{\not\bot}(P)$}\\
$\Rightarrow$
\formula{\forall P\in\wp(\Sigma)\mathrel{.}\theta(P)\subseteq\bar{\theta}({\alpha}_{\not\bot}(P))}\\[-0.5ex]
\explanation{since ${\gamma}_{\not\bot}\comp{\alpha}_{\not\bot}$ is extensive by Galois connection $\pair{{\alpha}_{\mskip-1mu\not\bot}}{{\gamma}_{\mskip-1mu\not\bot}}$ (\ref{eq:GC-assertional-not-bot}) and $\theta$ is increasing by hypothesis}\\
$\Leftrightarrow$
\formulaexplanation{\theta\stackrel{.}{\subseteq}\bar{\theta}({\alpha}_{\not\bot})}{pointwise def\@. $\stackrel{.}{\subseteq}$}\\
\quad Conversely,\\
\formula{\forall P\in\wp(\Sigma)\mathrel{.}\theta(P)\subseteq\bar{\theta}({\alpha}_{\not\bot}(P))}\\
$\Rightarrow$
\formulaexplanation{\forall Q\in\wp(\Sigma_{\bot})\mathrel{.}\theta({\gamma}_{\not\bot}(Q))\subseteq\bar{\theta}({\alpha}_{\not\bot}({\gamma}_{\not\bot}(Q)))}{for $P={\gamma}_{\not\bot}(Q)$}\\
$\Rightarrow$
\formula{\forall Q\in\wp(\Sigma_\bot)\mathrel{.}\theta({\gamma}_{\not\bot}(Q))\subseteq\bar{\theta}(Q)}\\[-0.5ex]
\lastexplanation{since ${\alpha}_{\not\bot}\comp{\gamma}_{\not\bot}$ is reductive by Galois connection $\pair{{\alpha}_{\mskip-1mu\not\bot}}{{\gamma}_{\mskip-1mu\not\bot}}$ (\ref{eq:GC-assertional-not-bot}) and $\bar\theta$ is increasing by hypothesis}{\mbox{\qed}}
\end{calculus}\let\qed\relax
\end{proof}
\end{toappendix}
\vspace*{-1.5em}
\subsection{Abstract Logics}
Finally logics may refer to any abstraction of the antecedents and consequents of a transformational logics. For example, \cite{DBLP:conf/oopsla/CousotCLB12} is an abstraction of Hoare logic such
that $\{\bar{P}\}\,\texttt{\small S}\,\{\bar{Q}\}$ means Hoare triple $\{\gamma_1(\bar{P})\}\,\texttt{\small S}\,\{\gamma_2(\bar{Q})\}$. Without appropriate
hypotheses on the abstraction, some rules of Hoare logic like disjunction and conjunction may be invalid in the abstract, see counter-examples and sufficient hypotheses in \cite[pages 219--221]{DBLP:conf/oopsla/CousotCLB12}. Similarly, \cite{DBLP:journals/entcs/GotsmanBC11} provides a counterexample showing the unsoundness of the conjunction rule. This is an argument for the use of a principled method for designing logics.

Another abstract logic \cite{DBLP:journals/jacm/BruniGGR23} combines an over approximation (for correctness) and an under approximation (for incorrectness) in the same abstract domain. The ``(relax)'' rule  requires that the under approximation uses abstract properties $\alpha(P)$ that exactly represent concrete properties $P$ by requiring that $\gamma\comp\alpha(P)=P$. This restricts the concrete points that can be used in the under approximation, and will be a source of incompleteness and imprecision for most static analyses.

Under approximation is the order semi\-dual of an over approximation, with abstraction $\pair{\wp(\Sigma_{\bot})}{\subseteq}\galois{\alpha}{\gamma}\pair{\mathcal{A}}{\sqsupseteq}$ exploited e.g\@. in \cite{DBLP:conf/cav/BallKY05}. The study by
\cite{DBLP:conf/fossacs/AscariBG22} provides a number of classic abstract domain examples showing the imprecision of such under approximation static analyses, but for few exceptions like \cite{DBLP:journals/scp/Mine14,DBLP:conf/pldi/AsadiC0GM21}.

These under approximation approaches are based on Th\@. \ref{th:Fixpoint-Underapproximation} for
 fixpoint under approximation by transfinite iterates. Termination proofs do not use an under approximation but instead an over approximation and a variant function as, e.g., in Th\@. \ref{th:Fixpoint-Underapproximation-Variant}. Alternatively, over approximating static analysis is classic and variant functions can also be inferred by abstract interpretation \cite{DBLP:conf/sas/Urban13,DBLP:conf/sas/UrbanM14,DBLP:conf/esop/UrbanM14,DBLP:conf/vmcai/UrbanM15,DBLP:conf/tacas/Urban15,DBLP:conf/cav/DSilvaU15,DBLP:conf/tacas/UrbanGK16}.
\subsection*{}
\vskip-2em
\subsection{The Subhierarchy of Assertional Logics}\label{sec:subhierarchy-assertional-logics}

Comparing logics means comparing their theories, that is their expressivity, through their respective abstractions of the collecting semantics (as formalized by fixpoint abstraction
in Sect\@. \ref{sect:FixpointAbstraction}), and comparing the induction principles induced by their abstractions (as formalized in Sect\@. \ref{sect:FixpointInduction} by fixpoint induction). For example, figure \ref{fig:taxonomy-assertional} shows that Hoare logic and subgoal induction are different but equivalent abstractions of the collecting semantics so have the same theory and equivalent but different proof systems.

\begin{wrapfigure}{r}{0.525\textwidth}
\vskip-1.5em
\includegraphics[width=0.515\textwidth]{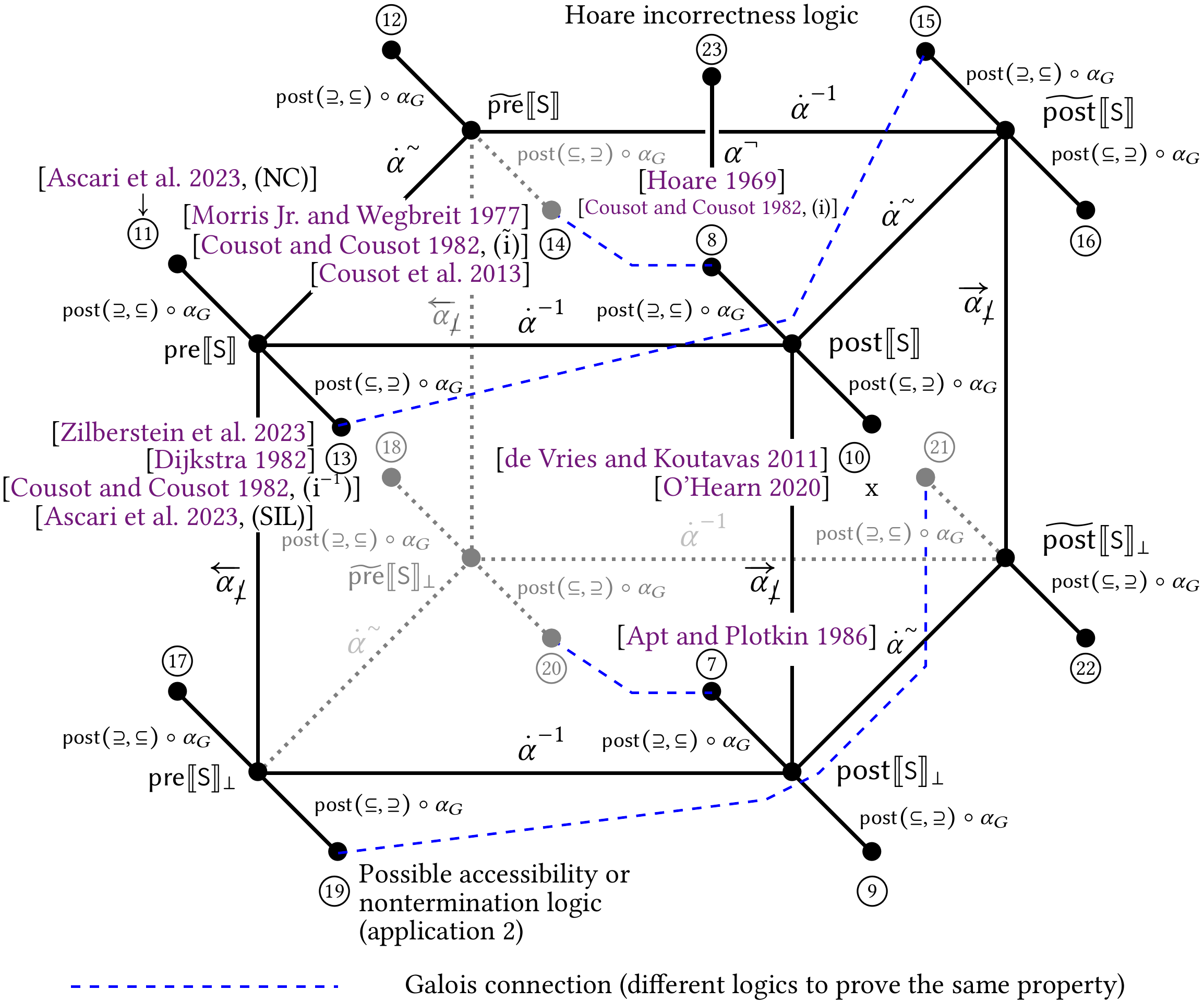}
\caption{Taxonomy of assertional  logics\label{fig:taxonomy-assertional}}
\end{wrapfigure}
These abstractions yield the hierarchical taxonomy of assertional transformational logics of Fig\@. \ref{fig:taxonomy-assertional},
which is a subset of Fig\@. \ref{fig:taxonomy}. Fig\@. \ref{fig:taxonomy-assertional}\ifshort, with a larger instance in the appendix \proofinapx,\fi\ is commented thereafter.
\begin{toappendix}
\begin{figure}[ht]
\mbox{}\hskip-1cm\begin{tikzpicture}[x=5mm,y=5mm,scale=0.6]%
\draw (3,3)   node[name=LFG] {\raisebox{-4pt}[0pt][0pt]{\begin{tabular}[t]{c}\llap{$\textsf{pre}\sqb{\texttt{\small S}}_{\bot}$\ }\huge$\bullet$\end{tabular}}};
\draw (11,11) node[name=LBG,opacity=0.5] {\raisebox{-4pt}[0pt][0pt]{\begin{tabular}[t]{c}\llap{\raisebox{-4pt}[0pt][0pt]{$\widetilde{\textsf{pre}}\sqb{\texttt{\small S}}_{\bot}$\ \ \ }}\huge$\color{gray}\bullet$\end{tabular}}};
\draw (3,19)  node[name=HFG] {\raisebox{-4pt}[0pt][0pt]{\begin{tabular}[t]{c}\llap{$\textsf{pre}\sqb{\texttt{\small S}}$\ }\huge$\bullet$\end{tabular}}};
\draw (11,27) node[name=HBG] {\raisebox{-4pt}[0pt][0pt]{\begin{tabular}[t]{c}\rlap{\hskip1em\raisebox{10pt}[0pt][0pt]{$\widetilde{\textsf{pre}}\sqb{\texttt{\small S}}$}}\huge$\bullet$\end{tabular}}};
\draw (23,3) node[name=LFD] {\raisebox{-4pt}[0pt][0pt]{\begin{tabular}[t]{c}\huge$\bullet$\rlap{\raisebox{2pt}[0pt][0pt]{\normalsize\ \ \ \ $\textsf{\textup{post}}\sqb{\texttt{\small S}}_{\bot}$}}\end{tabular}}};
\draw (31,11) node[name=LBD] {\raisebox{-4pt}[0pt][0pt]{\begin{tabular}[t]{c}\huge$\bullet$\rlap{\raisebox{1ex}{\normalsize\ \ \ $\widetilde{\textsf{\textup{post}}}\sqb{\texttt{\small S}}_{\bot}$}}\end{tabular}}};
\draw (23,19) node[name=HFD] {\small\raisebox{-4pt}[0pt][0pt]{\begin{tabular}[t]{c}\huge$\bullet$\rlap{\raisebox{2pt}[0pt][0pt]{\normalsize\ \ \ $\textsf{\textup{post}}\sqb{\texttt{\small S}}$}}\end{tabular}}};
\draw (31,27) node[name=HBD] {\small\raisebox{-4pt}[0pt][0pt]{\begin{tabular}[t]{c}\huge$\bullet$\rlap{\raisebox{1ex}{\normalsize\ \ \ $\widetilde{\textsf{\textup{post}}}\sqb{\texttt{\small S}}$}}\end{tabular}}};
\draw[very thick] (LFG.center) --  (LFD.center) node[draw=none,fill=none,font=\scriptsize,midway,above] {\large\quad$\stackrel{.}{\alpha}^{-1}$};
\draw[very thick] (LFD.center) --  (LBD.center) 
node[draw=none,fill=none,font=\scriptsize,midway,above] {\normalsize\ $\mathord{\stackrel{.}{\alpha}}^{\sim}$};
\draw[very thick,dotted,opacity=0.5] (LBD.center) --  (LBG.center) node[draw=none,fill=none,font=\scriptsize,midway,above] {\large$\color{gray}\mathord{\stackrel{.}{\alpha}}^{-1}$\qquad\qquad};
\draw[very thick,dotted,opacity=0.5] (LBG.center) --  (LFG.center) node[draw=none,fill=none,font=\scriptsize,midway,above] {\large$\color{gray}\mathord{\stackrel{.}{\alpha}}^{\sim}$};
\draw[very thick] (HFG.center) --  (HFD.center) node[draw=none,fill=none,font=\scriptsize,midway,above] {{\large\quad$\stackrel{.}{\alpha}^{-1}$}};
\draw[very thick] (HFD.center) --  (HBD.center) node[draw=none,fill=none,font=\scriptsize,midway,above] {\large$\mathord{\stackrel{.}{\alpha}}^{\sim}$};
\draw[very thick] (HBD.center) --  (HBG.center) node[draw=none,fill=none,font=\scriptsize,midway,above] {\rlap{\large\quad\ \ $\stackrel{.}{\alpha}^{-1}$}};
\draw[very thick] (HBG.center) --  (HFG.center) node[draw=none,fill=none,font=\scriptsize,midway,above] {\raisebox{17pt}[0pt][0pt]{\hskip3.6em\large$\mathord{\stackrel{.}{\alpha}}^{\sim}$}};
\draw[very thick] (LFG.center) --  (HFG.center) node[draw=none,fill=none,font=\scriptsize,midway,above] {\llap{\raisebox{-2em}[0pt][0pt]{\large$\overleftarrow{\alpha}_{\mskip-7mu\not\bot}$~~}}};
\draw[very thick] (LFD.center) --  (HFD.center) node[draw=none,fill=none,font=\scriptsize,midway,above] {\llap{\raisebox{-2em}[0pt][0pt]{\large$\overrightarrow{\alpha}_{\mskip-7mu\not\bot}$~~}}};
\draw[very thick] (LBD.center) --  (HBD.center) node[draw=none,fill=none,font=\scriptsize,midway,above] {\llap{\raisebox{1em}[0pt][0pt]{\large$\overrightarrow{\alpha}_{\mskip-7mu\not\bot}$~~}}};
\draw[very thick,dotted,opacity=0.5] (LBG.center) --  (HBG.center) node[draw=none,fill=none,font=\scriptsize,midway,above] {\llap{\raisebox{0.45em}[0pt][0pt]{\normalsize$\overleftarrow{\alpha}_{\mskip-7mu\not\bot}$~~}}};
\draw (0,6)   node[name=LFGO] {\small\raisebox{-4pt}[0pt][0pt]{\begin{tabular}[b]{c}\circled{17}\\\huge$\bullet$\\\end{tabular}}};
\draw (6,0)   node[name=LFGU] {\small\raisebox{-4pt}[0pt][0pt]{\begin{tabular}[t]{c}\huge$\bullet$\\\circled{19} \rlap{\raisebox{1.5ex}[0pt][0pt]{Possible accessibility or}}\\[-0.5ex]
\rlap{\raisebox{1.5ex}[0pt][0pt]{\ \ \ nontermination logic}}\\[-0.5ex]\rlap{\raisebox{1.5ex}[0pt][0pt]{\ \ \ (application 2)}}\end{tabular}}};
\draw[very thick] (LFG.center) --  (LFGO.center) node[draw=none,fill=none,font=\scriptsize,midway,below] {\llap{\raisebox{0pt}[0pt][0pt]{$\textsf{\textup{post}}({\supseteq},{\subseteq})\comp\alpha_G$\ }}}; 
\draw[very thick] (LFG.center) --  (LFGU.center) node[draw=none,fill=none,font=\scriptsize,midway,above]  {\rlap{\raisebox{-3pt}[0pt][0pt]{\ \ \ $\textsf{\textup{post}}({\subseteq},{\supseteq})\comp\alpha_G$}}};
\draw (14,8) node[name=LBGU] {\small\raisebox{-4pt}[0pt][0pt]{\begin{tabular}[t]{c}\huge$\color{gray}\bullet$\\[-0.75ex]\circledgray{20}\end{tabular}}};
\draw (8,14) node[name=LBGO] {\small\raisebox{-4pt}[0pt][0pt]{\begin{tabular}[b]{c}\circledgray{18}\\\huge$\color{gray}\bullet$\end{tabular}}};
\draw[very thick,dotted,opacity=0.5] (LBG.center) --  (LBGO.center) node[draw=none,fill=none,font=\scriptsize,midway,below] {\llap{\raisebox{-5pt}[0pt][0pt]{$\textsf{\textup{post}}({\supseteq},{\subseteq})\comp\alpha_G$}}} ; 
\draw[very thick,dotted,opacity=0.5] (LBG.center) --  (LBGU.center) node[draw=none,fill=none,font=\scriptsize,midway,above] {$\rlap{\ \raisebox{-2pt}[0pt][0pt]{$\textsf{\textup{post}}({\subseteq},{\supseteq})\comp\alpha_G$}}$};
\draw (0,22)  node[name=HFGO] {\small\raisebox{-4pt}[0pt][0pt]{\begin{tabular}[b]{c}{\raisebox{1ex}[0pt][0pt]{\cite[(NC)]{DBLP:journals/corr/abs-2310-18156}}} \\[-1.75ex]\llap{$\downarrow$\hskip1em}\\[-0.33ex]
\llap{\circled{11}\hskip0.6em}\\\huge$\bullet$\end{tabular}}};
\draw [ultra thick, white] (3,12) --  (3,16.3);;
\draw (6,16)  node[name=HFGU] {\small\raisebox{-5pt}[0pt][0pt]{\begin{tabular}[t]{c}\llap{\cite{DBLP:journals/pacmpl/ZilbersteinDS23}\hskip0.2em} \huge$\bullet$\\[-0.7ex]\llap{\cite{Dijkstra-EWD-576}} \circled{13}\\[-0.5ex]\llap{\cite[($\mskip1mu{\textrm{\upshape i}^{\scriptscriptstyle -1}}\mskip1mu$)]{CousotCousot82-TNPC}\hskip-2ex}\\[-0.5ex]\llap{\cite[(SIL)]{DBLP:journals/corr/abs-2310-18156}\hskip0.8em}\end{tabular}}};
\draw[very thick] (HFG.center) --  (HFGO.center) node[draw=none,fill=none,font=\scriptsize,midway,below] {\llap{\raisebox{0pt}[0pt][0pt]{$\textsf{\textup{post}}({\supseteq},{\subseteq})\comp\alpha_G$\ }}} ;
\draw[very thick] (HFG.center) --  (HFGU.center) node[draw=none,fill=none,font=\scriptsize,midway,above] {{$\rlap{\ \ \ \raisebox{-5pt}[0pt][0pt]{$\textsf{\textup{post}}({\subseteq},{\supseteq})\comp\alpha_G$}}$}};
\draw [fill=white,white] (5,21) rectangle (8.5,24.5); 
\draw [ultra thick, white] (11,21) --  (11,24.5);;
\draw (14,24) node[name=HBGO] {\small\raisebox{-4pt}[0pt][0pt]{\begin{tabular}[t]{c}\llap{\cite{DBLP:journals/cacm/MorrisW77}\hskip0.3em}\huge$\color{gray}\bullet$\\[-0.5ex]\llap{\cite[($\mskip1mu\tilde{\textrm{\upshape i}}\mskip1mu$)]{CousotCousot82-TNPC}} \circled{14}\\[-0.5ex]\llap{\cite{DBLP:conf/vmcai/CousotCFL13}\hskip0.7em}\end{tabular}}}; 
\draw (8,30) node[name=HBGU] {\small\raisebox{-4pt}[0pt][0pt]{\begin{tabular}[b]{c}\circled{12}\\\huge$\bullet$\end{tabular}}}; 
\draw[very thick,dotted,opacity=0.5] (HBG.center) --  (HBGO.center) node[draw=none,fill=none,font=\scriptsize,midway,above] {{$\rlap{\ \raisebox{0pt}[0pt][0pt]{$\textsf{\textup{post}}({\subseteq},{\supseteq})\comp\alpha_G$}}$}};
\draw[very thick] (HBG.center) --  (HBGU.center) node[draw=none,fill=none,font=\scriptsize,midway,below] {\llap{\raisebox{0pt}[0pt][0pt]{$\textsf{\textup{post}}({\supseteq},{\subseteq})\comp\alpha_G$\ }}}; 
\draw [ultra thick, white] (23,7) --  (23,9);;
\draw (20,6)   node[name=LFDO] {\small\raisebox{-4pt}[0pt][0pt]{\begin{tabular}[b]{c}\hskip2.5em\cite{DBLP:journals/jacm/AptP86}\\[-0.5ex]\circled{7}\\[-0.25ex]\huge$\bullet$\end{tabular}}};
\draw (26,0)   node[name=LFDU] {\small\raisebox{-4pt}[0pt][0pt]{\begin{tabular}[t]{c}\huge$\bullet$\\\circled{9}\end{tabular}}};
\draw[very thick] (LFD.center) --  (LFDO.center) node[draw=none,fill=none,font=\scriptsize,midway,below] {$\llap{\raisebox{-1pt}[0pt][0pt]{$\textsf{\textup{post}}({\supseteq},{\subseteq})\comp\alpha_G$}}$};
\draw[very thick] (LFD.center) --  (LFDU.center) node[draw=none,fill=none,font=\scriptsize,midway,below] {$\rlap{\raisebox{0pt}[0pt][0pt]{\ \ \ \ \ $\textsf{\textup{post}}({\subseteq},{\supseteq})\comp\alpha_G$}}$};
\draw (34,8) node[name=LBDU] {\small\raisebox{-4pt}[0pt][0pt]{\begin{tabular}[t]{c}\huge$\bullet$\\[-0.75ex]\circled{22}\end{tabular}}};
\draw (28,14) node[name=LBDO] {\small\raisebox{-4pt}[0pt][0pt]{\begin{tabular}[b]{c}\quad \circledgray{21}\\\huge$\color{gray}\bullet$\end{tabular}}};
\draw[very thick] (LBD.center) --  (LBDU.center) node[draw=none,fill=none,font=\scriptsize,midway,above] {$\rlap{\ \raisebox{0pt}[0pt][0pt]{$\textsf{\textup{post}}({\subseteq},{\supseteq})\comp\alpha_G$}}$}; 
\draw [ultra thick, white] (23,13) --  (23,15.5);;
\draw (20,22)  node[name=HFDO] {};
\draw (20,29)   node[name=negHL] {\raisebox{-4pt}[0pt][0pt]{\begin{tabular}[b]{c}Hoare incorrectness logic\\\circled{23}\\[-0.25ex]\huge$\bullet$\end{tabular}}};
\draw [ultra thick, white] (19.75,27) --  (20.25,27);;
\draw[very thick] (HFDO.center) --  (negHL.center) node[draw=none,fill=none,font=\scriptsize,midway,right] {\raisebox{3pt}[0pt][0pt]{\large$\alpha^{\neg}$}};
\draw [ultra thick, white] (20,21.75) --  (20,25.65);;
\draw (20,22)  node[name=HFDO] {\small\raisebox{-5pt}[0pt][0pt]{\begin{tabular}[b]{c}\cite{DBLP:journals/cacm/Hoare69}\\[-0.76ex]{\hskip-2.5pt\scriptsize\cite[(i)]{CousotCousot82-TNPC}}\\\circled{8}\\[-1pt]\huge$\bullet$\end{tabular}}};
\draw (26,16)  node[name=HFDU] {\small\raisebox{-4pt}[0pt][0pt]{\begin{tabular}[t]{c}\huge$\bullet$\\[-0.5ex]\llap{\cite{DBLP:conf/sefm/VriesK11}\ \circled{10}}\\[-0.5ex]\llap{\cite{DBLP:journals/pacmpl/OHearn20}\hskip1.2em}x\end{tabular}}};
\draw[very thick] (HFD.center) --  (HFDO.center) node[draw=none,fill=none,font=\scriptsize,midway,below] {\llap{\raisebox{0pt}[0pt][0pt]{$\textsf{\textup{post}}({\supseteq},{\subseteq})\comp\alpha_G$\ }}} ; 
\draw[very thick] (HFD.center) --  (HFDU.center) node[draw=none,fill=none,font=\scriptsize,midway,above] {{$\rlap{\ \ \ \raisebox{-5pt}[0pt][0pt]{$\textsf{\textup{post}}({\subseteq},{\supseteq})\comp\alpha_G$}}$}};
\draw (28,30) node[name=HBDO] {\small\raisebox{10pt}{\begin{tabular}[b]{c}\circled{15}\\\huge$\color{black}\bullet$\end{tabular}}};
\draw (34,24) node[name=HBDU] {\small\raisebox{-4pt}[0pt][0pt]{\begin{tabular}[t]{c}\huge$\bullet$\\[-0.75ex]\circled{16}\end{tabular}}};
\draw[very thick] (HBD.center) --  (HBDO.center) node[draw=none,fill=none,font=\scriptsize,midway,above] {\rlap{\raisebox{0pt}[0pt][0pt]{$\textsf{\textup{post}}({\supseteq},{\subseteq})\comp\alpha_G$\ }}} ; 
\draw[very thick] (HBD.center) --  (HBDU.center) node[draw=none,fill=none,font=\scriptsize,midway,above] {{$\rlap{\ \raisebox{0pt}[0pt][0pt]{$\textsf{\textup{post}}({\subseteq},{\supseteq})\comp\alpha_G$}}$}};
\draw [dashed, blue, thick] (HFGU.center) -- (23,20) -- (HBDO.center) ; 
\draw [dashed, blue, thick] (LFGU.center) -- (22,2) -- (24,3) -- (28,7) -- (LBDO.center) ; 
\draw [ultra thick, white] (28,11.5) --  (28,12.5);
\draw[very thick,dotted,opacity=0.5] (LBD.center) --  (LBDO.center) node[draw=none,fill=none,font=\scriptsize,midway,below] {$\llap{\raisebox{-3pt}[0pt][0pt]{$\textsf{\textup{post}}({\supseteq},{\subseteq})\comp\alpha_G$}}$};\draw [dashed, blue, thick] (HBGO.center) --(17,22) -- (HFDO.center) ;
\draw [dashed, blue, thick] (LBGU.center) --(17,6) -- (LFDO.center) ;
\draw (11,11) node {\small\raisebox{-4pt}[0pt][0pt]{\begin{tabular}[t]{c}\huge$\color{gray}\bullet$\end{tabular}}};
\draw (14,8) node {\small\raisebox{-4pt}[0pt][0pt]{\begin{tabular}[t]{c}\huge$\color{gray}\bullet$\end{tabular}}};
\draw (8,14) node {\small\raisebox{-4pt}[0pt][0pt]{\begin{tabular}[t]{c}\huge$\color{gray}\bullet$\end{tabular}}};
\draw (14,24) node{\small\raisebox{-4pt}[0pt][0pt]{\begin{tabular}[t]{c}\huge$\color{gray}\bullet$\end{tabular}}};
\draw (28,14) node{\small\raisebox{-4pt}[0pt][0pt]{\begin{tabular}[b]{c}\huge$\color{gray}\bullet$\end{tabular}}};
\draw [dashed, blue, thick] (-4,-5) -- (6,-5);
\draw (22,-5) node {Galois connection (different logics to prove the same property)};
\end{tikzpicture}
\vskip-2mm
\newcounter{myfigure}
\setcounter{myfigure}{\value{figure}}
\setcounter{figure}{2}%
\caption{Hierarchical taxonomy of transformational assertional logics\label{fig:taxonomy-assertional}}
\setcounter{figure}{\value{myfigure}}
\end{figure}
\end{toappendix}

\makeatletter
\renewcommand\subsection{\@startsection{subsection}{2}{\z@}%
  {-0.01\baselineskip \@plus -1\p@ \@minus -.1\p@}%
  {0.01\baselineskip}%
  {\ACM@NRadjust\@subsecfont}}
\makeatother

We use \emph{universal} to mean for all initial or final states and \emph{existential} to mean there exists at least one initial or final state. We use \emph{reachability} (often forward) for initial to final states and \emph{accessibility} (often backward) for final to initial states. We use  \emph{definite} to mean ``for all executions'' and \emph{possible} to mean ``for some execution'' (maybe none). In both cases, the qualification does not exclude possible nontermination or blocking states, which is emphasized by \emph{partial\/}. We use \emph{total} to mean that all executions must be finite. We use \emph{blocking} to mean a state, which is not final, but from which execution cannot go on. No such blocking states exist in the semantics $\sqb{\texttt{\small S}}_\bot$ of statements $\texttt{\small S}$ in Sect\@. \ref{sec:natural-relational-semantics-deductive} and \ref{sect:FixpointNaturalRelationalSemantics} but would correspond e.g.\ to an aborted execution after a runtime error (like a division by zero). 

\smallskip

The taxonomy for \textbf{direct proofs} (the hypothesis implies the conclusion) is illustrated in Fig\@.~\ref{fig:taxonomy-assertional}.
\newcommand{\PartialDefiniteAccessibilityOfSomeFinalStateFromAllInitialStates}{\normalfont$\textsf{\upshape post}\sqb{\texttt{\small S}}P\subseteq Q \Leftrightarrow P\subseteq\widetilde{\textsf{\upshape pre}}\sqb{\texttt{\small S}}Q$, $P, Q\in\wp(\Sigma)$}
\subsubsection{Partial Definite Accessibility of Some Final State From All Initial States \protect\PartialDefiniteAccessibilityOfSomeFinalStateFromAllInitialStates}\label{PartialDefiniteAccessibilityOfSomeFinalStateFromAllInitialStates}
Partial correctness, allowing blocking states, characterizes executions starting from any initial state in $P$, which, if  terminating normally, do terminate in a state of $Q$ and no other one. So blocking states are not excluded. This is Naur \cite{Naur66-1}, Hoare \cite{DBLP:journals/cacm/Hoare69} partial correctness, and Dijkstra weakest liberal preconditions \cite{DBLP:books/ph/Dijkstra76} and the partial correctness part of
Turing \cite{Turing49-program-proof} and Floyd
\cite{Floyd67-1} total correctness.

\smallskip

\noindent \circled{8} 
 $\textsf{\textup{post}}({\supseteq},{\subseteq})\comp\alpha_G(\textsf{\upshape post}\sqb{\texttt{\small S}})\triangleq\{\pair{P}{Q}\in \wp(\Sigma)\times \wp(\Sigma)\mid \textsf{\upshape post}\sqb{\texttt{\small S}}P\subseteq Q\}$ yields the theory of Hoare logic \cite{DBLP:journals/cacm/Hoare69}. 
This claim can be substantiated by (re)constructing Hoare logic by  abstracting the angelic relational semantics $\sqb{\texttt{\small S}}=\alpha^{\demoniac}(\sqb{\texttt{\small S}}_{\bot})$ by $\textsf{\textup{post}}({\supseteq},{\subseteq})\comp\alpha_G\comp\textsf{\upshape post}$. This has been done, e.g., in \cite[Theorem 1, page 79]{DBLP:journals/siamcomp/Cook78} (modulo a later correction in \cite{DBLP:journals/siamcomp/Cook81}) as well as in 
\cite[Chapter 26]{Cousot-PAI-2021}, although using the intermediate abstractions into an equational semantics and then verification conditions to explain Turing-Floyd's transition based invariance proof method.

\smallskip

\noindent \circled{14} By Galois connection (\ref{eq:def:post:GC}), $\textsf{\textup{post}}({\subseteq},{\supseteq})\comp\alpha_G(\widetilde{\textsf{\upshape pre}}\sqb{\texttt{\small S}})\triangleq\{\pair{P}{Q}\in \wp(\Sigma)\times \wp(\Sigma)\mid P\subseteq\widetilde{\textsf{\upshape pre}}\sqb{\texttt{\small S}} Q\}$ is equivalent and yields the theory of a logic axiomatizing subgoal induction \cite{DBLP:journals/cacm/MorrisW77} or necessary preconditions \cite{DBLP:conf/vmcai/CousotCL11,DBLP:conf/vmcai/CousotCFL13}. 

\smallskip

Hoare and subgoal induction logics can be used to prove universal partial correctness ($Q$ is good, as in static accessibility analysis \cite{DBLP:conf/popl/CousotC77}) and universal partial incorrectness ($Q$ is bad, as in necessary preconditions analyses \cite{DBLP:conf/vmcai/CousotCL11,DBLP:conf/vmcai/CousotCFL13}).
Both logics can be also  used to prove bounded termination, by introducing a counter incremented in loops and proved to be bounded \cite{DBLP:journals/acta/LuckhamS77}. However, this is incomplete for unbounded nondeterminism. $\textsf{\upshape post}\sqb{\texttt{\small S}}P\subseteq \emptyset$ $\Leftrightarrow$
$P\subseteq\widetilde{\textsf{\upshape pre}}\sqb{\texttt{\small S}}\emptyset$ $\Leftrightarrow$
$P\subseteq\neg{\textsf{\upshape pre}}\sqb{\texttt{\small S}}\Sigma$ $\Leftrightarrow$ ${\textsf{\upshape pre}}\sqb{\texttt{\small S}}\Sigma \subseteq \neg P$
is
\emph{definite nontermination from all initial states} (executions from any initial state of $P$ do not terminate).

Subgoal induction is exploited in necessary preconditions analyses \cite{DBLP:conf/vmcai/CousotCL11,DBLP:conf/vmcai/CousotCFL13}. 
Finding $P$ such that ${\textsf{\upshape post}\sqb{\texttt{\small S}}}P\subseteq Q$ is equivalent to finding $P$ such that $P\subseteq{\widetilde{\textsf{\upshape pre}}\sqb{\texttt{\small S}}}Q$ for the given error postcondition $Q$, which the necessary precondition analysis does by under approximating ${\widetilde{\textsf{\upshape pre}}\sqb{\texttt{\small S}}}$ defined structurally on the programming language and using fixpoint under approximation to handle iteration and recursion.

\newcommand{\TotalDefiniteAccessibilityOfSomeFinalStateFromAllInitialStates}{\normalfont$\textsf{\upshape post}\sqb{\texttt{\small S}}_{\bot}P\subseteq Q \Leftrightarrow P\subseteq\widetilde{\textsf{\upshape pre}}\sqb{\texttt{\small S}}_{\bot}Q$, $P, Q\in\wp(\Sigma)$}

\subsubsection{Total Definite Accessibility of Some Final States From All Initial States   \protect\TotalDefiniteAccessibilityOfSomeFinalStateFromAllInitialStates}\label{TotalDefiniteAccessibilityOfSomeFinalStateFromAllInitialStates}
Total correctness, allowing blocking states, characterizes executions from any initial state in $P$ that do terminate normally in a final state satisfying $Q$ or block. Taking $Q=\Sigma$ is universal definite termination.

\smallskip

\noindent\circled{7} The 
Turing \cite{Turing49-program-proof} \& Floyd \cite{Floyd67-1} proof method uses an invariant and a variant function into a well-founded set. The abstraction  $\textsf{\textup{post}}({\supseteq},{\subseteq})\comp\alpha_G(\textsf{\upshape post}\sqb{\texttt{\small S}}_{\bot})\triangleq\{\pair{P}{Q}\in \wp(\Sigma)\times \wp(\Sigma)\mid \textsf{\upshape post}\sqb{\texttt{\small S}}_{\bot}P\subseteq Q\}$ yields the theory of
Apt and Plotkin \cite{DBLP:journals/jacm/AptP86} logic in the assertional case (and that of Manna \& Pnueli logic \cite{DBLP:journals/acta/MannaP74} in the relational case). This claims follows from  \cite{DBLP:journals/jacm/AptP86} for an imperative language and \cite{DBLP:journals/tcs/Cousot02} for arbitrary transition systems. The logic can be used to prove definite correctness or incorrectness.

\newcommand{\PartialpossibleAccessibilityOfAllFinalStatesFromSomeInitialState}{\normalfont$Q\subseteq\textsf{\upshape post}\sqb{\texttt{\small S}}_{\bot}P$ $\Leftrightarrow$ $Q\subseteq\textsf{\upshape post}\sqb{\texttt{\small S}}P$, $P,Q\in\wp(\Sigma)$}
\subsubsection{Partial Possible Accessibility of All Final States From Some Initial State 
\protect\PartialpossibleAccessibilityOfAllFinalStatesFromSomeInitialState}\label{PartialpossibleAccessibilityOfAllFinalStatesFromSomeInitialState}
This means that for any final state $\sigma'$ in $Q$ there exists at least one initial state $\sigma$ in $P$  and an execution from $\sigma$ that will terminate in state $\sigma'$. Blocking states $\sigma$ may be included in $P$. Moreover,
this does not preclude executions from $\sigma$ to make nondeterministic choices terminating normally with $\neg Q$ or do not terminate at all.

\smallskip

\noindent \circled{10} By \cite[Definition 1]{DBLP:conf/sefm/VriesK11}, $\textsf{\textup{post}}({\subseteq},{\supseteq})\comp\alpha_G(\textsf{\upshape post}\sqb{{\texttt{\small S}}})$ is the theory of De Vries and Koutavas reversed Hoare logic.  This is also confirmed by the soundness and completeness proofs in  \cite[section 6]{DBLP:conf/sefm/VriesK11} based on a ``weakest postcondition calculus'' defined in \cite[section 5]{DBLP:conf/sefm/VriesK11} as ``wpo($P$, c), [is] the weakest postcondition given a precondition $P$ and program c''. So ``wpo'' is nothing other than $\textsf{\upshape post}$ and ``$\singleton{P}\:c\:\singleton{Q}$ is a valid triple if and only if $Q$ $\Rightarrow$ wpo($P$,c)''.

By \cite[\textsc{Fact 13}]{DBLP:journals/pacmpl/OHearn20}, this is also incorrectness logic requiring any bug in $Q$ to be possibly reachable in finitely many steps from $P$ thus discarding infinite executions as possible errors. 

The difference is in the examples handled where $Q$ is ``good'' for De Vries and Koutavas and ``bad''
for O'Hearn.

\newcommand{\PartialpossibleAccessibilityOfSomeFinalStateFromAllInitialStates}{\normalfont$P\subseteq\textsf{\upshape pre}\sqb{\texttt{\small S}}Q$, $P,Q\in\wp(\Sigma)$}
\subsubsection{Partial Possible Accessibility of Some Final State From All Initial States 
\protect\PartialpossibleAccessibilityOfSomeFinalStateFromAllInitialStates}\label{PartialpossibleAccessibilityOfSomeFinalStateFromAllInitialStates}
This prescribes that all initial states in $P$ have at least one execution that does reach $Q$. 

\smallskip

\noindent \circled{14}  Dijkstra \cite{Dijkstra-EWD-576} shown the equivalence of $\textsf{\upshape post}{\sqb{\texttt{\small S}}}P\subseteq Q$ (that is, Turing-Floyd-Naur-Hoare partial correctness and $P\subseteq\textsf{\upshape pre}{\sqb{\texttt{\small S}}}Q$ (that is, Morris  and Wegbreit subgoal induction, claiming ``subgoal induction is indeed the next variation on an old theme''). By (\ref{eq:def:post:GC}) this should have been $P\subseteq\widetilde{\textsf{\upshape pre}}{\sqb{\texttt{\small S}}}Q$ in general, but Dijkstra considers total deterministic programs for which $\textsf{\upshape pre}$ =
$\widetilde{\textsf{\upshape pre}}$.  This is also the incorrectness part of the outcome logic \cite{DBLP:journals/pacmpl/ZilbersteinDS23}, the induction principle ($\mskip1mu{\textrm{\upshape i}^{\scriptscriptstyle -1}}\mskip1mu$) of \cite[p\@. 100]{CousotCousot82-TNPC}, and (SIL) in \cite{DBLP:journals/corr/abs-2310-18156}.

\newcommand{\PossibleAccessibilityOfSomeFinalStateOrNonterminationFromAllInitialStates}{\normalfont$
P\subseteq\textsf{\upshape pre}\sqb{\texttt{\small S}}_{\bot}Q$, $P\in\wp(\Sigma), Q\in\wp(\Sigma_\bot)$}
\subsubsection{Possible Accessibility of Some Final State or Nontermination From All Initial States \protect\PossibleAccessibilityOfSomeFinalStateOrNonterminationFromAllInitialStates}\label{PossibleAccessibilityOfSomeFinalStateOrNonterminationFromAllInitialStates}
For $Q=\Sigma$, this is \emph{possible termination from all initial states} \circled{17}.  For $Q=\{\bot\}$, this is \emph{possible nontermination from all initial states}. Similarly, \circled{11} is $
P\subseteq\textsf{\upshape pre}\sqb{\texttt{\small S}}Q$, named (NC) in
\cite{DBLP:journals/corr/abs-2310-18156}.

\smallskip

\noindent \circled{19} This logic will be formally developed by calculus in Sect\@. \ref{sec:design:possibleAccessibilityNontermination}.

\smallskip

We can also consider \textbf{disproofs of program properties} by the abstraction $\alpha^{\neg}$ (\ref{eq-complement-GC}) of the theory of a program logic.

\newcommand{\PartialpossibleAccessibilityOfSomeFinalStateFromSomeInitialState}{\normalfont$\textsf{\upshape post}\sqb{\texttt{\small S}}P\cap Q\neq\emptyset$
for $P,Q\in\wp(\Sigma)$ (or 
$\textsf{\upshape post}\sqb{\texttt{\small S}}_{\bot}P\cap Q\neq\emptyset$ for $Q\in\wp(\Sigma_\bot$))}
\subsubsection{Partial Possible Accessibility of Some Final States (or Nontermination) From Some Initial States 
\protect\PartialpossibleAccessibilityOfSomeFinalStateFromSomeInitialState}\label{PartialpossibleAccessibilityOfSomeFinalStateFromSomeInitialState}
This means that at least one execution from at least one initial state in $P$ does terminate in a final state satisfying Q. Taking $Q=\Sigma$ is 
\emph{possible termination from some initial states}. 

\smallskip

\noindent\circled{23} Disproving a Hoare triple using the proof system would require to show that no proof does exist for this triple, a method no one
ever consider. One can use incorrectness logic \cite{DBLP:journals/pacmpl/OHearn20} or provide a counter-example (not supported by a logic).
The \emph{Hoare incorrectness logic} \circled{23} can be used to prove that a Hoare specification is violated with a possible counter-example, since
$\neg\bigl(\{P\}\texttt{\small S}\{Q\}\bigr)$ 
= 
$\neg\bigl(\textsf{\textup{post}}\sqb{S}P\subseteq Q\bigr)$ 
= 
$\textsf{\textup{post}}\sqb{S}P\cap \neg Q\neq\emptyset$. It's nothing but debugging in logic form.

This is weaker that the requirements of incorrectness logic, for which the principle of denial \cite[Fig\@. 1]{DBLP:journals/pacmpl/OHearn20}
states that if $Q\subseteq\textsf{\textup{post}}\sqb{\texttt{\small S}}P \wedge \neg(Q\subseteq Q')$ then $Q\cap\neg Q'\neq \emptyset$ and therefore 
$\textsf{\textup{post}}\sqb{\texttt{\small S}}P\cap \neg Q' \neq \emptyset$
that is, $\neg\bigl(\{P\}\texttt{\small S}\{Q'\}\bigr)$. However the converse is not true since the violation of $\{P\}\texttt{\small S}\{Q\}$
only require one state of $P$ definitely reaches one state not satisfying $Q$.

\ifshort Other contrapositive logics or logics for disproving program properties are considered in the appendix \proofinapx.\fi
\begin{toappendix}
\ifshort 

\section{The subhierarchy of assertional logics, continued}

Continuing Sect\@. \ref{sec:subhierarchy-assertional-logics} on ``The subhierarchy of assertional logics'', we consider other \textbf{contrapositive} logics or properties \textbf{disproving program properties}.\bgroup\let\subsubsection\subsection\fi
\newcommand{\PartialDefiniteInaccessibilityOfAllFinalStatesFromAllNoninitialStates}{\normalfont$\textsf{\upshape pre}\sqb{\texttt{\small S}}Q\subseteq P
\Leftrightarrow
Q\subseteq \widetilde{\textsf{\upshape post}}\sqb{\texttt{\small S}}P
\Leftrightarrow
\neg P\subseteq \widetilde{\textsf{\upshape pre}}\sqb{\texttt{\small S}}\neg Q
\Leftrightarrow
\textsf{\upshape pre}\sqb{\texttt{\small S}}_{\bot}Q\subseteq P
\Leftrightarrow
Q\subseteq \widetilde{\textsf{\upshape post}}\sqb{\texttt{\small S}}_{\bot}P
\Leftrightarrow
\neg P\subseteq \widetilde{\textsf{\upshape pre}}\sqb{\texttt{\small S}}_{\bot}\neg Q$, $P,Q\in\wp(\Sigma)$}
\subsubsection{Partial Definite Inaccessibility of All Final States From All Non-Initial States 
\protect\PartialDefiniteInaccessibilityOfAllFinalStatesFromAllNoninitialStates}
\label{PartialDefiniteInaccessibilityOfAllFinalStatesFromAllNoninitialStates}
This specifies the fact that any executions from initial states not in $P$ can never reach a state in $Q$. This provides a necessary precondition for partial possible accessibility of  some final states from some initial states but does not prevent nontermination. This formalizes by Galois connections the relation inductive principles ($\overline{\textup{I}}$), ($\widetilde{\overline{\textup{I}}}$), ($\overline{\textup{I}^{-1}}$), ($\widetilde{\overline{\textup{I}^{-1}}}$) and  the assertional ones 
($\overline{\textup{i}}$), ($\widetilde{\overline{\textup{i}}}$), ($\overline{\textup{i}^{-1}}$), ($\widetilde{\overline{\textup{i}^{-1}}}$) of \cite{CousotCousot82-TNPC}.

Notice that complement duality can be used to perform under (respectively over) approximation by over (respectively under) approximations of the complement.

\newcommand{\PartialDefiniteInaccessibilityOfSomeFinalStateFromAllNoninitialStates}{\normalfont$\widetilde{\textsf{\upshape post}}\sqb{\texttt{\small S}}P \cap Q\neq\emptyset$
$\Leftrightarrow$ 
$\widetilde{\textsf{\upshape post}}\sqb{\texttt{\small S}}_{\bot}P \cap Q\neq\emptyset$, $P, Q\in\wp(\Sigma)$}
\subsubsection{Partial Definite Inaccessibility of Some Final State From All Non-Initial States 
 \protect\PartialDefiniteInaccessibilityOfSomeFinalStateFromAllNoninitialStates}\label{PartialDefiniteInaccessibilityOfSomeFinalStateFromAllNoninitialStates}
This expresses that there exists at least one state $\sigma'$ in $Q$, which is not accessible by 
any execution from an initial state not in $P$.

\newcommand{\PartialDefiniteAccessibilityOfSomeFinalStateFromSomeInitialState}{\normalfont$\textsf{\upshape pre}\sqb{\texttt{\small S}}Q\cap P\neq\emptyset
\Leftrightarrow
\textsf{\upshape pre}\sqb{\texttt{\small S}}_{\bot}Q\cap P\neq\emptyset$, $P, Q\in\wp(\Sigma)$}
\subsubsection{Partial Definite Accessibility of Some Final State From Some Initial State \protect\PartialDefiniteAccessibilityOfSomeFinalStateFromSomeInitialState}
\label{PartialDefiniteAccessibilityOfSomeFinalStateFromSomeInitialState}
This indicates that there is at least one initial state in $P$ from which there is at least one execution that does terminate in state $Q$. We have $\neg(\{P\}\texttt{\small S}\{Q\})$ $\Leftrightarrow$ 
$\neg(P\subseteq \widetilde{\textsf{\upshape pre}}\sqb{\texttt{\small S}} Q)$ $\Leftrightarrow$ 
$(P\cap \neg(\widetilde{\textsf{\upshape pre}}\sqb{\texttt{\small S}} Q))\neq\emptyset$ $\Leftrightarrow$ 
$P\cap  \textsf{\upshape pre}\sqb{\texttt{\small S}}\neg Q\neq\emptyset$ which is another incorrectness logic
\cite[Ch\@. 50]{Cousot-PAI-2021}.

\newcommand{\PartialpossibleAccessibilityOfSomeNonfinalStateFromAllNoninitialStates}{\normalfont$\widetilde{\textsf{\upshape pre}}\sqb{\texttt{\small S}}Q \subseteq P$
$\Leftrightarrow$
$\widetilde{\textsf{\upshape pre}}\sqb{\texttt{\small S}}_{\bot}Q \subseteq P$
$\Leftrightarrow$
$\neg P\subseteq{\textsf{\upshape pre}}\sqb{\texttt{\small S}}\neg Q$
$\Leftrightarrow$
$\neg P\subseteq{\textsf{\upshape pre}}\sqb{\texttt{\small S}}_{\bot}\neg Q
$, $P, Q\in\wp(\Sigma)$}
\subsubsection{Partial Possible Accessibility of Some Non-Final State From All Non-Initial States \protect\PartialpossibleAccessibilityOfSomeNonfinalStateFromAllNoninitialStates}\label{PartialpossibleAccessibilityOfSomeNonfinalStateFromAllNoninitialStates}
The meaning is that executions from initial states not in $P$ have at least one execution not reaching a state in $Q$. This provides a necessary precondition for universal possible accessibility \cite{DBLP:conf/vmcai/CousotCFL13}. 

\newcommand{\PartialpossibleAccessibilityOfAllNonfinalStatesFromSomeNoninitialState}{\normalfont$\widetilde{\textsf{\upshape post}}\sqb{\texttt{\small S}}P \subseteq Q$
$\Leftrightarrow$
$\neg Q\subseteq {\textsf{\upshape post}}\sqb{\texttt{\small S}}\neg P$ 
$\Leftrightarrow$ 
$\neg Q \subseteq{\textsf{\upshape post}}\sqb{\texttt{\small S}}_{\bot} \neg P$
$\Leftrightarrow$ 
$\neg Q \subseteq \neg \widetilde{\textsf{\upshape post}}\sqb{\texttt{\small S}}_{\bot}   P $
$\Leftrightarrow$ 
$\widetilde{\textsf{\upshape post}}\sqb{\texttt{\small S}}_{\bot}P \subseteq Q
$}
\subsubsection{Partial Possible Accessibility of All Non-Final States From Some Non-Initial State \protect\PartialpossibleAccessibilityOfAllNonfinalStatesFromSomeNoninitialState}
\label{PartialpossibleAccessibilityOfAllNonfinalStatesFromSomeNoninitialState}
The signification is that for any state $\sigma'$ not in $Q$ there exists at least one initial state $\sigma$ not in $P$ and an execution from $\sigma$ that will terminate in state $\sigma'$. Letting $P'=\neg P$ and $Q'=\neg Q$, this is partial possible accessibility of all final states from some initial state  $Q'\subseteq\textsf{\upshape post}\sqb{\texttt{\small S}}P'$ from Sect\@. \ref{PartialpossibleAccessibilityOfAllFinalStatesFromSomeInitialState}. This shows that the under approximation $Q\subseteq\textsf{\upshape post}\sqb{\texttt{\small S}}P$ is equivalent to an over approximation  $\widetilde{\textsf{\upshape post}}\sqb{\texttt{\small S}}\neg P \subseteq \neg Q$ of the complement, that is, a proof by contradiction.

\newcommand{\TotalDefiniteAccessibilityOfSomeFinalStateFromSomeInitialState}{\normalfont$\widetilde{\textsf{\upshape pre}}\sqb{\texttt{\small S}}Q \cap P\neq\emptyset$ 
$\Leftrightarrow$
$\widetilde{\textsf{\upshape pre}}\sqb{\texttt{\small S}}_{\bot}Q \cap P\neq\emptyset$, $P, Q\in\wp(\Sigma)$}
\subsubsection{Total Definite Accessibility of Some Final State From Some Initial State  \protect\TotalDefiniteAccessibilityOfSomeFinalStateFromSomeInitialState}
\label{TotalDefiniteAccessibilityOfSomeFinalStateFromSomeInitialState}
This states that there is at least one initial state in $P$ from which all executions do terminate in $Q$. 
\ifshort\egroup\fi
\end{toappendix}

\smallskip

\subsection{The Combination of Logics}
Program logics are generally composite that is, the result of combining elementary logics which are different abstractions of program executions e.g\@. \cite{DBLP:journals/jacm/BruniGGR23,DBLP:journals/pacmpl/ZilbersteinDS23}.

\subsubsection{The Conjunction/Disjunction of Logics}
We have $\textsf{wlp}(\texttt{\small S},Q)=\textsf{pre}\sqb{\texttt{\small S}}Q\cap\widetilde{\textsf{pre}}\sqb{\texttt{\small S}}Q$ while $\textsf{wp}(\texttt{\small S},Q)=\textsf{pre}\sqb{\texttt{\small S}}_\bot Q\cap\widetilde{\textsf{pre}}\sqb{\texttt{\small S}}_\bot Q$ since blocking states must be prevented
as well as nontermination for \textsf{wp}, see Fig\@. \ref{fig:Property-transformers}\ifshort\ in the appendix\fi.
The relevant abstractions of transformers ${\tau_1}$, ${\tau_2}$ are \proofinapx
\bgroup\arraycolsep2pt\abovedisplayskip0.5\abovedisplayskip\belowdisplayskip0pt%
\begin{eqntabular}[fl]{rcl@{\quad}L@{\quad}l}
\alpha^{\cap}\pair{\tau_1}{\tau_2}r &\triangleq&\tau_1(r) \mathbin{\stackrel{.}{\cap}} \tau_2(r)&{\Large\strut} meet (or conjunction) where $\delta(x)=\pair{x}{x}$ is duplication\label{eq:GC-meet-join}\\[-1ex]
\renumber{$\pair{(\wp(\mathcal{X}\times \mathcal{Y})\rightarrow(\wp(\mathcal{X})\rightarrow\wp(\mathcal{Y})))^2}{\stackrel{..}{\supseteq}}\galoiS{\alpha^{\cap}}{\delta}\pair{\wp(\mathcal{X}\times \mathcal{Y})\rightarrow(\wp(\mathcal{X})\rightarrow\wp(\mathcal{Y}))}{\stackrel{..}{\supseteq}}$\qquad\qquad}\\[1ex]
\alpha^{\cup}\pair{\tau_1}{\tau_2}r &\triangleq&\tau_1(r)\mathbin{\stackrel{.}{\cup}} \tau_2(r)&{\Large\strut} join (or disjunction)\nonumber\\[-1ex]
\renumber{$\pair{(\wp(\mathcal{X}\times \mathcal{Y})\rightarrow(\wp(\mathcal{X})\rightarrow\wp(\mathcal{Y})))^2}{\stackrel{..}{\subseteq}}\galoiS{\alpha^{\cup}}{\delta}\pair{\wp(\mathcal{X}\times \mathcal{Y})\rightarrow(\wp(\mathcal{X})\rightarrow\wp(\mathcal{Y}))}{\stackrel{..}{\subseteq}}$\qquad\qquad}
\end{eqntabular}\egroup
\begin{toappendix}
\begin{proof}[Proof of (\ref{eq:GC-meet-join})]
\begin{calculus}[$\Leftrightarrow$\ ]
\hyphen{6} \formula{\alpha^{\cap}\pair{\tau_1}{\tau_2} \stackrel{..}{\supseteq} \bar{\tau}}\\
$\Leftrightarrow$
\formulaexplanation{\alpha^{\cap}\pair{\tau_1}{\tau_2}(r) \stackrel{.}{\supseteq} \bar{\tau}(r)}{pointwise def\@. $\stackrel{..}{\subseteq}$}\\
$\Leftrightarrow$
\formulaexplanation{\tau_1(r) \mathbin{\stackrel{.}{\cap}} \tau_2(r) \stackrel{.}{\supseteq} \bar{\tau}(r)}{def\@. $\alpha^{\cap}$}\\
$\Leftrightarrow$
\formulaexplanation{\tau_1(r) \stackrel{.}{\supseteq} \bar{\tau}(r) \wedge \tau_2(r) \stackrel{.}{\supseteq} \bar{\tau}(r)}{def\@. $\mathbin{\stackrel{.}{\cap}}$}\\
$\Leftrightarrow$
\formulaexplanation{\pair{\tau_1(r)}{\tau_2(r)} \stackrel{..}{\supseteq} \pair{\bar{\tau}(r)}{\bar{\tau}(r)}}{componentwise def\@. $\stackrel{..}{\supseteq}$ for pairs}\\
$\Leftrightarrow$
\formulaexplanation{\pair{\tau_1(r)}{\tau_2(r)} \stackrel{..}{\supseteq} \delta({\bar{\tau}(r)})}{def\@. $\delta$}\\[1ex]

\hyphen{6} \lastdiscussion{Similarly $\alpha^{\cup}\pair{\tau_1}{\tau_2} \stackrel{..}{\subseteq} \bar{\tau}$
$\Leftrightarrow$ $\pair{\tau_1(r)}{\tau_2(r)} \stackrel{..}{\subseteq} \delta({\bar{\tau}(r)})$ by $\subseteq$-order duality.}{\mbox{\qed}}
\end{calculus}\let\qed\relax
\end{proof}
\end{toappendix}

\subsubsection{The Product of Logics}

One can imagine a Cartesian product $\{DT,PT,NT\}\texttt{\small S}\{Q,R\}$ meaning that every execution of \texttt{\small S} starting with an initial state of $DT$ will definitely terminate in a final state in $Q$, every execution of \texttt{\small S} starting with an initial state of $PT$ will either terminate in a final state in R or not terminate, and every execution of \texttt{\small S} starting with an initial state of $NT$ will never terminate. $Q$ and $R$ could further be decomposed into a product of good and bad states. 

Similarly, \cite[section 4]{DBLP:journals/pacmpl/OHearn20} uses the notation $[p]C[ok:q][er:r]$ as a shorthand for $[p]C[ok:q]$ and $[p]C[er:r]$ resulting in a single deductive system instead of two independent ones. The definition of the relational semantics in (\ref{eq:def:semantics}) will use such a grouping to set apart \texttt{\small break}s. 

The relevant Cartesian abstraction $\alpha^{\times}$ merges two transformers into a single one. We assume that 
$\pair{\mathcal{X}\rightarrow \mathcal{Y}_1}{{\sqsubseteq}_1}$ and $\pair{\mathcal{X}\rightarrow \mathcal{Y}_2}{{\sqsubseteq}_2}$ are posets,
$\tau_1\in\mathcal{X}\rightarrow \mathcal{Y}_1$ and $\tau_2\in\mathcal{X}\rightarrow \mathcal{Y}_2$. \proofinapx

\bgroup\arraycolsep2pt\abovedisplayskip0.5\abovedisplayskip\belowdisplayskip0.5\belowdisplayskip
\begin{eqntabular}[fl]{L@{}rcl}
&\alpha^{\times}\pair{\tau_1}{\tau_2}(P)&\triangleq& \pair{\tau_1(P)}{\tau_2(P)}\label{eq:Cartesian-product}\stepcounter{equation}
\renumber{{\Large\strut} Cartesian product\quad(\ref{eq:Cartesian-product})}\\[-0.5ex]
&\gamma^{\times}(\bar{\tau})&\triangleq&\pair{\LAMBDA{P}\textsf{let\ }\pair{P_1}{P_2}=\bar{\tau}(P)\textsf{\ in }P_1}{\LAMBDA{P}\textsf{let\ }\pair{P_1}{P_2}=\bar{\tau}(P)\textsf{\ in }P_2}\nonumber\\
\rlap{with Galois connection}\renumber{$\pair{\mathcal{X}\rightarrow \mathcal{Y}_1\times \mathcal{X}\rightarrow \mathcal{Y}_2}{\mathord{\stackrel{.}{\sqsubseteq}_1}.\mathord{\stackrel{.}{\sqsubseteq}_2}}
\galois{\alpha^{\times}}{\gamma^{\times}}\pair{\mathcal{X}\rightarrow (\mathcal{Y}_1\times \mathcal{Y}_2)}{\stackrel{\mbox{\raisebox{-3pt}[0pt][0pt]{\huge .}}}{(\mathord{\sqsubseteq}_1,\mathord{\sqsubseteq}_2)}}$\quad\phantom{(\ref{eq:Cartesian-product})}}
\end{eqntabular}\egroup
\begin{toappendix}
\begin{proof}[Proof of (\ref{eq:Cartesian-product})]
\begin{calculus}[$\Leftrightarrow$\ ]
\formula{\alpha^{\times}\pair{\tau_1}{\tau_2} \stackrel{\mbox{\raisebox{-3pt}[0pt][0pt]{\huge .}}}{(\mathord{\sqsubseteq}_1,\mathord{\sqsubseteq}_2)} \bar{\tau}}\\
$\Leftrightarrow$
\formulaexplanation{\forall P\mathrel{.}\alpha^{\times}\pair{\tau_1}{\tau_2}(P) \mathrel{(\mathord{\sqsubseteq}_1,\mathord{\sqsubseteq}_2)} \bar{\tau}(P)}{pointwise def\@. $\stackrel{\mbox{\raisebox{-3pt}[0pt][0pt]{\huge .}}}{(\mathord{\sqsubseteq}_1,\mathord{\sqsubseteq}_2)}$}\\
$\Leftrightarrow$
\formulaexplanation{\forall P\mathrel{.}\pair{\tau_1(P)}{\tau_2(P)} \mathrel{(\mathord{\sqsubseteq}_1,\mathord{\sqsubseteq}_2)} \bar{\tau}(P)}{def\@. $\alpha^{\times}$}\\
$\Leftrightarrow$
\formulaexplanation{\forall P\mathrel{.}\pair{\tau_1(P)}{\tau_2(P)} \mathrel{(\mathord{\sqsubseteq}_1,\mathord{\sqsubseteq}_2)} \textsf{\ let\ }\pair{P_1}{P_2} = \bar{\tau}(P)\textsf{\ in\ } \pair{P_1}{P_2}}{$\bar{\tau}(P)\in\mathcal{Y}_1\times \mathcal{Y}_2$}\\
$\Leftrightarrow$
\formulaexplanation{\forall P\mathrel{.}\textsf{\ let\ }\pair{P_1}{P_2} = \bar{\tau}(P)\textsf{\ in\ } 
\tau_1(P)\mathrel{\mathord{\sqsubseteq}_1}P_1
\wedge
\alpha_2(P)\mathrel{\mathord{\sqsubseteq}_2}P_2}{componentwise def\@. of $\mathrel{(\mathord{\sqsubseteq}_1,\mathord{\sqsubseteq}_2)}$}\\
$\Leftrightarrow$
\formula{\tau_1\mathrel{\mathord{\stackrel{.}{\sqsubseteq}_1}}
\LAMBDA{P}\textsf{\ let\ }\pair{P_1}{P_2} = \bar{\tau}(P)\textsf{\ in\ } P_1
\wedge
\tau_2\mathrel{\mathord{\stackrel{.}{\sqsubseteq}_2}}
\LAMBDA{P}\textsf{\ let\ }\pair{P_1}{P_2} = \bar{\tau}(P)\textsf{\ in\ } P_2}\\\rightexplanation{pointwise def\@. $\mathord{\stackrel{.}{\sqsubseteq}_1}$ and $\mathord{\stackrel{.}{\sqsubseteq}_2}$}\\
$\Leftrightarrow$
\formulaexplanation{\textsf{\ let\ }\bar{\tau}_1=(\LAMBDA{P}\textsf{\ let\ }\pair{P_1}{P_2} = \bar{\tau}(P)\textsf{\ in\ } P_1)
\textsf{\ and\ }\bar{\tau}_2=(\LAMBDA{P}\textsf{\ let\ }\pair{P_1}{P_2} = \bar{\tau}(P)\textsf{\ in\ } P_2)
\textsf{\ in\ }
\tau_1\mathrel{\mathord{\stackrel{.}{\sqsubseteq}_1}}
\bar{\tau}_1
\wedge
\tau_2\mathrel{\mathord{\stackrel{.}{\sqsubseteq}_2}}
\bar{\tau}_2}{def\@. $\textsf{let}$}\\
$\Leftrightarrow$
\formulaexplanation{\textsf{let\ }\pair{\bar{\tau}_1}{\bar{\tau}_2} = \gamma^{\times}(\bar{\tau}) \textsf{\ in\ }\tau_1\mathrel{\stackrel{.}{\sqsubseteq}_1} \bar{\tau}_1\wedge\tau_2\mathrel{\stackrel{.}{\sqsubseteq}_2} \bar{\tau}_2}{def\@. $\gamma^{\times}$}\\
$\Leftrightarrow$
\formulaexplanation{\textsf{let\ }\pair{\bar{\tau}_1}{\bar{\tau}_2} = \gamma^{\times}(\bar{\tau}) \textsf{\ in\ }\pair{\tau_1}{\tau_2} \mathrel{\mathord{\stackrel{.}{\sqsubseteq}_1}.\mathord{\stackrel{.}{\sqsubseteq}_2}} \pair{\bar{\tau}_1}{\bar{\tau}_2}}{componentwise def\@. $\mathrel{\mathord{\stackrel{.}{\sqsubseteq}_1}.\mathord{\stackrel{.}{\sqsubseteq}_2}}$}\\
$\Leftrightarrow$
\lastformulaexplanation{\pair{\tau_1}{\tau_2} \mathrel{\mathord{\stackrel{.}{\sqsubseteq}_1}.\mathord{\stackrel{.}{\sqsubseteq}_2}} \gamma^{\times}(\bar{\tau})}{def\@. \textsf{let}}{\mbox{\qed}}
\end{calculus}\let\qed\relax
\end{proof}
\end{toappendix}
\begin{example}
We mentioned the origin \cite{DBLP:conf/ac/Park79} of relational semantics that Park encodes by 
$\alpha^{\times}\pair{\alpha^{\angelic}}{\LAMBDA{S}\alpha^{\neg}((\mathord{\stackrel{.}{\alpha}}^{-1}(\textsf{\upshape post}))(S)\{\bot\})}\sqb{\texttt{\small S}}_{\bot}$ i.e\@.
the input-output relation $\sqb{\texttt{\small S}}$ computed by \texttt{\small S} and the definite termination domain of \texttt{\small S} which is the complement of  possible nontermination.
\end{example}
\begin{example}\label{ex:Dijkstra-weakest-preconditions}
Dijkstra's weakest precondition $\textsf{\upshape wp}(\texttt{\small S},Q)$ \cite{DBLP:books/ph/Dijkstra76} is 
$\LAMBDA{Q}\alpha^{\cap}(
\textsf{\upshape pre}\sqb{\texttt{\small S}}_{\bot},\allowbreak
\widetilde{\textsf{\upshape pre}}\sqb{\texttt{\small S}}_{\bot})$, with  $Q\in\wp(\Sigma)$,
$\textsf{\upshape pre}=\mathord{\stackrel{.}{\alpha}^{-1}}(\textsf{\upshape post})$, and
$\widetilde{\textsf{\upshape pre}}=\stackrel{.}{\alpha^{\sim}}(\textsf{\upshape pre})$.
The weakest liberal condition $\textsf{\upshape wlp}(\texttt{\small S}, Q)$  is
 $\LAMBDA{Q}\allowbreak\alpha^{\cap}((\widetilde{\textsf{\upshape pre}}\sqb{\texttt{\small S}}_{\bot})\comp\alpha^{\demoniac},\allowbreak
 (\widetilde{\textsf{\upshape pre}}\sqb{\texttt{\small S}}_{\bot})\comp\alpha^{\demoniac})$
 =
$\alpha^{\cap}(
\textsf{\upshape pre}\sqb{\texttt{\small S}},\allowbreak
\widetilde{\textsf{\upshape pre}}\sqb{\texttt{\small S}})$. 
\end{example}

\subsubsection{The Reduced Product of Logics}
The components are usually not independent. For example one uses invariants of Hoare logic to prove termination, or definite termination implies possible termination. Another example is adversarial logic \cite{DBLP:conf/sas/Vanegue22} to describe the possible interaction between a program and an attacker. These are reductions (\ref{eq:reduced-product}) that have been studied in the context of program analysis \cite[chapter 29]{Cousot-PAI-2021} but also apply to any abstraction, including logics, e.g.\ \cite{DBLP:journals/jacm/BruniGGR23}. 

The functor $\alpha^{\ostar}$, inspired by the reduced product in abstract interpretation \cite[section 10.1]{DBLP:conf/popl/CousotC79},  is the Cartesian product where the information of one component is propagated, in abstract form, to the other. This is useful for combining program logics dealing with properties that are not independent.

Assume two abstractions of a (collecting) semantics in $\pair{\mathcal{S}}{\sqsubseteq}$ into different transformers $\pair{\mathcal{S}}{\sqsubseteq}\galois{\alpha_1}{\gamma_1}\pair{\mathcal{X}\rightarrow \mathcal{Y}_1}{\mathrel{\stackrel{.}{\mathord{\sqsubseteq}}_1}}$
and $\pair{\mathcal{S}}{\sqsubseteq}\galois{\alpha_2}{\gamma_2}\pair{\mathcal{X}\rightarrow \mathcal{Y}_1}{\mathrel{\stackrel{.}{\mathord{\sqsubseteq}}_2}}$. Assume that $\triple{\mathcal{X}\rightarrow (\mathcal{Y}_1\times \mathcal{Y}_2)}{\stackrel{\mbox{\raisebox{-3pt}[0pt][0pt]{\huge .}}}{(\mathord{\sqsubseteq}_1,\mathord{\sqsubseteq}_2)}}{\sqcap}$ is a complete lattice.

The reduced product combines two abstractions of the semantics $S$ into transformers $\tau_1$ and
$\tau_2$ into an abstraction of the semantics $S$ into a single transformer with $\alpha^{\ostar,\sqcap}\triangleq\rho\,\comp\,\alpha^{\times}$ where the reduction operator is
$\rho(\bar{\tau})\triangleq
\bigsqcap\{\bar{\tau}'\mid\textsf{let\ }\pair{\tau_1}{\tau_2}=\gamma^{\times}(\bar{\tau})\textsf{\ and\ }\pair{\tau'_1}{\tau'_2}=\gamma^{\times}(\bar{\tau}')\textsf{\ in\ }
\gamma_1(\tau_1)\sqcap\gamma_2(\tau_2)\sqsubseteq \gamma_1(\tau'_1)\wedge\gamma_1(\tau_1)\sqcap\gamma_2(\tau_2)\sqsubseteq \gamma_2(\tau'_2)\}$. By \cite[Theorem 36.24]{Cousot-PAI-2021}, we have the Galois connection
\bgroup\abovedisplayskip0.5\abovedisplayskip\begin{eqntabular}{c}
\pair{\mathcal{X}\rightarrow \mathcal{Y}_1\times \mathcal{X}\rightarrow \mathcal{Y}_2}{\stackrel{.}{\sqsubseteq}_1.\stackrel{.}{\sqsubseteq}_2}
\galois{\rho\,\comp\,\alpha^{\times}}{\gamma^{\times}}\pair{\mathcal{X}\rightarrow (\mathcal{Y}_1\times \mathcal{Y}_2)}{\stackrel{\mbox{\raisebox{-3pt}[0pt][0pt]{\huge .}}}{(\mathord{\sqsubseteq}_1,\mathord{\sqsubseteq}_2)}}
\label{eq:reduced-product}
\end{eqntabular}\egroup
\begin{example}Continuing example \ref{ex:fact:spec}, the reduced product of Hoare logic \cite{DBLP:journals/jacm/Hoare78} (abstracting \textsf{\textup{post}}) and subgoal induction logic \cite{DBLP:journals/cacm/MorrisW77} (in Dijkstra's version \cite{Dijkstra-EWD-576} abstracting \textsf{pre}) for the factorial with consequent specification $f=!\underline{n}$ is $\{ n=\underline{n}\geqslant0\wedge f=1\}\,\texttt{\small fact}\,\{\underline{n}\geqslant0\wedge f=!\underline{n}\}$.
\end{example}

\subsection{Symbolic Inversion}\label{sec:SymbolicInversion}

Let us consider one more useful abstraction of transformers allowing for their inversion using symbolic execution. This reversal abstraction $\alpha^{\leftrightarrow}$ from \cite[Theorem 10-13]{Cousot81-1} allows to prove  backward properties using a forward proof system by using auxiliary variables for initial values of variables (as in symbolic execution) 
and conversely (as an inverse symbolic execution starting with symbolic final values of variables). 
Given $\mathcal{D}\triangleq\wp(\mathcal{X}\times\mathcal{Y})\rightarrow(\wp(\mathcal{X})\rightarrow\wp(\mathcal{Y}))$, $\mathcal{D}'\triangleq\wp(\mathcal{X}\times\mathcal{Y})\rightarrow(\wp(\mathcal{Y})\rightarrow\wp(\mathcal{X}))$,  and $\tau=\textsf{\textup{post}}$, we have~\proofinapx\ (and similarly for  $\tau\in\{\textsf{pre},\textsf{\textup{post}},\textsf{Pre}\}$)
\bgroup\arraycolsep=2pt
\begin{eqntabular}{c@{\qquad}}
\alpha^{\leftrightarrow}(\tau)(r)P\colsep{\triangleq}{\{\sigma'\mid\exists\sigma\in P\mathrel{.}
\sigma\in
\tau(r^{-1})\{\overline{\varsigma}\mid
\overline{\varsigma}=\sigma'\}\}}\qquad\quad
\pair{\mathcal{D}}{\stackrel{...}{\subseteq}}\GaloiS{\alpha^{\leftrightarrow}}{\alpha^{\leftrightarrow}}\pair{\mathcal{D}'}{\stackrel{...}{\subseteq}}\label{eq:alpha-leftrightarrow}
\end{eqntabular}\egroup
\begin{toappendix}
\begin{proof}[Proof of (\ref{eq:alpha-leftrightarrow})] We use $\{\varsigma\mid\varsigma=\sigma \}$ rather than $\{\sigma\}$ to make bindings more clear.
\begin{calculus}
\formula{\textsf{\upshape post}(r)P}\\
=\formulaexplanation{\{\sigma'\mid\exists\sigma\in P\mathrel{.}\pair{\sigma}{\sigma'}\in r\}}{def\@. $\textsf{\upshape post}$}\\
=\formulaexplanation{\{\sigma'\mid\exists\sigma\in P\mathrel{.}
\exists\overline{\sigma}\mathrel{.}\overline{\sigma}=\sigma'\wedge\pair{\sigma}{\overline{\sigma}}\in r\}}{def\@. $=$}\\
=\formulaexplanation{\{\sigma'\mid\exists\sigma\in P\mathrel{.}
\sigma\in\{\sigma\mid\exists\overline{\sigma}\mathrel{.}
\overline{\sigma}=\sigma'\wedge\pair{\sigma}{\overline{\sigma}}\in r\}\}}{def\@. $\in$}\\
=\formulaexplanation{\{\sigma'\mid\exists\sigma\in P\mathrel{.}
\sigma\in\{\sigma\mid\exists\overline{\sigma}\mathrel{.}
\overline{\sigma}\in\{\overline{\varsigma}\mid
\overline{\varsigma}=\sigma'\}\wedge\pair{\sigma}{\overline{\sigma}}\in r\}\}}{def\@. $\in$}\\
=\formulaexplanation{\{\sigma'\mid\exists\sigma\in P\mathrel{.}
\sigma\in\{\sigma\mid\exists\overline{\sigma}\mathrel{.}
\pair{\sigma}{\overline{\sigma}}\in r\wedge
\overline{\sigma}\in\{\overline{\varsigma}\mid
\overline{\varsigma}=\sigma'\}\}\}}{commutativity $\wedge$}\\
=\formulaexplanation{\{\sigma'\mid\exists\sigma\in P\mathrel{.}
\sigma\in
\textsf{\upshape pre}(r)\{\overline{\varsigma}\mid
\overline{\varsigma}=\sigma'\}\}}{def\@. $\textsf{\upshape pre}$}\\
=\formulaexplanation{\{\sigma'\mid\exists\sigma\in P\mathrel{.}
\sigma\in
\textsf{\upshape post}(r^{-1})\{\overline{\varsigma}\mid
\overline{\varsigma}=\sigma'\}\}}{def\@. $\textsf{\upshape post}$}\\
= \formulaexplanation{\alpha^{\leftrightarrow}(\textsf{\upshape post})(r)P}{def\@. $\alpha^{\leftrightarrow}$}
\end{calculus}

\medskip

\noindent
The proof for \textsf{\upshape pre} is similar with
$\mathcal{D}=(\wp(\mathcal{X}\times\mathcal{Y})\rightarrow(\wp(\mathcal{Y})\rightarrow\wp(\mathcal{X})))\rightarrow(\wp(\mathcal{X}\times\mathcal{Y})\rightarrow(\wp(\mathcal{Y})\rightarrow\wp(\mathcal{X})))$. For \textsf{\upshape Post}, we have
\begin{calculus}
\formula{\textsf{\upshape Post}(r)P}\\
=\formulaexplanation{\{\pair{\sigma_0}{\sigma'}\mid\exists\sigma\mathrel{.}
\pair{\sigma_0}{\sigma}\in P\wedge
\pair{\sigma}{\sigma'}\in r\}}{def\@. $\textsf{\upshape post}$}\\
=\formulaexplanation{\{\pair{\sigma_0}{\sigma'}\mid\exists\sigma\mathrel{.}
\pair{\sigma_0}{\sigma}\in P\wedge\pair{\sigma'}{\sigma}\in r^{-1}\}}{def\@. inverse $r^{-1}$}\\
=\formulaexplanation{\{\pair{\sigma_0}{\sigma'}\mid\exists\sigma\mathrel{.}
\pair{\sigma_0}{\sigma}\in P\wedge
\pair{\sigma_0}{\sigma'}\in\{\pair{\varsigma_0}{\varsigma'}\mid \exists \pair{\varsigma_0}{\varsigma}\in\{\pair{\overline{\varsigma}_0}{\overline{\varsigma}}\mid\overline{\varsigma}_0 =\sigma_0\wedge \overline{\varsigma}=\sigma \}\mathrel{.}\pair{\varsigma'}{\varsigma}\in r^{-1}\}\}}{def\@. membership $\in$}\\
=\formulaexplanation{\{\pair{\sigma_0}{\sigma'}\mid\exists\sigma\mathrel{.}
\pair{\sigma_0}{\sigma}\in P\wedge\pair{\sigma_0}{\sigma'}\in\textsf{\upshape Post}(r^{-1})\{\pair{\overline{\varsigma}_0}{\overline{\varsigma}}\mid\overline{\varsigma}_0 =\sigma_0\wedge \overline{\varsigma}=\sigma \}\}}{def\@. $\textsf{\upshape Post}$}\\
= \formulaexplanation{\alpha^{\leftrightarrow}(\textsf{\upshape Post})(r)P}{def\@. $\alpha^{\leftrightarrow}$}
\end{calculus}

\medskip

\noindent where in this case, we have the definition
\begin{eqntabular*}{rcl}
\alpha^{\leftrightarrow}(\tau)(r)P
&\triangleq&
\{\pair{\sigma_0}{\sigma'}\mid\exists\sigma\mathrel{.}
\pair{\sigma_0}{\sigma}\in P\wedge\pair{\sigma_0}{\sigma'}\in\tau(r^{-1})(\{\pair{\overline{\varsigma}_0}{\overline{\varsigma}}\mid\overline{\varsigma}_0 =\sigma_0\wedge \overline{\varsigma}=\sigma \}) \}
\end{eqntabular*}
The case of \textsf{\upshape Pre} is similar with 
$\mathcal{D}=(\wp(\mathcal{X}\times\mathcal{Y})\rightarrow(\wp(\mathcal{Y}\times\mathcal{Z})\rightarrow\wp(\mathcal{X}\times\mathcal{Z})))\rightarrow(\wp(\mathcal{X}\times\mathcal{Y})\rightarrow(\wp(\mathcal{Y}\times\mathcal{Z})\rightarrow\wp(\mathcal{X}\times\mathcal{Z})))$.
\end{proof}
\end{toappendix}
\begin{example}\label{ex:forward-symbolic-execution}
Consider the straight-line program \texttt{\small x = x+y; y = 2*x+y}. A forward symbolic execution $\textsf{\upshape post}\sqb{\texttt{\small S}}\{\sigma\mid {\sigma}_{\texttt{x}}=\underline{\sigma}_{\texttt{x}}\wedge
{\sigma}_{\texttt{y}}=\underline{\sigma}_{\texttt{y}}\}$ with Hoare logic for  initial auxiliary variables $\underline{x}$, $\underline{y}$ is 
\begin{center}
\ttfamily\small
$\{ x=\underline{x}\wedge y=\underline{y} \} 
~\texttt{\small x = x + y;}~
\{ x = \underline{x}+\underline{y} \wedge y=\underline{y} \} 
~\texttt{\small y = 2*x+y;}~
\{ x = \underline{x}+\underline{y} \wedge y=2\underline{x}+3\underline{y} \} $
\end{center}
This information can be inferred automatically by forward static analyses using affine equalities \cite{DBLP:journals/acta/Karr76} or inequalities \cite{DBLP:conf/popl/CousotH78}.
This can be used to get a precondition $\textsf{\upshape pre}\sqb{\texttt{\small S}}Q$ 
ensuring that a postcondition $Q$ holds be defining
$\textsf{\upshape pre}\sqb{\texttt{\small S}}Q$ = 
$\{\underline{\sigma}\mid\exists \overline{\sigma}\in Q\mathrel{.}\overline{\sigma}\in\textsf{\upshape post}\sqb{\texttt{\small S}}\{\sigma\mid {\sigma}_{\texttt{x}}=\underline{\sigma}_{\texttt{x}}\wedge
{\sigma}_{\texttt{y}}=\underline{\sigma}_{\texttt{y}}\}\}$ which, in our example, is
$\{\pair{x}{y}\mid\exists\pair{\overline{x}}{\overline{y}}\in Q\wedge \overline{x} = {x}+{y} \wedge \overline{y}=2{x}+3{y}\}$ so that, e.g., for $Q=\{\pair{\overline{x}}{\overline{y}}\mid \overline{x}=\overline{y}\}$ stating that $x=y$ on exit, we get the precondition
$\{\pair{z}{-z/3}\mid z\in\mathbb{Z}\}$.

Inversely, using subgoal induction, a backward execution $\textsf{\upshape pre}\sqb{\texttt{\small S}}\{\sigma\mid {\sigma}_{\texttt{x}}=\overline{\sigma}_{\texttt{x}}\wedge
{\sigma}_{\texttt{y}}=\overline{\sigma}_{\texttt{y}}\}$ is
\begin{center}
\ttfamily\small
$\{ x=3\overline{x}-\overline{y}\wedge y=\overline{y}-2\overline{x}  \} 
~x = x + y;~
\{ x = \overline{x}\wedge y=\underline{y}-2\overline{x} \} 
~y = 2*x+y;~
\{ x=\overline{x}\wedge y=\overline{y} \} $
\end{center}
This information can be used to get a postcondition $\textsf{\upshape post}\sqb{\texttt{\small S}}P$ hence holding for states reachable from the precondition $P$ as
$\{\sigma\mid\exists\underline{\sigma}\in P \mathrel{.}\underline{\sigma}\in\textsf{\upshape pre}\sqb{\texttt{\small S}}\{\sigma\mid {\sigma}_{\texttt{x}}=\overline{\sigma}_{\texttt{x}}\wedge
{\sigma}_{\texttt{y}}=\overline{\sigma}_{\texttt{y}}\}\}$.
For our example, we get
$\{\pair{x}{y}\mid\exists\pair{\underline{x}}{\underline{y}}\in P \mathrel{.}
\underline{x}=3{x}-{y}\wedge \underline{y}={y}-2{x}\}$ which, e.g., for $P=\{\pair{x}{y}\mid x=y \}$, yields
$\{\pair{x}{y}\mid 5x=2y\}$. This calculation is mechanizable using the operations of the abstract domains for affine equalities \cite{DBLP:journals/acta/Karr76}  or inequalities \cite{DBLP:conf/popl/CousotH78}.
\end{example}

\vspace*{-2em}

\part{}{{\LARGE\bfseries  Part II: Design of the Proof Rules of Logics by Abstraction of \phantom{{\LARGE\bfseries  Part II: }}Their Theory}}
\setcounter{section}{0}%
\def\thesection{II.\originalthesection}%
\vskip2ex

Given the theory $\alpha(\sqb{\texttt{\small S}}_\bot)$ of a logic defined by an abstraction $\alpha$  of the natural relational semantics $\sqb{\texttt{\small S}}_\bot$, we now consider the problem of designing
the proof/deductive system for that logic. 
The abstraction $\alpha$ can be decomposed into $\alpha_a\comp\alpha_t$ where $\alpha_t$ abstracts the natural relational semantics $\sqb{\texttt{\small S}}_\bot)$ into
an exact transformer (isomorphically its antecedant-consequent graph) which is then over or under approximated by $\alpha_a$.

We first express the natural relational semantics in structural fixpoint form in Sect\@. \ref{sect:FixpointNaturalRelationalSemantics}.
Then we use fixpoint abstraction of Sect\@. \ref{sect:FixpointAbstraction} and structural induction to express the exact transformer
$\alpha_t(\sqb{\texttt{\small S}}_\bot)$ in structural fixpoint form. The approximation abstraction $\alpha_a$ is then handled using the fixpoint
induction principles of Sect\@. \ref{sect:FixpointInduction} to under or over approximate the transformer by  $\alpha_a\comp\alpha_t(\sqb{\texttt{\small S}}_\bot)$.
 \cite{Aczel:1977:inductive-definitions} has shown that set theoretic fixpoints can be expressed as proof/deductive systems and conversely. We recall his method in Sect\@. \ref{sect:SemanticsDeductiveSystems}.
This yields a method of designing proof system by calculus in Sect\@. \ref{sec:equivalence-fxp-deductive-system}. This is applied to two new example logics. The first example in section 
\ref{Calculational-Design-of-the-Extended-Hoare-Logic} is a forward transformational logic to express correct reachability of a postcondition (as in  Hoare and Manna partial correctness logics), termination (as in
Apt \& Plokin and Manna \& Pnueli logics) as well as nontermination, all cases being expressible by a single formula of the logic (depending on initial values). The second example in section 
\ref{sec:design:possibleAccessibilityNontermination} is a backward transformational logic to express correct accessibility of a postcondition or nontermination.

\section{Structural Fixpoint Natural Relational Semantics}\label{sect:FixpointNaturalRelationalSemantics}

We define the relational natural semantic $\sqb{\texttt{\small S}}_{\bot}\in\wp(\Sigma\times\Sigma_\bot)$ of statements $\texttt{\small S}$ by structural induction on the program syntax and iteration defined as extremal fixpoints of increasing (monotone/isotone) functions on complete lattices \cite{Tarski-fixpoint}.

The definition is in Milner/Tofte style \cite{DBLP:journals/tcs/MilnerT91}, except that finite behaviors in $\wp(\Sigma\times\Sigma)$ are in inductive style with least fixpoints (\textsf{\upshape lfp}) and infinite behaviors in $\wp(\Sigma\times\{\bot\})$ are in co-inductive style with greatest fixpoints (\textsf{\upshape gfp}), as in \cite{DBLP:conf/popl/CousotC92,DBLP:journals/iandc/CousotC09}. Milner/Tofte define both finite and infinite behaviors in co-inductive style \cite{DBLP:journals/tcs/MilnerT91,DBLP:conf/esop/Leroy06}, which looks more uniform. However, some fixpoint approximation techniques are more precise for least fixpoints than for  greatest fixpoints \cite[Chapter 18]{Cousot-PAI-2021}, which will be essential to prove completeness of proof methods\footnote{For example, Park induction Th\@. \ref{th:Fixpoint-Overapproximation} can be used to over approximate least fixpoints with an invariant only while approximating greatest fixpoints in the dual of Th\@. \ref{th:Fixpoint-Underapproximation-Variant} involves a variant function.}.

Given the assignment $\sigma[\texttt{\small x} \leftarrow v]$ of value $v\in\mathbb{V}$ to variable $\texttt{\small x}\in\mathbb{X}$ in state $\sigma\in\Sigma\triangleq\mathbb{X}\rightarrow \mathbb{V}$ and the identity relation  $\textsf{\upshape id}\triangleq\{\pair{\sigma}{\sigma}\mid \sigma\in\Sigma_{\bot}\}
$, the basic statements have the following semantics. They all terminate and do not exit loops, but for \texttt{\small break}, that exits the closest outer  loop (which existence must be  checked syntactically) without changing the values of variables.
\ifshort
\bgroup\arraycolsep=0.45\arraycolsep\begin{eqntabular}[fl]{rcl@{\quad}rcl@{\quad}rcl}
\sqb{\texttt{\small x = A}}^e&\triangleq&\{\pair{\sigma}{\sigma[\texttt{\small x}\leftarrow\mathcal{A}\sqb{\texttt{\small A}}\sigma]}\mid \sigma\in\Sigma\}
&
\sqb{\texttt{\small x = A}}^b&\triangleq&\emptyset
&
\sqb{\texttt{\small x = A}}^\bot&\triangleq&\emptyset
\nonumber\\
\sqb{\texttt{\small x = [a, b]}}^e&\triangleq&
\{\pair{\sigma}{\sigma[\texttt{\small x}\leftarrow i]}\mid \smash{\begin{array}[t]{@{}l@{}}\sigma\in\Sigma\wedge{}\\[-0.5ex] a-1 <i< b+1\}\end{array}}
&
\sqb{\texttt{\small x = [a,b]}}^b&\triangleq&\emptyset
&
\sqb{\texttt{\small x = [a,b]}}^\bot&\triangleq&\emptyset
\label{eq:def:sem:basis}\\
\sqb{\texttt{\small break}}^e&\triangleq&\emptyset
&
\sqb{\texttt{\small break}}^b&\triangleq&\textsf{\upshape id}
&
\sqb{\texttt{\small break}}^\bot&\triangleq&\emptyset
\nonumber\\
\sqb{\texttt{\small skip}}^e&\triangleq&\textsf{\upshape id}
&
\sqb{\texttt{\small skip}}^b&\triangleq&\emptyset
&
\sqb{\texttt{\small skip}}^\bot&\triangleq&\emptyset
\nonumber
\end{eqntabular}\egroup
\else

\bgroup\arraycolsep=0.475\arraycolsep\begin{eqntabular}[fl]{lcl@{\quad}lcl}
\sqb{\texttt{\small x = A}}^e&\triangleq&\{\pair{\sigma}{\sigma[\texttt{\small x}\leftarrow\mathcal{A}\sqb{\texttt{\small A}}\sigma]}\mid \sigma\in\Sigma\}
&
\sqb{\texttt{\small x = [a, b]}}^e&\triangleq&
\{\pair{\sigma}{\sigma[\texttt{\small x}\leftarrow i]}\mid \smash{\begin{array}[t]{@{}l@{}}\sigma\in\Sigma\wedge{}\\[-0.5ex]\quad a-1 <i< b+1\}\end{array}}
\nonumber\\
\sqb{\texttt{\small x = A}}^b&\triangleq&\emptyset
&
\sqb{\texttt{\small x = [a,b]}}^b&\triangleq&\emptyset
\nonumber\\
\sqb{\texttt{\small x = A}}^\bot&\triangleq&
\emptyset
&
\sqb{\texttt{\small x = [a,b]}}^\bot&\triangleq&
\emptyset
\nonumber\\
\sqb{\texttt{\small break}}^e&\triangleq&\emptyset
&
\sqb{\texttt{\small skip}}^e&\triangleq&\textsf{\upshape id}
\label{eq:def:sem:basis}\\
\sqb{\texttt{\small break}}^b&\triangleq&\textsf{\upshape id}
&
\sqb{\texttt{\small skip}}^b&\triangleq&\emptyset
\nonumber\\
\sqb{\texttt{\small break}}^\bot&\triangleq&
\emptyset
&
\sqb{\texttt{\small skip}}^\bot&\triangleq&
\emptyset
\nonumber
\end{eqntabular}\egroup
\fi
 For the conditional, we let $\sqb{\texttt{\small B}}\triangleq\{\pair{\sigma}{\sigma}\mid\sigma\in\mathcal{B}\sqb{\texttt{\small B}}\}$ be the relational semantics of Boolean expressions.
We define ($\fatsemi$ is the composition of relations, see Sect\@. \ref{sect:Relations} in the appendix)
\bgroup\arraycolsep=0.475\arraycolsep\begin{eqntabular}[fl]{lcl@{\qquad\quad}lcl}
\sqb{\texttt{\small S$_1$;S$_2$}}^e&\triangleq&\sqb{\texttt{\small S$_1$}}^e\fatsemi\sqb{\texttt{\small S$_2$}}^e
&
\sqb{\texttt{\small if (B) S$_1$ else S$_2$}}^e&\triangleq&\sqb{\texttt{\small B}}\fatsemi\sqb{\texttt{\small S$_1$}}^e\cup\sqb{\neg\texttt{\small B}}\fatsemi\sqb{\texttt{\small S$_2$}}^e\nonumber\\
\sqb{\texttt{\small S$_1$;S$_2$}}^b&\triangleq&\sqb{\texttt{\small S$_1$}}^b\cup(\sqb{\texttt{\small S$_1$}}^e\fatsemi\sqb{\texttt{\small S$_2$}}^b)
&
\sqb{\texttt{\small if (B) S$_1$ else S$_2$}}^b&\triangleq&\sqb{\texttt{\small B}}\fatsemi\sqb{\texttt{\small S$_1$}}^b\cup\sqb{\neg\texttt{\small B}}\fatsemi\sqb{\texttt{\small S$_2$}}^b
\label{eq:def:sem:seq:if}\\
\sqb{\texttt{\small S$_1$;S$_2$}}^\bot&\triangleq&
\sqb{\texttt{\small S$_1$}}^\bot\cup(\sqb{\texttt{\small S$_1$}}^e\fatsemi\sqb{\texttt{\small S$_2$}}^\bot)
&
\sqb{\texttt{\small if (B) S$_1$ else S$_2$}}^\bot&\triangleq&
\sqb{\texttt{\small B}}\fatsemi\sqb{\texttt{\small S$_1$}}^\bot\cup\sqb{\neg\texttt{\small B}}\fatsemi\sqb{\texttt{\small S$_2$}}^\bot
\nonumber\end{eqntabular}\egroup
For iteration, we define
\begin{eqntabular}{rcll}
{F^e}(X)&\triangleq&\textsf{\upshape id} \cup (\sqb{\texttt{\small B}}\fatsemi\sqb{\texttt{\small S}}^e\fatsemi (X\setminus\Sigma\times\{\bot\})), & X\in\wp(\Sigma\times(\Sigma\cup\{\bot\}))\label{eq:natural-transformer-finite}\\
{F^\bot}(X)&\triangleq&\sqb{\texttt{\small B}}\fatsemi\sqb{\texttt{\small S}}^e\fatsemi X, &X\in\wp(\Sigma\times\{\bot\})\label{eq:natural-transformer-infinite}\\
\sqb{\texttt{\small while (B) S}}^e&\triangleq&
\Lfp{\subseteq}{F^e}\fatsemi(\sqb{\neg\texttt{\small B}}\cup\sqb{\texttt{\small B}}\fatsemi\sqb{\texttt{\small S}}^b)\label{eq:natural-finite}\\
\sqb{\texttt{\small while (B) S}}^b&\triangleq&\emptyset\label{eq:natural-break}\\
\sqb{\texttt{\small while (B) S}}^\bot&\triangleq&(\Lfp{\subseteq}{F^e}\fatsemi\sqb{\texttt{\small B}}\fatsemi\sqb{\texttt{\small S}}^\bot)\cup\Gfp{\subseteq}{F^\bot}
\label{eq:natural-oo}
\end{eqntabular}
The transformers are defined on complete lattices, ${F^e}$ on
$\sextuple{\wp(\Sigma\times\Sigma)}{\subseteq}
{\emptyset}{\Sigma \times\Sigma }{\cup}{\cap}$ and ${F^\bot}$ on
$\sextuple{\wp(\Sigma_\bot\times\{\bot\})}{\subseteq}
{\emptyset}{\vec{\infty}}{\cup}{\cap}$ with $\vec{\infty}\triangleq\Sigma \times\{\bot\}$ and are $\subseteq$-increasing, so do exist \cite{Tarski-fixpoint}. 

Moreover, the natural transformer $F^e$ in (\ref{eq:natural-transformer-finite}) preserves arbitrary joins, so is continuous. By Scott-Kleene fixpoint theorem \cite{ScottStrachey71-PRG6}, its least fixpoint is the reflexive transitive closure  $\Lfp{\subseteq}{F^e}=\bigcup_{n\geqslant 0}(\sqb{\texttt{\small B}}\fatsemi\sqb{\texttt{\small S}}^e)\strut^n={(\sqb{\texttt{\small B}}\fatsemi\sqb{\texttt{\small S}}^e)\strut}^{\ast}$. So $\Lfp{\subseteq}{F^e}$ is a relation between initial states before entering the loop and successive states at loop reentry after any number $n\geqslant 0$ of iterations. If, after $n$ iterations, the test $\texttt{\small B}$ ever becomes false then $\sqb{\texttt{\small B}}=\emptyset$ and so all later terms in the infinite disjunction are empty. 

Then composing $\Lfp{\subseteq}{F^e}={(\sqb{\texttt{\small B}}\fatsemi\sqb{\texttt{\small S}}^e)\strut}^{\ast}$ with $\sqb{\neg\texttt{\small B}}\cup\sqb{\texttt{\small B}}\fatsemi\sqb{\texttt{\small S}}^b$ in 
(\ref{eq:natural-finite})
 yields the relation between initial and final states in case of termination or in case of a break when excuting the loop body \texttt{\small S}.  (\ref{eq:natural-break}) states that a \texttt{\small break} exits the immediately enclosing loop, not any of the outer ones.
 
Composing $\Lfp{\subseteq}{F^e}={(\sqb{\texttt{\small B}}\fatsemi\sqb{\texttt{\small S}}^e)\strut}^{\ast}$ with $\sqb{\texttt{\small B}}\fatsemi\sqb{\texttt{\small S}}^\bot$ in (\ref{eq:natural-oo}) yields the possible cases of nontermination
when the loop body $\texttt{\small S}$ does not terminate after finitely many finite iterations in the loop. 

Finally, the term $\Gfp{\subseteq}{F^\bot}$ in (\ref{eq:natural-oo}) represents infinitely many iterations of terminating body executions. Again if \texttt{\small B} becomes false after finitely many iterations then $\sqb{\texttt{\small B}}=\emptyset$ so that this infinite iteration term is $\emptyset$ (since $\emptyset$ is absorbant for $\fatsemi$). 
As shown by \cite[Example 22]{DBLP:journals/tcs/Cousot02}, ${F^\bot}$ may not be co-continuous when considering unbounded nondeterminism so that transfinite decreasing fixpoint iterations from the supremum might be necessary
\cite{CousotCousot-PJM-82-1-1979}. The following lemma makes clear that $\Gfp{\subseteq}{F^\bot}$ characterises (non)termination \proofinapx
\begin{lemma}[Termination]\label{lem:Fbot:nontermination}
$\Gfp{\subseteq}{F^\bot} = \emptyset$
$\Leftrightarrow$
$\{\sigma\in\mathbb{N}\rightarrow\Sigma\mid\forall i\in\mathbb{N}\mathrel{.}\pair{\sigma_i}{\sigma_{i+1}}\in \sqb{\texttt{\small B}}\fatsemi\sqb{\texttt{\small S}^e}\}=\emptyset
$.
\end{lemma}
\begin{toappendix}
\begin{proof}[proof of lemma \ref{lem:Fbot:nontermination}] Let $\pair{X^n}{n\in\mathbb{N}}$ be the iterates of 
${F^\bot}(X)\triangleq\sqb{\texttt{\small B}}\fatsemi\sqb{\texttt{\small S}}^e\fatsemi X$ in (\ref{eq:natural-transformer-infinite}) on the complete lattice $\sextuple{\wp(\Sigma\times\{\bot\})}{\subseteq}{\emptyset}{\Sigma\times\{\bot\}}{\cup}{\cap}$ from the supremum $X^0=\Sigma\times\{\bot\}$. 

Assume that $X^n=(\sqb{\texttt{\small B}}\fatsemi\sqb{\texttt{\small S}}^e)^n\fatsemi(\Sigma\times\{\bot\})$ by induction hypothesis. 

Then by def.\ (\ref{eq:natural-transformer-infinite}) of $F^{\bot}$, we have
$X^{n+1}$ = ${F^\bot}(X^n)$ = $\sqb{\texttt{\small B}}\fatsemi\sqb{\texttt{\small S}}^e\fatsemi(\sqb{\texttt{\small B}}\fatsemi\sqb{\texttt{\small S}}^e)^n\fatsemi(\Sigma\times\{\bot\})$
=
$(\sqb{\texttt{\small B}}\fatsemi\sqb{\texttt{\small S}}^e)^{n+1}\fatsemi\Sigma\times\{\bot\}$. 

Since ${F^\bot}$ preserves joins, by \cite{CousotCousot-PJM-82-1-1979}, the limit is
$\Gfp{\subseteq}{F^\bot}$
=
$X^\omega$
=
$\bigcap_{n\in\mathbb{N}}X^n$
=
$\bigcap_{n\in\mathbb{N}}(\sqb{\texttt{\small B}}\fatsemi\sqb{\texttt{\small S}^e})^n\fatsemi\Sigma\times\{\bot\}$. 

For $X^\omega$ not to be empty, none of the iterates $(\sqb{\texttt{\small B}}\fatsemi\sqb{\texttt{\small S}^e})^n\fatsemi\Sigma\times\{\bot\}$ must be empty, meaning that
$\forall n\in\mathbb{N}\mathrel{.}(\sqb{\texttt{\small B}}\fatsemi\sqb{\texttt{\small S}^e})^n\neq\emptyset$. By definition of the power of a relation, this implies that $\exists\sigma\in\mathbb{N}\rightarrow\Sigma\mathrel{.}\forall i\in\mathbb{N}\mathrel{.}\pair{\sigma_i}{\sigma_{i+1}}\in \sqb{\texttt{\small B}}\fatsemi\sqb{\texttt{\small S}^e}$. The result follows by contraposition.
\end{proof}
\end{toappendix}
Since $\bot\notin\Sigma$, $(\sqb{\texttt{\small S}}^e\cup\sqb{\texttt{\small S}}^\bot)\cap\Sigma=\sqb{\texttt{\small S}}^e$ and  $(\sqb{\texttt{\small S}}^e\cup\sqb{\texttt{\small S}}^\bot)\cap\{\bot\}=\sqb{\texttt{\small S}}^\bot$,  the semantics can be defined as
\bgroup\abovedisplayskip0pt\belowdisplayskip0.75\belowdisplayskip
\begin{eqntabular}{rclcl@{\quad}L}
\sqb{\texttt{\small S}}_{\bot}&\triangleq&\pair{\sqb{\texttt{\small S}}^e\cup\sqb{\texttt{\small S}}^\bot}{\sqb{\texttt{\small S}}^b}
&\in&\wp(\Sigma\times\Sigma_\bot)\times\wp(\Sigma\times\Sigma)
& natural semantics
\label{eq:def:semantics}\\
\sqb{\texttt{\small S}}&\triangleq&\pair{\sqb{\texttt{\small S}}^e}{\sqb{\texttt{\small S}}^b}
&\in&\wp(\Sigma\times\Sigma)\times\wp(\Sigma\times\Sigma)
& angelic semantics\nonumber
\end{eqntabular}\egroup
where $\sextuple{\wp(\Sigma\times\Sigma_\bot)}{\sqsubseteq}{\Sigma\times\{\bot\}}{\Sigma\times\Sigma}{\sqcup}{\sqcap}$ is a complete lattice for the computational ordering $X\sqsubseteq Y\triangleq(X\cap (\Sigma\times\Sigma))\subseteq (Y\cap (\Sigma\times\Sigma))\wedge(X\cap (\Sigma\times\{\bot\}))\supseteq (Y\cap (\Sigma\times\{\bot\}))$. It follows that the definition of termination on normal exit or nontermination can be defined by a single transformer \cite[Theorem 9]{DBLP:journals/tcs/Cousot02} (but termination $\sqb{\texttt{\small S}}^e$ and break $\sqb{\texttt{\small S}}^b$ cannot be mixed without losing information).

This relational natural semantics can be extended to record a relation between the initial and current values of variables. This consists in  considering the Galois connections $\pair{\alpha_{\downarrow^2}}{\gamma_{\downarrow^2}}$ for assertions and
$\pair{\mathord{\stackrel{.}{\alpha}}_{\downarrow^2}}{\mathord{\stackrel{.}{\gamma}}_{\downarrow^2}}$ for relations in (\ref{eq:def:alpha-gamma-d2}).
This can be implemented using auxiliary variables without modification of the semantics \proofinapx. 
\begin{toappendix}

\smallskip

\subsection{Auxiliary Variables}\label{sec:auxiliary:variables}
The relational natural semantics $\sqb{\texttt{\small S}}_{\bot}\in\wp(\Sigma\times\Sigma_\bot)$
of Sect\@. \ref{sect:FixpointNaturalRelationalSemantics} can be extended to record the values of variables
of entry of statements \texttt{\small S}. The states $\sigma\in\Sigma\triangleq\mathbb{X}\rightarrow \mathcal{V}$
are extended with fresh auxiliary variables $\overline{\mathbb{X}}\triangleq \{\bar{x}\mid x\in\mathbb{X}\}$ to
$\sigma\in\overline{\Sigma}\triangleq\overline{\mathbb{X}}\cup\mathbb{X}\rightarrow \mathcal{V}$. We define the projection on auxiliary variables as $\sigma_{\overline{\mathbb{X}}}\in\overline{\mathbb{X}}\rightarrow \mathcal{V}$ such that
$\forall \bar{x}\in\overline{\mathbb{X}}\mathrel{.}\sigma_{\overline{\mathbb{X}}}( \bar{x})=\sigma( \bar{x})$ while the projection on program  variables is $\sigma_{{\mathbb{X}}}\in{\mathbb{X}}\rightarrow \mathcal{V}$ such that
$\forall {x}\in{\mathbb{X}}\mathrel{.}\sigma_{{\mathbb{X}}}( {x})=\sigma( {x})$. A predicate $P\in\wp(\Sigma)$ is extended to $P\in\wp(\overline{\Sigma})=\{\sigma\mid\sigma_{\mathbb{X}}\in P\}$. 
For initialization of fresh auxiliary variables to the values of the program variables, we define $\mathbb{I}_{\overline{\mathbb{X}}}(P)\triangleq\{\sigma\in\overline{\Sigma}\mid \sigma_{\overline{\mathbb{X}}}=\sigma_{{\mathbb{X}}}\wedge\sigma_{{\mathbb{X}}}\in P\}$.
The relations $r\in\wp(\Sigma\times\Sigma_\bot)$ involved in the relational natural semantics
of Sect\@. \ref{sect:FixpointNaturalRelationalSemantics} are extended to $r\in\wp(\overline{\Sigma}\times\overline{\Sigma}_\bot)=\{\pair{\sigma}{\sigma'}\mid \sigma'_{\overline{\mathbb{X}}}=\sigma_{\overline{\mathbb{X}}}\wedge \pair{\sigma_{\mathbb{X}}}{\sigma'_{\mathbb{X}}}\in r\}$ since the values of initial values are kept unchanged by program execution. We leave implicit the fact that batches of fresh auxiliary variables can be successively added to program variables, for example on entry of each loop body in imbricated loops.
\end{toappendix}

\ifshort Nondeterminism can be unbounded, as discussed in the appendix \proofinapx\fi.
\begin{toappendix}
\subsection{Bounded Versus Unbounded Nondeterminism}\label{sec:Nondeterminism}
Unbounded nondeterminism was strongly rejected by  Dijkstra in \cite[chapter 9]{DBLP:books/ph/Dijkstra76} as unimplementable.  Park was rather critical about this restriction \cite{Park69-MI5} since time is unbounded. Later  Dijkstra changed his mind \cite{DBLP:journals/acta/DijkstraG86}, \cite[pages 174--180 of chapter 9]{DBLP:books/daglib/0067387} using arbitrary well-founded sets as by Turing \cite{Turing49-program-proof} and Floyd \cite{Floyd67-1}. This is the approach we use, using ordinals \cite[Ch\@. 2]{Monk-Set-Theory} to formalize well-founded sets used in termination proofs.
\end{toappendix}

\section{Fixpoint Abstraction}\label{sect:FixpointAbstraction}\label{sec:abstraction}

We recall classic  fixpoint abstraction theorems \cite{DBLP:journals/tcs/Cousot02}, \cite[Ch\@. 18]{Cousot-PAI-2021} to abstract the fixpoint definition of the program  relational semantics into a fixpoint definition of transformers (or their graph). \ifshort Abstraction can also be applied to deductive systems \proofinapx.\fi
\begin{theorem}[Fixpoint abstraction \cite{DBLP:conf/popl/CousotC79}]\label{th:fixpoint-abstraction}If $\pair{C}{\sqsubseteq}\galois{\alpha}{\gamma}\pair{A}{\preceq}$ is a Galois connection between complete lattices $\pair{C}{\sqsubseteq}$ and $\pair{A}{\preceq}$, $f\in C\mathrel{\smash{\stackrel{i}{\longrightarrow}}} C$ and $\bar{f}\in A \mathrel{\smash{\stackrel{i}{\longrightarrow}}} A$ are increasing and commuting, that is, $\alpha\comp f=\bar{f}\comp\alpha$, then $\alpha(\Lfp{\sqsubseteq}{f})=\Lfp{\preceq}{\bar{f}}$ (while semi-commutation $\alpha\comp f\mathrel{\dot\preceq}\bar{f}\comp\alpha$ implies $\alpha(\Lfp{\sqsubseteq}{f})\preceq\Lfp{\preceq}{\bar{f}}$).{\Large\strut}
\end{theorem}
As a simple application, we will need the following  corollary \proofinapx.
\begin{corollary}[Pointwise abstraction]\label{th:Pointwise-abstraction}
Let $\quadruple{L}{\sqsubseteq}{\top}{\sqcup}$ and $\quadruple{L'}{\mathrel{{\sqsubseteq}'}}{\top'}{\sqcup'}$ be complete lattices. Assume that $F\in(L\rightarrow L')\mathrel{\smash{\stackrel{i}{\longrightarrow}}}(L\rightarrow L')$ is increasing and that for all $Q\in L$,  $\bar{F}_Q\in L'\mathrel{\smash{\stackrel{i}{\longrightarrow}}}L'$ is increasing. Assume $\forall Q\in L\mathrel{.}\forall f\in L\rightarrow L'\mathrel{.}F(f)Q=\bar{F}_Q(f(Q))$. Then $\forall Q\in L\mathrel{.}(\Lfp{\mathrel{\stackrel{.}{{\sqsubseteq}}'}}F)Q=\Lfp{\mathrel{{\sqsubseteq}'}}\bar{F}_Q$.
\end{corollary}
\begin{toappendix}
\label{aec:apx:proof:th:Pointwise-abstraction}
\begin{proof}[Proof of corollary \ref{th:Pointwise-abstraction}] Given any $Q\in L$, define $\alpha_Q(f)\triangleq f(Q)$ and $\gamma_Q(y)\triangleq\LAMBDA{x}\si x=Q \alors y \sinon \top\fsi$. We have
$\pair{L\rightarrow L'}{\mathrel{\stackrel{.}{{\sqsubseteq}}'}}\galois{\alpha_Q}{\gamma_Q}\pair{L'}{{\mathrel{{\sqsubseteq}'}}}$ as follows
\begin{calculus}[$\Leftrightarrow$~]
\formula{\alpha_Q(f)\mathrel{{\sqsubseteq}'}y}\\
$\Leftrightarrow$
\formulaexplanation{f(Q)\mathrel{{\sqsubseteq}'}y}{def.\ $\alpha_Q$}\\
$\Leftrightarrow$
\formulaexplanation{f\mathrel{\stackrel{.}{{\sqsubseteq}}'}\LAMBDA{x}\si x=Q \alors y\sinon\top'\fsi}{pointwise def.\ ${\mathrel{\stackrel{.}{{\sqsubseteq}}'}}$ and $\top'$ is supremum for ${\mathrel{{\sqsubseteq}'}}$}\\
$\Leftrightarrow$
\formulaexplanation{f\mathrel{\stackrel{.}{{\sqsubseteq}}'}\gamma_Q(y)}{def.\ $\gamma_Q(y)$}
\end{calculus}
\smallskip

\noindent We have the commutation
\begin{calculus}[$\Leftrightarrow$~]
\formula{\alpha_Q(F(f))}\\
=
\formulaexplanation{F(f)(Q)}{def.\ $\alpha_Q$}\\
=
\formulaexplanation{\bar{F}_Q(f(Q))}{hypothesis}\\
=
\formulaexplanation{\bar{F}_Q(\alpha_Q(f))}{def.\ $\alpha_Q$}
\end{calculus}
\smallskip

\noindent By the fixpoint abstraction Th\@. \ref{th:fixpoint-abstraction} and definition of $\alpha_Q$, we conclude that 
$\Lfp{\mathrel{{\sqsubseteq}'}}\bar{F}_Q = \alpha_Q(\Lfp{\mathrel{\stackrel{.}{{\sqsubseteq}}'}}F) = (\Lfp{\mathrel{\stackrel{.}{{\sqsubseteq}}'}}F)Q$.
\end{proof}
\end{toappendix}
When the abstraction involves the negation abstraction $\alpha^{\neg}$, Park's classic fixpoint theorem \cite[equation (4,1,2)]{DBLP:conf/ac/Park79} is useful (and generalizes to complete Boolean lattices).
\begin{theorem}[Complement dualization]\label{th:Complement-dualization}If $\,X$ is a set and $f\in\wp(X)\mathrel{\smash{\stackrel{i}{\longrightarrow}}} \wp(X)$ is $\,\subseteq$-increasing then 
$\Lfp{\subseteq}{\alpha^{\sim}(f)}=\alpha^{\neg}(\Gfp{\subseteq}{f}).$
\end{theorem}
\begin{toappendix}
\section{Abstractions of Deductive Systems}
In order to abstract a deductive system $R$, we can consider its model semantics $\Lfp{\subseteq}F_{R}$ as in Sect\@. \ref{sec:Deductive-system-Model-theoretic-definition}, abstract this fixpoint into $\Lfp{\subseteq}\bar{F}_{R}$ by abstracting its transformer $F_{R}$ into an abstract one $\bar{F}_{R}$  using Th\@. \ref{th:fixpoint-abstraction}, and then this abstract fixpoint back to an abstract deductive system, as explained in Sect\@. \ref{sec:equivalence-fxp-deductive-system}. A more direct way is as follows
\begin{theorem}[Deductive system abstraction]\label{th:Deductive-system-abstraction}
Let $S\in\wp(X)$ be the set defined by the deductive system $R$. Let $\pair{\wp(X)}{\subseteq}\galois{\alpha}{\gamma}\pair{\wp(Y)}{\subseteq}$. Let $\bar{S}\in\wp(Y)$ be defined by the abstract deductive system $\bar{R}\triangleq\bigl\{\frac{\alpha(P)}{\bar{c}}\bigm|\frac{P}{c}\in R\wedge \bar{c}\in\alpha(\{c\})\bigr\}$. Then $\alpha(S)\subseteq\bar{S}$. For a Galois isomorphism, $\alpha(S)=\bar{S}$.
\end{theorem}
\begin{proof}[Proof of Th\@. \ref{th:Deductive-system-abstraction}]Consider the consequence operators $F_{R}$ and $F_{\bar{R}}$ of the two deductive systems $R$ and $\bar{R}$. We have
\begin{calculus}[${\subseteq}/{=}$~]
\formula{\alpha(F_{R}(X))}\\
=\formulaexplanation{\alpha(\bigl\{c\bigm|\frac{P}{c}\in R\wedge P\subseteq X\bigr\}}{def.\ consequence operator $F_{R}$}\\
=\formulaexplanation{\alpha(\bigcup\bigl\{\{c\}\bigm|\frac{P}{c}\in R\wedge P\subseteq X\bigr\})}{def.\ $\bigcup$}\\
=\formulaexplanation{\bigcup\bigl\{\alpha(\{c\})\bigm|\frac{P}{c}\in R\wedge P\subseteq X\bigr\})}{in a Galois connection $\alpha$ preserves existing joins}\\
=\formulaexplanation{\bigl\{\bar{c}\bigm|\bar{c}\in\alpha(\{c\})\wedge\frac{P}{c}\in R\wedge P\subseteq X\bigr\}}{def.\ $\bigcup$}\\
${\subseteq}/{=}$\formula{\bigl\{\bar{c}\bigm|\bar{c}\in\alpha(\{c\})\wedge\frac{P}{c}\in R\wedge \alpha(P)\subseteq\alpha(X)\bigr\}}\\[-0.5ex]
\explanation{for $\subseteq$ in case of a  Galois connection, $\alpha$  is increasing\\
for $=$ in case of a Galois isomorphism such that, conversely, $\alpha(X)\subseteq\alpha(Y)$ implies $X\subseteq Y$}\\
=\formulaexplanation{F_{\bar{R}}(\alpha(X))}{def.\ consequence operator $F_{\bar{R}}$}
\end{calculus}
\raisebox{0.5ex}{\strut}We conclude, by semi-commutativity and fixpoint abstraction Th\@. \ref{th:fixpoint-abstraction} that $S=\Lfp{\subseteq}F_{R}\subseteq\Lfp{\subseteq}F_{\bar{R}}=\bar{S}$ with equality in case of a Galois isomorphism implying commutativity $\alpha\comp F_{R}=F_{\bar{R}}\comp\alpha$.
\end{proof}
\end{toappendix}

\section{Fixpoint Induction}\label{sect:FixpointInduction}
Least or greatest fixpoint definitions of the graph of transformers provide strongest or antecedent-consequent (or weakest consequent-antecedent) pairs. Then we need to take into account consequence rules, that is, approximations discussed
in Sect\@. \ref{sec:Weakening-strengthening-abstractions}. In this section, and in addition to \cite{DBLP:conf/lopstr/Cousot19} and \cite[Ch\@. 24]{Cousot-PAI-2021}, we introduce fixpoint induction methods to handle such approximations $\textsf{\textup{post}}({\supseteq},{\subseteq})$,
$\textsf{\textup{post}}({\subseteq},{\supseteq})$, etc. In this section \ref{sect:FixpointInduction}, $\bot$ is the infimum of a poset and possibly unrelated to nontermination.

\subsection{Least Fixpoint Over Approximation}\label{sec:Induction-Principles}

The classic least fixpoint ($\textsf{\upshape lfp}$) over approximation theorem (and order dually over approximation of greatest fixpoints ($\textsf{\upshape gfp}$)), called ``fixpoint induction'', is due to Park \cite{Park69-MI5} and follows directly from Tarski's fixpoint theorem \cite{Tarski-fixpoint},
$\Lfp{\sqsubseteq}{f}=\bigsqcap\{x\in L\mid f(x)\sqsubseteq x\}$.
\begin{theorem}[Least fixpoint over approximation]\label{th:Fixpoint-Overapproximation}
Let $\sextuple{L}{\sqsubseteq}{\bot}{\top}{\sqcup}{\sqcap}$ be a complete lattice,  $f\in L\mathrel{\smash{\stackrel{i}{\longrightarrow}}} L$ be increasing, and $p\in L$. Then $\Lfp{\sqsubseteq}{f}\sqsubseteq p$ if and only if $\,\exists i\in L\mathrel{.} f(i)\sqsubseteq i \wedge i\sqsubseteq p$.
\end{theorem}
\begin{example}\label{ex:invariant:turcing-floyd}An invariant of a conditional iteration \texttt{\small while(B)\,S} with precondition $P$ must satisfy $\Lfp{\subseteq}{\LAMBDA{X}P\cup\textsf{\textup{post}}\sqb{S}(B\cap X)}\subseteq I$. The proof method provided by Park's Th\@. \ref{th:Fixpoint-Overapproximation} is $\exists J\mathrel{.}P\subseteq J\wedge \textsf{\textup{post}}\sqb{S}(B\cap J)\subseteq J\wedge J\subseteq I$ which is Turing \cite{Turing49-program-proof}/Floyd \cite{Floyd67-1} invariant proof method.
\end{example}
By order-duality, this is sound and complete greatest fixpoint under approximation  $p\sqsubseteq\Gfp{\sqsubseteq}{f}$ proof method. $i$ is called an invariant (a co-invariant for greatest fixpoints).
\begin{example}Continuing example \ref{ex:invariant:turcing-floyd}, by contraposition, the invariant must satisfy $\neg I\subseteq \neg\Lfp{\subseteq}{\LAMBDA{X}P\cup\textsf{\textup{post}}\sqb{S}(B\cap X)}$ that is
$\neg I\subseteq \Gfp{\subseteq}{\LAMBDA{X}\neg P\cap\widetilde{\textsf{\textup{post}}}\sqb{S}(\neg B\cup X)}$ by Park's Th\@. \ref{th:Complement-dualization}. The dual of Th\@. \ref{th:Fixpoint-Overapproximation} suggest the proof method $\exists J\mathrel{.}J\subseteq \neg P\wedge J\subseteq\widetilde{\textsf{\textup{post}}}\sqb{S}(\neg B\cup J)\wedge I\subseteq \neg J$ which is methods (i${}^{\widetilde{-1}}$) and (I${}^{\widetilde{-1}}$) of \cite{CousotCousot82-TNPC}.
\end{example}

\subsection{Ordinals}
\ifshort
We let $\sextuple{\mathbb{O}}{\in}{\emptyset}{\mathbb{O}}{\cup}{\cap}$ be the von Neumann's ordinals \cite{vonNeumann-1923-ordinals}, writing the more
intuitive $<$ for $\in$, $0$ for $\emptyset$, ${}+1$ for the successor function, sometimes $\max$ for $\cup$, $\min$ for $\cap$, and $\omega$ for the first infinite limit ordinal. If necessary, a short refresher on ordinals is given in Sect\@. \ref{sec:apx:refresher-on-ordinals} of the appendix \proofinapx.
\fi
\begin{toappendix}
\label{sec:apx:refresher-on-ordinals}
We let $\sextuple{\mathbb{O}}{\in}{\emptyset}{\mathbb{O}}{\cup}{\cap}$ be the von Neumann's ordinals \cite{vonNeumann-1923-ordinals} (see \cite{Monk-Set-Theory} for an introduction). Von Neumann's ordinals are $0 = \emptyset$, $1=\{\emptyset\}$, $2= \{\emptyset,\{\emptyset\}\} = \{0,1\}$, $3=\{0,1,2\}$, \ldots, $\delta+1=\{\delta\} \cup\delta$ for successor ordinals, and $\lambda = \bigcup_{\beta\in\delta} \beta$ for limit ordinals ordered by $\in$. We let $\omega$ be the first infinite limit ordinal. For better intuition, we will denote the strict order $\in$ by $<$, $0$ for $\emptyset$, ${}+1$ for the successor function, sometimes $\max$ for $\cup$, $\min$ for $\cap$. For this total order $\pair{\mathbb{O}}{\leqslant}$ where $\leqslant$ is $<$ or equality, $\cup$ is the least upper bound and $\cap$ is the greatest lower bound. 

We let $\mathfrak{Wf}$ be the class of all well-founded sets $\pair{W}{\leqslant}$, that is, a set $W$ equipped with a binary relation $\leqslant$, with no infinite strictly decreasing chain for $\leqslant$. $\leqslant$ is often a partial order, but this is not necessary. In that case $\pair{W}{\leqslant}$ is called a well-ordered set. Any well-founded sets $\pair{W}{\leqslant}\in\mathfrak{Wf}$ can be mapped to a unique ordinal. Define the ranking function $\rho(w)=\bigcup_{w'\in W\wedge w'<w} \rho(w')$. Notice that for minimal elements $m$ of $W$ (i.e.\ $\forall w'\in W: w'\nless m$ and there must be some by wellfoundedness), $\rho(m)=\bigcup\emptyset=\emptyset$ that is, $0$. For any subset $X\in\wp(W)$ define $\rho(X)=\bigcup_{w\in X} \rho(w)$. The order of the well-founded set $\pair{W}{\leqslant}$, is $\rho(W)$. The elements $w$ of $W$ are mapped to an ordinal $\rho(w)\in\rho(W)$, The crucial property, that we will use, is that $w<w'\Rightarrow \rho(w)\in\rho(w')$ i.e.\ $\rho(w)<\rho(w')$. So ordinals are a well founded set (indeed a well founded total ordering) and any well founded set can be abstracted into an ordinal, that is
$\pair{\mathfrak{Wf}}{\subseteq}\galoiS{\rho}{\textsf{\upshape id}}\pair{\mathbb{O}}{\leqslant}$ where $\textsf{\upshape id}$ is the identity function and all well founded set of the same rank are abstracted to the same ordinal, while larger well-founded sets can only be mapped to larger ordinals.
\end{toappendix}

\smallskip 

\subsection{Over Approximation of the Abstraction of a Least Fixpoint}

To solve the problem $\alpha(\Lfp{\sqsubseteq}F)\sqsubseteq P$ where $\alpha$ is a function on the domain of 
$F$, we can try to use fixpoint abstraction Th\@. \ref{th:fixpoint-abstraction} to get 
$\alpha(\Lfp{\sqsubseteq}F)=\Lfp{\sqsubseteq}\bar{F}$ and then check 
$\Lfp{\sqsubseteq}\bar{F}\sqsubseteq P$ by fixpoint induction Th\@. \ref{th:Fixpoint-Overapproximation}. But Th\@. \ref{th:fixpoint-abstraction} requires $\alpha$
 to preserves joins, which is not always the case (for the dual problem $\alpha=\textsf{pre}$ in remark
 \ref{rem:pre-does-not-preserve-meets} is a counter-example). If $\alpha$ does not preserves joins, we can nevertheless use the following theorem \proofinapx.
\begin{theorem}[Overapproximation of a least fixpoint image]\label{th:OverapproximationLeastFixpointImage}Let $\quadruple{L}{\sqsubseteq}{\bot}{\sqcup}$ and $\quadruple{\bar{L}}{\mathrel{\bar{\sqsubseteq}}}{\bar{\top}}{\bar{\sqcup}}$ be complete lattices$\,$\footnote{or CPOs.}, $F\in L\mathrel{\smash{\stackrel{i}{\longrightarrow}}}L$ and $\alpha\in L\mathrel{\smash{\stackrel{i}{\longrightarrow}}}\bar{L}$ be increasing functions, and $P\in\bar{L}$. 

\smallskip 

Then $\alpha(\Lfp{\sqsubseteq}F)\mathrel{\bar{\sqsubseteq}} P$ if and only if there exists $I\in\bar{L}$ such that (1) $\alpha(\bot)\mathrel{\bar{\sqsubseteq}}I$ (2) $\forall X\in L\mathrel{.}\alpha(X)\mathrel{\bar{\sqsubseteq}}I \Rightarrow \alpha(F(X))\mathrel{\bar{\sqsubseteq}}I$, (3) for any $\sqsubseteq$-increasing chain $\pair{X^\delta}{\delta\in\mathbb{O}}$ of elements $X^\delta\sqsubseteq\Lfp{\sqsubseteq}F$, $\forall\beta<\lambda\mathrel{.}\alpha(X^\beta)\mathrel{\bar{\sqsubseteq}}I$ implies $\alpha(\bigsqcup_{\beta<\lambda}X^\beta)\mathrel{\bar{\sqsubseteq}}I$,  and (4) $I\mathrel{\bar{\sqsubseteq}}P$. 
\end{theorem}
Let $\pair{F^\delta}{\delta\in\mathbb{O}}$ be the increasing iterates of $F$ from $\bot$ ultimately stationary at rank $\epsilon$~\cite{CousotCousot-PJM-82-1-1979}. Then condition \ref{th:OverapproximationLeastFixpointImage}.(2) in is only necessary for all $X=F^\delta$, $\delta\leqslant\epsilon$ while condition (3) is only necessary for  $\pair{X^\delta}{\delta\leqslant\epsilon}$ = $\pair{F^\delta}{\delta\leqslant\epsilon}$. These weaker conditions are assumed to prove completeness (``only if'' in Th\@. \ref{th:OverapproximationLeastFixpointImage}). 
\begin{toappendix}
\begin{proof}[Proof of Th\@. \ref{th:OverapproximationLeastFixpointImage}] Let $\pair{F^\delta}{\delta\in\mathbb{O}}$ be the increasing iterates of $F$ from $\bot$ ultimately stationary at rank $\epsilon$\cite{CousotCousot-PJM-82-1-1979}. 

Soundness ($\Leftarrow$).\quad We have $\alpha(F^0)=\alpha(\bot)\mathrel{\bar{\sqsubseteq}}I$ by \ref{th:OverapproximationLeastFixpointImage}.(1).
Assume $\alpha(F^\delta)\mathrel{\bar{\sqsubseteq}}I$ by induction hypothesis. Then 
$\alpha(F^{\delta+1}) = \alpha(F(F^\delta))\mathrel{\bar{\sqsubseteq}}I$ by \ref{th:OverapproximationLeastFixpointImage}.(2). 
If $\lambda$ is limit ordinal and $\forall\beta<\lambda\mathrel{.}\alpha(F^\beta)\mathrel{\bar{\sqsubseteq}}I$. 
Then $\alpha(F^\lambda)=\alpha(\bigsqcup_{\beta<\lambda}F^\beta)\mathrel{\bar{\sqsubseteq}}I$ by \ref{th:OverapproximationLeastFixpointImage}.(3). By transfinite induction  $\forall \delta\in\mathbb{O}\mathrel{.}\alpha(F^\delta)\mathrel{\bar{\sqsubseteq}}I$. 
So $\alpha(\Lfp{\sqsubseteq}F)=\alpha(F^\epsilon)\mathrel{\bar{\sqsubseteq}} I\mathrel{\bar{\sqsubseteq}}P$ by \ref{th:OverapproximationLeastFixpointImage}.(4).

Completeness ($\Rightarrow$), for the weaker condition \ref{th:OverapproximationLeastFixpointImage}.(2) and \ref{th:OverapproximationLeastFixpointImage}.(3).\quad Assume that $\alpha(\Lfp{\sqsubseteq}F)\mathrel{\bar{\sqsubseteq}} P$ and define $I=\alpha(\Lfp{\sqsubseteq}F)$, $\pair{X^\delta}{\delta\in\mathbb{O}}$ = $\pair{F^\delta}{\delta\in\mathbb{O}}$. We have \ref{th:OverapproximationLeastFixpointImage}.(1) $\alpha(\bot)\sqsubseteq\alpha(\Lfp{\sqsubseteq}F)=I$ since $\bot$ is the infimum and $\alpha$ is increasing. For \ref{th:OverapproximationLeastFixpointImage}.(2) with $X=F^\delta$ we have $F^\delta\sqsubseteq\Lfp{\sqsubseteq}F$ so $\alpha(F(X))=\alpha(F(F^\delta))=\alpha(F^{\delta+1})\mathrel{\bar{\sqsubseteq}}\alpha(\Lfp{\sqsubseteq}F)=I$. For \ref{th:OverapproximationLeastFixpointImage}.(3) with $\pair{X^\delta}{\delta\in\mathbb{O}}$ = $\pair{F^\delta}{\delta\in\mathbb{O}}$ we have $\bigsqcup_{\beta<\lambda}X^\beta\sqsubseteq\Lfp{\sqsubseteq}F$  so 
$\alpha(\bigsqcup_{\beta<\lambda}X^\beta)\sqsubseteq\alpha(\Lfp{\sqsubseteq}F)=I$. Finally \ref{th:OverapproximationLeastFixpointImage}.(4), by hypothesis $I=\alpha(\Lfp{\sqsubseteq}F)\mathrel{\bar{\sqsubseteq}} P$.
\end{proof}
\end{toappendix}

\smallskip 

\subsection{Fixpoint Under Approximation by Transfinite Iterates}
For under approximation of least fixpoints  (or order dually over approximation of greatest fixpoints), we can use the generalization \cite{DBLP:conf/lopstr/Cousot19} of Scott-Kleene induction based on transfinite induction when continuity does not apply and follows directly from the constructive version of Tarski's fixpoint theorem \cite{CousotCousot-PJM-82-1-1979}.

\begin{definition}[Ultimately Over Approximating Transfinite Sequence]\label{def:Underapproximation-Sequence}
We say that ``the transfinite sequence $\pair{X^\delta}{\delta\in\mathbb{O}}$ of elements of poset $\pair{L}{\sqsubseteq}$ for $f\in L\rightarrow L$ ultimately over approximates $P\in L$''  if and only if $X^0=\bot$, $X^{\delta+1}\sqsubseteq f(X^{\delta})$ for successor ordinals, $\bigsqcup_{\delta<\lambda} X^\delta$ exists for limit ordinals $\lambda$ such that $X^{\lambda} \sqsubseteq \bigsqcup_{\delta<\lambda} X^\delta$, and $\exists\delta\in\mathbb{O}\mathrel{.}P\sqsubseteq X^\delta$.
\end{definition}
The condition can equivalently be expressed as $\forall \delta\in\mathbb{O}\mathrel{.}X^{\delta} \sqsubseteq f(\bigsqcup_{\beta<\delta} X^\beta+1)$ which avoids to have to make the distinction between successor and limit ordinals \proofinapx. 
\begin{theorem}[Fixpoint Under Approximation by Transfinite Iterates]\label{th:Fixpoint-Underapproximation}
Let $f\in L\mathrel{\smash{\stackrel{i}{\longrightarrow}}}L$ be an increasing function on a \textsc{cpo} $\quadruple{L}{\sqsubseteq}{\bot}{\sqcup}$ (i.e\@. every increasing chain in $L$ has a least upper bound in $L$, including $\bot=\sqcup\emptyset$). $P\in L$ is a fixpoint underapproximation, i.e\@. 
$P \sqsubseteq\Lfp{\sqsubseteq}{f}$, if and only if there exists an increasing transfinite sequence
 $\pair{X^\delta}{\delta\in\mathbb{O}}$ for $f$ ultimately over approximating $P$ (Def\@. \ref{def:Underapproximation-Sequence}).
\end{theorem}
Notice that  ordinals are an abstraction $\pair{\mathfrak{Wf}}{\subseteq}\galoiS{\rho}{\textsf{\upshape id}}\pair{\mathbb{O}}{\leqslant}$ of well-founded sets by their rank $\rho$, so that Th\@. \ref{th:Fixpoint-Underapproximation} could have assumed the existence of a well-founded set to replace the ordinals.
\begin{toappendix}
\label{sec:apx:proof:th:Fixpoint-Underapproximation}
\begin{proof}[Proof of Th\@. \ref{th:Fixpoint-Underapproximation}] Since $f$ in increasing on $L$, the iterates $I^{0}=\bot$,
$I^{\delta+1}=f(I^{\delta})$, and $I^{\lambda}=\bigsqcup_{\beta<\lambda}I^{\beta}$ are an increasing chain in $L$ so, its 
cardinality being bounded by that of $L$, the sequence must be ultimately stationary at rank $\epsilon$, to a fixpoint 
$I^{\delta}$ for all $\delta\geqslant\epsilon$, which, $f$ being increasing, is the least one \cite{CousotCousot-PJM-82-1-1979}.
 By transfinite induction, 
 $\pair{X^{\delta}}{\delta\in\mathbb{O}}$ 
 is pointwise bounded by 
 $\pair{I^{\delta}}{\delta\in\mathbb{O}}$ i.e. 
 $\forall \delta\in\mathbb{O}\mathrel{.}X^{\delta}\sqsubseteq I^{\delta}$. 
Therefore, $\pair{X^{\delta}}{\delta\in\mathbb{O}}$ being increasing,
we have $\exists\delta\in\mathbb{O}\mathrel{.}P$
$\sqsubseteq$
$X^{\delta}$ 
$\sqsubseteq$ 
$I^{\delta}$ 
$\sqsubseteq $
$I^{\max(\delta,\epsilon)}$
=
$\Lfp{\sqsubseteq}{f}$.
Conversely, choosing $\pair{X^\delta}{\delta\in\mathbb{O}}$ = $\pair{I^\delta}{\delta\in\mathbb{O}}$ will do
since $P\sqsubseteq \Lfp{f} = I^{\epsilon}$.
\end{proof}
\end{toappendix}
The hypothesis that $\pair{X^{\delta}}{\delta\in\mathbb{O}}$ is increasing is necessary in a \textsc{cpo} but not in a complete lattice, in which case this non-increasing sequence can be used to build an increasing one \proofinapx.
\begin{lemma}\label{rmk:increasing-sequence}Let $\pair{X^\delta}{\delta\in\mathbb{O}}$ be a sequence in a complete lattice satisfying the hypotheses of Def\@. \ref{def:Underapproximation-Sequence}, then there is an increasing one satisfying these same hypotheses.
\end{lemma}
\begin{toappendix}
\label{sec:apx:proof:rmk:increasing-sequence}
\begin{proof}[Proof of lemma \ref{rmk:increasing-sequence}]Let $\pair{X^\delta}{\delta\in\mathbb{O}}$ be a sequence in a complete lattice $\quadruple{L}{\sqsubseteq}{\bot}{\sqcup}$ satisfying the hypotheses of Def\@. \ref{def:Underapproximation-Sequence}. Define $Y^0=\bot$, $Y^{\delta+1}=f(Y^{\delta})$, and $Y^{\lambda}=\bigsqcup_{\beta<\lambda}Y^{\beta}$ which is well defined in the complete lattice $L$ and, by transfinite induction, increasing since for the basis $X^0=\bot\sqsubseteq X^1$ and for the induction $f$ is increasing \cite{CousotCousot-PJM-82-1-1979}. To prove that this sequence $\pair{Y^\delta}{\delta\in\mathbb{O}}$ is an upper bound of $\pair{X^\delta}{\delta\in\mathbb{O}}$, observe, for the basis, that $X^0=\bot\sqsubseteq \bot=Y^0$ by reflexivity and definition of the iterates. Assume $X^{\delta}\sqsubseteq Y^{\delta}$ by induction hypothesis. Then  $X^{\delta+1}\sqsubseteq f(X^{\delta})\sqsubseteq f(Y^{\delta})=Y^{\delta+1}$ since $f$ is increasing and by definition of the iterates. Moreover,  $\forall\beta<\lambda\mathrel{.}X^{\beta}\sqsubseteq Y^{\beta}$ implies $\forall\beta<\lambda\mathrel{.}X^{\beta}\sqsubseteq \bigsqcup_{\beta<\lambda}Y^{\beta}$ implies $X^{\lambda}\sqsubseteq\bigsqcup_{\beta<\lambda}X^{\beta}\sqsubseteq \bigsqcup_{\beta<\lambda}Y^{\beta}=Y^{\lambda}$ by definition of the least upper bound and the iterates. By transfinite induction,  $\forall \delta\in\mathbb{O}\mathrel{.}X^{\delta}\sqsubseteq Y^{\delta}$. By hypothesis $\exists\delta\in\mathbb{O}\mathrel{.}P\sqsubseteq X^\delta$ so $\exists\delta'\in\mathbb{O}\mathrel{.}P\sqsubseteq Y^{\delta'}$ with $\delta'=\delta$, $X^{\delta}\sqsubseteq Y^{\delta}$, and transitivity. In conclusion, $\pair{Y^\delta}{\delta\in\mathbb{O}}$ is well-defined, increasing, and satisfies Def\@. \ref{def:Underapproximation-Sequence}.
\end{proof}
\end{toappendix}

\subsection{Fixpoint Under Approximation by Bounded Iterates}
For iterations, under approximations such as $P\subseteq \textsf{\upshape post}\sqb{\texttt{\small S}}_{\bot}Q$ (incorrectness logic), $P\subseteq \textsf{\upshape pre}\sqb{\texttt{\small S}}_{\bot}\Sigma$ (possible termination),
$P\subseteq \neg\textsf{\upshape pre}\sqb{\texttt{\small S}}_{\bot}\{\bot\}=\widetilde{\textsf{\upshape pre}}\sqb{\texttt{\small S}}_{\bot}\Sigma$ (definite termination), 
and $P\subseteq \textsf{\upshape pre}\sqb{\texttt{\small S}}_{\bot}Q\cap \widetilde{\textsf{\upshape pre}}\sqb{\texttt{\small S}}_{\bot}Q$ 
(weakest precondition, starting from any initial state of $P$, $\texttt{\small S}_{\bot}$ ``is certain 
to establish eventually the truth of'' $Q$ \cite[page 17]{DBLP:books/ph/Dijkstra76}) are fixpoint under approximations. Programmers almost never use Th\@. \ref{th:Fixpoint-Underapproximation} for proving termination using ordinals (or a well-founded set). They cannot use Hoare logic either since nontermination $\{P\}\texttt{\small S}\{\textsf{\upshape false}\}$ is provable by the logic but its negation 
$\neg(\{P\}\texttt{\small S}\{\textsf{\upshape false}\})$ is not in the logic. A first method for bounded iteration uses a loop counter incremented on each iteration and an invariant proving that the counter is bounded (``time clocks'' in \cite{DBLP:books/aw/Knuth68}, \cite{DBLP:journals/acta/LuckhamS77,DBLP:journals/acta/Sokolowski77}). This is sound but incomplete for unbounded nondeterminism. The most popular method uses well-founded sets, which can  be generalized to fixpoints \proofinapx.

\begin{theorem}[Least Fixpoint Under Approximation with a Variant Function]\label{th:Fixpoint-Underapproximation-Variant}
We assume that (1) $f$ is increasing on a \textsc{cpo} $\quadruple{L}{\sqsubseteq}{\bot}{\sqcup}$;  (2) that $P\in L$; (3) that there exists a sequence $\pair{X^{\delta}}{\delta\in\mathbb{O}}$ of elements of $L$ such that $X^{0}=\bot$, $X^{\delta+1} \sqsubseteq f(X^{\delta})$ for successor ordinals, and $X^{\lambda}\sqsubseteq \bigsqcup_{\beta<\lambda}X^{\beta}$ for limit ordinals $\lambda$; and (4) that there exists a well-founded set $\pair{W}{\preceq}$ and a variant function $\nu\in\{X^{\delta}\mid \delta\in\mathbb{O}\}\rightarrow W$ such that  for all $\beta<\delta$, we have $P\not\sqsubseteq X^{\beta}$ implies $\nu(X^{\beta})\succ\nu(X^{\delta})$.

Hypotheses(1) to (4) imply that $\exists \delta<\omega\mathrel{.}P\sqsubseteq X^{\delta}\sqsubseteq f^{\delta}\sqsubseteq\Lfp{\sqsubseteq}f$.
\end{theorem}
\begin{toappendix}
\begin{proof}[Proof of Th\@. \ref{th:Fixpoint-Underapproximation-Variant}]
Observe that, by hypothesis \ref{th:Fixpoint-Underapproximation-Variant}.(1), the transfinite iterates
 $f^{0}=\bot$,
$f^{\delta+1}=f(f^{\delta})$, and $f^{\lambda}=\bigsqcup_{\beta<\lambda}f^{\beta}$ of $f$, starting at $\bot$, are increasing and converging at rank $\epsilon\in\mathbb{O}$ to 
$\Lfp{\subseteq}f$ \cite{CousotCousot-PJM-82-1-1979}. Consider $\pair{X^{\delta}}{\delta\leqslant\epsilon}$ as defined by hypothesis \ref{th:Fixpoint-Underapproximation-Variant}.(3). Since $f$ is increasing, it follows, by transfinite induction, that
$\forall\delta\leqslant\epsilon\mathrel{.}X^{\delta}\sqsubseteq f^{\delta}$. If $\exists\delta\in\mathbb{O}\mathrel{.}P\sqsubseteq X^{\delta}$ then we are done since $P\sqsubseteq X^{\delta}\sqsubseteq f^{\delta}\sqsubseteq \Lfp{\sqsubseteq}f$. Otherwise, $\forall\delta\in\mathbb{O}\mathrel{.}P\not\sqsubseteq X^{\delta}$ and so, by hypothesis \ref{th:Fixpoint-Underapproximation-Variant}.(4), the chain $\pair{\nu(X^{\delta})}{\delta\in\mathbb{O}}$ is strictly $\prec$-decreasing, in contradiction with the hypothesis that $\pair{W}{\preceq}$ is well-founded, proving $\exists\delta\in\mathbb{O}\mathrel{.}P\sqsubseteq X^{\delta}\sqsubseteq f^{\delta}\sqsubseteq \Lfp{\sqsubseteq}f$. Notice that, by well-foundedness, $\pair{\nu(X^{\delta})}{\delta\in\mathbb{O}}$ reaches a minimal element at some rank less than $\omega$ so $\delta<\omega$.
\end{proof}
\end{toappendix}
Because $\delta<\omega$ in Th\@. \ref{th:Fixpoint-Underapproximation-Variant}, the proof method is sound but incomplete, as shown by the following counter example where the property holds but the proof method of Th\@. \ref{th:Fixpoint-Underapproximation-Variant},  is inapplicable.
\begin{example}Consider the complete lattice $\pair{\wp(\mathbb{Z})}{\subseteq}$. Define $f\in\wp(\mathbb{Z})\rightarrow\wp(\mathbb{Z})$ by $f(X)= \{0\}\cup\{x\in \mathbb{Z}\mid x-1\in X\}$. The iterates are $f^{0}=\emptyset$, $f^{n}=\{k\in\mathbb{N}\mid 0\leqslant x< n\}$. The limit is $f^{\omega}=\bigcup_{n\in\mathbb{N}}f^{n}=\mathbb{N}=\Lfp{\subseteq}f$. Take $P=\mathbb{N}$ such that $P\subseteq\Lfp{\subseteq}f$. Then $\forall n\in \mathbb{N}\mathrel{.}P\not\subseteq f^{n}$.
It follows that Def\@. \ref{th:Fixpoint-Underapproximation-Variant} is infeasible since $\forall n\in \mathbb{N}\mathrel{.}P\not\subseteq f^{n}$ implies for all $\beta<\delta$ that $\nu(X^{\beta})\succ\nu(X^{\delta})$. This infinite strictly decreasing chain is in contradiction with the well-foundness hypothesis.
\end{example}

\smallskip

\subsection{Void Intersection With Fixpoint Using Variant Functions}

Turing and Floyd \cite{Turing49-program-proof,Floyd67-1} method for unbounded nondeterminism, uses \emph{reductio ad absurdum}, proving that nontermination is impossible. This idea can also be generalized to fixpoints.

An atom of a poset $\pair{L}{\sqsubseteq}$ is either a minimal element of $L$  if $L$ has no infimum  or covers the infimum $\bot$ otherwise. So the set  of atoms of a poset $\pair{L}{\sqsubseteq}$ is $\textsf{\upshape atoms}(L)$ $\triangleq$ $\{a\in L\mid{}\not\exists x'\in L\mathrel{.}x'\sqsubset a\}$ if $L$ has no infimum  and $\textsf{\upshape atoms}(L)$ $\triangleq$ $\{a\in L\mid \not\exists x'\in L\mathrel{.}\bot\sqsubset x'\sqsubset a\}$ if $\bot$ is the infimum of $L$. The atoms of an element $x$ of $L$ are $\textsf{\upshape atoms}(x)\triangleq\{a\in\textsf{\upshape atoms}(L)\mid a\sqsubseteq x\}$.
A poset is atomic if the atoms of any element $x$ of $L$ have a join which exists and is $x$, that is, $\forall x\in L\mathrel{.}x=\bigsqcup\textsf{\upshape atoms}(x)$. Co-atomicity is $\sqsubseteq$-order-dual. We have \proofinapx
\begin{theorem}[Void intersection with least fixpoint]\label{th:Intersection-Lfp}
We assume that 
(1) $\sextuple{L}{\sqsubseteq}{\bot}{\top}{\sqcap}{\sqcup}$ is an atomic complete lattice;
(2) $f\in L\rightarrow L$ preserves non-empty joins; 
(3) there exists an invariant $I\in L$ of $f$ (i.e.\ such that $f(I)\sqsubseteq I$); 
(4) that there exists a well-founded set $\pair{W}{\preceq}$ and a variant function $\nu\in I\rightarrow W$ such that $\forall x\in\textsf{\upshape atoms}(I)\mathrel{.}(x\neq f(x))\Rightarrow (\nu(x)\succ\nu(f(x)))$;
(5) $Q\in L$; and 
(6) $\forall x\in\textsf{\upshape atoms}( I)\mathrel{.}(\nu(x)\not\succ\nu(f(x)))\Rightarrow (x\sqcap Q=\bot)$.
Then, hypotheses (1) to (6) imply $\Lfp{\sqsubseteq}f \sqcap Q=\bot$.
\end{theorem}
Th\@. \ref{th:Intersection-Lfp}  is useful, in particular, to prove $\Lfp{\sqsubseteq}f=\bot$ for $Q=\top$. If $L=\wp(\Sigma_\bot)$ then $P\subseteq Q$ is $P\cap \neg Q=\emptyset$, another possible use of this theorem.
\begin{toappendix}
\begin{proof}[Proof of Th\@. \ref{th:Intersection-Lfp}]By ref{th:Intersection-Lfp}.(2), $f$ preserves non-empty joins implies that the
 iterates $f^0(\bot)\triangleq\bot$ and $f^{n+1}(\bot)\triangleq f(f^{n}(\bot))$ of $f$ from the infimum $\bot$ on the $L$ form an increasing chain which limit, by ref{th:Intersection-Lfp}.(1) exists in the complete lattice L, and is $\Lfp{\sqsubseteq}f$ = $f^{\omega}(\bot)$ $\triangleq$ $\bigsqcup_{n\in\mathbb{N}}f^{n}(\bot)$  \cite{CousotCousot-PJM-82-1-1979}. Moreover, by ref{th:Intersection-Lfp}.(3), $f(I)\sqsubseteq I$ implies, by recurrence, that $\forall n\in\mathbb{N}\cup\{\omega\}\mathrel{.}f^{n}(\bot)\sqsubseteq I$ since by ref{th:Intersection-Lfp}.(2) $f$ preserves non-empty joins so is increasing. 

\smallskip

If $f(\bot)=\bot$ then obviously $\Lfp{\sqsubseteq}f=\bot=\bot\sqcap Q$. So in the following, we assume that $\bot\neq f(\bot)$ that is, $\bot\sqsubset f(\bot)$ by definition of the infimum $\bot$.

In an atomic complete lattice, all elements have a non-empty set of atoms but $\bot$. Since $\bot\neq f(\bot)$, we have  $\textsf{\upshape atoms}(f(\bot))\neq\emptyset$. 

We have
$\bot=f^{0}(\bot)\sqsubset f^{1}(\bot)=\bigsqcup\{f^0(x)\mid x\in \textsf{\upshape atoms}(f(\bot))\}$ by definition of the identity $f^0$. 
Assume that, for $n\geqslant 1$,   $f^n(\bot)=\bigsqcup\{f^{n-1}(x)\mid x\in \textsf{\upshape atoms}(f(\bot))\}$ by recurrence hypothesis.
By ref{th:Intersection-Lfp}.(2), $f$ preserves non-empty joins so  $f^{n+1}(\bot)=f(f^{n}(\bot))$ =
$f(\bigsqcup\{f^{n-1}(x)\mid x\in \textsf{\upshape atoms}(f(\bot))\})$ =
$\bigsqcup\{f(f^{n-1}(x))\mid x\in \textsf{\upshape atoms}(f(\bot))\}$ =
$\bigsqcup\{f^{n}(x)\mid x\in \textsf{\upshape atoms}(f(\bot))\}$. By recurrence, $\forall n\geqslant 1\mathrel{.}f^n(\bot)=\bigsqcup\{f^{n-1}(x)\mid x\in \textsf{\upshape atoms}(f(\bot))\}$. 

By ref{th:Intersection-Lfp}.(1) and definition of an atom $x\in \textsf{\upshape atoms}(f(\bot))$, we have $\bot\sqsubset x\sqsubseteq f(\bot)$ that is
 $f^0(\bot)\sqsubset f^0(x)\sqsubseteq f(\bot)$. Since, by ref{th:Intersection-Lfp}.(2), $f$ preserves non-empty joins it is increasing and therefore $\forall n\in\mathbb{N}\mathrel{.}
f^n(\bot)\sqsubseteq f^n(x)\sqsubseteq f^{n+1}(\bot)$ by recurrence. By definition of a least upper bound, it follows that $\bigsqcup_{n\in\mathbb{N}}f^n(\bot)\sqsubseteq\bigsqcup_{n\in\mathbb{N}} f^n(x)\sqsubseteq\bigsqcup_{n\in\mathbb{N}} f^{n+1}(\bot)=\bigsqcup_{n\in\mathbb{N}}f^n(\bot)$ since $\bot\sqsubseteq f(\bot)=f^{1}(\bot)$. This implies that
$\forall x\in\textsf{\upshape atoms}(f(\bot))\mathrel{.}\bigsqcup_{n\in\mathbb{N}} f^n(x)=\Lfp{\sqsubseteq}f$.

Because $\forall n\geqslant 1\mathrel{.}f^n(x)=f(f^{n-1}(x))\sqsubseteq I$, ref{th:Intersection-Lfp}.(4) implies that the chain $\pair{\nu(f^n(x))}{n\geqslant 1}$ in the well-founded set $\pair{W}{\preceq}$ cannot infinitely decrease so is ultimately stationary at some $n=\ell_x$. By contraposition of ref{th:Intersection-Lfp}.(4), $\forall y\sqsubseteq I\mathrel{.} (\nu(y)\not\succ\nu(f(y)))\implies(y= f(y))\sqsubseteq I$ so that $f^{\ell_x}(x)=f^{m}(x)\sqsubseteq I$ for all $m\geqslant \ell_x$. By ref{th:Intersection-Lfp}.(5) and ref{th:Intersection-Lfp}.(6), it follows that $\forall m\geqslant \ell_x\mathrel{.}f^{m}(x)\sqcap Q=\bot$

It follows, by definition of the least upper bound, that $\Lfp{\sqsubseteq}f$ =
$\bigsqcup_{n\in\mathbb{N}}f^{n}$ =
$\bigsqcup_{n\geqslant 1}\bigsqcup\{f^{n-1}(x)\mid x\in \textsf{\upshape atoms}(f(\bot))\}$ =
$\bigsqcup\{\bigsqcup_{n\geqslant 1}f^{n-1}(x)\mid x\in \textsf{\upshape atoms}(f(\bot))\}$ =
$\bigsqcup\{f^{\ell_x}(x)\mid x\in \textsf{\upshape atoms}(f(\bot))\}$.\
We have shown that $\forall x\in \textsf{\upshape atoms}(f(\bot))\mathrel{.}f^{\ell_x}(x)\sqcap Q=\bot$ so
$\Lfp{\sqsubseteq}f\sqcap Q$ =
$\bigsqcup\{f^{\ell_x}(x)\mid x\in \textsf{\upshape atoms}(f(\bot))\}\sqcap Q=\bot$ in the atomic complete lattice $L$.
\end{proof}
\end{toappendix}

The proof method of Th\@. \ref{th:Intersection-Lfp} is incomplete, as shown by counter-example \ref{counter-example-th:Intersection-Lf} in the appendix.
\begin{toappendix}
\begin{example}[Counter-example to the completeness of Th\@. \ref{th:Intersection-Lfp}]\label{counter-example-th:Intersection-Lf}
\begin{figure}[h]
\includegraphics[scale=0.15]{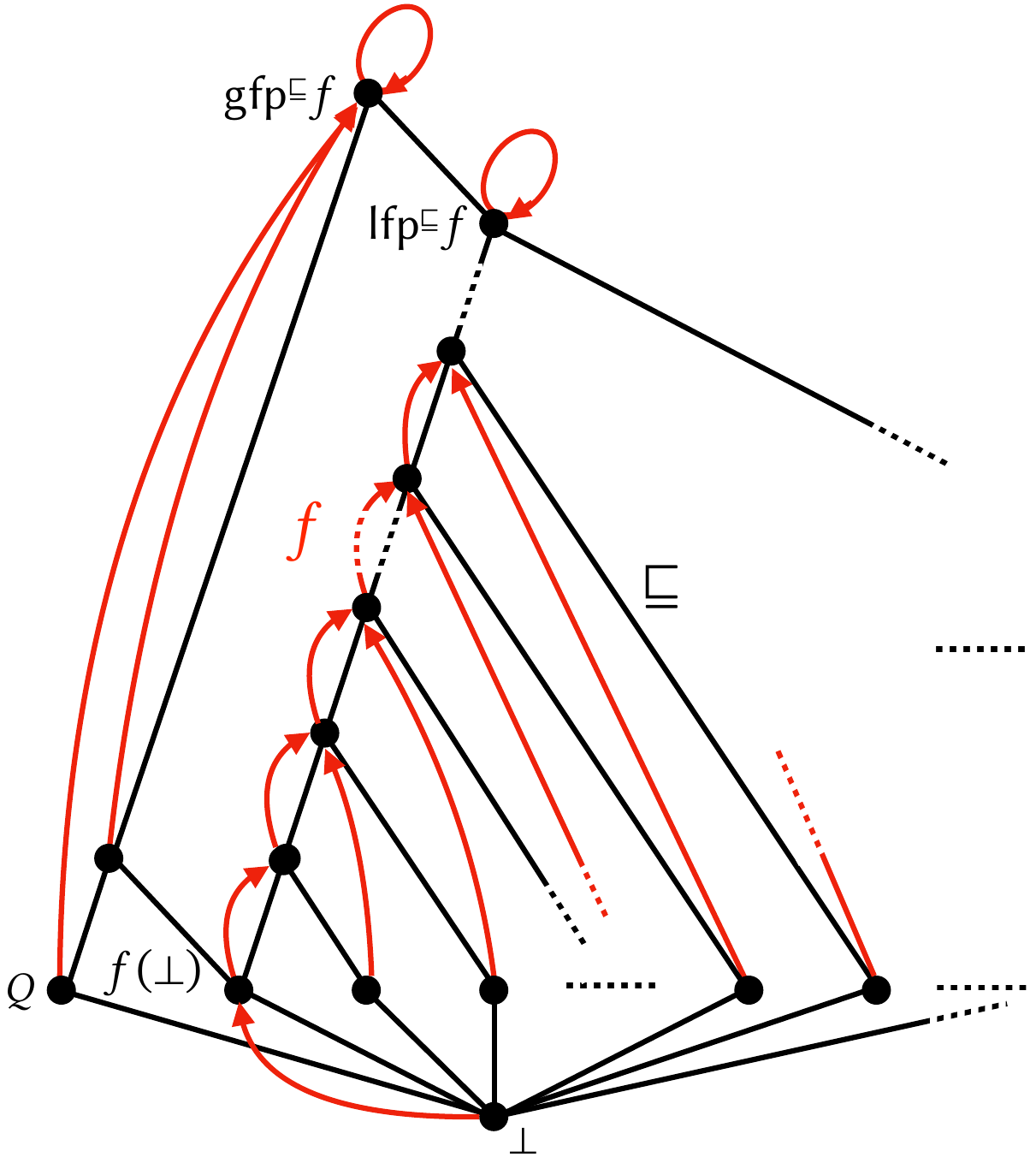}
\caption{Counter-example to the completeness of Th\@. \ref{th:Intersection-Lfp}}
\end{figure}
The iterates $\pair{f^n(\bot)}{n\in\mathbb{N}}$ of $f$ from $\bot$ are the same as those for the atoms of $f(\bot)$ which are nothing but $f(\bot)$ itself and are infinitely strictly increasing, so the hypothesis (4) of Th\@. \ref{th:Intersection-Lfp} cannot be satisfied since it would imply convergence in finitely many steps. 
\end{example}
\end{toappendix}
The completeness of Turing/Floyd variant function method is due to the additional property that the inverse of the transition relation of a terminating program is well-founded \proofinapx\ (see example \ref{ex:contructing-f} in the appendix).
\begin{toappendix}
\begin{example}[Transfinite variant function]\label{ex:contructing-f} The construction of the variant function $\nu$ for the program \texttt{\small x=[0,$\infty$]; while (x>0) \{x=x-1; y=[0,$\infty$]; while (y>0) y=y-1;\}} with unbounded nondeterminism is illustrated on Fig\@. \ref{fig:contructing-f}.
\begin{figure}[ht]
\includegraphics[width=0.6\textwidth]{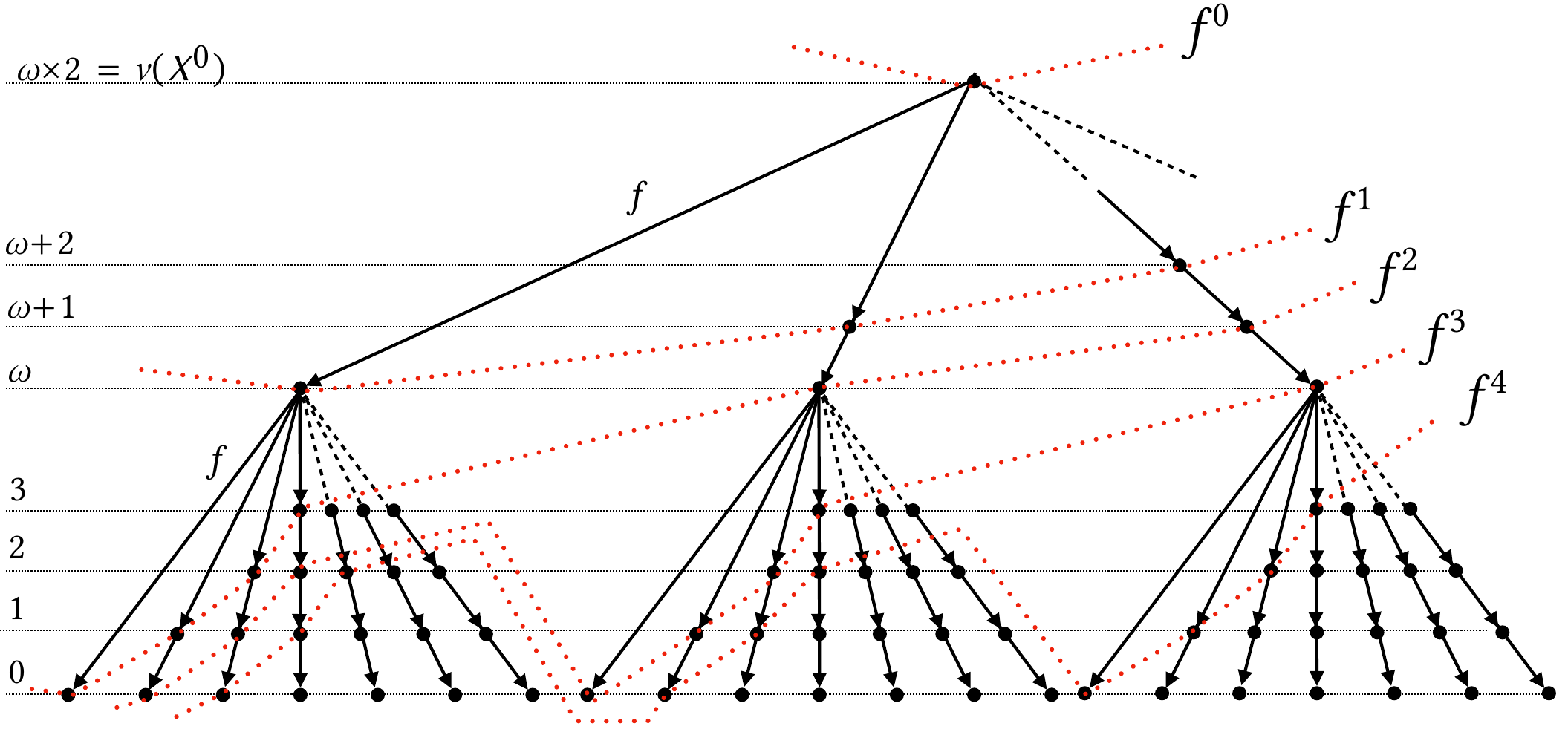}
\caption{Construction of the variant function $\nu$\label{fig:contructing-f}}
\end{figure}
\end{example}
\end{toappendix}
\begin{theorem}[Turing/Floyd]\label{th:Turing:Floyd}Let $r\in\wp(\mathcal{X}\times\mathcal{X})$ be a relation on a set $\mathcal{X}$ and $P\in\wp(\mathcal{X})$. Then
\abovedisplayskip0.5\abovedisplayskip\belowdisplayskip0.5\belowdisplayskip\begin{eqntabular}[fl]{@{\qquad}rl}
&\{x\in\mathcal{X}\mid \exists\sigma\in\mathbb{N}\rightarrow\mathcal{X}\mathrel{.}\sigma_0=x\in P\wedge\forall i\in\mathbb{N}\mathrel{.}\pair{\sigma_i}{\sigma_{i+1}}\in r\}=\emptyset\nonumber\\
\Leftrightarrow&\{x\in\mathcal{X}\mid x\in P\wedge\exists\pair{W}{\preceq}\in\mathfrak{Wf}\mathrel{.}\exists
I\in\wp(\mathcal{X})\mathrel{.}
 P\cup \textsf{\upshape post}(r)I\subseteq I\wedge\exists\nu\in I\rightarrow W\mathrel{.}{}\nonumber\\
\renumber{$\forall y\in I\mathrel{.}\forall y'\in \mathcal{X}.\pair{y}{y'}\in r\Rightarrow \nu(y)\succ\nu(y')\}$}
\end{eqntabular}
\end{theorem}
Notice that the soundness proof \ifshort given in the appendix \fi uses (the dual of) Th\@. \ref{th:Fixpoint-Underapproximation-Variant} which shows that it is a generalization for Turing/Floyd variant method.
\begin{toappendix}
\label{sec:apx:proof:th:Turing:Floyd}
\begin{proof}[Proof of Th\@. \ref{th:Turing:Floyd}]
For soundness ($\supseteq$), we look for $P\cap \neg\textsf{\upshape pre}(r)\{\bot\}$ = $P\cap \widetilde{\textsf{\upshape pre}}(r)\Sigma$ and apply the order dual of Th\@. \ref{th:Fixpoint-Underapproximation-Variant} observing that $I\subseteq P\cap \widetilde{\textsf{\upshape pre}}(r)I$ is equivalent to $P\cup \textsf{\upshape post}(r)I\subseteq I$.

For completeness ($\supseteq$), we let $r^{\ast}=\Lfp{\subseteq}\LAMBDA{X}\textsf{\upshape id}\cup X\comp r$ and choose $W=I=\textsf{\upshape post}(r^{\ast})P$, $\nu(x)=x$, and $x\succ y\triangleq(x\in W\wedge \pair{x}{y}\in r)$. We must prove the following three conditions (we write $a\stackrel{r}{\longrightarrow}b$ for $\pair{a}{b}\in r$).
\begin{enumerate}[leftmargin=*]

\item $\pair{W}{\preceq}\in\mathfrak{Wf}$.

By reductio ad absurdum, assume that there exists an infinite sequence strictly decreasing sequence $x_0\succ x_1\succ\ldots$ of elements of $W=I$. 
By definition of $W$, $x_0\in\textsf{\upshape post}(r^{\ast})P$ so that there exists a finite sequence $y_1\in P$, ${y_1}\stackrel{r}{\longrightarrow}{y_2}\stackrel{r}{\longrightarrow}\ldots \stackrel{r}{\longrightarrow}{y_{n-1}}\stackrel{r}{\longrightarrow}{y_n}$ with $y_i\in W = I$, $i\in\interval{1}{n}$, $n\geqslant 0$ ($0$ if $x_0\in P$) and $y_n=x_0$. By definition of $x\succ y$ implying $x\stackrel{r}{\longrightarrow}y$, there exists an infinite sequence 
$y_1\in P$, ${y_1}\stackrel{r}{\longrightarrow}{y_2}\stackrel{r}{\longrightarrow}\ldots \stackrel{r}{\longrightarrow}{y_{n-1}}\stackrel{r}{\longrightarrow}x_0\stackrel{r}{\longrightarrow} x_1\stackrel{r}{\longrightarrow}\ldots$ in contradiction with $\neg(\exists\sigma\in\mathbb{N}\rightarrow\mathcal{X}\mathrel{.}\sigma_0=x\in P\wedge\forall i\in\mathbb{N}\mathrel{.}\pair{\sigma_i}{\sigma_{i+1}}\in r)$.

\item By the fixpoint abstraction Th\@. \ref{th:fixpoint-abstraction}, the abstraction of $r^{\ast}=\Lfp{\subseteq}\LAMBDA{X}\textsf{\upshape id}\cup X\comp r$ by $\LAMBDA{X}\textsf{\upshape post}(X)P$ is  $W=I=\textsf{\upshape post}(r^{\ast})P = \Lfp{\subseteq}\LAMBDA{X}\textsf{\upshape P}\cup \textsf{\upshape post}(r)X$ so that $I=\textsf{\upshape P}\cup \textsf{\upshape post}(r)I$ hence $P\cup \textsf{\upshape post}(r)I\subseteq I$ by reflexivity.

\item if $y\in W=I=\textsf{\upshape post}(r^{\ast})P$ and $y'\in \mathcal{X}$ is such that $\pair{y}{y'}\in r$ then $y'\in \textsf{\upshape post}(r^{\ast})P=I=W$. Then $(y\in W\wedge \pair{y}{y'}\in r)$
implies $y\succ y'$ by definition of $\succ$ and therefore $\nu(y)\succ\nu(y')$ by definition of $\nu$.\qed
\end{enumerate}\let\qed\relax
\end{proof}
\end{toappendix}
Notice that if the intersection of $\Gfp{\sqsubseteq}f$ with $Q$ is empty ($\bot$ in the lattice) then so is 
the intersection of $\Lfp{\sqsubseteq}f$ with $Q$ but not conversely, so in addition to theorem 
\ref{th:Intersection-Lfp}, we also need the following \proofinapx
\begin{theorem}[Void intersection with greatest fixpoint]\label{th:Fixpoint-Gfp-intersection-Variant}
We assume that 
(1) $\sextuple{L}{\sqsubseteq}{\bot}{\top}{\sqcap}{\sqcup}$ is a coatomic complete lattice;
(2) $f\in L\rightarrow L$ preserves non-empty meets; 
(3) there exists a coinvariant $I\in L$ of $f$ (i.e.\ such that $I\sqsubseteq f(I)$); 
(4) that there exists a well-founded set $\pair{W}{\preceq}$ and a variant function $\nu\in I\rightarrow W$ such that $\forall x\in\textsf{\upshape coatoms}( I)\mathrel{.}(x\neq f(x))\Rightarrow (\nu(x)\succ\nu(f(x)))$;
(5) $Q\in L$; and 
(6) $\forall x\in\textsf{\upshape coatoms}( I)\mathrel{.}(\nu(x)\not\succ\nu(f(x)))\Rightarrow (x\sqcap Q=\bot)$.
Then, hypotheses (1) to (6) imply $\Gfp{\sqsubseteq}f \sqcap Q=\bot$.
\end{theorem}
Notice that Th\@. \ref{th:Fixpoint-Gfp-intersection-Variant}, as well as its proof  \ifshort in Sect\@. \ref{sec:apx:proof:th:Fixpoint-Gfp-intersection-Variant} of the appendix, \fi are not the order dual of 
Th\@. \ref{th:Intersection-Lfp} since (6) have the same conclusion $x\sqcap Q=\bot$ and the dual of the conclusion $\Lfp{\sqsubseteq}f \sqcap Q=\bot$ would be $\Gfp{\sqsubseteq}f \sqcup Q=\top$.
\begin{toappendix}
\label{sec:apx:proof:th:Fixpoint-Gfp-intersection-Variant}
\begin{proof}[Proof of Th\@. \ref{th:Fixpoint-Gfp-intersection-Variant}]By (2), $f$ preserves non-empty meets implies that the
 iterates $f^0(\top)\triangleq\top$ and $f^{n+1}(\bot)\triangleq f(f^{n}(\bot))$ of $f$ from the supremum $\top$ on the $L$ form an decreasing chain which limit, by (1) exists in the complete lattice L, and is $\Gfp{\sqsubseteq}f$ = $f^{\omega}(\bot)$ $\triangleq$ $\bigsqcap_{n\in\mathbb{N}}f^{n}(\top)$  \cite{CousotCousot-PJM-82-1-1979}. Moreover, by (3), $I \sqsubseteq f(I)$ implies, by recurrence, that $\forall n\in\mathbb{N}\cup\{\omega\}\mathrel{.}I\sqsubseteq f^{n}(\top)$ since by (2) $f$ preserves non-empty meets so is increasing. 

\smallskip

If $f(\top)=\top$ then obviously $\Gfp{\sqsubseteq}f=\top$. By (6) $\forall \top \sqsupseteq I$ and $(\nu(\top)\succ\nu(f(\top)))$ implies $ (\top\sqcap Q=\bot)$. So in the following, we assume that $\top\neq f(\top)$ that is, $\top\sqsupset f(\top)$ by definition of the supremum $\top$.

In a coatomic complete lattice, all elements have a non-empty set of coatoms but $\top$. Since $\top\neq f(\top)$, we have  $\textsf{\upshape coatoms}(f(\top))\neq\emptyset$. 

We have
$\top=f^{0}(\top)\sqsupset f^{1}(\top)=\bigsqcap\{f^0(x)\mid x\in \textsf{\upshape coatoms}(f(\top))\}$ by definition of the identity $f^0$. 
Assume that, for $n\geqslant 1$,   $f^n(\top)=\bigsqcap\{f^{n-1}(x)\mid x\in \textsf{\upshape coatoms}(f(\top))\}$ by recurrence hypothesis.
By (2), $f$ preserves non-empty meets so  $f^{n+1}(\top)=f(f^{n}(\top))$ =
$f(\bigsqcap\{f^{n-1}(x)\mid x\in \textsf{\upshape coatoms}(f(\top))\})$ =
$\bigsqcap\{f(f^{n-1}(x))\mid x\in \textsf{\upshape coatoms}(f(\top))\}$ =
$\bigsqcap\{f^{n}(x)\mid x\in \textsf{\upshape coatoms}(f(\top))\}$. By recurrence, $\forall n\geqslant 1\mathrel{.}f^n(\top)=\bigsqcup\{f^{n-1}(x)\mid x\in \textsf{\upshape coatoms}(f(\top))\}$. 

By (1) and definition of a coatom $x\in \textsf{\upshape coatoms}(f(\top))$, we have $\top\sqsupset x\sqsupset f(\top)$ that is
 $f^0(\top)\sqsupset f^0(x)\sqsupset f(\top)$. Since, by (2), $f$ preserves non-empty meets it is increasing and therefore $\forall n\in\mathbb{N}\mathrel{.}
f^n(\top)\sqsupseteq f^n(x)\sqsupseteq f^{n+1}(\top)$ by recurrence. By definition of a greatest lower bound, it follows that $\bigsqcap_{n\in\mathbb{N}}f^n(ç)\sqsupseteq\bigsqcap_{n\in\mathbb{N}} f^n(x)\sqsubseteq\bigsqcap_{n\in\mathbb{N}} f^{n+1}(\top)=\bigsqcap_{n\in\mathbb{N}}f^n(\top)$ since $ç\sqsupseteq f(ç)=f^{1}(ç)$. This implies that
$\forall x\in\textsf{\upshape coatoms}(f(ç))\mathrel{.}\bigsqcap_{n\in\mathbb{N}} f^n(x)=\Gfp{\sqsubseteq}f$.

Because $\forall n\geqslant 1\mathrel{.}f^n(x)=f(f^{n-1}(x))\sqsupseteq I$, (4) implies that the chain $\pair{\nu(f^n(x))}{n\geqslant 1}$ in the well-founded set $\pair{W}{\preceq}$ cannot infinitely decrease so is ultimately stationary at some $n=\ell_x$. By contraposition of (4), $\forall y\sqsupseteq I\mathrel{.} (\nu(y)\not\succ\nu(f(y)))\implies(y= f(y))\sqsupseteq I$ so that $f^{\ell_x}(x)=f^{m}(x)\sqsupseteq I$ for all $m\geqslant \ell_x$. By (5) and (6), it follows that $\forall m\geqslant \ell_x\mathrel{.}f^{m}(x)\sqcap Q=\bot$

It follows, by definition of the greatest lower bound, that $\Gfp{\sqsubseteq}f$ =
$\bigsqcap_{n\in\mathbb{N}}f^{n}$ =
$\bigsqcap_{n\geqslant 1}\bigsqcap\{f^{n-1}(x)\mid x\in \textsf{\upshape atoms}(f(\top))\}$ =
$\bigsqcap\{\bigsqcap_{n\geqslant 1}f^{n-1}(x)\mid x\in \textsf{\upshape atoms}(f(\top))\}$ =
$\bigsqcap\{f^{\ell_x}(x)\mid x\in \textsf{\upshape coatoms}(f)(\top)\}$.
We have shown that $\forall x\in \textsf{\upshape coatoms}(f(\top))\mathrel{.}f^{\ell_x}(x)\sqcap Q=\bot$ so
$\Lfp{\sqsubseteq}f\sqcap Q$ =
$\bigsqcap\{f^{\ell_x}(x)\mid x\in \textsf{\upshape coatoms}(f(\top))\}\sqcap Q=\bot$ in the coatomic complete lattice $L$.
\end{proof}
\end{toappendix}

\smallskip

\subsection{Fixpoint Non Emptiness}\label{sec:FixpointNonEmptiness}
Another result to handle greatest fixpoints, e.g.\ to prove definite nontermination, is the following theorem \proofinapx.
\begin{theorem}[Greatest fixpoint non emptiness]\label{th:Greatest-fixpoint-non-emptiness}
Let $f\in L\mathrel{\smash{\stackrel{i}{\longrightarrow}}}L$ be an increasing function of a complete lattice $\sextuple{L}{\sqsubseteq}{\bot}{\top}{\sqcup}{\sqcap}$ and $P\in L\setminus\{\bot\}$. Then $\Gfp{\sqsubseteq}f\sqcap P\neq \bot$ if and only if $\forall X\in L\mathrel{.} (\Gfp{\sqsubseteq}f\sqsubseteq X\wedge f(X)\sqsubseteq X \wedge X\sqcap P \neq \bot)\Rightarrow(f(X)\sqcap P \neq \bot)$.
\end{theorem}
\begin{toappendix}
\label{sec:apx:proof:th:Greatest-fixpoint-non-emptiness}
\begin{proof}[Proof of Th\@. \ref{th:Greatest-fixpoint-non-emptiness}]The transfinite iterates $X^0=\top$, $X^{\delta+1}=f(X^{\delta})$, and $X^{\lambda}=\bigsqcap_{\beta<\lambda}X^{\beta}$ of $f$ from $\top$ are decreasing, ultimately stationary, with limit $\Gfp{\sqsubseteq}f$ \cite{CousotCousot-PJM-82-1-1979}, so $\forall \delta\in\mathbb{O}\mathrel{.}\Gfp{\sqsubseteq}f\sqsubseteq X^{\delta}$.

For soundness, we have  $X^0\sqcap P \neq \bot$ since $P\neq\bot$. If $X^{\delta}\sqcap P \neq \bot$ then $f(X^{\delta})=X^{\delta+1}\sqsubseteq X^{\delta}$ implies 
$f(X^{\delta})\sqcap P \neq \bot$ by hypothesis, that is, $X^{\delta+1}\sqcap P \neq \bot$. If, by induction hypothesis, $X^{\beta}\sqcap P \neq \bot$ for all $\beta<\lambda$ then
$X^{\lambda}\sqcap P=(\bigsqcap_{\beta<\lambda}X^{\beta})\sqcap P=\bigsqcap_{\beta<\lambda}(X^{\beta}\sqcap P)\neq \bot$ since $\pair{X^{\beta}}{\beta<\lambda}$ is decreasing and so is $\pair{X^{\beta}\sqcap  P}{\beta<\lambda}$. By transfinite induction, we conclude that 
$\forall\delta\in\mathbb{O}\mathrel{.}X^{\delta}\sqcap P\neq \bot$ and so $\Gfp{\sqsubseteq}f\sqcap P\neq \bot$.

For completeness, $\Gfp{\sqsubseteq}f\sqcap P\neq \bot$ then $P\neq \bot$. If $\Gfp{\sqsubseteq}f\sqsubseteq X\wedge f(X)\sqsubseteq X \wedge X\sqcap P \neq \bot$ then $\Gfp{\sqsubseteq}f=
f(\Gfp{\sqsubseteq})\sqsubseteq f(X)$ since $f$ is increasing so $\Gfp{\sqsubseteq}f\sqcap P\neq\bot$ implies $f(X)\sqcap P\neq\bot$.
\end{proof}
\end{toappendix}
\noindent A fixpoint induction principle \ref{th:abstract-least-fixpoint-non-emptiness} for $\alpha(\Lfp{\sqsubseteq}f)\sqcap P\neq \bot$ in (\ref{eq:post-t-pre-t}.d) is \ifshort given in the appendix. \proofinapx\fi
\begin{toappendix}
\begin{theorem}[Non empty intersection with abstraction of least fixpoint]\label{th:abstract-least-fixpoint-non-emptiness}
Assume that 
(1) $\sextuple{L}{\sqsubseteq}{\bot}{\top}{\sqcap}{\sqcup}$ is an atomic complete lattice;
(2) $f\in L\rightarrow L$ preserves nonempty joins $\sqcup$; 
(3) $\pair{L}{\sqsubseteq}\galoiS{\alpha}{\gamma}\triple{\bar{L}}{\preceq}{\curlywedge}$;
(4) $\bar{Q}\in \bar{L}\setminus\{0\}$ where $0\triangleq\alpha(\bot)$; 
(5) There exists an inductive invariant $I\in L$ of $f$ (i.e\@. $f(I)\sqsubseteq I$);
(6) $\pair{W}{\leqslant}$ is a well-founded set and $\nu\in \textsf{\upshape atoms}(I)\rightarrow W$ is a (variant) function;
(7) There exists a sequence $\pair{a_i\in\textsf{\upshape atoms}(I)}{i\in\interval{1}{\infty}}$ that 
(7.a) $a_1\in f(\bot)$,
(7.b) $\forall  i\in\interval{1}{\infty}\mathrel{.}a_{i+1}\in\textsf{\upshape atoms}(f(a_i))$,
(7.c) $\forall  i\in\interval{1}{\infty}\mathrel{.} (a_i\neq a_{i+1})\Rightarrow(\nu(a_i)>\nu(a_{i+1})$,
(7.d) $\forall  i\in\interval{1}{\infty}\mathrel{.} (\nu(a_i)\not>\nu(a_{i+1})\Rightarrow\alpha(a_i)\curlywedge \bar{Q}\neq 0$;
Then, hypotheses (1) to (7) imply $\alpha(\Lfp{\sqsubseteq}f) \curlywedge \bar{Q}\neq 0$. Conversely (1) to (4) and
$\Lfp{\sqsubseteq}f \sqcap \gamma(\bar{Q})\neq \bot$ imply (5) to (7).
\end{theorem}
Notice that if $L=\wp(\Sigma)$ then $\textsf{\upshape atoms}(L)=\{\{x\}\mid x\in L\}$ so that $I\in\wp(\Sigma)$ and $\nu$ can be chosen in $I\rightarrow W$ instead of $\{\{x\}\mid x\in I\}\rightarrow W$.
\begin{proof}[\hyperlink{th:theory:not:Hoare}{Proof of Th\@. 6}]
By (1) and (2), $\Lfp{\sqsubseteq}f=\bigsqcup_{n\in\mathbb{N}}f^n(\bot)$ where the iterates of $f$ from $x\in L$ are
$f^0(x)=x$ and $f^{n+1}(x)=f(f^n(x))$ \cite{CousotCousot-PJM-82-1-1979}. By (5), $f(I)\sqsubseteq I$ so that
$\Lfp{\sqsubseteq}f\sqsubseteq I$ by Tarski's fixpoint theorem \cite{Tarski-fixpoint}. Consider 
$\pair{a_i\in\textsf{\upshape atoms}(I)}{i\in\interval{1}{\infty}}$. By (7.a), $a_1\in f(\bot)$ so $a_1\in \textsf{\upshape atoms}(f^1(\bot))$. Assume $a_n\in \textsf{\upshape atoms}(f^n(\bot))$ so that $a_n\sqsubseteq f^n(\bot)$.
By (2), $f$ is increasing so $f(a_n)\sqsubseteq f(f^n(\bot))=f^{n+1}(\bot)$. By (7.b), $a_{n+1}\in \textsf{\upshape atoms}(f(a_n))\subseteq  \textsf{\upshape atoms}(f^{n+1}(\bot))$. By recurrence $\forall n\in\mathbb{N}\mathrel{.}a_n\in \textsf{\upshape atoms}(f^n(\bot))$. This implies $a_n\sqsubseteq f^n(\bot)\sqsubseteq \bigsqcup_{n\in\mathbb{N}}f^n(\bot) = \Lfp{\sqsubseteq}f\subseteq I$ so that $a_n\in\textsf{\upshape atoms}(I)$ proving that $\nu(a_n)$ is well-defined for all $n\in\mathbb{N}$. By (6), the sequence
$\pair{\nu(a_n)}{n\in\mathbb{N}}$ cannot be strictly $>$-decreasing. So there is some $\ell\in\mathbb{N}$ such that $\nu(a_{\ell})=
 \nu(a_{\ell+1})$. By (7.d), this implies that $\alpha(a_{\ell})\curlywedge \bar{Q}\neq 0$. By (3), 
$\sextuple{\alpha(L)}{\preceq}{0}{1}{\curlywedge}{\curlyvee}$ is a complete lattice. We have 
$\alpha(\Lfp{\sqsubseteq}f)$
=
$\alpha(\bigsqcup_{n\in\mathbb{N}}f^n(\bot))$
=
$\bigcurlyvee_{n\in\mathbb{N}}\alpha(f^n(\bot))$
$\succeq$
$\alpha(f^{\ell}(\bot))$
$\succeq$
$\alpha(a_{\ell})$
so that
$\alpha(a_{\ell})\curlywedge \bar{Q}\neq 0$
implies
$\alpha(\Lfp{\sqsubseteq}f) \curlywedge \bar{Q}\neq 0$.

Conversely assume (1) to (4) and
$\Lfp{\sqsubseteq}f \sqcap \gamma(\bar{Q})\neq \bot$. Let $x_0$ be an atom common to $\Lfp{\sqsubseteq}f$ and $ \gamma(\bar{Q})$ so that $x_0\in \textsf{\upshape atoms}(f^n(\bot))$ for some $n>0$. Assume we have constructed $x_0,\ldots,x_{n-i}$, $0<i\leqslant n$ such that $\forall k\in\interval{0}{i}\mathrel{.}x_k\in
\textsf{\upshape atoms}(f^{n-k})$ which elements are two by two different. There are two cases. 
\begin{enumerate}[leftmargin=*]
\item If $x_{n-i}\in f(\bot)$ then define the finite sequence $a_1=x_{n-i}$, 
 $a_2=x_{n-i+1}$, \ldots, $a_{n-i+1}=x_0$.
 Define $I=\Lfp{\sqsubseteq}f$, $\pair{W}{\leqslant}=\pair{\interval{1}{n-i+1}}{\leqslant}$, $\nu(x)=\si x=a_i\alors i\sinon 1\fsi$, which is well-founded since the elements of $a_1$, \ldots, $a_{n-i+1}$ are two by two different. Then (5) to (7) are satisfied, Q.E.D.
 \item Otherwise $x_{n-i}\not\in f(\bot)$ and $x_{n-i}\in\textsf{\upshape atoms}(f^{n-i})=\textsf{\upshape atoms}(f(f^{n-i-1}))$.
 Pick $x_{n-i-1}$ as an atom of $f^{n-i-1}$ different from  $x_0$, \ldots, $x_{n-i}$.  Notice that if there no such $x_{n-i-1}$, we are in the previous case (1). This extends the sequence by one element, and we must terminate ultimately at $f^{1}(\bot)$ for
 which case (1) concludes. \qed
\end{enumerate}\let\qed\relax
\end{proof}
\end{toappendix}
\label{last-induction-principle}

\section{Deductive Systems of Program Logics}\label{sec:ProgramLogics}
Logics define the valid properties of a program as all provable facts by the formal proof system of the logic. These formal systems,
introduced by Hilbert \cite[\textsection\ 5]{Hilbert-Ackermann38}, are ``a system of axioms from which the remaining true sentences may be obtained by means of certain rules''. Such a formal system is  a finitely presented set of axioms $c$  and rules $\frac{P}{c}$ where the axioms and conclusions $c$ of the rules are terms with variables and the premisses $P$ are formulas of a logic. 

The semantics/interpretation of the logic maps logical terms to elements of a mathematical structure with  universe  $\mathcal{U}$. Logical formulas are interpreted as the subsets of $\,\mathcal{U}$ of elements satisfying the formulas. Therefore logical implication is subset inclusion $\subseteq$ in the complete Boolean lattice $\septuple{\wp(\mathcal{U})}{\subseteq}{\emptyset}{\mathcal{U}}{\cup}{\cap}{\neg}$ where $\emptyset$ is false, ${\mathcal{U}}$ true, $\cup$ disjunction, $\cap$ conjunction, and $\neg$ negation. The semantics/interpretation of the formal rules is a
deductive system $R=\bigl\{\frac{P_i}{c_i}\bigm| i\in\Delta\bigr\}$  where $P_i\in\wp(\mathcal{U})$ is the finite premise and $c_i\in\mathcal{U}$ the conclusion of the rule. The axioms have $P_i=\emptyset$ (false) as premises. We have $R\in\wp(\wp(\mathcal{U})\times\mathcal{U})$ where pairs $\pair{P}{c}$ are conventionally written $\frac{P}{c}$.  
\begin{example}\label{ex:odd-numbers}The formal system $1\in\mathcal{O}$ and inductive rule 
$\frac{n\:\in\:\mathcal{O}}{n+2\:\in\:\mathcal{O}}$ (defining the odd naturals $\mathcal{O}$ on universe $\mathbb{N}$) has the interpretation $\{\frac{\emptyset}{1}\}\cup\{\frac{\{n\}}{n+2}\mid n\in\mathbb{N}\}$. For example if $2\in\mathbb{N}$ is odd then $4$ is odd.
To prove that $2$ is odd, the only way is to prove that $0$ is odd which is not an axiom nor the conclusion of a rule, proving $2$ not to be odd. 
 \end{example}

\section{The Semantics of Deductive Systems}\label{sec:SemanticsLogics}\label{sect:SemanticsDeductiveSystems}

Aczel \cite{Aczel:1977:inductive-definitions} has shown that there are two equivalent ways of defining the subset $\alpha^{\mathcal{I}}(R)$ of the universe $\mathcal{U}$ defined by a deductive system $R=\bigl\{\frac{P_i}{c_i}\bigm| i\in\Delta\bigr\}$.

\subsection{Proof-Theoretic Semantics of Deductive Systems}
In the proof-theoretic approach, $\alpha^{\mathcal{I}}(R)$ is the set of provable elements where
a formal proof is a finite sequence $t_1 \ldots\ t_n$ of terms (i.e.\ elements of the universe $\mathcal{U}$) such that any term is the conclusion of a rule which premise is implied by (i.e.\ included in $\subseteq$) the set of previous terms in the sequence (which have been already proved, starting with axioms). Therefore $\alpha^{\mathcal{I}}(R)=\{t_n\in\mathcal{U}\mid \exists t_1, \ldots,t_{n-1}\in \mathcal{U}\mathrel{.}\forall k\in\interval{1}{n}\mathrel{.}\exists \frac{P}{c}\in R\mathrel{.}P\subseteq\{t_1, \ldots,t_{k-1}\} \wedge t_k=c\}$ (this requires $P$ to be finite). 

It follows that there is a Galois connection $\pair{\wp(\wp(\mathcal{U})\times\mathcal{U})}{\subseteq}\galoiS{\alpha^{\mathcal{I}}}{\gamma^{\mathcal{I}}}\pair{\wp(\mathcal{U})}{\subseteq}$ where $\alpha^{\mathcal{I}}$ is $\subseteq$-increasing (the more rules the larger is the defined set) and $\gamma^{\mathcal{I}}(X)=\{\frac{P}{c}\mid P\in\wp(\mathcal{U}) \land c\in X\}$ including axioms $\frac{\emptyset}{c}$ collecting all elements $c$ of $X$. (As discussed thereafter, there are other, more natural and effective, possible deductive systems. Proof systems are not unique.)

\subsection{Model-Theoretic Semantics of Deductive Systems}\label{sec:Deductive-system-Model-theoretic-definition}

In the  model-theoretic approach,  the same $\alpha^{\mathcal{I}}(R)$ is defined  as $\alpha^{\mathcal{I}}(R)$ = $\Lfp{\subseteq}\alpha^{F}(R)$ where the consequence operator is $\alpha^{F}(R)X\triangleq\{c\mid\exists\frac{P}{c}\in R\mathrel{.}P\subseteq X\}$. $\alpha^{F}(R)X$ is the set of consequences derivable from
the hypotheses $X\in\wp(\mathcal{U})$ by one application of an axiom (with $P=\emptyset$) or a rule of the deductive system. $\alpha^{F}(R)\in\wp(\mathcal{U})\stackrel{\sqcup}{\longrightarrow}\wp(\mathcal{U})$ preserve nonempty joins and so is increasing.
The least fixpoint $\Lfp{\subseteq}\alpha^{F}(R)$ of the consequence operator is well-defined \cite{Tarski-fixpoint} and is the set of all provable terms, that is, $\alpha^{\mathcal{I}}(R)$. For example \ref{ex:odd-numbers},  $\mathcal{O}=\Lfp{\subseteq}\LAMBDA{X}\{1\}\cup\{n+2\mid\{n\}\subseteq X\}$.

\subsection{Equivalence of the Two Definitions of the Semantics of Deductive Systems}\label{sec:equivalence-fxp-deductive-system}

The definitions of a subset of the universe by a deductive system or by a fixpoint are equivalent \cite{Aczel:1977:inductive-definitions}. We have recalled that a deductive system can be expressed in fixpoint form. Conversely, given any increasing operator $F$ on $\pair{\mathcal{U}}{\subseteq}$, the terms provable by the deductive system $\gamma^{F}(F)=\bigl\{\frac{P}{c}\bigm|P\in\wp(\mathcal{U})\land c\in F(P)\bigr\}$ (or $\bigl\{\frac{P}{c}\bigm|P\in\wp(\mathcal{U})\land c\in F(P)\wedge \forall P'\in\wp(\mathcal{U})\mathrel{.}c\in F(P')\Rightarrow P\subseteq P'\bigr\}$) are exactly its least fixpoint $\Lfp{\subseteq}F$. This yields a Galois connection between deductive systems and increasing consequence operators $\pair{\wp(\wp(\mathcal{U})\times\mathcal{U})}{\subseteq}\galoiS{\alpha^{F}}{\gamma^{F}}\pair{\wp(\mathcal{U})\mathrel{\smash{\stackrel{i}{\longrightarrow}}}\wp(\mathcal{U})}{\stackrel{.}{\subseteq}}$ where $\stackrel{.}{\subseteq}$ is $\subseteq$, pointwise. Note that there is also a Galois connection between increasing operators and fixpoints $\pair{\wp(\mathcal{U})\mathrel{\smash{\stackrel{i}{\longrightarrow}}}\wp(\mathcal{U})}{\stackrel{.}{\subseteq}}\galoiS{\Lfp{\subseteq}}{\LAMBDA{y}\LAMBDA{x}\si x\,\subseteq\, y \,\alors\, y \,\sinon\, \mathcal{U}\fsi}\pair{\wp(\mathcal{U})}{\subseteq}$  such that $\alpha^{\mathcal{I}}$ is the composition of these two Galois connections. 

The order dual of this result is defined by co-induction leading to greatest fixpoints  $\pair{\wp(\mathcal{U})\mathrel{\smash{\stackrel{i}{\longrightarrow}}}\wp(\mathcal{U})}{\stackrel{.}{\subseteq}}\galoiS{\Gfp{\subseteq}}{\LAMBDA{y}\LAMBDA{x}\si x\,\supseteq\, y \,\alors\, y \,\sinon\, \emptyset\fsi}\pair{\wp(\mathcal{U})}{\subseteq}$, we get the coinductive interpretation of proof systems. It can also be biinductive, a mix of the two, taking the $\Lfp{}$ of $\alpha^{F}(R)$ restricted to a subset of $\mathbb{V}\subseteq\mathcal{U}$ of the universe and $\Gfp{}$ on $\alpha^{F}(R)$ restricted to the complement
$\mathcal{U}\setminus\mathbb{V}$ \cite{DBLP:conf/popl/CousotC92,DBLP:conf/cav/CousotC95,DBLP:journals/iandc/CousotC09}.

More generally, the results hold for  any complete lattice $\pair{L}{\preceq}$ thus generalizing the powerset case $\pair{\wp(\mathcal{U})}{\subseteq}$ and its order dual \cite{DBLP:conf/cav/CousotC95}.

The take away is that, knowing the fixpoint semantics of the logic, there is a method for constructing the deductive system for that logic, which is both sound and complete, by construction. \ifshort An example \ref{ex:fix-to-deductiv-semantics} is given in the appendix showing how to construct the deductive natural relational semantics Sect\@. \ref{sec:natural-relational-semantics-deductive} from its fixpoint definition of Sect\@. \ref{sect:FixpointNaturalRelationalSemantics} \proofinapx.\fi
\begin{toappendix}
\begin{example}[Design of the deductive natural relational semantics]\label{ex:fix-to-deductiv-semantics}
The rule-based deductive natural relational semantics of Sect\@. \ref{sec:natural-relational-semantics-deductive} is derived from its fixpoint definition of Sect\@. \ref{sect:FixpointNaturalRelationalSemantics}, by structural induction. The base cases in (\ref{eq:def:sem:basis}) are understood as constant fixpoints $S=\Lfp{\subseteq}\LAMBDA{X}S$ so that 
Sect\@. \ref{sec:Deductive-system-Model-theoretic-definition} yields axioms. For the assignment, we get $\sigma\vdash\texttt{\small  x = A}\stackrel{e}{\Rightarrow}{\sigma[\texttt{\small x}\leftarrow\mathcal{A}\sqb{\texttt{\small A}}\sigma]}$. Since there are no rules for $\stackrel{b}{\Rightarrow}$ and $\stackrel{\infty}{\Rightarrow}$, $\sqb{\texttt{\small x = A}}^b$ and 
$\sqb{\texttt{\small x = A}}^{\infty}$  are empty.

For the induction cases, consider for example $\sqb{\texttt{\small S$_1$;S$_2$}}^e\triangleq\sqb{\texttt{\small S$_1$}}^e\fatsemi\sqb{\texttt{\small S$_2$}}^e$ in (\ref{eq:def:sem:seq:if}). By structural induction hypothesis and definition of $\fatsemi$, we get $\frac{\sigma\vdash\texttt{\small  S$_1$}\stackrel{e}{\Rightarrow}{\sigma'},\ \sigma'\vdash\texttt{\small  S$_2$}\stackrel{e}{\Rightarrow}{\sigma''}}{\sigma\vdash\texttt{\small S$_1$;S$_2$}\stackrel{e}{\Rightarrow}{\sigma''}}$ where the comma means conjunction. We are in the constant fixpoint case, so the rule is actually an axiom for \texttt{\small S$_1$;S$_2$} and, more rigorously, the premiss should be a side condition. 

For iteration \texttt{\small W} =\texttt{\small while (B) S}, $\sqb{\texttt{\small W}}^e$ in (\ref{eq:natural-finite}) involves $\Lfp{\subseteq}{F^e}$. In (\ref{eq:while:invariant}), we write $\sigma\vdash\texttt{\small W}\stackrel{i}{\Rightarrow}\sigma'$ for $\pair{\sigma}{\sigma'}\in\Lfp{\subseteq}{F^e}$. By ${F^e}(X)\triangleq\textsf{\upshape id} \cup (\sqb{\texttt{\small B}}\fatsemi\sqb{\texttt{\small S}}^e\fatsemi (X\setminus\Sigma\times\{\bot\}))$ in (\ref{eq:natural-transformer-finite}), we decompose the union into two rules (i.e\@. $\frac{x\in X}{x\in X\cup Y}$ and $\frac{y\in Y}{y\in X\cup Y}$). This decomposition yields the axiom $\sigma\vdash\texttt{\small W}\stackrel{i}{\Rightarrow}\sigma$ for $\textsf{\upshape id}$ and $\frac{\mathcal{B}\sqb{\texttt{\small B}}\sigma,\quad \sigma\vdash\texttt{\small S}\stackrel{e}{\Rightarrow}\sigma',\quad \sigma'\vdash\texttt{\small W}\stackrel{i}{\Rightarrow}\sigma''}{\sigma\vdash\texttt{\small W}\stackrel{i}{\Rightarrow}\sigma''}$ for $\sqb{\texttt{\small B}}\fatsemi\sqb{\texttt{\small S}}^e\fatsemi (X\setminus\Sigma\times\{\bot\})$. Then $\sqb{\texttt{\small W}}^e\triangleq
\Lfp{\subseteq}{F^e}\fatsemi(\sqb{\neg\texttt{\small B}}\cup\sqb{\texttt{\small B}}\fatsemi\sqb{\texttt{\small S}}^b)$ in (\ref{eq:natural-finite}) is handled \ulstrut like the sequential composition and union case to get (\ref{eq:W:e}).

For the case $\Gfp{\subseteq}{F^\bot}$ of non termination, we use the dual  interpretation $\alpha^{\mathcal{I}}(R)$ = $\Gfp{\subseteq}\alpha^{F}(R)$ of rules $R$ so that ${F^\bot}(X)\triangleq\sqb{\texttt{\small B}}\fatsemi\sqb{\texttt{\small S}}^e\fatsemi X$ yields the coinductive rule $\frac{\mathcal{B}\sqb{\texttt{\small B}}\sigma,\quad \sigma\vdash\texttt{\small S}\stackrel{e}{\Rightarrow}\sigma', \quad
\sigma'\vdash\texttt{\small W}\stackrel{\infty}{\Rightarrow}}{\sigma\vdash\texttt{\small W}\stackrel{\infty}{\Rightarrow}}\infty$ of (\ref{eq:W:infty}). The other nontermination rule $\frac{\sigma\vdash\texttt{\small W}\stackrel{i}{\Rightarrow}\sigma',\quad\mathcal{B}\sqb{\texttt{\small B}}\sigma',\quad \sigma'\vdash\texttt{\small S}\stackrel{\infty}{\Rightarrow}}{\sigma\vdash\texttt{\small W}\stackrel{\infty}{\Rightarrow}}$ follows, by structural induction, from the term $\Lfp{\subseteq}{F^e}\fatsemi\sqb{\texttt{\small B}}\fatsemi\sqb{\texttt{\small S}}^\bot$ of the union.
\end{example}
\end{toappendix}

\section{Calculational Design of Proof Systems}
After defining the theory of a logic by abstraction $\alpha_a\comp\alpha_t(\sqb{\texttt{\small S}}_\bot)$ of the relational semantics $\sqb{\texttt{\small S}}_\bot$, we use the fixpoint abstraction theorems of Sect\@. \ref{sec:abstraction} to provide a fixpoint definition of $\alpha_t(\sqb{\texttt{\small S}}_\bot)$, which is most often a transformer or its graph. Then to handle $\alpha_a$, which is an approximation abstraction like  $\textsf{\textup{post}}({\subseteq},{\supseteq})$ or  $\textsf{\textup{post}}({\supseteq},{\subseteq})$, we use the fixpoint induction theorems of Sect\@. \ref{sect:FixpointInduction}
to provide a set-theoretic of the theory of the logic which is then translated in a proof system by Aczel method of Sect\@. \ref{sec:equivalence-fxp-deductive-system}.
\begin{example}Assume that $\alpha_t(\sqb{\texttt{\small S}}_\bot)=\Lfp{\subseteq}F_P$ and that we must derive the abstract theory $T=\alpha_a\comp\alpha_t(\sqb{\texttt{\small S}}_\bot)=\{\pair{P}{Q}\mid\Lfp{\subseteq}F_P\subseteq Q\}$ (e.g\@. to handle the $\subseteq$ part in 
$\textsf{\textup{post}}({\subseteq},{\supseteq})=\textsf{\textup{post}}({=},{\supseteq})\comp\textsf{\textup{post}}({\subseteq},{=})$ or  $\textsf{\textup{post}}({\supseteq},{\subseteq})$, the other part $\supseteq$ being dual). By Th\@. \ref{th:Fixpoint-Overapproximation}, $T=\{\pair{P}{Q}\mid \exists I\mathrel{.} F_P(I)\subseteq I \wedge I\subseteq Q\}$. By Sect\@. \ref{sec:Deductive-system-Model-theoretic-definition}. this  set $T$ is defined by the axiom $\frac{F_P(I)\,\subseteq\, I,\ I\,\subseteq\, Q}{\pair{P}{Q}\,\in\, T}$.
\end{example}
\begin{remark}\label{rem:abstraction-induction}\textit{(On abstraction versus induction)}\quad Hoare logic is the $\textsf{\textup{post}}({\subseteq},{\supseteq})$ and it's reverse is the $\textsf{\textup{post}}({\supseteq},{\subseteq})$ abstraction of the transformer graph $T=\{\pair{P}{\textsf{\textup{post}}\sqb{\texttt{\small S}}P}\mid P\in \wp(\Sigma)\}$. Both proof systems can be designed, by the rules for $T$ plus the consequence rules for  $\textsf{\textup{post}}({\subseteq},{\supseteq})$ and $\textsf{\textup{post}}({\supseteq},{\subseteq})$. By (\ref{eq:natural-finite}), the theory $T$  of the conditional iteration \texttt{\small W} without \texttt{\small break}s would involve $T'=\{\pair{P}{\textsf{\textup{post}}(\Lfp{\subseteq}{F^e})P}\mid P\in \wp(\Sigma)\}$. The rule would be (\ref{eq:hoare-while-strongest}), using ordinals for unbounded nondeterminism. So to prove $\{P\}\,\texttt{\small S}\,\{Q\}$, $P\neq\emptyset$, we would have to find a postcondition $Q'$, prove that it is the strongest, and then use the consequence rule to prove that $Q'\subseteq Q$/  This is sound and complete but much too demanding. The fixpoint induction theorems of Sect\@. \ref{sect:FixpointInduction} solve this problem by weakening the rules for iteration while preserving soundness and completeness. Contrary to fixpoint abstraction, fixpoint induction allows us to take the consequence rule into account in the design of proof rules for fixpoint semantics. So partial correctness need not be a consequence of total correctness and nontermination.
\end{remark}

\section{On the comparison of logics}
To compare logics, we  first relate their theories, that is compare their expressivity, through their respective abstractions of the collecting semantics (as formalized by fixpoint abstraction in Sect\@. \ref{sect:FixpointAbstraction}). Different abstractions yield different logics, compared though their relation by Galois connections. The logics are equivalent when their theories are linked by a Galois isomorphism. An example is given in Sect\@. \ref{PartialpossibleAccessibilityOfSomeFinalStateFromAllInitialStates} where Hoare logic and subgoal induction have the same theory but different proof method (as shown in figure \ref{fig:taxonomy-assertional}). 

The proof system of a logic is entirely determined by its theory (as proved in Sect\@. \ref{sec:ProgramLogics}), but up to an equivalence, since different induction principles may be used, as formalized in Sect\@. \ref{sect:FixpointInduction}, to exploit approximation so as to simplify induction. This is exemplified by Rem\@. \ref{rem:abstraction-induction}. Which induction principle is used is the second characteristic to compare logics.

\section{Applications}
The development of Hoare incorrectness logic in Ex\@. \ref{ex-Hoare-incorrectness} is relegated to the appendix \proofinapx.
\begin{toappendix}
\subsection{Application 0: Calculational Design of Hoare Incorrectness Logic}\label{apx:DesignHoareIncorrectnessLogic}
We design by calculus the Hoare incorrectness logic of Ex\@. \ref{ex-Hoare-incorrectness} which theory is $\mathcal{T}_{\overline{\textrm{HL}}}(\texttt{\small S}) =\textsf{\textup{post}}({\subseteq},{\supseteq})\comp\alpha^{\neg}\comp\mathcal{T}_{\textrm{HL}}(\texttt{\small S})$.  The proof is by structural induction. We consider the case of the conditional iteration 
\texttt{\small W} = \texttt{\small while (B) S} (without \texttt{\small break}, to simplify).  All other cases are similar and simpler.

\subsubsection{Strongest Postcondition Over Approximation}
We start by characterizing the theory of classic Hoare logic.
\begin{lemma}[Strongest postcondition]\label{lem:strongest:postcondition}
$\mathcal{T}(\texttt{\small S})=\alpha_{\textup{G}}\comp\textsf{\upshape post}\sqb{\texttt{\small S}}=\{\pair{P}{\textsf{\upshape post}\sqb{\texttt{\small S}}P}\mid P\in\wp(\Sigma)\}$.
\end{lemma}
\begin{proof}[Proof of lem\@. \ref{lem:strongest:postcondition}]
\begin{calculus}
\formula{\mathcal{T}(\texttt{\small S})}\\
=
\formulaexplanation{\alpha_{\textup{G}}\comp\textsf{\upshape post}\comp\alpha_{\not\bot}\comp\alpha_C(\{\sqb{\texttt{\small S}}_{\bot}\})}{def\@. $\mathcal{T}$}\\
=
\formulaexplanation{\alpha_{\textup{G}}\comp\textsf{\upshape post}\comp\alpha_{\not\bot}(\sqb{\texttt{\small S}}_{\bot})}{def\@. $\alpha_C$}\\
=
\formulaexplanation{\alpha_{\textup{G}}\comp\textsf{\upshape post}(\sqb{\texttt{\small S}}_{\bot}\cap (\Sigma\times\Sigma))}{def\@. $\alpha_{\not\bot}$}\\
=
\formulaexplanation{\alpha_{\textup{G}}\comp\textsf{\upshape post}\sqb{\texttt{\small S}}}{def\@. (\ref{eq:angelic-semantics}) of the angelic semantics $\sqb{\texttt{\small S}}$}\\
=
\lastformulaexplanation{\{\pair{P}{\textsf{\upshape post}\sqb{\texttt{\small S}}P}\mid P\in\wp(\Sigma)\}}{def\@. $\alpha_{\textup{G}}$}{\mbox{\qed}}
\end{calculus}
\let\qed\relax\end{proof}
\begin{lemma}[Strongest postcondition over approximation]\label{lem:strongest:postcondition:over:approximation}
\begin{eqntabular*}{rclclcl}
\mathcal{T}_{\textup{HL}}(\texttt{\small S})&\triangleq&\textsf{\upshape post}(\mathord{{\supseteq}.{\subseteq}})\comp\mathcal{T}(\texttt{\small S})
&=&
\{\pair{P}{Q}\mid\textsf{\upshape post}\sqb{\texttt{\small S}}P \subseteq Q\}
&=&\textsf{\upshape post}(\mathord{{=},{\subseteq}})\comp\mathcal{T}(\texttt{\small S})
\end{eqntabular*}
\end{lemma}
\begin{proof}[Proof of Lem\@. \ref{lem:strongest:postcondition:over:approximation}]
\begin{calculus}
\formula{\textsf{\upshape post}(\mathord{{\supseteq}.{\subseteq}})\comp\mathcal{T}(\texttt{\small S})}\\
=
\formulaexplanation{\textsf{\upshape post}(\mathord{{\supseteq}.{\subseteq}})(\mathcal{T}(\texttt{\small S}))}{def\@. function composition $\comp$}\\
=
\formulaexplanation{\textsf{\upshape post}(\mathord{{\supseteq}.{\subseteq}})(\{\pair{P}{\textsf{\upshape post}\sqb{\texttt{\small S}}P}\mid P\in\wp(\Sigma)\})}{lem\@. \ref{lem:strongest:postcondition}}\\
=
\formulaexplanation{\{\pair{P'}{Q'}\mid\exists\pair{P}{Q}\in\{\pair{P}{\textsf{\upshape post}\sqb{\texttt{\small S}}P}\mid P\in\wp(\Sigma)\}\mathrel{.}\pair{\pair{P}{Q}}{\pair{P'}{Q'}}\in\mathord{{\supseteq}.{\subseteq}}\}}{def\@. (\ref{eq:def:post}) of \textsf{\upshape post}}\\
=
\formulaexplanation{\{\pair{P'}{Q'}\mid\exists P\mathrel{.}\pair{\pair{P}{\textsf{\upshape post}\sqb{\texttt{\small S}}P}}{\pair{P'}{Q'}}\in\mathord{{\supseteq}.{\subseteq}}\}}{def\@. $\in$}\\
=
\formulaexplanation{\{\pair{P'}{Q'}\mid\exists P\mathrel{.}\pair{P}{\textsf{\upshape post}\sqb{\texttt{\small S}}P}\mathrel{{\supseteq}.{\subseteq}}\pair{P'}{Q'}\}}{def\@. $\in$}\\
=
\formulaexplanation{\{\pair{P'}{Q'}\mid\exists P\mathrel{.}P\supseteq P'\wedge \textsf{\upshape post}\sqb{\texttt{\small S}}P \subseteq Q'\}}{def\@. $\mathrel{{\supseteq}.{\subseteq}}$}\\
=
\formulaexplanation{\{\pair{P'}{Q'}\mid\exists P\mathrel{.}P'\subseteq P\wedge \textsf{\upshape post}\sqb{\texttt{\small S}}P \subseteq Q'\}}{def\@. $\supseteq$}\\
=
\formula{\{\pair{P'}{Q'}\mid\textsf{\upshape post}\sqb{\texttt{\small S}}P' \subseteq Q'\}}\\
\explanation{($\subseteq$) by Galois connection (\ref{eq:def:post:GC}), $\textsf{\upshape post}$ is increasing
so that $P'\subseteq P\wedge \textsf{\upshape post}\sqb{\texttt{\small S}}P \subseteq Q'$ implies $\textsf{\upshape post}\sqb{\texttt{\small S}}P'\subseteq \textsf{\upshape post}\sqb{\texttt{\small S}}P\wedge \textsf{\upshape post}\sqb{\texttt{\small S}}P \subseteq Q'$ hence $\textsf{\upshape post}\sqb{\texttt{\small S}}P'\subseteq Q'$ by transitivity;\\
($\supseteq$) take $P=P'$}\\
=
\formulaexplanation{\{\pair{P'}{Q'}\mid\exists P\mathrel{.}P'=P\wedge \textsf{\upshape post}\sqb{\texttt{\small S}}P \subseteq Q'\}}{def\@. $=$}\\
=
\formulaexplanation{\{\pair{P'}{Q'}\mid\exists P\mathrel{.}\pair{P}{\textsf{\upshape post}\sqb{\texttt{\small S}}P}\mathrel{{=},{\subseteq}}\pair{P'}{Q'}\}}{def\@. $\mathrel{{=},{\subseteq}}$}\\
=
\formulaexplanation{\{\pair{P'}{Q'}\mid\exists P\mathrel{.}\pair{\pair{P}{\textsf{\upshape post}\sqb{\texttt{\small S}}P}}{\pair{P'}{Q'}}\in\mathord{{=},{\subseteq}}\}}{def\@. $\in$}\\
=
\formulaexplanation{\{\pair{P'}{Q'}\mid\exists\pair{P}{Q}\in\{\pair{P}{\textsf{\upshape post}\sqb{\texttt{\small S}}P}\mid P\in\wp(\Sigma)\}\mathrel{.}\pair{\pair{P}{Q}}{\pair{P'}{Q'}}\in\mathord{{=},{\subseteq}}\}}{def\@. $\in$}\\
=
\formulaexplanation{\{\pair{P'}{Q'}\mid\exists\pair{P}{Q}\in\mathcal{T}(\texttt{\small S})\mathrel{.}\pair{\pair{P}{Q}}{\pair{P'}{Q'}}\in\mathord{{=},{\subseteq}}\}}{lem\@. \ref{lem:strongest:postcondition}}\\
=
\formulaexplanation{\textsf{\upshape post}(\mathord{{=},{\subseteq}})(\mathcal{T}(\texttt{\small S}))}{def\@. (\ref{eq:def:post}) of \textsf{\upshape post}}\\
=
\lastformulaexplanation{\textsf{\upshape post}(\mathord{{=},{\subseteq}})\comp\mathcal{T}(\texttt{\small S})}{def\@. function composition $\comp$}{\mbox{\qed}}

\end{calculus}
\let\qed\relax\end{proof}

\subsubsection{Calculational Design of Hoare Incorrectness Logic Theory}
\begin{theorem}[Equivalent definitions of $\overline{\textrm{HL}}$ theories]\label{th:theory:not:Hoare:no:consequence}
\begin{eqntabular}[fl]{@{\quad}rcl}
\mathcal{T}_{\overline{\textrm{HL}}}(\texttt{\small W})
&\triangleq&\textsf{\textup{post}}({\subseteq},{\supseteq})\comp\alpha^{\neg}\comp\mathcal{T}_{\textrm{HL}}(\texttt{\small W})\colsep{=}\alpha^{\neg}\comp\mathcal{T}_{\textrm{HL}}(\texttt{\small W})
\renumber{{\color{black}\texttt{\small W} = \texttt{\small while (B) S}}}
\end{eqntabular}
\end{theorem}
\noindent Observe that Th\@. \ref{th:theory:not:Hoare:no:consequence} shows that $\textsf{\textup{post}}({\subseteq},{\supseteq})$ can be dispensed with. This implies that the consequence rule is useless for Hoare incorrectness logic.
\begin{proof}[Proof of Th\@. \ref{th:theory:not:Hoare:no:consequence}]
\begin{calculus}
\formulaexplanation{\mathcal{T}_{\overline{\textrm{HL}}}(\texttt{\small W}) \colsep{=}\textsf{\textup{post}}({\subseteq},{\supseteq})\comp\alpha^{\neg}\comp\mathcal{T}_{\textrm{HL}}(\texttt{\small W})}{def\@. $\mathcal{T}_{\overline{\textrm{HL}}}$}\\
=
\formula{\textsf{\textup{post}}(({\subseteq},{\supseteq})(\neg\{\pair{P}{Q}\mid\textsf{\upshape post}\sqb{\texttt{\small W}}P \subseteq Q\})}\\[-0.5ex]
\rightexplanation{lem\@. \ref{lem:strongest:postcondition:over:approximation} and def\@. (\ref{eq-complement-GC}) of $\alpha^{\neg}$}\\
=
\formulaexplanation{\textsf{\textup{post}}({\subseteq},{\supseteq})(\{\pair{P}{Q}\mid\neg(\textsf{\upshape post}\sqb{\texttt{\small W}}P \subseteq Q)\})}{def\@. $\neg$}\\
=
\formulaexplanation{\textsf{\textup{post}}({\subseteq},{\supseteq})(\{\pair{P}{Q}\mid\textsf{\upshape post}\sqb{\texttt{\small W}}P \cap \neg Q\neq\emptyset\})}{def\@. $\subseteq$ and $\neg$}\\
=
\formulaexplanation{\{\pair{P'}{Q'}\mid\exists\pair{P}{Q}\in\{\pair{P}{Q}\mid\textsf{\upshape post}\sqb{\texttt{\small W}}P \cap \neg Q\neq\emptyset\}\mathrel{.}
\pair{P}{Q}\mathrel{{\subseteq},{\supseteq}}\pair{P'}{Q'}\}}{def\@. \textsf{\textup{post}}}\\
=
\formulaexplanation{\{\pair{P'}{Q'}\mid\exists\pair{P}{Q}\mathrel{.}\textsf{\upshape post}\sqb{\texttt{\small W}}P \cap \neg Q\neq\emptyset\wedge
\pair{P}{Q}\mathrel{{\subseteq},{\supseteq}}\pair{P'}{Q'}\}}{def\@. $\in$}\\
=
\formulaexplanation{\{\pair{P'}{Q'}\mid\exists\pair{P}{Q}\mathrel{.}\textsf{\upshape post}\sqb{\texttt{\small W}}P \cap \neg Q\neq\emptyset\wedge
P\subseteq P'\wedge Q\supseteq Q'\}}{component wise def\@. of $\mathrel{{\subseteq},{\supseteq}}$}\\
=
\formula{\{\pair{P'}{Q'}\mid\exists{Q}\mathrel{.}\textsf{\upshape post}\sqb{\texttt{\small W}}P' \cap \neg Q\neq\emptyset\wedge
Q\supseteq Q'\}}\\[-0.5ex]
\explanation{($\subseteq$)\hskip1em if $P\subseteq P'$ then $\textsf{\upshape post}\sqb{\texttt{\small W}}P \subseteq \textsf{\upshape post}\sqb{\texttt{\small W}}P'$ by
(\ref{eq:def:post:GC}) so that $\textsf{\upshape post}\sqb{\texttt{\small W}}P \cap \neg Q\neq\emptyset$ implies $\textsf{\upshape post}\sqb{\texttt{\small W}}P' \cap \neg Q\neq\emptyset$;
\\
($\supseteq$)\hskip1em conversely, if $\exists{Q}\mathrel{.}\textsf{\upshape post}\sqb{\texttt{\small W}}P'$, then $\exists P\mathrel{.}\textsf{\upshape post}\sqb{\texttt{\small W}}P \cap \neg Q\neq\emptyset\wedge
P\subseteq P'$ by choosing $P=P'$.
}\\
=
\formula{\{\pair{P'}{Q'}\mid\textsf{\upshape post}\sqb{\texttt{\small W}}P' \cap \neg Q'\neq\emptyset\}}\\[-0.5ex]
\explanation{($\subseteq$)\quad if $Q\supseteq Q'$ then $\neg Q'\supseteq \neg Q$ so
$\textsf{\upshape post}\sqb{\texttt{\small W}}P' \cap \neg Q\neq\emptyset$ implies $\textsf{\upshape post}\sqb{\texttt{\small W}}P' \cap \neg Q'\neq\emptyset$;\\
($\supseteq$)\quad conversely $\textsf{\upshape post}\sqb{\texttt{\small W}}P' \cap \neg Q'\neq\emptyset$ implies $\exists{Q}\mathrel{.}\textsf{\upshape post}\sqb{\texttt{\small W}}P' \cap \neg Q\neq\emptyset\wedge
Q\supseteq Q'$ by choosing $Q=Q'$.
}\\
=
\formulaexplanation{{\{\pair{P}{Q}\mid\neg(\textsf{\upshape post}\sqb{\texttt{\small W}}P \subseteq Q)\}}}{def\@. $\subseteq$ and $\neg$}\\
=
\lastformulaexplanation{\alpha^{\neg}\comp\mathcal{T}_{\textrm{HL}}(\texttt{\small W})}{def\@. $\alpha^{\neg}$ and $\mathcal{T}_{\textrm{HL}}$ for Hoare logic}{\mbox{\qed}}
\end{calculus}\let\qed\relax
\end{proof}
\begin{theorem}[Theory of $\overline{\textrm{HL}}$]\label{th:theory:not:Hoare}
\begin{eqntabular}[fl]{@{\quad}rcl}
\mathcal{T}_{\overline{\textrm{HL}}}(\texttt{\small W})
&=&
\{\pair{P}{Q}\mid \begin{array}[t]{@{}l@{}}
\exists n\geqslant 1\mathrel{.}
\exists\pair{\sigma_i\in I}{i\in\interval{1}{n}}\mathrel{.}\sigma_1\in P\wedge{}\\
\forall  i\in\interval[open right]{1}{n}\mathrel{.}\pair{\mathcal{B}\sqb{\texttt{\small B}}\cap\{\sigma_i\}}{\neg\{\sigma_{i+1}\}}\in\mathcal{T}_{\overline{\textrm{HL}}}(\texttt{\small S})
\wedge
\sigma_n\not\in\mathcal{B}\sqb{\texttt{\small B}} \wedge\sigma_n\not\in Q\}
\end{array}\nonumber
\end{eqntabular}
\end{theorem}
\begin{proof}[Proof of Th\@. \ref{th:theory:not:Hoare}]
\begin{calculus}
\formula{\mathcal{T}_{\overline{\textrm{HL}}}(\texttt{\small W})}\\
=
\formulaexplanation{\{\pair{P}{Q}\mid
\textsf{\upshape post}\sqb{\neg\texttt{\small B}}(\Lfp{\subseteq}\bar{\bar{F}}^e_P) \cap \neg Q\neq\emptyset\}}{lem\@. \ref{lem:strongest:postcondition}, where $\bar{\bar{F}}^e_P(X)\triangleq P \cup \textsf{\upshape post}(\sqb{\texttt{\small B}}\fatsemi\sqb{\texttt{\small S}}^e)X$ }\\
=
\formulaexplanation{\{\pair{P}{Q}\mid \Lfp{\subseteq}\bar{\bar{F}}^e_P \cap \textsf{\upshape pre}\sqb{\neg\texttt{\small B}}(\neg Q)\neq\emptyset\}}{(\ref{eq:post-t-pre-t}.d)}\\
=
\formulaexplanation{
\{\pair{P}{Q}\mid\exists I\in\wp(\Sigma)\mathrel{.}\bar{\bar{F}}^e_P(I)\subseteq I\wedge
\exists\pair{W}{\leqslant}\in\mathfrak{Wf}\mathrel{.}\exists\nu\in I\rightarrow W\mathrel{.}
\exists\pair{\sigma_i\in I}{i\in\interval{1}{\infty}}\mathrel{.}
\sigma_1\in \bar{\bar{F}}^e_P(\emptyset)\wedge
\forall  i\in\interval{1}{\infty}\mathrel{.}\sigma_{i+1}\in\bar{\bar{F}}^e_P(\{\sigma_i\})\wedge
\forall  i\in\interval{1}{\infty}\mathrel{.} (\sigma_i\neq \sigma_{i+1})\Rightarrow(\nu(\sigma_i)>\nu(\sigma_{i+1})\wedge
\forall  i\in\interval{1}{\infty}\mathrel{.} (\nu(\sigma_i)\not>\nu(\sigma_{i+1})\Rightarrow\{\sigma_i\}\cap \textsf{\upshape pre}\sqb{\neg\texttt{\small B}}(\neg Q)\neq 0
\}}{Th\@. \ref{th:abstract-least-fixpoint-non-emptiness}}\\[1ex]
=
\formula{
\{\pair{P}{Q}\mid\exists I\in\wp(\Sigma)\mathrel{.}
P\subseteq I\wedge \textsf{\upshape post}(\sqb{\texttt{\small B}}\fatsemi\sqb{\texttt{\small S}}^e)I\subseteq I\wedge
\exists\pair{W}{\leqslant}\in\mathfrak{Wf}\mathrel{.}\exists\nu\in I\rightarrow W\mathrel{.}
\exists\pair{\sigma_i\in I}{i\in\interval{1}{\infty}}\mathrel{.}
\sigma_1\in P\wedge
\forall  i\in\interval{1}{\infty}\mathrel{.}(\sigma_{i+1}\in P\vee \{\sigma_{i+1}\}\subseteq\textsf{\upshape post}(\sqb{\texttt{\small B}}\fatsemi\sqb{\texttt{\small S}}^e)\{\sigma_i\})\wedge
\forall  i\in\interval{1}{\infty}\mathrel{.} (\sigma_i\neq \sigma_{i+1})\Rightarrow(\nu(\sigma_i)>\nu(\sigma_{i+1})\wedge
\forall  i\in\interval{1}{\infty}\mathrel{.} (\nu(\sigma_i)\not>\nu(\sigma_{i+1})\Rightarrow\sigma_i\in\textsf{\upshape pre}\sqb{\neg\texttt{\small B}}(\neg Q)
\}}\\[0.5ex]
\rightexplanation{def\@. $\bar{\bar{F}}^e_P(X)\triangleq P \cup \textsf{\upshape post}(\sqb{\texttt{\small B}}\fatsemi\sqb{\texttt{\small S}}^e)X$,  $\subseteq$, and \textsf{\upshape post}, which is $\emptyset$-strict}\\
=
\formula{
\{\pair{P}{Q}\mid\exists I\in\wp(\Sigma)\mathrel{.}
P\subseteq I\wedge \textsf{\upshape post}(\sqb{\texttt{\small B}}\fatsemi\sqb{\texttt{\small S}}^e)I\subseteq I\wedge
\exists\pair{W}{\leqslant}\in\mathfrak{Wf}\mathrel{.}\exists\nu\in I\rightarrow W\mathrel{.}
\exists\pair{\sigma_i\in I}{i\in\interval{1}{\infty}}\mathrel{.}
\sigma_1\in P\wedge
\forall  i\in\interval{1}{\infty}\mathrel{.}\{\sigma_{i+1}\}\subseteq\textsf{\upshape post}(\sqb{\texttt{\small B}}\fatsemi\sqb{\texttt{\small S}}^e)\{\sigma_i\}\wedge
\forall  i\in\interval{1}{\infty}\mathrel{.} (\sigma_i\neq \sigma_{i+1})\Rightarrow(\nu(\sigma_i)>\nu(\sigma_{i+1})\wedge
\forall  i\in\interval{1}{\infty}\mathrel{.} (\nu(\sigma_i)\not>\nu(\sigma_{i+1})\Rightarrow\sigma_i\in\textsf{\upshape pre}\sqb{\neg\texttt{\small B}}(\neg Q)
\}}\\[0.5ex]
\rightexplanation{since if $\sigma_{i+1}\in P$, we can equivalently consider the sequence $\pair{\sigma_j\in I}{j\in\interval{i+1}{\infty}}$}\\
=
\formula{
\{\pair{P}{Q}\mid\exists I\in\wp(\Sigma)\mathrel{.}
P\subseteq I\wedge \textsf{\upshape post}(\sqb{\texttt{\small B}}\fatsemi\sqb{\texttt{\small S}}^e)I\subseteq I\wedge
\exists n\geqslant 1\mathrel{.}
\exists\pair{\sigma_i\in I}{i\in\interval{1}{n}}\mathrel{.}
\sigma_1\in P\wedge
\forall  i\in\interval[open right]{1}{n}\mathrel{.}\{\sigma_{i+1}\}\subseteq\textsf{\upshape post}(\sqb{\texttt{\small B}}\fatsemi\sqb{\texttt{\small S}}^e)\{\sigma_i\}\wedge
\sigma_n\in\textsf{\upshape pre}\sqb{\neg\texttt{\small B}}(\neg Q)
\}}\\[0.5ex]
\explanation{($\subseteq$)\quad By $\pair{W}{\leqslant}\in\mathfrak{Wf}$, $\nu\in I\rightarrow W$, $\forall  i\in\interval{1}{\infty}\mathrel{.} (\sigma_i\neq \sigma_{i+1})\Rightarrow(\nu(\sigma_i)>\nu(\sigma_{i+1})$, the sequence is ultimately stationary at some rank $n$. For then on, $\sigma_{i+1}=\sigma_{i}$, $i\geqslant n$ and so $\nu(\sigma_i)=\nu(\sigma_{i+1})$.
Therefore 
 $\forall  i\in\interval{1}{\infty}\mathrel{.} (\nu(\sigma_i)\not>\nu(\sigma_{i+1})\Rightarrow\sigma_i\not\in{Q}$ implies that $\sigma_n\in\textsf{\upshape pre}\sqb{\neg\texttt{\small B}}(\neg Q)$;\\
($\supseteq$)\quad Conversely, from $\pair{\sigma_i\in I}{i\in\interval{1}{n}}$ we can define $W=\{\sigma_i\mid i\in\interval{1}{n}\} \cup \{-\infty\}$ with
$-\infty<\sigma_i<\sigma_{i+1}$ and $\nu(x)=\si x\in\{\sigma_i\mid i\in\interval{1}{n}\alors x\sinon-\infty\fsi$ and the sequence $\pair{\sigma_j\in I}{j\in\interval{1}{\infty}}$ repeats $\sigma_n$ ad infimum for $j\geqslant n$.}\\[0.5ex]
=
\formulaexplanation{
\{\pair{P}{Q}\mid\exists I\in\wp(\Sigma)\mathrel{.}
P\subseteq I\wedge \textsf{\upshape post}(\sqb{\texttt{\small B}}\fatsemi\sqb{\texttt{\small S}}^e)I\subseteq I\wedge
\exists n\geqslant 1\mathrel{.}
\exists\pair{\sigma_i\in I}{i\in\interval{1}{n}}\mathrel{.}
\sigma_1\in P\wedge
\forall  i\in\interval[open right]{1}{n}\mathrel{.}\{\sigma_{i+1}\}\subseteq\textsf{\upshape post}(\sqb{\texttt{\small B}}\fatsemi\sqb{\texttt{\small S}}^e)\{\sigma_i\}\wedge
\sigma_n\not\in\mathcal{B}\sqb{\texttt{\small B}} \wedge\sigma_n\not\in Q\}}{def\@. \textsf{\upshape pre}}\\[0.5ex]
=
\formulaexplanation{
\{\pair{P}{Q}\mid 
\exists n\geqslant 1\mathrel{.}
\exists\pair{\sigma_i\in I}{i\in\interval{1}{n}}\mathrel{.}
\sigma_1\in P\wedge
\forall  i\in\interval[open right]{1}{n}\mathrel{.}\{\sigma_{i+1}\}\subseteq\textsf{\upshape post}(\sqb{\texttt{\small B}}\fatsemi\sqb{\texttt{\small S}}^e)\{\sigma_i\}\wedge
\sigma_n\not\in\mathcal{B}\sqb{\texttt{\small B}} \wedge\sigma_n\not\in Q\}}{$I$ is not used and can always be chosen to be $\Sigma$}\\[0.5ex]
=
\formulaexplanation{
\{\pair{P}{Q}\mid 
\exists n\geqslant 1\mathrel{.}
\exists\pair{\sigma_i\in I}{i\in\interval{1}{n}}\mathrel{.}
\sigma_1\in P\wedge
\forall  i\in\interval[open right]{1}{n}\mathrel{.}\textsf{\upshape post}(\sqb{\texttt{\small B}}\fatsemi\sqb{\texttt{\small S}}^e)\{\sigma_i\}\cap \{\sigma_{i+1}\}\neq\emptyset\wedge
\sigma_n\not\in\mathcal{B}\sqb{\texttt{\small B}} \wedge\sigma_n\not\in Q\}}{since $x\in X\Leftrightarrow X\cap\{x\}\neq\emptyset$}\\[0.5ex]
=
\formulaexplanation{
\{\pair{P}{Q}\mid 
\exists n\geqslant 1\mathrel{.}
\exists\pair{\sigma_i\in I}{i\in\interval{1}{n}}\mathrel{.}
\sigma_1\in P\wedge
\forall  i\in\interval[open right]{1}{n}\mathrel{.}\textsf{\upshape post}(\sqb{\texttt{\small B}}\fatsemi\sqb{\texttt{\small S}}^e)\{\sigma_i\}\cap \neg(\neg\{\sigma_{i+1}\})\neq\emptyset\wedge
\sigma_n\not\in\mathcal{B}\sqb{\texttt{\small B}} \wedge\sigma_n\not\in Q\}}{def\@. $\neg X=\Sigma\setminus X$}\\[0.5ex]
=
\formulaexplanation{
\{\pair{P}{Q}\mid 
\exists n\geqslant 1\mathrel{.}
\exists\pair{\sigma_i\in I}{i\in\interval{1}{n}}\mathrel{.}
\sigma_1\in P\wedge
\forall  i\in\interval[open right]{1}{n}\mathrel{.}\neg(\textsf{\upshape post}(\sqb{\texttt{\small B}}\fatsemi\sqb{\texttt{\small S}}^e)\{\sigma_i\}\subseteq(\neg\{\sigma_{i+1}\}))\wedge
\sigma_n\not\in\mathcal{B}\sqb{\texttt{\small B}} \wedge\sigma_n\not\in Q\}}{$\neg(X\subseteq Y)\Leftrightarrow(X\cap\neg Y\neq\emptyset$}\\[0.5ex]
=
\formulaexplanation{
\{\pair{P}{Q}\mid 
\exists n\geqslant 1\mathrel{.}
\exists\pair{\sigma_i\in I}{i\in\interval{1}{n}}\mathrel{.}
\sigma_1\in P\wedge
\forall  i\in\interval[open right]{1}{n}\mathrel{.}\neg(\textsf{\upshape post}(\sqb{\texttt{\small S}}^e)(\mathcal{B}\sqb{\texttt{\small B}}\cap\{\sigma_i\})\subseteq(\neg\{\sigma_{i+1}\}))\wedge
\sigma_n\not\in\mathcal{B}\sqb{\texttt{\small B}} \wedge\sigma_n\not\in Q\}}{def\@. $\textsf{\upshape post}$, $\sqb{\texttt{\small B}}$, and $\fatsemi$}\\[0.5ex]

=
\formulaexplanation{
\{\pair{P}{Q}\mid 
\exists n\geqslant 1\mathrel{.}
\exists\pair{\sigma_i\in I}{i\in\interval{1}{n}}\mathrel{.}
\sigma_1\in P\wedge
\forall  i\in\interval[open right]{1}{n}\mathrel{.}\pair{\mathcal{B}\sqb{\texttt{\small B}}\cap\{\sigma_i\}}{\neg\{\sigma_{i+1}\}}\in\{\pair{P}{Q}\mid\neg(\textsf{\upshape post}(\sqb{\texttt{\small S}}^e)P\subseteq Q)\}
\wedge
\sigma_n\not\in\mathcal{B}\sqb{\texttt{\small B}} \wedge\sigma_n\not\in Q\}}{def\@. $\in$}\\[0.5ex]

=
\lastformulaexplanation{
\{\pair{P}{Q}\mid 
\exists n\geqslant 1\mathrel{.}
\exists\pair{\sigma_i\in I}{i\in\interval{1}{n}}\mathrel{.}
\sigma_1\in P\wedge
\forall  i\in\interval[open right]{1}{n}\mathrel{.}\pair{\mathcal{B}\sqb{\texttt{\small B}}\cap\{\sigma_i\}}{\neg\{\sigma_{i+1}\}}\in\mathcal{T}_{\overline{\textrm{HL}}}(\texttt{\small S})
\wedge
\sigma_n\not\in\mathcal{B}\sqb{\texttt{\small B}} \wedge\sigma_n\not\in Q\}}{def\@. $\mathcal{T}_{\overline{\textrm{HL}}}(\texttt{\small S})$}{\mbox{\qed}}
\end{calculus}\let\qed\relax
\end{proof}
\vskip1ex
\subsubsection{Calculational Design of\/ $\overline{\mbox{\normalfont HL}}$ Proof Rules}~\\[-1em]

\noindent\hypertarget{th:slides:7}{\textsc{Theorem 7 ($\overline{\mbox{\normalfont HL}}$ rules for conditional iteration)}}. 
\begin{eqntabular}{c}
\frac {\displaystyle\exists\pair{\sigma_i\in I}{i\in\interval{1}{n}}\mathrel{.}\sigma_1\in P\wedge{}
\forall  i\in\interval[open right]{1}{n}\mathrel{.}
\llparenthesis\,\mathcal{B}\sqb{\texttt{\small B}}\cap\{\sigma_i\}\,\rrparenthesis\,\texttt{\small S}\,\llparenthesis\,\neg\{\sigma_{i+1}\}\,\rrparenthesis
\wedge
\sigma_n\not\in\mathcal{B}\sqb{\texttt{\small B}} \wedge\sigma_n\not\in Q}{\displaystyle\llparenthesis\,P\,\rrparenthesis\,\texttt{\small while (B) S}\,\llparenthesis\, Q\,\rrparenthesis}
\label{eq:not:Hoare:while:rule}
\end{eqntabular}

\begin{proof}[Proof of (\ref{eq:not:Hoare:while:rule})]

We write $\llparenthesis\,P\,\rrparenthesis\,\texttt{\small S}\,\llparenthesis\, Q\,\rrparenthesis \triangleq\pair{P}{Q}\in{\overline{\textrm{HL}}}(\texttt{\small S})$;

By structural induction (\texttt{\small S} being a strict component of \texttt{\small while (B) S}), the rule for $\llparenthesis\,P\,\rrparenthesis\,\texttt{\small S}\,\llparenthesis\, Q\,\rrparenthesis$ have already been defined; 

By Aczel method, the (constant) fixpoint $\Lfp{\subseteq}\LAMBDA{X}S$ is defined by
$\{\frac{\emptyset}{c}\mid c\in S\}$; 

So for \texttt{\small while (B) S} we have an axiom $\frac {\displaystyle\emptyset}{\displaystyle\llparenthesis\,P\,\rrparenthesis\,\texttt{\small while (B) S}\,\llparenthesis\, Q\,\rrparenthesis}$ with side condition $\exists\pair{\sigma_i\in I}{i\in\interval{1}{n}}\mathrel{.}\sigma_1\in P\wedge{}
\forall  i\in\interval[open right]{1}{n}\mathrel{.}
\llparenthesis\,\mathcal{B}\sqb{\texttt{\small B}}\cap\{\sigma_i\}\,\rrparenthesis\,\texttt{\small S}\,\llparenthesis\,\neg\{\sigma_{i+1}\}\,\rrparenthesis
\wedge
\sigma_n\not\in\mathcal{B}\sqb{\texttt{\small B}} \wedge\sigma_n\not\in Q$  where $\llparenthesis\,\mathcal{B}\sqb{\texttt{\small B}}\cap\{\sigma_i\}\,\rrparenthesis\,\texttt{\small S}\,\llparenthesis\,\neg\{\sigma_{i+1}\}\,\rrparenthesis$ is well-defined by structural induction;

\vskip1mm

Traditionally, the side condition is written as a premiss, to get (\ref{eq:not:Hoare:while:rule}).
\end{proof}
This is nothing but debugging formalized as a logic since $\pair{\sigma_i\in I}{i\in\interval{1}{n}}$ is a finite iteration in the loop starting with $P$ true and finishing with $Q$ false, which is obviously a counter example to Hoare triple $\{P\}\,\texttt{\small while (B) S}\,\{Q\}$. Notice that recursively $\llparenthesis\,\mathcal{B}\sqb{\texttt{\small B}}\cap\{\sigma_i\}\,\rrparenthesis\,\texttt{\small S}\,\llparenthesis\,\{\sigma_{i+1}\}\,\rrparenthesis$ enforces the execution of the loop body \texttt{\small S}  to start in state $\sigma_i$ and terminate in state $\sigma_{i+1}$.

\end{toappendix}

\subsection{Application I: Calculational Design of a New Forward Logic for Termination with Correct Reachability of a Postcondition or Nontermination}\label{Calculational-Design-of-the-Extended-Hoare-Logic}

Using $\bot$ to denote nontermination, we write $Q_{\bot}\triangleq Q\cup\{\bot\}$ and $Q_{\angelic}\triangleq Q\setminus\{\bot\}$. The semantics and predicates/assertions are relational. They can establish a relation between initial and final values of a loop body to show that a variant function in a well founded set is decreasing (as in
Turing/Floyd method formalized by Th\@. \ref{th:Turing:Floyd}). See example \ref{ex:fact:spec}.

The language includes a \texttt{\small break} out of the closest enclosing loop, so  the specifications have the form $\{P\}\,\texttt{\small S}\,\{ok:Q,br:T\}$ meaning that any execution of \texttt{\small S} started in a state of $P$ will terminate in a state of $Q_{\angelic}$, or not terminate if $\bot\in Q$, or break out of \texttt{\small S} to the closest enclosing loop in a state satisfying $T$.
So $Q=\{\bot\}$ and $T=\emptyset$ would mean definite non termination (when $P\neq\emptyset$).

To design the logic, we first formally define the meaning of specifications as an abstraction of \textsf{\textup{post}}. Then we proceed by structural induction on the syntax of the language. Using fixpoint  over approximation Th\@. \ref{th:Fixpoint-Overapproximation}, the iteration rule is  ($\overline{\Sigma}$ is $\Sigma$ extended to an auxiliary variable in $\overline{\mathbb{X}}$ for each variable in $\mathbb{X}$)\ifshort\ \proofinapx\fi
\bgroup\abovedisplayskip-4pt\belowdisplayskip-3pt\begin{eqntabular}{@{}c@{}}
\frac{\begin{array}[t]{@{}c@{}}
\{\sigma\in\wp(\overline{\Sigma})\mid \sigma_{\overline{\mathbb{X}}}=\sigma_{\mathbb{X}} \wedge \sigma_{\mathbb{X}}\in P\}\subseteq {I}\quad
\{{\mathcal{B}\sqb{\texttt{\small B}}\cap I_{\angelic}}\}\,\texttt{\small S}\,\{ok:R, br:T\}\\[2pt] R_{\angelic}\subseteq I
\quad
(\mathcal{B}\sqb{\neg\texttt{\small B}}\cap I)\subseteq Q
\quad
T \subseteq Q
\quad
R_{\bot}\subseteq Q
\\[2pt]
(\bot\notin Q)
\Rightarrow
(\exists\pair{W}{\preceq}\in\mathfrak{Wf}\mathrel{.}
\exists\nu\in I\rightarrow W\mathrel{.}
\forall  \pair{\underline{\sigma}}{\sigma'}\in I\mathrel{.}\nu( \underline{\sigma})\succ\nu(\sigma'))
\end{array}}
{\{P\}\,\texttt{\small while  (B) S}\,\{ok:Q, br:T\}}\label{eq:ehl:while}
\end{eqntabular}\egroup
\begin{example}\label{ex:fact:post}For factorial \texttt{\small fact}, we choose the invariant $I=I_{\angelic}\cup I_{\bot}$ with 
$I_{\angelic}=\{n=\underline{n}\wedge f=1\}\cup\{\underline{n}>{n}\geqslant 0\wedge f=\prod_{i=n}^{\underline{n}}i\}$ with $\prod\emptyset=1$
for termination and $I_{\bot}= \{n\leqslant\underline{n}<0\wedge f=\bot\}$ for nontermination, 
$\pair{W}{\preceq}=\pair{\mathbb{N}}{\leqslant}\in\mathfrak{Wf}$, $\nu(\underline{n},n)=n$. 
We have $\mathcal{B}\sqb{\texttt{\small B}}=n\neq0$ 
so that 
$\{\mathcal{B}\sqb{\texttt{\small B}}\cap I_{\angelic}\}\,\texttt{\small f = f*n; n = n-1;}\,\{ok:R, br:T\}$ is 
$\{(n=\underline{n}\neq0\wedge f=1)\vee(n>\underline{n}> 0\wedge f=\prod_{i=n}^{\underline{n}}i)\}\,\texttt{\small f = f*n; n = n-1;}\,\{ok:R, br:\emptyset\}$ with $R$ = $R_{\angelic}$ =
$(n=\underline{n}-1\neq0\wedge f=\underline{n})\vee(n>\underline{n}> 0\wedge f=\prod_{i=n}^{\underline{n}}i)$ $\subseteq$ $I$, $R_{\bot}=\emptyset$
and $T=\emptyset$ by termination and absence of \texttt{\small break}.
\end{example}
\begin{toappendix}
\label{sec:apx:Calculational-Design-of-the-Extended-Hoare-Logic}
\subsection{The post image transformer with breaks}
To handle breaks we extend the definition of $\textsf{\upshape post}\sqb{\texttt{\small S}}$ componentwise.
\begin{eqntabular}{rcl}
    \textsf{\upshape post}^{\times}(\sqb{\texttt{\small S}}_\bot)
    &\triangleq&
\pair{
        \textsf{\upshape post}(\sqb{\texttt{\small S}}^e\cup
        \sqb{\texttt{\small S}}^\bot)}
       {\textsf{\upshape post} \sqb{\texttt{\small S}}^b}
       \label{eq:def:post:x}
\end{eqntabular}
The semantic of extended Hoare triples is
\begin{eqntabular}{rcl}
\{P\}\,\texttt{\small S}\,\{ok:Q, br:T\}
&\triangleq&
\triple{P}{Q}{T}\in\alpha_{\textup{HL}}(\textsf{\upshape post}^{\times}(\sqb{\texttt{\small S}}_\bot))
\label{eq:semantics-Hoare-logic}
\end{eqntabular} where
\begin{eqntabular}[fl]{@{}rcl}
\alpha_{\textup{HL}}(\textsf{\upshape post}^{\times}(\sqb{\texttt{\small S}}_\bot))
&\triangleq &\{\triple{P}{Q}{T}\mid
    \textsf{\upshape post}^{\times} \sqb{\texttt{\small S}}_\bot P  \stackrel{\text{.}}{\subseteq}
    \pair{Q}{T}\}\label{eq:def:aHLopost:times}\\
&=&\{\triple{P}{Q}{T}\mid
    \textsf{\upshape post}(\sqb{\texttt{\small S}}^e\cup
        \sqb{\texttt{\small S}}^\bot)P \subseteq Q \land
    {\textsf{\upshape post} \sqb{\texttt{\small S}}^bP \subseteq T}
        \}\renumber{\textexplanation{by def.\ of $\textsf{\upshape post}^{\times}$}}
\end{eqntabular}
Note that this may involve several batches of never modified auxiliary variables.

\subsection{Auxiliary Propositions}
\noindent We will use the following auxiliary lemmas.
\bgroup\arraycolsep0.5\arraycolsep\begin{eqntabular}[fl]{L@{\quad}rcl@{\quad and\quad}rcl}
\labelitemi\quad join preservation
&
{\textsf{\upshape post}}(\bigcup_{i\in\Delta} \tau_i)Q 
&=& 
\bigcup_{i\in\Delta}{\textsf{\upshape post}}(\tau_i)Q 
&
{\textsf{\upshape post}}(\tau)\bigcup_{i\in\Delta} Q _i
&=&
\bigcup_{i\in\Delta} {\textsf{\upshape post}}(\tau)Q _i
 \label{eq:post:join-preservation}
\end{eqntabular}\egroup
\begin{proof}[proof of (\ref{eq:post:join-preservation})] By the Galois connections $\pair{\wp(X\times Y)}{\subseteq}\galois{\textsf{\upshape post}}{{\textsf{\upshape post}^{-1}}}\pair{\wp(X)\stackrel{\sqcupdot}{\longrightarrow}\wp(Y)}{\stackrel{.}{\subseteq}}$ and
$\pair{\wp(X)}{\subseteq}\galois{\textsf{\upshape post}(\tau)}{\widetilde{\textsf{\upshape pre}}(\tau)}\pair{\wp(Y)}{\subseteq}$
where the lower adjoint preserves arbitrary joins.
\end{proof}
\begin{eqntabular}[fl]{L@{\qquad}rcl}
\labelitemi\quad composition
&
\textsf{\upshape post}(r_1\fatsemi r_2)P 
&=& 
\textsf{\upshape post}(r_2)\comp \textsf{\upshape post}(r_1)P 
 \label{eq:post:composition}
\end{eqntabular}
\begin{proof}[proof of (\ref{eq:post:composition})] 
\begin{calculus}[$\Rightarrow$~]<2pt>
\formula{\textsf{\upshape post}(r_1\fatsemi r_2)P }\\
=
\formulaexplanation{\{\sigma'\in Y\mid\exists\sigma\in P\mathrel{.}\pair{\sigma}{\sigma'}\in r_1\fatsemi r_2\}}{def.\ (\ref{eq:def:post}) of \textsf{\upshape post}}\\
=
\formulaexplanation{\{\sigma'\in Y\mid\exists\sigma\in P\mathrel{.}\exists \sigma''\mathrel{.}\pair{\sigma}{\sigma''}\in r_1
\land  \pair{\sigma''}{\sigma'}\in r_2\}}{def.\ relation composition $\fatsemi$}\\
=
\formulaexplanation{\{\sigma'\in Y\mid\exists \sigma''\mathrel{.}\pair{\sigma''}{\sigma'}\in r_2
\land \exists\sigma\in P\mathrel{.} \pair{\sigma}{\sigma''}\in r_1\}}{associativity of disjunction}\\
=
\formulaexplanation{\{\sigma'\in Y\mid\exists \sigma''\mathrel{.}\pair{\sigma''}{\sigma'}\in r_2
\land \sigma''\in\{\sigma''\mid\exists\sigma\in P\mathrel{.} \pair{\sigma}{\sigma''}\in r_1\}\}}{def.\ $\in$}\\
=
\formulaexplanation{\{\sigma'\in Y\mid\exists \sigma''\mathrel{.}\pair{\sigma''}{\sigma'}\in r_2
\land \sigma''\in\textsf{\upshape post}(r_1)P\}}{def.\ (\ref{eq:def:post}) of \textsf{\upshape post}}\\
=
\formulaexplanation{\textsf{\upshape post}(r_2)(\textsf{\upshape post}(r_1)P)}{def.\ (\ref{eq:def:post}) of \textsf{\upshape post}}\\
=
\formulaexplanation{\textsf{\upshape post}(r_2)\comp\textsf{\upshape post}(r_1)P}{def.\ function composition $\comp$}
\end{calculus}
\end{proof}
\begin{eqntabular}[fl]{L@{\qquad}rcl}
\labelitemi\quad Boolean postcondition
&
\textsf{\upshape post}(\sqb{\texttt{\small B}})P
=&
P\cap\mathcal{B}\sqb{\texttt{\small B}}
 \label{eq:Boolean-postcondition}
\end{eqntabular}
\begin{proof}[proof of (\ref{eq:Boolean-postcondition})]
\begin{calculus}[($\Rightarrow$)~]<2pt>
\formula{\textsf{\upshape post}(\sqb{\texttt{\small B}})Q}\\
=
\formulaexplanation{\{\sigma'\mid\exists\sigma\in P\mathrel{.}\pair{\sigma}{\sigma'}\in\sqb{\texttt{\small B}}\}}{def.\ (\ref{eq:def:post}) of \textsf{\upshape post}}\\
=
\formulaexplanation{\{\sigma'\mid\exists\sigma\in P\mathrel{.}\sigma=\sigma'\in\mathcal{B}\sqb{\texttt{\small B}}\}}{def.\  $\sqb{\texttt{\small B}}\triangleq\{\pair{\sigma}{\sigma}\mid\sigma\in\mathcal{B}\sqb{\texttt{\small B}}\}$}\\
=
\formulaexplanation{\{\sigma\mid\sigma\in P\land\sigma\in\mathcal{B}\sqb{\texttt{\small B}}\}}{def.\ equality}\\

=
\lastformulaexplanation{P\cap \mathcal{B}\sqb{\texttt{\small B}}}{def.\ intersection $\cap$}{\mbox{\qed}}
\end{calculus}\let\qed\relax
\end{proof}

\begin{eqntabular}[fl]{L@{\qquad}rcl}
\labelitemi\quad commutativity
&
\Lfp{\subseteq}{F^e}
&=& 
\Lfp{\subseteq}F'^e
\label{eq:Fe:inversion}\\
&{F^e}&\triangleq&\LAMBDA{X}\textsf{\upshape id} \cup (\sqb{\texttt{\small B}}\fatsemi\sqb{\texttt{\small S}}^e\fatsemi X),\quad X\in\wp(\Sigma\times\Sigma)\renumber{(\ref{eq:natural-transformer-finite}) }\\
&F'^e &\triangleq&\LAMBDA{X}\textsf{\upshape id} \cup (X\fatsemi\sqb{\texttt{\small B}}\fatsemi\sqb{\texttt{\small S}}^e)
 \nonumber
\end{eqntabular}
\begin{proof}[proof of (\ref{eq:Fe:inversion})] Let $\pair{X^n}{n\in\mathbb{N}}$ and
$\pair{X'^n}{n\in\mathbb{N}}$ be the respective iterates of ${F^e}$ and $F'^e$ (which preserve joins) starting from the infimum $\emptyset$. We have $X^0=X'^0$ so $X^1=X'^1=\textsf{id}=(\sqb{\texttt{\small B}}\fatsemi\sqb{\texttt{\small S}}^e)^0$. Assume $X^n=X'^n=\bigcup_{k<n}(\sqb{\texttt{\small B}}\fatsemi\sqb{\texttt{\small S}}^e)^k$, $n>0$ by induction hypothesis. Then
\begin{calculus}
\formula{X^{n+1}}\\
=\formulaexplanation{F^e(X^{n})}{def.\ iterates}\\
=\formulaexplanation{\textsf{\upshape id} \cup (\sqb{\texttt{\small B}}\fatsemi\sqb{\texttt{\small S}}^e\fatsemi X^n)}{def.\ (\ref{eq:natural-transformer-finite}) of $F^e$ and $X\cap\Sigma\times\{\bot\}=\emptyset$}\\
=\formulaexplanation{\textsf{\upshape id} \cup (\sqb{\texttt{\small B}}\fatsemi\sqb{\texttt{\small S}}^e\fatsemi \bigcup_{k<n}(\sqb{\texttt{\small B}}\fatsemi\sqb{\texttt{\small S}}^e)^k)}{induction hypothesis}\\
=\formulaexplanation{(\sqb{\texttt{\small B}}\fatsemi\sqb{\texttt{\small S}}^e)^0 \cup \bigcup_{k<n}(\sqb{\texttt{\small B}}\fatsemi\sqb{\texttt{\small S}}^e)^{k+1})}{def.\ power and relation composition $\fatsemi$}\\
=\formulaexplanation{\bigcup_{k<n+1}(\sqb{\texttt{\small B}}\fatsemi\sqb{\texttt{\small S}}^e)^{k}}{which is the induction hypotehsis}\\
=\formulaexplanation{\textsf{\upshape id} \cup (\bigcup_{k<n}(\sqb{\texttt{\small B}}\fatsemi\sqb{\texttt{\small S}}^e)^k\fatsemi\sqb{\texttt{\small B}}\fatsemi\sqb{\texttt{\small S}}^e )}{def.\ power and relation composition $\fatsemi$}\\
=\formulaexplanation{\textsf{\upshape id} \cup (X'^n\fatsemi\sqb{\texttt{\small B}}\fatsemi\sqb{\texttt{\small S}}^e )}{induction hypothesis
}\\
=\formulaexplanation{F'^e(X'^{n})}{def.\ (\ref{eq:Fe:inversion}) of $F'^e$}\\
=\formulaexplanation{X'^{n+1}}{def.\ iterates}
\end{calculus}
It follows that $\Lfp{\subseteq}{F^e}$
=
$\bigcup_{k\in \mathbb{X}}X^k$
=
$\bigcup_{k\in \mathbb{X}}X'^k$
=
$\Lfp{\subseteq}F'^e$.
\end{proof}

\begin{eqntabular}[fl]{L@{\qquad}rcl}
\labelitemi\quad commutation
&
\textsf{\upshape post}({F'^e}(X))P
&=& 
{\bar{F}'^e_P}(\textsf{\upshape post}(X)P)\label{eq:post:commutation:Fe}\\
&{\bar{F}'^e_P}(X)&\triangleq&
P\cup\textsf{\upshape post}(\sqb{\texttt{\small B}}\fatsemi\sqb{\texttt{\small S}}^e)X,\quad X\in\wp(\Sigma)\rightarrow\wp(\Sigma)
 \nonumber
\end{eqntabular}
\begin{proof}[proof of (\ref{eq:post:commutation:Fe})] 
\begin{calculus}[$\Rightarrow$~]<2pt>
\formula{\textsf{\upshape post}({F'^e}(X))P}\\
=
\formulaexplanation{\textsf{\upshape post}(\textsf{\upshape id} \cup (X\fatsemi\sqb{\texttt{\small B}}\fatsemi\sqb{\texttt{\small S}}^e))P}{def.\ (\ref{eq:Fe:inversion}) of ${F'^e}$}\\ 
=
\formulaexplanation{\textsf{\upshape post}(\textsf{\upshape id})P \cup \textsf{\upshape post}(X\fatsemi\sqb{\texttt{\small B}}\fatsemi\sqb{\texttt{\small S}}^e )P}{join preservation (\ref{eq:post:join-preservation})}\\ 
=
\formulaexplanation{P \cup \textsf{\upshape post}(X\fatsemi\sqb{\texttt{\small B}}\fatsemi\sqb{\texttt{\small S}}^e)P}{def.\ (\ref{eq:def:post}) of \textsf{\upshape post} and def.\  identity relation \textsf{\upshape id}}\\ 
=
\formulaexplanation{P \cup \textsf{\upshape post}(\sqb{\texttt{\small B}}\fatsemi\sqb{\texttt{\small S}}^e)(\textsf{\upshape post}(X)P)}{composition (\ref{eq:post:composition})}\\ 
=
\lastformulaexplanation{{\bar{F}'^e_P}(\textsf{\upshape post}(X)P)}{def.\ (\ref{eq:post:commutation:Fe}) of ${\bar{F}'^e}$}{\mbox{\qed}}
\end{calculus}
\let\qed\relax
\end{proof}
\subsection{Design of the Extended Hoare Logic}
The deductive rules of Hoare logic are derived by structural induction on the program syntax and abstraction of the semantics for each statement \texttt{\small S}.
\begin{calculus}<2pt>
\hyphen{6}\discussion{Iteration \texttt{\small while(B) S}}\\

\formula{\alpha_{\textup{HL}}(\textsf{\upshape post}^{\times}(\sqb{\texttt{\small while(B) S}}_\bot))}\\
=
\formula{\{\triple{P}{Q}{T}\mid
    \textsf{\upshape post}(\sqb{\texttt{\small while(B) S}}^e\cup
        \sqb{\texttt{\small while(B) S}}^{\bot})P \subseteq Q \land
    {\textsf{\upshape post} \sqb{\texttt{\small while(B) S}}^bP \subseteq T}
    \}
}\\
\explanation{def.\ $\alpha_{\textup{HL}}\comp\textsf{\upshape post}^{\times}$}\\
=
\formula{\{\triple{P}{Q}{T}\mid
    \textsf{\upshape post}((\Lfp{\subseteq}{F^e}\fatsemi(\sqb{\neg\texttt{\small B}}\cup\sqb{\texttt{\small B}}\fatsemi\sqb{\texttt{\small S}}^b))\cup
        (\Lfp{\subseteq}{F^e}\fatsemi\sqb{\texttt{\small B}}\fatsemi\sqb{\texttt{\small S}}^\bot)\cup\Gfp{\subseteq}{F^\bot})P \subseteq Q \land
    {\textsf{\upshape post}(\emptyset)P \subseteq T}
    \}
}\\
\explanation{
by definition (\ref{eq:natural-finite}), (\ref{eq:natural-oo}), and (\ref{eq:natural-break}) of the relational semantics of \sqb{\texttt{\small  while(B) S}}
}\\
=
\formula{\{\triple{P}{Q}{T}\mid
    \textsf{\upshape post}((\Lfp{\subseteq}{F^e}\fatsemi(\sqb{\neg\texttt{\small B}}\cup\sqb{\texttt{\small B}}\fatsemi\sqb{\texttt{\small S}}^b))\cup
        (\Lfp{\subseteq}{F^e}\fatsemi\sqb{\texttt{\small B}}\fatsemi\sqb{\texttt{\small S}}^\bot)\cup\Gfp{\subseteq}{F^\bot})P \subseteq Q 
\}
}\\
\explanation{$\textsf{\upshape post}(\emptyset)P=\emptyset$}\\
= 
\formula{\{\triple{P}{Q}{T}\mid
    \textsf{\upshape post}
        (\Lfp{\subseteq}{F^e}\fatsemi(\sqb{\neg\texttt{\small B}} \cup \sqb{\texttt{\small B}}\fatsemi\sqb{\texttt{\small S}}^b)P \subseteq Q  \land
    \textsf{\upshape post}(
        \Lfp{\subseteq}{F^e}\fatsemi\sqb{\texttt{\small B}}\fatsemi\sqb{\texttt{\small
            S}}^\bot))P\subseteq Q  \land
       \textsf{\upshape post}(\Gfp{\subseteq}{F^\bot}
        )P\subseteq Q
    \}
}\\
\explanation{
$\textsf{\upshape post}$ preserves arbitrary joins (\ref{eq:post:join-preservation}), and $(A \cup B) \subseteq q \Leftrightarrow A \subseteq q \land B \subseteq q $}\\
= 
\formula{\{\triple{P}{Q}{T}\mid
    \textsf{\upshape post}(\sqb{\neg\texttt{\small B}}\cup\sqb{\texttt{\small B}}\fatsemi\sqb{\texttt{\small S}}^b)(\textsf{\upshape post}(\Lfp{\subseteq}{F^e}) P)\subseteq Q\land
    \textsf{\upshape post}(\sqb{\texttt{\small B}}\fatsemi\sqb{\texttt{\small S}}^\bot)(\textsf{\upshape post}(
        \Lfp{\subseteq}{F^e})P)\subseteq Q  \land
       \textsf{\upshape post}(\Gfp{\subseteq}{F^\bot}
        )P\subseteq Q
    \}
}\\
\explanation{
     by composition (\ref{eq:post:composition})
}\\
=
\formula{\{\triple{P}{Q}{T}\mid\exists I\in\wp(\Sigma_{\bot})\mathrel{.}\textsf{\upshape post}(\Lfp{\subseteq}{F^e}) P\subseteq I\land
    \textsf{\upshape post}(\sqb{\neg\texttt{\small B}}\cup(\sqb{\texttt{\small B}}\fatsemi\sqb{\texttt{\small S}}^b))I\subseteq Q\land
    \textsf{\upshape post}(\sqb{\texttt{\small B}}\fatsemi\sqb{\texttt{\small S}}^\bot)I\subseteq Q  \land
       \textsf{\upshape post}(\Gfp{\subseteq}{F^\bot}
        )P\subseteq Q
    \}
}\\
\explanation{introducing an over approximation $I\in\wp(\Sigma_{\bot})$.\\
($\subseteq$)\quad $\textsf{\upshape post}(r)$ preserves joins so is increasing and transitivity;\\
($\supseteq$)\quad take $I=\textsf{\upshape post}(\Lfp{\subseteq}{F^e}) P$}\\
=
\formula{\{\triple{P}{Q}{T}\mid\exists I\in\wp(\Sigma_{\bot})\mathrel{.}\textsf{\upshape post}(\Lfp{\subseteq}{F'^e}) P\subseteq I\land
    \textsf{\upshape post}(\sqb{\neg\texttt{\small B}}\cup(\sqb{\texttt{\small B}}\fatsemi\sqb{\texttt{\small S}}^b))I\subseteq Q\land
    \textsf{\upshape post}(\sqb{\texttt{\small B}}\fatsemi\sqb{\texttt{\small S}}^\bot)I\subseteq Q  \land
       \textsf{\upshape post}(\Gfp{\subseteq}{F^\bot}
        )P\subseteq Q
    \}
}\\
\explanation{by $\Lfp{\subseteq}{F^e}
=
\Lfp{\subseteq}F'^e$ 
at (\ref{eq:Fe:inversion})}\\
=
\formula{\{\triple{P}{Q}{T}\mid\exists I\in\wp(\Sigma_{\bot})\mathrel{.}\Lfp{\subseteq}{\bar{F}'^e_P}\subseteq I\land
    \textsf{\upshape post}(\sqb{\neg\texttt{\small B}}\cup(\sqb{\texttt{\small B}}\fatsemi\sqb{\texttt{\small S}}^b))I\subseteq Q\land
    \textsf{\upshape post}(\sqb{\texttt{\small B}}\fatsemi\sqb{\texttt{\small S}}^\bot)I\subseteq Q  \land
       \si \bot\in Q\alors \textsf{true}\sinon \Gfp{\subseteq}{F^\bot}\cap(P\times\{\bot\})=\emptyset\fsi
    \}
}\\
\explanation{\hyphen{6}~by Galois connection(\ref{eq:def:post:GC}), commutativity (\ref{eq:post:commutation:Fe}), and Th\@. \ref{th:fixpoint-abstraction} for $\alpha(X)=\textsf{\upshape post}(X)P$\\
\hyphen{6}~By def.\ (\ref{eq:natural-transformer-infinite}) of ${F^\bot}(X)\triangleq\sqb{\texttt{\small B}}\fatsemi\sqb{\texttt{\small S}}^e\fatsemi X$ where $X\in\wp(\Sigma\times\{\bot\})$, we have $\Gfp{\subseteq}{F^\bot}\in\wp(\Sigma\times\{\bot\})$ and so by def.\ (\ref{eq:def:post}) of \textsf{\upshape post}, we have 
$\textsf{\upshape post}(r)P$ = 
$\{\bot\mid\exists\sigma\mathrel{.}\sigma\in P\land\pair{\sigma}{\bot}\in \Gfp{\subseteq}{F^\bot}\}$
=
$\{\bot\mid \Gfp{\subseteq}{F^\bot}\cap(P\times\{\bot\})\neq\emptyset\}$.
If follows that $\textsf{\upshape post}(\Gfp{\subseteq}{F^\bot}
        )P\subseteq Q$ when $\bot\in Q$ and otherwise $\Gfp{\subseteq}{F^\bot}\cap(P\times\{\bot\})=\emptyset$
}\\
=
\formula{\{\triple{P}{Q}{T}\mid\exists I\in\wp(\Sigma_{\bot})\mathrel{.}\Lfp{\subseteq}\LAMBDA{X}P\cup\textsf{\upshape post}(\sqb{\texttt{\small B}}\fatsemi\sqb{\texttt{\small S}}^e)X \subseteq I
\land
    \textsf{\upshape post}(\sqb{\neg\texttt{\small B}}\cup(\sqb{\texttt{\small B}}\fatsemi\sqb{\texttt{\small S}}^b))I\subseteq Q\land
    \textsf{\upshape post}(\sqb{\texttt{\small B}}\fatsemi\sqb{\texttt{\small S}}^\bot)I\subseteq Q  \land
       (\bot\notin Q)\Rightarrow (\Gfp{\subseteq}{F^\bot}\cap(P\times\{\bot\})=\emptyset)
    \}
}\\
\explanation{\hyphen{6}~def.\ (\ref{eq:post:commutation:Fe}) of ${\bar{F}'^e_P}$\\
\hyphen{6}~def.\ conditional $\si\ldots\alors\ldots\sinon\ldots\fsi$ and implication $\Rightarrow$}\\
=
\formula{\{\triple{P}{Q}{T}\mid\exists I\in\wp(\Sigma_{\bot})\mathrel{.}P\subseteq I\land
\textsf{\upshape post}(\sqb{\texttt{\small B}}\fatsemi\sqb{\texttt{\small S}}^e)I\subseteq I
\land
    \textsf{\upshape post}(\sqb{\neg\texttt{\small B}}\cup(\sqb{\texttt{\small B}}\fatsemi\sqb{\texttt{\small S}}^b))I\subseteq Q\land
    \textsf{\upshape post}(\sqb{\texttt{\small B}}\fatsemi\sqb{\texttt{\small S}}^\bot)I\subseteq Q  \land
              (\bot\notin Q)\Rightarrow \{\sigma\in\mathbb{N}\rightarrow\Sigma\mid\sigma_0\in P\land\forall i\in\mathbb{N}\mathrel{.}\pair{\sigma_i}{\sigma_{i+1}}\in \sqb{\texttt{\small B}}\fatsemi\sqb{\texttt{\small S}^e}\}=\emptyset
    \}
}\\
\explanation{\hyphen{6}~induction for fixpoint  over approximation Th\@. \ref{th:Fixpoint-Overapproximation} noting that the invariant $I$ need not be strengthened since it can be always be chosen to be strong enough\\
\hyphen{6}~by lemma \ref{lem:Fbot:nontermination}}\\
=
\formula{\{\triple{P}{Q}{T}\mid\exists I\in\wp(\Sigma_{\bot})\mathrel{.}P\subseteq I\land
\textsf{\upshape post}(\sqb{\texttt{\small B}}\fatsemi\sqb{\texttt{\small S}}^e)I\subseteq I
\land
    \textsf{\upshape post}\sqb{\neg\texttt{\small B}}I\subseteq Q \land\textsf{\upshape post}(\sqb{\texttt{\small B}}\fatsemi\sqb{\texttt{\small S}}^b)I\subseteq Q\land
    \textsf{\upshape post}(\sqb{\texttt{\small B}}\fatsemi\sqb{\texttt{\small S}}^\bot)I\subseteq Q  \land
       (\bot\notin Q)\Rightarrow \{\sigma\in\mathbb{N}\rightarrow\Sigma\mid \sigma_0\in P\land\forall i\in\mathbb{N}\mathrel{.}\pair{\sigma_i}{\sigma_{i+1}}\in \sqb{\texttt{\small B}}\fatsemi\sqb{\texttt{\small S}^e}\}=\emptyset
    \}
}\\
\explanation{\textsf{\upshape post} preserves joins and union subsumption law, $(A \cup B) \subseteq q \Leftrightarrow A \subseteq q \land B \subseteq q$}\\
=
\formula{\{\triple{P}{Q}{T}\mid\exists I\in\wp(\Sigma_{\bot})\mathrel{.}P\subseteq I\land
\textsf{\upshape post}\sqb{\texttt{\small B}}(\textsf{\upshape post}\sqb{\texttt{\small S}}^e I)\subseteq I
\land
    \textsf{\upshape post}\sqb{\neg\texttt{\small B}}I\subseteq Q \land\textsf{\upshape post}\sqb{\texttt{\small B}}(\textsf{\upshape post}\sqb{\texttt{\small S}}^b I)\subseteq Q\land
    \textsf{\upshape post}\sqb{\texttt{\small B}}(\textsf{\upshape post}\sqb{\texttt{\small S}}^\bot I)\subseteq Q  \land
       (\bot\notin Q)\Rightarrow (\exists\pair{W}{\preceq}\in\mathfrak{Wf}\mathrel{.}\exists
J\in\wp(\Sigma)\mathrel{.}
\exists\nu\in J\rightarrow W\mathrel{.}
 P\cup \textsf{\upshape post}(\sqb{\texttt{\small B}}\fatsemi\sqb{\texttt{\small S}^e})J\subseteq J\land{}
\forall \underline{\sigma}\in J\mathrel{.}\forall \sigma'\in \Sigma\mathrel{.}\pair{ \underline{\sigma}}{\sigma'}\in \sqb{\texttt{\small B}}\fatsemi\sqb{\texttt{\small S}^e}\Rightarrow \nu( \underline{\sigma})\succ\nu(\sigma'))
    \}
}\\
\explanation{\hyphen{6}~composition (\ref{eq:post:composition})\\
\hyphen{6}~by Turing/Floyd Th\@. \ref{th:Turing:Floyd}}\\
=
\formula{\{\triple{P}{Q}{T}\mid\exists I\in\wp(\Sigma_{\bot})\mathrel{.}P\subseteq I\land
 \textsf{\upshape post}\sqb{\texttt{\small S}}^e(\mathcal{B}\sqb{\texttt{\small B}}\cap I)\subseteq I
\land
    (\mathcal{B}\sqb{\neg\texttt{\small B}}\cap I)\subseteq Q \land\textsf{\upshape post}\sqb{\texttt{\small S}}^b(\mathcal{B}\sqb{\texttt{\small B}}\cap I)\subseteq Q\land
    \textsf{\upshape post}\sqb{\texttt{\small S}}^\bot(\mathcal{B}\sqb{\texttt{\small B}}\cap I)\subseteq Q  \land
              (\bot\notin Q)\Rightarrow (\exists\pair{W}{\preceq}\in\mathfrak{Wf}\mathrel{.}\exists
J\in\wp(\Sigma)\mathrel{.}
\exists\nu\in J\rightarrow W\mathrel{.}
P\subseteq J\land ( \mathcal{B}\sqb{\texttt{\small B}}\cap\textsf{\upshape post}(\sqb{\texttt{\small S}^e})J\subseteq J\land{}
\forall  \underline{\sigma}\in J\mathrel{.}\forall \sigma'\in \Sigma\mathrel{.}\pair{ \underline{\sigma}}{\sigma'}\in \sqb{\texttt{\small B}}\fatsemi\sqb{\texttt{\small S}^e}\Rightarrow \nu( \underline{\sigma})\succ\nu(\sigma'))
    \}
}\\
\explanation{\hyphen{6}~Boolean composition (\ref{eq:Boolean-postcondition})\\
\hyphen{6}~by composition (\ref{eq:post:composition}), Boolean postcondition (\ref{eq:Boolean-postcondition}),  and union subsumption law, $(A \cup B) \subseteq q \Leftrightarrow A \subseteq q \land B \subseteq q$}\\
=
\formula{\{\triple{P}{Q}{T}\mid\exists I'\mathrel{.}P\subseteq I'\land
 \textsf{\upshape post}\sqb{\texttt{\small S}}^e(\mathcal{B}\sqb{\texttt{\small B}}\cap I')\subseteq I'
\land
    (\mathcal{B}\sqb{\neg\texttt{\small B}}\cap I')\subseteq Q \land\textsf{\upshape post}\sqb{\texttt{\small S}}^b (\mathcal{B}\sqb{\texttt{\small B}}\cap I')\subseteq Q\land
    \textsf{\upshape post}\sqb{\texttt{\small S}}^\bot (\mathcal{B}\sqb{\texttt{\small B}}\cap I')\subseteq Q  \land
              (\bot\notin Q)\Rightarrow (\exists\pair{W}{\preceq}\in\mathfrak{Wf}\mathrel{.}
\exists\nu\in I'\rightarrow W\mathrel{.}
\forall  \underline{\sigma}\in I'\mathrel{.}\forall \sigma'\in \Sigma\mathrel{.}\pair{ \underline{\sigma}}{\sigma'}\in \sqb{\texttt{\small B}}\fatsemi\sqb{\texttt{\small S}^e}\Rightarrow \nu( \underline{\sigma})\succ\nu(\sigma'))
    \}
}\\
\explanation{($\supseteq$)\quad take $I=J=I'$\\
($\supseteq$)\quad take $I'$ to be the greatest lower bound of $I$ and $J$ in the complete lattice
of post fixpoints of $\LAMBDA{X}P\cup \textsf{\upshape post}\sqb{\texttt{\small S}}^e(\mathcal{B}\sqb{\texttt{\small B}}\cap X)$, \cite[Corollary 4.2]{CousotCousot-PJM-82-1-1979}; $\cap$ and $\textsf{\upshape post}$ are increasing in both of their parameters}\\
=
\formula{\{\triple{P}{Q}{T}\mid\exists I\in\wp(\Sigma_{\bot})\mathrel{.}P\subseteq I\land
\textsf{\upshape post}\sqb{\texttt{\small S}}^e(\mathcal{B}\sqb{\texttt{\small B}}\cap I)\subseteq I
\land
    (\mathcal{B}\sqb{\neg\texttt{\small B}}\cap I)\subseteq Q \land\textsf{\upshape post}\sqb{\texttt{\small S}}^b (\mathcal{B}\sqb{\texttt{\small B}}\cap I)\subseteq Q\land
    \textsf{\upshape post}\sqb{\texttt{\small S}}^\bot(\mathcal{B}\sqb{\texttt{\small B}}\cap I)\subseteq Q  \land
              (\bot\notin Q)\Rightarrow (\exists\pair{W}{\preceq}\in\mathfrak{Wf}\mathrel{.}
\exists\nu\in I\rightarrow W\mathrel{.}
\forall  \underline{\sigma}\in I\mathrel{.}\forall \sigma'\in 
\textsf{\upshape post}\sqb{\texttt{\small S}^e}(\mathcal{B}\sqb{\texttt{\small B}}\cap\{\varsigma\mid\varsigma= \underline{\sigma}\})\mathrel{.}\nu( \underline{\sigma})\succ\nu(\sigma'))
    \}
}\\
\explanation{by composition (\ref{eq:post:composition}) and Boolean postcondition (\ref{eq:Boolean-postcondition}), writing $\{\underline{\sigma}\}=\{\varsigma\mid\varsigma= \underline{\sigma}\}$ to make the binding of $ \underline{\sigma}$ explicit}\\
=
\formula{\{\triple{P}{Q}{T}\mid\exists \bar{I}\in\wp(\overline{\Sigma}_{\bot})\mathrel{.}\mathbb{I}_{\overline{\mathbb{X}}}(P)\subseteq \bar{I}
\land
\textsf{\upshape post}\sqb{\texttt{\small S}}^e(\mathcal{B}\sqb{\texttt{\small B}}\cap \bar{I})\subseteq \bar{I}
\land
(\mathcal{B}\sqb{\neg\texttt{\small B}}\cap \bar{I})_{\mathbb{X}}\subseteq Q
\land
\textsf{\upshape post}\sqb{\texttt{\small S}}^b (\mathcal{B}\sqb{\texttt{\small B}}\cap \bar{I})_{\mathbb{X}}\subseteq Q
\land
\textsf{\upshape post}\sqb{\texttt{\small S}}^\bot(\mathcal{B}\sqb{\texttt{\small B}}\cap \bar{I})_{\mathbb{X}}\subseteq Q  
\land
(\bot\notin Q)
\Rightarrow 
(\exists\pair{W}{\preceq}\in\mathfrak{Wf}\mathrel{.}
\exists\nu\in \bar{I}\rightarrow W\mathrel{.}
\forall  \pair{\underline{\sigma}}{\sigma'}\in 
\textsf{\upshape post}\sqb{\texttt{\small S}^e}(\mathcal{B}\sqb{\texttt{\small B}}\cap\bar{I})\mathrel{.}\nu( \underline{\sigma})\succ\nu(\sigma'))
    \}
}\\
\explanation{by introducing auxiliary variables, with $I=\bar{I}_{\mathbb{X}}$}\\
=
\formula{\{\triple{P}{Q}{T}\mid\exists \bar{I}\in\wp(\overline{\Sigma}_{\bot})\mathrel{.}\mathbb{I}_{\overline{\mathbb{X}}}(P)\subseteq \bar{I}
\land
\textsf{\upshape post}\sqb{\texttt{\small S}}^e(\mathcal{B}\sqb{\texttt{\small B}}\cap \bar{I})\subseteq \bar{I}
\land
(\mathcal{B}\sqb{\neg\texttt{\small B}}\cap \bar{I})_{\mathbb{X}}\subseteq Q
\land
\textsf{\upshape post}\sqb{\texttt{\small S}}^b (\mathcal{B}\sqb{\texttt{\small B}}\cap \bar{I})_{\mathbb{X}}\subseteq Q
\land
\textsf{\upshape post}\sqb{\texttt{\small S}}^\bot(\mathcal{B}\sqb{\texttt{\small B}}\cap \bar{I})_{\mathbb{X}}\subseteq Q  
\land
(\bot\notin Q)
\Rightarrow 
(\exists\pair{W}{\preceq}\in\mathfrak{Wf}\mathrel{.}
\exists\nu\in \bar{I}\rightarrow W\mathrel{.}
\forall  \pair{\underline{\sigma}}{\sigma'}\in \bar{I}\mathrel{.}\nu( \underline{\sigma})\succ\nu(\sigma'))
    \}
}\\
\explanation{($\subseteq$)\quad by $x\in A\land (A\subseteq B)\Rightarrow (x\in B)$\\
($\supseteq$)\quad by taking the greatest lower bound of $\bar(I)$ and $\textsf{\upshape post}\sqb{\texttt{\small S}}^e(\mathcal{B}\sqb{\texttt{\small B}}\cap \bar{I})$ in the complete lattice
of post fixpoints of $\LAMBDA{X}\textsf{\upshape post}\sqb{\texttt{\small S}}^e(\mathcal{B}\sqb{\texttt{\small B}}\cap P\cap X)$, \cite[Corollary 4.2]{CousotCousot-PJM-82-1-1979} and increasingness
}\\
=
\formula{\{\triple{P}{Q}{T}\mid\exists I\in\wp(\overline{\Sigma}_{\bot})\mathrel{.}
\mathbb{I}_{\overline{\mathbb{X}}}(P)\subseteq \bar{I}
\land
\exists R^e\mathrel{.}
\textsf{\upshape post}\sqb{\texttt{\small S}}^e(\mathcal{B}\sqb{\texttt{\small B}}\cap I)\subseteq R^e\land R^e\subseteq I
\land
(\mathcal{B}\sqb{\neg\texttt{\small B}}\cap I)_{\mathbb{X}}\subseteq Q
\land
\exists R^b\mathrel{.}
\textsf{\upshape post}\sqb{\texttt{\small S}}^b (\mathcal{B}\sqb{\texttt{\small B}}\cap I)\subseteq R^b \land R^b_{\mathbb{X}}  \subseteq Q
\land
\exists R^\bot\mathrel{.}
\textsf{\upshape post}\sqb{\texttt{\small S}}^\bot(\mathcal{B}\sqb{\texttt{\small B}}\cap I)\subseteq R^\bot \land R^\bot_{\mathbb{X}}\subseteq Q  
\land
(\bot\notin Q)
\Rightarrow 
(\exists\pair{W}{\preceq}\in\mathfrak{Wf}\mathrel{.}
\exists\nu\in I\rightarrow W\mathrel{.}
\forall  \pair{\underline{\sigma}}{\sigma'}\in I\mathrel{.}\nu( \underline{\sigma})\succ\nu(\sigma'))
    \}
}\\
\explanation{choosing $I=\bar{I}$ and introducing possible approximations by $X\subseteq Y\Leftrightarrow\exists Z\mathrel{.}X\subseteq Z\land Z\subseteq Y$}\\
=
\formula{\{\triple{P}{Q}{T}\mid\exists I\in\wp(\overline{\Sigma}_{\bot})\mathrel{.}
\mathbb{I}_{\overline{\mathbb{X}}}(P)\subseteq \bar{I}
\land
\exists R,R^b\mathrel{.}
\triple{\mathcal{B}\sqb{\texttt{\small B}}\cap I_{\angelic}}{R}{R^b}\in\alpha_{\textup{HL}}(\textsf{\upshape post}^{\times}\sqb{\texttt{\small S}}_\bot)
\land R_{\angelic}\subseteq I
\land
(\mathcal{B}\sqb{\neg\texttt{\small B}}\cap I)_{\mathbb{X}}\subseteq Q
\land
R^b_{\mathbb{X}}  \subseteq Q
\land
R_{\bot_{\mathbb{X}}}\subseteq Q  
\land
(\bot\notin Q)
\Rightarrow 
(\exists\pair{W}{\preceq}\in\mathfrak{Wf}\mathrel{.}
\exists\nu\in I\rightarrow W\mathrel{.}
\forall  \pair{\underline{\sigma}}{\sigma'}\in I\mathrel{.}\nu( \underline{\sigma})\succ\nu(\sigma'))
    \}
}\\
\explanation{def.\ (\ref{eq:def:post:x}) of $\textsf{\upshape post}^{\times}$ and (\ref{eq:def:aHLopost:times}) of $\alpha_{\textup{HL}}$}\\
=
\formula{\{\triple{P}{Q}{T}\mid\exists I\in\wp(\overline{\Sigma}_{\bot})\mathrel{.}
\{\sigma\in\wp(\overline{\Sigma})\mid \sigma_{\overline{\mathbb{X}}}=\sigma_{\mathbb{X}} \wedge \sigma_{\mathbb{X}}\in P\}\subseteq {I}
\land
\exists R,R^b\mathrel{.}\{{\mathcal{B}\sqb{\texttt{\small B}}\cap I_{\angelic}}\}\,\texttt{\small S}\,\{ok:R, br:R^b\}
\land R_{\angelic}\subseteq I
\land
(\mathcal{B}\sqb{\neg\texttt{\small B}}\cap I)_{\mathbb{X}}\subseteq Q
\land
R^b_{\mathbb{X}}  \subseteq Q
\land
R_{\bot_{\mathbb{X}}}\subseteq Q  
\land
(\bot\notin Q)
\Rightarrow 
(\exists\pair{W}{\preceq}\in\mathfrak{Wf}\mathrel{.}
\exists\nu\in I\rightarrow W\mathrel{.}
\forall  \pair{\underline{\sigma}}{\sigma'}\in I\mathrel{.}\nu( \underline{\sigma})\succ\nu(\sigma'))
    \}
}\\
\explanation{def.\ $\mathbb{I}_{\overline{\mathbb{X}}}(P)$ and (\ref{eq:semantics-Hoare-logic}) of the semantics of Hoare triples}
\end{calculus}

\noindent
The computation clearly shows that $Q$ can be chosen in $\wp(\overline{\Sigma}\times\overline{\Sigma}_{\bot})$ to relate initial to final states or $\bot$, in which case we get
\begin{calculus}<2pt>
\formula{\{\triple{P}{Q}{T}\mid\exists I\in\wp(\overline{\Sigma}_{\bot})\mathrel{.}
\{\sigma\in\wp(\overline{\Sigma})\mid \sigma_{\overline{\mathbb{X}}}=\sigma_{\mathbb{X}} \wedge \sigma_{\mathbb{X}}\in P\}\subseteq {I}
\land
\exists R,T\mathrel{.}\{{\mathcal{B}\sqb{\texttt{\small B}}\cap I_{\angelic}}\}\,\texttt{\small S}\,\{ok:R, br:T\}
\land R_{\angelic}\subseteq I
\land
(\mathcal{B}\sqb{\neg\texttt{\small B}}\cap I)\subseteq Q
\land
T \subseteq Q
\land
R_{\bot}\subseteq Q  
\land
(\bot\notin Q)
\Rightarrow 
(\exists\pair{W}{\preceq}\in\mathfrak{Wf}\mathrel{.}
\exists\nu\in I\rightarrow W\mathrel{.}
\forall  \pair{\underline{\sigma}}{\sigma'}\in I\mathrel{.}\nu( \underline{\sigma})\succ\nu(\sigma'))
    \}
}\\
\explanation{$Q\in\wp(\overline{\Sigma}\times\overline{\Sigma}_{\bot})$}
\end{calculus}
\vskip1ex

\noindent
Following Sect\@. \ref{sec:SemanticsLogics}, this set $\{\triple{P}{Q}{T}\mid\triple{P}{Q}{T}\in\alpha_{\textup{HL}}(\textsf{\upshape post}^{\times}(\sqb{\texttt{\small while(B) S}}_\bot))\}$ = $\{\triple{P}{Q}{T}\mid\{P\}\,\texttt{\small S}\,\{ok:Q, br:T\}$ by (\ref{eq:semantics-Hoare-logic}) can be equivalently defined by the following deductive system.
\begin{eqntabular}[fl]{@{}c@{}}
\frac{\begin{array}[t]{@{}c@{}}
\{\sigma\in\wp(\overline{\Sigma})\mid \sigma_{\overline{\mathbb{X}}}=\sigma_{\mathbb{X}} \wedge \sigma_{\mathbb{X}}\in P\}\subseteq {I}\\[2pt]
\{{\mathcal{B}\sqb{\texttt{\small B}}\cap I_{\angelic}}\}\,\texttt{\small S}\,\{ok:R, br:T\}\\[2pt]
R_{\angelic}\subseteq I
\quad
(\mathcal{B}\sqb{\neg\texttt{\small B}}\cap I_{\angelic})\subseteq Q
\quad
T \subseteq Q
\quad
R_{\bot}\subseteq Q  
\\[2pt]
(\bot\notin Q)
\Rightarrow 
(\exists\pair{W}{\preceq}\in\mathfrak{Wf}\mathrel{.}
\exists\nu\in I\rightarrow W\mathrel{.}
\forall  \pair{\underline{\sigma}}{\sigma'}\in I\mathrel{.}\nu( \underline{\sigma})\succ\nu(\sigma'))
\end{array}}
{\{P\}\,\texttt{\small while(B) S}\,\{ok:Q, br:T\}}
\end{eqntabular}
\end{toappendix}

\subsection{Application II: Calculational Design of a New Program Logic for possible Accessibility of a Postcondition or Nontermination}\label{sec:design:possibleAccessibilityNontermination}
As a second example, we design a logic $\{P\}\,\texttt{\small S}\,\cev{\{}ok:Q, br:T\}$ for the language of Sect\@. \ref{sect:FixpointNaturalRelationalSemantics} with natural semantics (\ref{eq:def:semantics}). A quadruple  $\{P\}\,\texttt{\small S}\,\cev{\{}ok:Q, br:T\}$ means that for any state in $P$ there exists at least one execution from that state that terminates in a state of $Q$, or $T$ through a break, or
does not terminate (contrary to incorrectness logic \cite{DBLP:journals/pacmpl/OHearn20} requiring termination). For example if $Q$ is bad, $\bot\notin Q$, and $T=\emptyset$ then from any state of $P$ there must be a finite execution reaching a bad state in $Q$ (unless all executions from that state in $P$ do not terminate or $P$ is empty), which corresponds to the incorrectness component of the outcome logic \cite{DBLP:journals/pacmpl/ZilbersteinDS23} in case $\bot\notin Q$. $\{P\}\,\texttt{\small S}\,\cev{\{}ok:\{\bot\}, br:\emptyset\}$ stipulates that any initial state in $P$ can lead to at least one nonterminating execution (as opposed to the extended Hoare specification  $\{P\}\,\texttt{\small S}\,\{ok:\{\bot\}, br:\emptyset\}$ stating that no execution from $P$ can terminate). Formally,
\bgroup\abovedisplayskip0.5\abovedisplayskip\belowdisplayskip0.5\belowdisplayskip
\begin{eqntabular}{rcl}
\{P\}\,\texttt{\small S}\,\cev{\{}ok:Q, br:T\}
&\triangleq&
\triple{P}{Q}{T}\in
\alpha_{\textsf{\upshape pre}}(\sqb{\texttt{\small S}}_{\bot})
\label{eq:def:pre-logic-notation}
\end{eqntabular}\egroup
where the abstraction $\alpha_{\textsf{\upshape pre}}(\sqb{\texttt{\small S}}_{\bot})$ is
\bgroup\abovedisplayskip0.5\abovedisplayskip\belowdisplayskip0.5\belowdisplayskip
\begin{eqntabular}[fl]{Llcl@{\qquad}L}
& \rlap{$\{\triple{P}{Q}{T}\in\wp(\Sigma)\times\wp(\Sigma_{\bot})\times\wp(\Sigma)\mid
(P\subseteq\textsf{\upshape pre}({\sqb{\texttt{\small S}}^e\cup\sqb{\texttt{\small S}}^\bot})Q)
        \vee
(P\subseteq \textsf{\upshape pre}({\sqb{\texttt{\small S}}^b})T)\}$}
\label{eq:def:alpha:pre}\\
that is&\alpha_{\textsf{\upshape pre}}(\sqb{\texttt{\small S}}_{\bot})&=&
\{\triple{P}{Q}{T}\mid
P\subseteq\textsf{\upshape pre}({\sqb{\texttt{\small S}}^e\cup\sqb{\texttt{\small S}}^\bot})Q\}&(A)\label{def:alpha:pre:A:B}\\
&&&\llap{${}\cup{}$}
\{\triple{P}{Q}{T}\mid P\subseteq \textsf{\upshape pre}({\sqb{\texttt{\small S}}^b})T\}&(B)
\renumber{\textexplanation{def.\ union $\cup$}}
\end{eqntabular}\egroup
We proceed by structural induction on statements\ifshort\ (details are found in the appendix \proofinapx)\fi. We consider the iteration \texttt{\small while(B) S} which is the most difficult case. Let us start with the easy case (B) of (\ref{def:alpha:pre:A:B}) for the iteration \texttt{\small while(B) S}.
\begin{calculus}[$\Rightarrow$~]
\formula{\{\triple{P}{Q}{T}\mid P\subseteq \textsf{\upshape pre}({\sqb{\texttt{\small while(B) S}}^b})T\}}\\[-0.5ex]
=
\formulaexplanation{\{\triple{P}{Q}{T}\mid P\subseteq \textsf{\upshape pre}(\emptyset)T\}=\{\triple{P}{Q}{T}\mid P\subseteq \emptyset\}}{def. (\ref{eq:natural-break}) of ${\sqb{\texttt{\small while(B) S}}^b}$ and (\ref{eq:def:psott-pre-pret}) of \textsf{\upshape pre}}\\[-0.5ex]
=
\formulaexplanation{\{\triple{\emptyset}{Q}{T}\mid Q\in\wp(\Sigma_{\bot})\wedge T\in\wp(\Sigma)\}}{def.\ inclusion $\subseteq$ and empty set $\emptyset$}
\end{calculus}
\hskip1em\uLstrut Following Sect\@. \ref{sec:SemanticsLogics},  (A) $\cup$ (B) can be defined by a deductive system with separate rules for (A) and (B). With notation (\ref{eq:def:pre-logic-notation}), the deductive system for (B) of (\ref{def:alpha:pre:A:B}) for the iteration \texttt{\small while(B) S} is the axiom
\bgroup\abovedisplayskip-2\abovedisplayskip\belowdisplayskip0.5\belowdisplayskip\begin{eqntabular}{c}
\frac{\emptyset}{\{\emptyset\}\,\texttt{\small while(B) S}\,\cev{\{}ok:Q, br:T\}}
\end{eqntabular}\egroup
meaning that if you never execute a program you can conclude anything on its executions. This is also
valid in Hoare logic but is not given as an explicit axiom since it can be derived from other rules (by an extensive induction on all program statements).

We now have to consider case (A) of (\ref{def:alpha:pre:A:B}) for the iteration \texttt{\small while(B) S}. This is more difficult and requires three pages of calculation plus two pages of auxiliary propositions, too much for most readers. So we will \ifshort\else first \fi sketch the main steps of this calculation for a global understanding\ifshort\ and refer to Sect\@. \ref{apx:sec:design:possibleAccessibilityNontermination} of the appendix for all further details\fi.

Starting from $\{\triple{P}{Q}{T}\mid
P\subseteq\textsf{\upshape pre}({\sqb{\texttt{\small while(B) S}}^e\cup\sqb{\texttt{\small while(B) S}}^\bot})Q\}$
we get
\bgroup\abovedisplayskip0.75\abovedisplayskip\belowdisplayskip0.75\belowdisplayskip
\begin{eqntabular*}{c}
\{\triple{P^e\cup P^{\bot}}{Q}{T}\mid\begin{array}[t]{@{}l@{}}
P^e\subseteq(\textsf{\upshape pre}(\Lfp{\subseteq}{F^e}\fatsemi\sqb{\neg\texttt{\small B}})Q_{\angelic}
\cup\textsf{\upshape pre}(\Lfp{\subseteq}{F^e}\fatsemi\sqb{\texttt{\small B}}\fatsemi\sqb{\texttt{\small S}}^b)Q_{\angelic})
\wedge{}\\
P^{\bot}\subseteq(\textsf{\upshape pre}(\Lfp{\subseteq}{F^e}\fatsemi\sqb{\texttt{\small B}}\fatsemi\sqb{\texttt{\small S}}^\bot)\{\bot\mid \bot\in Q\}\cup
\textsf{\upshape pre}(\Gfp{\subseteq}{F^\bot})\{\bot\mid \bot\in Q\}))\}\end{array}
\end{eqntabular*}\egroup
\noindent
by case analysis and expanding definitions (\ref{eq:natural-finite}) and (\ref{eq:natural-oo}) of the semantics. For states in $P^e$ there is an execution terminating in $Q_{\angelic}\triangleq Q\setminus\{\bot\}$, possibly by a break. For states in $P^{\bot}_{b}$ there is an infinite execution consisting of zero or more finite iterations followed by a nonterminating execution of the loop body. Finally, for states in  $P^{\bot}_{\ell}$ there is an execution consisting of infinitely many finite iterations. These cases are not exclusive and might be empty.

Grouping the cases $P=P^e\cup P^{\bot}_{b}$, we get (by case analysis and $\textsf{\upshape pre}(r)\emptyset=\emptyset$)
\bgroup\abovedisplayskip0.75\abovedisplayskip\belowdisplayskip0.75\belowdisplayskip
\begin{eqntabular*}{l}
\{\triple{P\cup P^{\bot}_{\ell}}{Q}{T}\mid\begin{array}[t]{@{}l@{}}
P\subseteq\Lfp{\subseteq}\LAMBDA{X}(\textsf{\upshape pre}\sqb{\neg\texttt{\small B}}Q_{\angelic})\cup
(\textsf{\upshape pre}(\sqb{\texttt{\small B}}\fatsemi\sqb{\texttt{\small S}}^b)Q_{\angelic})
\cup{}(\textsf{\upshape pre}(\sqb{\texttt{\small B}}\fatsemi\sqb{\texttt{\small S}}^\bot)\{\bot\mid \bot\in Q\})\cup{}\\(\textsf{\upshape pre}(\sqb{\texttt{\small B}}\fatsemi\sqb{\texttt{\small S}}^e)X)
\wedge
\si \bot\in Q \alors P^{\bot}_{\ell}\subseteq\textsf{\upshape pre}(\Gfp{\subseteq}{F^\bot})\{\bot\}\sinon P^{\bot}_{\ell}=\emptyset\fsi\}\end{array}
\end{eqntabular*}\egroup
\noindent We must now handle fixpoint under approximations using induction principles. Since the \textsf{pre} transformer preserves arbitrary joins, its least fixpoint iterations are stable at $\omega$. 
So the hypotheses of Th\@. \ref{th:Fixpoint-Underapproximation} for $P\subseteq\Lfp{\subseteq}\LAMBDA{X}A\cup\textsf{pre}(r)X$  become
\bgroup\abovedisplayskip0.75\abovedisplayskip\belowdisplayskip0.5\belowdisplayskip\begin{eqntabular}{c}
\exists\pair{I^n}{n\in\mathbb{N}}\mathrel{.}I^0=\emptyset\wedge \forall n\in\mathbb{N}\mathrel{.}I^n\subseteq I^{n+1}\subseteq 
 A\cup\textsf{pre}(r)I^n\wedge\exists\ell\in\mathbb{N}\mathrel{.} P\subseteq I^\ell
 \label{cal:Fixpoint-Underapproximation-pre} 
\end{eqntabular}\egroup
Since \textsf{pre} does not preserve meets, we cannot use the dual of the fixpoint abstraction Th\@. \ref{th:fixpoint-abstraction} to express $\textsf{\upshape pre}(\Gfp{\subseteq}{F^\bot})\{\bot\mid \bot\in Q\}$ in fixpoint form and then use the dual of fixpoint induction Th\@. \ref{th:OverapproximationLeastFixpointImage}.
This is where the order dual Th\@. \ref{th:OverapproximationLeastFixpointImage} is useful to under approximate the image of the greatest fixpoint $\Gfp{\subseteq}{F^\bot}$ by $\LAMBDA{r}\textsf{\upshape pre}(r)\{\bot\mid \bot\in Q\}$. The dual hypotheses of Th\@. \ref{th:OverapproximationLeastFixpointImage} are that there exists $J\in\wp(\Sigma_{\bot})$ such that, after simplifications for that particular case, are
\bgroup\abovedisplayskip0.5\abovedisplayskip\belowdisplayskip0.5\belowdisplayskip\begin{eqntabular}{c}
\exists J\in\wp(\Sigma_\bot)\mathrel{.}
\si \bot\in Q\alors
\textsf{pre}(\sqb{\texttt{\small B}}\fatsemi\sqb{\texttt{\small S}}^e)(J)\subseteq J\wedge P^{\bot}_{\ell}\subseteq J
\sinon\textsf{true}\fsi
\label{cal:Underapproximation-gfp-pre}
\end{eqntabular}\egroup
so that, by (\ref{cal:Fixpoint-Underapproximation-pre}), (\ref{cal:Underapproximation-gfp-pre}), and $\textsf{\upshape pre}(\sqb{\texttt{\small B}})R
=
\mathcal{B}\sqb{\texttt{\small B}}\cap R$, we get
\bgroup\abovedisplayskip0pt\belowdisplayskip3pt
\begin{eqntabular*}{l}
\{\triple{P\cup P^{\bot}_{\ell}}{Q}{T}\mid
\exists\pair{I^n}{n\in\mathbb{Q}}\mathrel{.}
 I^0=\emptyset\wedge
\forall n\in\mathbb{N}\mathrel{.}I^n\subseteq I^{n+1}\subseteq
(\mathcal{B}\sqb{\neg\texttt{\small B}}\cap Q_{\angelic})
\cup{}
(\mathcal{B}\sqb{\texttt{\small B}}\cap\textsf{\upshape pre}(\sqb{\texttt{\small S}}^b)Q_{\angelic}){}\\
{}\cup{}
(\mathcal{B}\sqb{\texttt{\small B}}\cap \textsf{\upshape pre}\sqb{\texttt{\small S}}^\bot\{\bot\mid \bot\in Q\})\cup{} (\sqb{\texttt{\small B}}\cap \textsf{\upshape pre}\sqb{\texttt{\small S}}^e(I^n))\wedge\exists\ell\in\mathbb{N}\mathrel{.} P\subseteq I^\ell
\wedge
\exists J\in\wp(\Sigma_\bot)\mathrel{.}
\si \bot\in Q\alors{}\\
\mathcal{B}\sqb{\texttt{\small B}}\cap\textsf{pre}(\sqb{\texttt{\small S}}^e)(J)\subseteq J\wedge P^{\bot}_{\ell}\subseteq J
\sinon P^{\bot}_{\ell}=\emptyset\fsi\}
\end{eqntabular*}\egroup
\noindent
Since we proceed by structural induction, we have, by definition (\ref{eq:def:pre-logic-notation}), to make the under approximations of $\textsf{\upshape pre}\sqb{\texttt{\small S}}$ for the loop body \texttt{\small S} to appear explicitly in the calculations, as follows
\begin{calculus}
\formula{\{\triple{P\cup P^{\bot}_{\ell}}{Q}{T}\mid
\exists R^b\mathrel{.}R^b\subseteq\textsf{\upshape pre}(\sqb{\texttt{\small S}}^b)Q_{\angelic}
\wedge
\exists R^\bot\mathrel{.}\triple{R^\bot}{\{\bot\mid \bot\in Q\}}{\emptyset}\in
\alpha_{\textsf{\upshape pre}}(\sqb{\texttt{\small S}}_{\bot})
\wedge
\exists\pair{I^n}{n\in\mathbb{Q}}\mathrel{.}
 I^0=\emptyset\wedge
\forall n\in\mathbb{N}\mathrel{.}
\exists R^e_n\mathrel{.}\triple{R^e_n}{I^n}{\emptyset}\in\alpha_{\textsf{\upshape pre}}(\sqb{\texttt{\small S}}_{\bot})\wedge
I^n\subseteq I^{n+1}\subseteq
(\mathcal{B}\sqb{\neg\texttt{\small B}}\cap Q_{\angelic})
\cup(\sqb{\texttt{\small B}}\cap R^b)\cup(\sqb{\texttt{\small B}}\cap R^\bot)\cup (\sqb{\texttt{\small B}}\cap R^e_n)\wedge\exists\ell\in\mathbb{N}\mathrel{.} P\subseteq I^\ell
\wedge
\exists J\in\wp(\Sigma_\bot)\mathrel{.}
\exists R^{\bot}_{\ell}\mathrel{.}\triple{R^{\bot}_{\ell}}{J}{\emptyset}\in\alpha_{\textsf{\upshape pre}}(\sqb{\texttt{\small S}})\wedge
\si \bot\in Q\alors
\mathcal{B}\sqb{\texttt{\small B}}\cap R^{\bot}_{\ell}\subseteq J\wedge P^{\bot}_{\ell}\subseteq J
\sinon P^{\bot}_{\ell}=\emptyset\fsi\}}
\end{calculus}
\noindent\uLstrut Following Sect\@. \ref{sec:SemanticsLogics} and using the notation (\ref{eq:def:pre-logic-notation}), the theory $\alpha_{\textsf{\upshape pre}}(\sqb{\texttt{\small while(B) S}}_{\bot})$ can be equivalently  defined by the following deductive system of the logic.
\bgroup\abovedisplayskip2pt\belowdisplayskip-7pt
\begin{eqntabular}[fl]{@{}c@{}}
\frac{\begin{array}{@{}c@{}}
I^0=\emptyset\quad%
\{R^b\}\,\texttt{\small S}\,\cev{\{}ok:\emptyset, br:Q_{\angelic}\}\quad%
\{R^\bot\}\,\texttt{\small S}\,\cev{\{}ok:\{\bot\mid \bot\in Q\}, br:\emptyset\}\\
\begin{array}[t]{@{}l@{}l@{}}
\forall n\in\mathbb{N}\mathrel{.}{}&\{R^e_n\}\,\sqb{\texttt{\small S}}\,\cev{\{}ok:I^n, br:
\emptyset\}\\[0.25ex]
& I^n\subseteq I^{n+1}\subseteq \,
(\mathcal{B}\sqb{\neg\texttt{\small B}}\cap Q_{\angelic})\cup(\mathcal{B}\sqb{\texttt{\small B}}\cap  R^b)\cup(\mathcal{B}\sqb{\texttt{\small B}}\cap  R^\bot)\cup (\mathcal{B}\sqb{\texttt{\small B}}\cap R^e_n)
\quad\exists\ell\in\mathbb{N}\mathrel{.} P\subseteq I^\ell
\end{array}\\[15pt]
\si \bot\in Q \alors \{R^{\bot}_{\ell}\}\,\sqb{\texttt{\small S}}\,\cev{\{}ok:J, br:\emptyset\}\wedge \mathcal{B}\sqb{\texttt{\small B}}\cap R^{\bot}_{\ell}\subseteq J\wedge P^{\bot}_{\ell}\subseteq J\sinon  P^{\bot}_{\ell}=\emptyset\fsi\\[0.25ex]
\end{array}
}
{
\mbox{\ustrut}\{P\cup P^{\bot}_{\ell}\}\,{\texttt{\small while(B) S}}\,\cev{\{}ok:Q, br:\emptyset\}
}\label{eq:deductive:system:pre:while}\addtocounter{equation}{1}\renumber{\raisebox{-2ex}[0pt][0pt]{(\ref{eq:deductive:system:pre:while})}}
\end{eqntabular}\egroup
\ifshort\else
Notice that termination is proved using naturals in a backward reasoning whereas ordinals would be necessary in a forward classic reasoning à la Turing/Floyd with unbounded nondeterminism. Two separate rules can also be provided for the  cases $\bot\in Q$ and $\bot\notin Q$. This was also a possibility for the extended Hoare logic.
\fi
\vspace*{-1pt}
\begin{example}\label{ex:fact:pre}Continuing  Ex\@. \ref{ex:fact:spec} and \ref{ex:fact:pre:spec}, consider the factorial with postcondition contract $\texttt{\small f}>0$. An interval analysis produces an alarm $Q=Q_{\angelic}=\texttt{\small f}\leqslant0$ where $\bot\notin Q$ so $Q_{\bot}=\emptyset$ and $P^{\bot}_{\ell}=\emptyset$. Take $R^{\bot}=R^{b}=\emptyset$ since the loop body terminates with no break. Let $I^k = n\leqslant k\wedge f\leqslant 0$ and $R^e_k=I^{k-1}$ so that $\{R^e_k\}\,\sqb{\texttt{\small f = f*n; n = n-1;}}\,\cev{\{}ok:I^k, br:\emptyset\}$. Take $P=I^{|\underline{n}|}$. By (\ref{eq:deductive:system:pre:while}), $\{P\}\,{\texttt{\small fact}}\,\cev{\{}ok:Q, br:\emptyset\}$. But $P$ implies $\texttt{\small f}\leqslant 0$ in contradiction $\{\emptyset\}\,{\texttt{\small f=1;}}\,\cev{\{}ok:P, br:\emptyset\}$  with the initialization \texttt{\small f=1} proving the unreachable alarm to be false, which \cite{DBLP:journals/pacmpl/OHearn20,DBLP:journals/pacmpl/ZilbersteinDS23} cannot do.
\end{example}
\begin{toappendix}
\label{apx:sec:design:possibleAccessibilityNontermination}
\subsection{Auxiliary Propositions}
\noindent We will use the following auxiliary lemmas.
\bgroup\arraycolsep0.5\arraycolsep\begin{eqntabular}[fl]{L@{\qquad}rcl@{\quad and\quad}rcl}
\labelitemi\quad join preservation
&
{\textsf{\upshape pre}}(\bigcup_{i\in\Delta} \tau_i)Q 
&=& 
\bigcup_{i\in\Delta}{\textsf{\upshape pre}}(\tau_i)Q 
&
{\textsf{\upshape pre}}(\tau)\bigcup_{i\in\Delta} Q _i
&=&
\bigcup_{i\in\Delta} {\textsf{\upshape pre}}(\tau)Q _i
 \label{eq:pre:join-preservation}
\end{eqntabular}\egroup
\begin{proof}[proof of (\ref{eq:pre:join-preservation})] By the Galois connections $\pair{\wp(X\times Y)}{\subseteq}\galois{\textsf{\upshape pre}}{{\textsf{\upshape pre}^{-1}}}\pair{\wp(X)\stackrel{\sqcupdot}{\longrightarrow}\wp(Y)}{\stackrel{.}{\subseteq}}$ and
$\pair{\wp(X)}{\subseteq}\galois{\textsf{\upshape pre}(\tau)}{\widetilde{\textsf{\upshape post}}(\tau)}\pair{\wp(Y)}{\subseteq}$
where the lower adjoint preserves arbitrary joins.
\end{proof}
\begin{eqntabular}[fl]{L@{\qquad}rcl}
\labelitemi\quad composition
&
\textsf{\upshape pre}(r_1\fatsemi r_2)Q 
&=& 
\textsf{\upshape pre}(r_1)\comp \textsf{\upshape pre}(r_2)Q 
 \label{eq:pre:composition}
\end{eqntabular}
\begin{proof}[proof of (\ref{eq:pre:composition})] 
\begin{calculus}[$\Rightarrow$~]<2pt>
\formula{\textsf{\upshape pre}(r_1\fatsemi r_2)Q }\\
=
\formulaexplanation{\{\sigma\in X\mid\exists\sigma'\in Q\mathrel{.}\pair{\sigma}{\sigma'}\in r_1\fatsemi r_2\}}{def.\ (\ref{eq:def:psott-pre-pret}) of \textsf{\upshape pre}}\\
=
\formulaexplanation{\{\sigma\in X\mid\exists\sigma'\in Q\mathrel{.}\exists \sigma''\mathrel{.}\pair{\sigma}{\sigma''}\in r_1
\wedge  \pair{\sigma''}{\sigma'}\in r_2\}}{def.\ relation composition $\fatsemi$}\\
=
\formulaexplanation{\{\sigma\in X\mid\exists \sigma''\mathrel{.}\pair{\sigma}{\sigma''}\in r_1
\wedge \exists\sigma'\in Q\mathrel{.} \pair{\sigma''}{\sigma'}\in r_2\}}{associativity of disjunction}\\
=
\formulaexplanation{\{\sigma\in X\mid\exists \sigma''\mathrel{.}\pair{\sigma}{\sigma''}\in r_1
\wedge \sigma''\in\{\sigma''\mid\exists\sigma'\in Q\mathrel{.} \pair{\sigma''}{\sigma'}\in r_2\}\}}{def.\ $\in$}\\
=
\formulaexplanation{\{\sigma\in X\mid\exists \sigma''\mathrel{.}\pair{\sigma}{\sigma''}\in r_1
\wedge \sigma''\in\textsf{\upshape pre}(r_2)Q\}}{def.\ (\ref{eq:def:psott-pre-pret}) of \textsf{\upshape pre}}\\
=
\formulaexplanation{\textsf{\upshape pre}(r_1)(\textsf{\upshape pre}(r_2)Q)}{def.\ (\ref{eq:def:psott-pre-pret}) of \textsf{\upshape pre}}\\
=
\formulaexplanation{\textsf{\upshape pre}(r_1)\comp\textsf{\upshape pre}(r_2)Q}{def.\ function composition $\comp$}
\end{calculus}
\end{proof}
\begin{eqntabular}[fl]{L@{\qquad}rcl}
\labelitemi\quad partitioning
&
\textsf{\upshape pre}(r)Q 
&=& 
\textsf{\upshape pre}(r\cap(\Sigma\times\Sigma))Q_{\angelic} 
\cup
\textsf{\upshape pre}(r\cap(\Sigma\times\{\bot\}))\{\bot\mid\bot\in Q\}
 \label{eq:pre:partitioning}
\end{eqntabular}
\begin{proof}[proof of (\ref{eq:pre:partitioning})] 
\begin{calculus}[$\Rightarrow$~]<2pt>
\formula{\textsf{\upshape pre}(r)Q}\\
=
\formulaexplanation{\{\sigma\mid\exists \sigma'\mathrel{.} \pair{\sigma}{\sigma'}\in r \wedge \sigma'\in Q\}}{def.\ (\ref{eq:def:psott-pre-pret}) of \textsf{\upshape pre}}\\
=
\formula{\{\sigma\mid\exists \sigma'\mathrel{.} \pair{\sigma}{\sigma'}\in (r\cap(\Sigma\times\Sigma)\cup r\cap(\Sigma\times\{\bot\})) \wedge \sigma'\in Q\}}\\
\explanation{$r\in\wp(\Sigma\times\Sigma_{\bot})$, $\Sigma_{\bot}=\Sigma\cup\{\bot\}$ and $\bot\notin \Sigma$}\\
=
\formulaexplanation{\{\sigma\mid\exists \sigma'\mathrel{.} (\pair{\sigma}{\sigma'}\in r\cap(\Sigma\times\Sigma)
\lor
\pair{\sigma}{\sigma'}\in  r\cap(\Sigma\times\{\bot\})) \wedge \sigma'\in Q\}}{def.\ $\cup$}\\
=
\formulaexplanation{\{\sigma\mid\exists \sigma'\mathrel{.} (\pair{\sigma}{\sigma'}\in r\cap(\Sigma\times\Sigma)\wedge \sigma'\in Q)
\lor
(\pair{\sigma}{\sigma'}\in  r\cap(\Sigma\times\{\bot\})\wedge \sigma'\in Q) \}}{distributivity}\\
=
\formulaexplanation{\{\sigma\mid (\exists \sigma'\mathrel{.}\pair{\sigma}{\sigma'}\in r\cap(\Sigma\times\Sigma)\wedge \sigma'\in Q)
\lor
(\exists \sigma'\mathrel{.}\pair{\sigma}{\sigma'}\in  r\cap(\Sigma\times\{\bot\})\wedge \sigma'\in Q) \} }{associativity}\\
=
\formulaexplanation{\{\sigma\mid \exists \sigma'\mathrel{.}\pair{\sigma}{\sigma'}\in r\cap(\Sigma\times\Sigma)\wedge \sigma'\in Q\}
\cup
\{\sigma\mid\exists \sigma'\mathrel{.}\pair{\sigma}{\sigma'}\in  r\cap(\Sigma\times\{\bot\})\wedge \sigma'\in Q \} }{def.\ $\cup$}\\
=
\formula{\{\sigma\mid \exists \sigma'\mathrel{.}\pair{\sigma}{\sigma'}\in r\cap(\Sigma\times\Sigma)\wedge \sigma'\in Q_{\angelic}\}
\cup
\{\sigma\mid\exists \sigma'\mathrel{.}\pair{\sigma}{\sigma'}\in  r\cap(\Sigma\times\{\bot\})\wedge \sigma'\in \{\bot\mid\bot\in Q\} \} }\\
\explanation{for the first term, $\sigma'\in\Sigma$ and $\sigma'\in Q$ so $\sigma'\in Q\cap\Sigma\triangleq Q_{\angelic}$, while, for the second term,  ${\sigma'}\in \{\bot\}$ and $\sigma'\in Q$ so $\sigma'\in Q\cap\{\bot\}= \{\bot\mid\bot\in Q\}$}\\
=
\lastformulaexplanation{\textsf{\upshape pre}(r\cap(\Sigma\times\Sigma))Q_{\angelic} 
\cup
\textsf{\upshape pre}(r\cap(\Sigma\times\{\bot\}))\{\bot\mid\bot\in Q\}}{def.\ (\ref{eq:def:psott-pre-pret}) of \textsf{\upshape pre}}{\mbox{\qed}}
\end{calculus}\let\qed\relax
\end{proof}
\begin{eqntabular}[fl]{L@{\qquad}rcl}
\labelitemi\quad join covering
&
P\subseteq A\cup B&
\Leftrightarrow& 
\exists P_A, P_B\mathrel{.}P=P_A\cup P_B\wedge P_A\subseteq A \wedge P_B\subseteq B
 \label{eq:join-covering}
\end{eqntabular}
\begin{proof}[proof of (\ref{eq:join-covering})] 
\begin{calculus}[($\Rightarrow$)~]<2pt>
($\Rightarrow$)\discussion{take $P_A=P\cap A$ and $P_B=P\cap B$ so that $P_A\subseteq P\cap A \subseteq A$ and $P_B\subseteq P\cap B\subseteq B$; }\\
($\Leftarrow$)\lastdiscussion{we have $P=P_A\cup P_B\subseteq A\cup B$.}{\mbox{\qed}}
\end{calculus}\let\qed\relax
\end{proof}
\begin{eqntabular}[fl]{L@{\qquad}rcl}
\labelitemi\quad Boolean precondition
&
\textsf{\upshape pre}(\sqb{\texttt{\small B}})Q
=&
\mathcal{B}\sqb{\texttt{\small B}}\cap Q
 \label{eq:Boolean-precondition}
\end{eqntabular}
\begin{proof}[proof of (\ref{eq:Boolean-precondition})] 
\begin{calculus}[($\Rightarrow$)~]<2pt>
\formula{\textsf{\upshape pre}(\sqb{\texttt{\small B}})Q}\\
=
\formulaexplanation{\{\sigma\mid\exists\sigma'\in Q\mathrel{.}\pair{\sigma}{\sigma'}\in\sqb{\texttt{\small B}}\}}{def.\ (\ref{eq:def:psott-pre-pret}) of \textsf{\upshape pre}}\\
=
\formulaexplanation{\{\sigma\mid\exists\sigma'\in Q\mathrel{.}\sigma=\sigma'\in\mathcal{B}\sqb{\texttt{\small B}}\}}{def.\  $\sqb{\texttt{\small B}}\triangleq\{\pair{\sigma}{\sigma}\mid\sigma\in\mathcal{B}\sqb{\texttt{\small B}}\}$}\\
=
\lastformulaexplanation{Q\cap \mathcal{B}\sqb{\texttt{\small B}}}{def.\ intersection $\cap$}{\mbox{\qed}}
\end{calculus}\let\qed\relax
\end{proof}
\begin{eqntabular}[fl]{L@{\qquad}L}
\labelitemi\quad $F^e$ commutation
&
$F^e$ and 
$F^e_{\textsf{\upshape pre}}\triangleq\LAMBDA{X}\LAMBDA{Q}Q\cup (\mathcal{B}\sqb{\texttt{\small B}}\cap  \textsf{\upshape pre}\sqb{\texttt{\small S}}^e(X(Q)))$
commute \label{eq:F-e-commutation}\\
& for $\alpha(X)=\LAMBDA{Q}\textsf{\upshape pre}(X)Q$ \nonumber
\end{eqntabular}
\begin{proof}[proof of (\ref{eq:F-e-commutation})] 
\begin{calculus}[($\Rightarrow$)~]<2pt>
\formula{\alpha(F^e(X))}\\
=
\formulaexplanation{\textsf{\upshape pre}(F^e(X))}{def.\ $\alpha$}\\
=
\formulaexplanation{\textsf{\upshape pre}(\textsf{\upshape id} \cup (\sqb{\texttt{\small B}}\fatsemi\sqb{\texttt{\small S}}^e\fatsemi X))}{def.\ (\ref{eq:natural-transformer-finite}) of $F^e$}\\
=
\formulaexplanation{\textsf{\upshape pre}(\textsf{\upshape id}) \stackrel{.}{\cup} \textsf{\upshape pre}(\sqb{\texttt{\small B}}\fatsemi\sqb{\texttt{\small S}}^e\fatsemi X))}{join preservation (\ref{eq:pre:join-preservation})}\\
=
\formulaexplanation{\LAMBDA{Q}\textsf{\upshape pre}(\textsf{\upshape id})Q \stackrel{.}{\cup} \textsf{\upshape pre}(\sqb{\texttt{\small B}}\fatsemi\sqb{\texttt{\small S}}^e\fatsemi X)Q}{def.\ \text{\boldmath$\lambda$} notation}\\
=
\formulaexplanation{\LAMBDA{Q}Q\cup \textsf{\upshape pre}\sqb{\texttt{\small B}}( \textsf{\upshape pre}\sqb{\texttt{\small S}}^e(\textsf{\upshape pre}(X)Q))}{def.\ (\ref{eq:def:psott-pre-pret}) of \textsf{\upshape pre} and composition (\ref{eq:pre:composition})}\\
=
\formulaexplanation{\LAMBDA{Q}Q\cup (\mathcal{B}\sqb{\texttt{\small B}}\cap  \textsf{\upshape pre}\sqb{\texttt{\small S}}^e(\textsf{\upshape pre}(X)Q))}
{Boolean precondition (\ref{eq:Boolean-precondition})}\\
=
\formulaexplanation{\LAMBDA{Q}Q\cup (\mathcal{B}\sqb{\texttt{\small B}}\cap  \textsf{\upshape pre}\sqb{\texttt{\small S}}^e(\alpha(X)Q))}{def.\ $\alpha=\textsf{\upshape pre}(X)$}\\
=
\lastformulaexplanation{F^e_{\textsf{\upshape pre}}(\alpha(X))}{def.\ (\ref{eq:F-e-commutation}) of $F^e_{\textsf{\upshape pre}}$}{\mbox{\qed}}
\end{calculus}\let\qed\relax
\end{proof}
\begin{calculus}[\labelitemi\quad]
{\labelitemi\quad}\numberedleftdiscussion{Observe that $F^e_{\textsf{\upshape pre}}(X)Q=\bar{F}^e_Q(X(Q))$ by defining $\bar{F}^e_Q\triangleq\LAMBDA{X}Q\cup (\mathcal{B}\sqb{\texttt{\small B}}\cap  \textsf{\upshape pre}\sqb{\texttt{\small S}}^e(X))$}\label{Fepre-FbareQ}
\end{calculus}
\smallskip

\begin{eqntabular}[fl]{L@{\qquad}L}
\labelitemi\quad $F^\bot$ commutation
&
$F^\bot$ and 
$F^\bot_{\textsf{\upshape pre}}\triangleq\LAMBDA{X}\LAMBDA{Q}\mathcal{B}\sqb{\texttt{\small B}}\cap(\textsf{\upshape pre}\sqb{\texttt{\small S}}^e(X(Q))$
commute for\label{eq:F-bot-commutation}\\
& $\alpha(X)=\LAMBDA{Q}\textsf{\upshape pre}(X)Q$ \nonumber
\end{eqntabular}
\begin{proof}[proof of (\ref{eq:F-bot-commutation})] 
\begin{calculus}[($\Rightarrow$)~]<2pt>
\formula{\alpha(F^\bot(X))}\\
=
\formulaexplanation{\textsf{\upshape pre}(F^\bot(X))}{def.\ $\alpha$}\\
=
\formulaexplanation{\textsf{\upshape pre}(\sqb{\texttt{\small B}}\fatsemi\sqb{\texttt{\small S}}^e\fatsemi X)}{def.\ (\ref{eq:natural-transformer-infinite}) of $F^\bot$}\\
=
\formulaexplanation{\textsf{\upshape pre}\sqb{\texttt{\small B}}\comp\textsf{\upshape pre}\sqb{\texttt{\small S}}^e\comp\textsf{\upshape pre}(X)}{composition (\ref{eq:pre:composition})}\\
=
\formulaexplanation{\LAMBDA{Q}\textsf{\upshape pre}\sqb{\texttt{\small B}}\comp\textsf{\upshape pre}\sqb{\texttt{\small S}}^e\comp\textsf{\upshape pre}(X)(Q)}{def.\ \text{\boldmath$\lambda$} notation}\\
=
\formulaexplanation{\LAMBDA{Q}\textsf{\upshape pre}\sqb{\texttt{\small B}}(\textsf{\upshape pre}\sqb{\texttt{\small S}}^e(\textsf{\upshape pre}(X)Q)}{def.\ function composition $\comp$}\\
=
\formulaexplanation{ \LAMBDA{Q}\mathcal{B}\sqb{\texttt{\small B}}\cap(\textsf{\upshape pre}\sqb{\texttt{\small S}}^e(\textsf{\upshape pre}(X)Q))}{composition (\ref{eq:pre:composition})}\\
=
\formulaexplanation{ \LAMBDA{Q}\mathcal{B}\sqb{\texttt{\small B}}\cap(\textsf{\upshape pre}\sqb{\texttt{\small S}}^e(\alpha(X)Q))}{def.\ $\alpha=\textsf{\upshape pre}(X)$}\\
=
\lastformulaexplanation{F^\bot_{\textsf{\upshape pre}}(\alpha(X))}{def.\ (\ref{eq:F-bot-commutation}) of $F^\bot_{\textsf{\upshape pre}}$}{\mbox{\qed}}
\end{calculus}\let\qed\relax
\end{proof}
\smallskip

\noindent \labelitemi\quad Let us now calculate $\alpha_{\textsf{\upshape pre}}(\sqb{\texttt{\small S}}_{\bot})$. We consider the case of the iteration \texttt{\small while(B) S} (which covers the conditional and sequential composition). Let us start with the easier case (B).
\begin{calculus}[$\Rightarrow$~]<2pt>
\formula{\{\triple{P}{Q}{T}\mid P\subseteq \textsf{\upshape pre}({\sqb{\texttt{\small while(B) S}}^b})T\}}\\
=
\formulaexplanation{\{\triple{P}{Q}{T}\mid P\subseteq \textsf{\upshape pre}(\emptyset)T\}}{def.\ (\ref{eq:natural-break}) of the natural relational semantics}\\
=
\formulaexplanation{\{\triple{P}{Q}{T}\mid P\subseteq \emptyset\}}{\textsf{\upshape pre} preserves arbitrary joins (\ref{eq:pre:join-preservation}), hence $\emptyset$}\\
=
\formulaexplanation{\{\triple{\emptyset}{Q}{T}\mid Q\in\wp(\Sigma_{\bot})\wedge T\in\wp(\Sigma)\}}{def.\ $\subseteq$}\\[-1em]
\end{calculus}
which we can interpret as ``if you believe that a loop breaks out of an outer loop, which is impossible, then you can state anything on the behavior of this outer loop''.

\subsection{Design of the Deductive System for Tthe iteration \texttt{\small while(B) S} in Case (A) of (\ref{def:alpha:pre:A:B})}
\begin{calculus}[$\Rightarrow$~]<2pt>
\formula{\{\triple{P}{Q}{T}\mid
P\subseteq\textsf{\upshape pre}({\sqb{\texttt{\small while(B) S}}^e\cup\sqb{\texttt{\small while(B) S}}^\bot})Q\}}\\
=
\formulaexplanation{\{\triple{P}{Q}{T}\mid
P\subseteq(\textsf{\upshape pre}({\sqb{\texttt{\small while(B) S}}^e)Q\cup\textsf{\upshape pre}(\sqb{\texttt{\small while(B) S}}^\bot})Q)\}}{join preservation (\ref{eq:pre:join-preservation})}\\
=
\formula{\{\triple{P}{Q}{T}\mid
P\subseteq(\textsf{\upshape pre}(\sqb{\texttt{\small while(B) S}}^e)Q_{\angelic}\cup\textsf{\upshape pre}(\sqb{\texttt{\small while(B) S}}^\bot)\{\bot\mid \bot\in Q\})\}}\\
\explanation{partitioning (\ref{eq:pre:partitioning}) with $\sqb{\texttt{\small while(B) S}}^e\in\wp(\Sigma\times\Sigma)$ and $\sqb{\texttt{\small while(B) S}}^{\bot}\in\wp(\Sigma\times\{\bot\})$}\\
=
\formula{\{\triple{P^e\cup P^{\bot}}{Q}{T}\mid
P^e\subseteq\textsf{\upshape pre}(\sqb{\texttt{\small while(B) S}}^e)Q_{\angelic}
\wedge
P^{\bot}\subseteq\textsf{\upshape pre}(\sqb{\texttt{\small while(B) S}}^\bot)\{\bot\mid \bot\in Q\})\}}\\
\explanation{By join covering (\ref{eq:join-covering}) of $P$ by the initial states $P^e$ from which executions may terminate in $Q_{\angelic}$ and those $P^{\bot}$with possible nonterminating executions.}\\
=
\formula{\{\triple{P^e\cup P^{\bot}}{Q}{T}\mid
P^e\subseteq\textsf{\upshape pre}(\Lfp{\subseteq}{F^e}\fatsemi\sqb{\neg\texttt{\small B}}\cup\sqb{\texttt{\small B}}\fatsemi\sqb{\texttt{\small S}}^b)Q_{\angelic}
\wedge
P^{\bot}\subseteq\textsf{\upshape pre}((\Lfp{\subseteq}{F^e}\fatsemi\sqb{\texttt{\small B}}\fatsemi\sqb{\texttt{\small S}}^\bot)\cup\Gfp{\subseteq}{F^\bot}
)\{\bot\mid \bot\in Q\})\}}\\
\explanation{by definitions (\ref{eq:natural-finite}) and (\ref{eq:natural-oo})}\\
=
\formulaexplanation{\{\triple{P^e\cup P^{\bot}}{Q}{T}\mid
P^e\subseteq(\textsf{\upshape pre}(\Lfp{\subseteq}{F^e}\fatsemi\sqb{\neg\texttt{\small B}})Q_{\angelic}
\cup\textsf{\upshape pre}(\Lfp{\subseteq}{F^e}\fatsemi\sqb{\texttt{\small B}}\fatsemi\sqb{\texttt{\small S}}^b)Q_{\angelic})
\wedge
P^{\bot}\subseteq(\textsf{\upshape pre}(\Lfp{\subseteq}{F^e}\fatsemi\sqb{\texttt{\small B}}\fatsemi\sqb{\texttt{\small S}}^\bot)\{\bot\mid \bot\in Q\}
\cup
\textsf{\upshape pre}(\Gfp{\subseteq}{F^\bot})\{\bot\mid \bot\in Q\}))\}}{by join preservation (\ref{eq:pre:join-preservation})}\\
=	
\formula{\{\triple{P^{ok}\cup P^{br}\cup P^{\bot}_{b}\cup P^{\bot}_{\ell}}{Q}{T}\mid
P^{ok}\subseteq\textsf{\upshape pre}(\Lfp{\subseteq}{F^e}\fatsemi\sqb{\neg\texttt{\small B}})Q_{\angelic}
\wedge
P^{br}\subseteq\textsf{\upshape pre}(\Lfp{\subseteq}{F^e}\fatsemi\sqb{\texttt{\small B}}\fatsemi\sqb{\texttt{\small S}}^b)Q_{\angelic})
\wedge
P^{\bot}_{b}\subseteq(\textsf{\upshape pre}(\Lfp{\subseteq}{F^e}\fatsemi\sqb{\texttt{\small B}}\fatsemi\sqb{\texttt{\small S}}^\bot)\{\bot\mid \bot\in Q\})
\wedge
P^{\bot}_{\ell}\subseteq(\textsf{\upshape pre}(\Gfp{\subseteq}{F^\bot})\{\bot\mid \bot\in Q\})\}}\\
\explanation{By join covering (\ref{eq:join-covering}) of $P^e$ into the initial states $P^{ok}$ for which the loop terminates normally and $P^{br}$ for which the loop terminates through a break
and join convering of $P^\bot$ into the initial states $P^\bot_b$ from which the loop body ultimately does not terminate and $P^\bot_\ell$ from which there are infinitely many terminating executions of the loop body}\\
=	
\formula{\{\triple{P^{ok}\cup P^{br}\cup P^{\bot}_{b}\cup P^{\bot}_{\ell}}{Q}{T}\mid
P^{ok}\subseteq\textsf{\upshape pre}(\Lfp{\subseteq}{F^e})(\textsf{\upshape pre}\sqb{\neg\texttt{\small B}}Q_{\angelic})
\wedge
P^{br}\subseteq\textsf{\upshape pre}(\Lfp{\subseteq}{F^e})(\textsf{\upshape pre}(\sqb{\texttt{\small B}}\fatsemi\sqb{\texttt{\small S}}^b)Q_{\angelic})
\wedge
P^{\bot}_{b}\subseteq\textsf{\upshape pre}(\Lfp{\subseteq}{F^e})(\textsf{\upshape pre}(\sqb{\texttt{\small B}}\fatsemi\texttt{\small S}^\bot)\{\bot\mid \bot\in Q\})
\wedge
P^{\bot}_{\ell}\subseteq(\textsf{\upshape pre}(\Gfp{\subseteq}{F^\bot})\{\bot\mid \bot\in Q\})\}}\\\rightexplanation{composition (\ref{eq:pre:composition})}\\
=	
\formula{\{\triple{P\cup P^{\bot}_{\ell}}{Q}{T}\mid
P\subseteq(\textsf{\upshape pre}(\Lfp{\subseteq}{F^e})(\textsf{\upshape pre}\sqb{\neg\texttt{\small B}}Q_{\angelic})
\cup\textsf{\upshape pre}(\Lfp{\subseteq}{F^e})(\textsf{\upshape pre}(\sqb{\texttt{\small B}}\fatsemi\sqb{\texttt{\small S}}^b)Q_{\angelic})
\cup\textsf{\upshape pre}(\Lfp{\subseteq}{F^e})(\textsf{\upshape pre}(\sqb{\texttt{\small B}}\fatsemi\sqb{\texttt{\small S}}^\bot)\{\bot\mid \bot\in Q\}))
\wedge
P^{\bot}_{\ell}\subseteq(\textsf{\upshape pre}(\Gfp{\subseteq}{F^\bot})\{\bot\mid \bot\in Q\})\}}\\
\rightexplanation{join covering (\ref{eq:join-covering}) with $P=P^{ok}\cup P^{br}\cup P^{\bot}_{b}$}\\
=
\formulaexplanation{\{\triple{P\cup P^{\bot}_{\ell}}{Q}{T}\mid
P\subseteq\textsf{\upshape pre}(\Lfp{\subseteq}{F^e})((\textsf{\upshape pre}\sqb{\neg\texttt{\small B}}Q_{\angelic})
\cup
(\textsf{\upshape pre}(\sqb{\texttt{\small B}}\fatsemi\sqb{\texttt{\small S}}^b)Q_{\angelic})
\cup
(\textsf{\upshape pre}(\sqb{\texttt{\small B}}\fatsemi\sqb{\texttt{\small S}}^\bot)\{\bot\mid \bot\in Q\}))
\wedge
P^{\bot}_{\ell}\subseteq(\textsf{\upshape pre}(\Gfp{\subseteq}{F^\bot})\{\bot\mid \bot\in Q\})\}}{join preservation (\ref{eq:pre:join-preservation})}\\
=
\formula{\{\triple{P\cup P^{\bot}_{\ell}}{Q}{T}\mid
P\subseteq(\Lfp{\stackrel{.}{\subseteq}}{F^e_{\textsf{\upshape pre}}})((\textsf{\upshape pre}\sqb{\neg\texttt{\small B}}Q_{\angelic})
\cup
(\textsf{\upshape pre}(\sqb{\texttt{\small B}}\fatsemi\sqb{\texttt{\small S}}^b)Q_{\angelic})
\cup
(\textsf{\upshape pre}(\sqb{\texttt{\small B}}\fatsemi\sqb{\texttt{\small S}}^\bot)\{\bot\mid \bot\in 
\cup(\textsf{\upshape pre}(\sqb{\texttt{\small B}}\fatsemi\sqb{\texttt{\small S}}^\bot)\{\bot\mid \bot\in Q\}))
\wedge
P^{\bot}_{\ell}\subseteq(\textsf{\upshape pre}(\Gfp{\subseteq}{F^\bot})\{\bot\mid \bot\in Q\})\}}\\
\explanation{by the Galois connection $\pair{\textsf{\upshape pre}}{\textsf{\upshape pre}^{-1}}$ of
(\ref{eq:def:post:GC}) $F^e$ and 
$F^e_{\textsf{\upshape pre}}$ are increasing, and, by (\ref{eq:F-e-commutation}), they commute so that
$\textsf{\upshape pre}(\Lfp{\subseteq}F^e) =\Lfp{\stackrel{.}{\subseteq}}F^e_{\textsf{\upshape pre}}$
by the fixpoint abstraction Th\@. \ref{th:fixpoint-abstraction}}\\
=
\formula{\{\triple{P\cup P^{\bot}_{\ell}}{Q}{T}\mid
P\subseteq\Lfp{\subseteq}{\bar{F}^e_{((\textsf{\upshape pre}\sqb{\neg\texttt{\small B}}Q_{\angelic})\cup
(\textsf{\upshape pre}(\sqb{\texttt{\small B}}\fatsemi\sqb{\texttt{\small S}}^b)Q_{\angelic})
\cup(\textsf{\upshape pre}(\sqb{\texttt{\small B}}\fatsemi\sqb{\texttt{\small S}}^\bot)\{\bot\mid \bot\in Q\}))}}
\wedge
P^{\bot}_{\ell}\subseteq(\textsf{\upshape pre}(\Gfp{\subseteq}{F^\bot})\{\bot\mid \bot\in Q\})\}}\\
\rightexplanation{by $F^e_{\textsf{\upshape pre}}(X)Q=\bar{F}^e_Q(X(Q))$ at (\ref{Fepre-FbareQ}) and the pointwise abstraction corollary \ref{th:Pointwise-abstraction}}\\
=
\formulaexplanation{\{\triple{P\cup P^{\bot}_{\ell}}{Q}{T}\mid
P\subseteq\Lfp{\subseteq}\LAMBDA{X}(\textsf{\upshape pre}\sqb{\neg\texttt{\small B}}Q_{\angelic})\cup
(\textsf{\upshape pre}(\sqb{\texttt{\small B}}\fatsemi\sqb{\texttt{\small S}}^b)Q_{\angelic})
\cup(\textsf{\upshape pre}(\sqb{\texttt{\small B}}\fatsemi\sqb{\texttt{\small S}}^\bot)\{\bot\mid \bot\in Q\})\cup (\textsf{\upshape pre}(\sqb{\texttt{\small B}}\fatsemi\sqb{\texttt{\small S}}^e)X)
\wedge
P^{\bot}_{\ell}\subseteq(\textsf{\upshape pre}(\Gfp{\subseteq}{F^\bot})\{\bot\mid \bot\in Q\})\}}{def.\ (\ref{Fepre-FbareQ}) of $\bar{F}^e_Q$}\\
=
\formulaexplanation{\{\triple{P\cup P^{\bot}_{\ell}}{Q}{T}\mid
P\subseteq\Lfp{\subseteq}\LAMBDA{X}(\textsf{\upshape pre}\sqb{\neg\texttt{\small B}}Q_{\angelic})\cup
(\textsf{\upshape pre}(\sqb{\texttt{\small B}}\fatsemi\sqb{\texttt{\small S}}^b)Q_{\angelic})
\cup(\textsf{\upshape pre}(\sqb{\texttt{\small B}}\fatsemi\sqb{\texttt{\small S}}^\bot)\{\bot\mid \bot\in Q\})\cup (\textsf{\upshape pre}(\sqb{\texttt{\small B}}\fatsemi\sqb{\texttt{\small S}}^e)X)
\wedge
\si \bot\in Q \alors P^{\bot}_{\ell}\subseteq(\textsf{\upshape pre}(\Gfp{\subseteq}{F^\bot})\{\bot\}\sinon P^{\bot}_{\ell}=\emptyset\fsi)\}}{since $\textsf{pre}(r)\emptyset=\emptyset$}
\end{calculus}

\smallskip

\hyphen{5}\quad Since we proceed by structural induction on statements, $\textsf{\upshape pre}\sqb{\texttt{\small S}}^\bot(\{\bot\mid \bot\in Q\})$ is already determined so can be considered to be a constant the same way that $\mathcal{B}\sqb{\texttt{\small B}}$ or $\mathcal{B}\sqb{\neg\texttt{\small B}}$ are. We now have to choose induction principles to under approximate the fixpoints. 

\smallskip

\hyphen{5}\quad Since \textsf{pre} preserves arbitrary joins, its least fixpoint iterations are stable at $\omega$. 
So the hypotheses of Th\@. \ref{th:Fixpoint-Underapproximation} for $P\subseteq\Lfp{\subseteq}\LAMBDA{X}A\cup\textsf{pre}(r)X$  become
\begin{eqntabular}{c}
\exists\pair{I^n}{n\in\mathbb{N}}\mathrel{.}I^0=\emptyset\wedge \forall n\in\mathbb{N}\mathrel{.}I^n\subseteq I^{n+1}\subseteq 
 A\cup\textsf{pre}(r)I^n\wedge\exists\ell\in\mathbb{N}\mathrel{.} P\subseteq I^\ell
\renumber{(\ref{cal:Fixpoint-Underapproximation-pre})}
\end{eqntabular}
since, $\pair{I^n}{n\in\mathbb{N}}$ being an increasing chain for $\subseteq$, we have $P\subseteq\bigcup_{n\in\mathbb{N}}{I^n}$ if and only if $\exists\ell\in\mathbb{N}\mathrel{.} P\subseteq I^\ell$.

\smallskip

\hyphen{5}\quad The hypotheses of the dual of Th\@. \ref{th:OverapproximationLeastFixpointImage} to under approximate the image of the greatest fixpoint $\Gfp{\subseteq}{F^\bot}$ by $\LAMBDA{r}\textsf{\upshape pre}(r)\{\bot\}$ are that there exists $J\in\wp(\Sigma_{\bot})$ such that
\begin{itemize}
\item[(1)] $\textsf{pre}(\emptyset)\{\bot\}\subseteq J$ which trivially holds since  $\textsf{pre}(\emptyset)\{\bot\}=\emptyset$. 
\item[(2)] $\forall X\in \wp(\Sigma\times\Sigma_\bot)\mathrel{.}\textsf{pre} (X)\{\bot\}\subseteq J \Rightarrow \textsf{pre} (F^\bot(X))\{\bot\}\subseteq J$
\begin{calculus}[$\Leftrightarrow$~]<3pt>
$\Leftrightarrow$
\formulaexplanation{\si \bot\in Q\alors\forall X\mathrel{.}
\textsf{pre} (X)\{\bot\}\subseteq J \Rightarrow \textsf{pre} (F^\bot(X))\{\bot\}\subseteq J
\sinon\textsf{true}\fsi}{since $\textsf{pre}(X)\emptyset=\emptyset$}\\
$\Leftrightarrow$
\formulaexplanation{\si \bot\in Q\alors\forall X\mathrel{.}
\textsf{pre} (X)\{\bot\}\subseteq J \Rightarrow \textsf{pre} (\sqb{\texttt{\small B}}\fatsemi\sqb{\texttt{\small S}}^e\fatsemi X)\{\bot\}\subseteq J
\sinon\textsf{true}\fsi}{def.\ (\ref{eq:natural-transformer-infinite}) of $F^{\bot}$}\\
$\Leftrightarrow$
\formula{\si \bot\in Q\alors\forall X\mathrel{.}
\textsf{pre} (X)\{\bot\}\subseteq J \Rightarrow \textsf{pre} (\sqb{\texttt{\small B}}\fatsemi\sqb{\texttt{\small S}}^e)(\textsf{pre} (X)\{\bot\})\subseteq J
\sinon\textsf{true}\fsi}\\[-0.5ex]
\rightexplanation{composition (\ref{eq:pre:composition})}\\[-1ex]
$\Leftrightarrow$
\formula{\si \bot\in Q\alors
\textsf{pre} (\sqb{\texttt{\small B}}\fatsemi\sqb{\texttt{\small S}}^e)(J)\subseteq J
\sinon\textsf{true}\fsi}\\
\explanation{($\Rightarrow$)\quad take $X=\{\pair{\sigma}{\bot}\mid \sigma\in J\}$ so that 
$\textsf{pre} (X)\{\bot\}=J$;\\ 
($\Leftarrow$)\quad for all $X\in \wp(\Sigma\times\Sigma_\bot)$, if $\textsf{pre} (X)\{\bot\}\subseteq J$ by hypothesis then $\textsf{pre} (\sqb{\texttt{\small B}}\fatsemi\sqb{\texttt{\small S}}^e)(\textsf{pre} (X))\{\bot\}\subseteq \textsf{pre} (\sqb{\texttt{\small B}}\fatsemi\sqb{\texttt{\small S}}^e)(J)$ since  $\textsf{pre}(r)$ preserves joins so is increasing, and therefore the hypothesis $\textsf{pre} (\sqb{\texttt{\small B}}\fatsemi\sqb{\texttt{\small S}}^e)(J)\subseteq J$ implies $\textsf{pre} (\sqb{\texttt{\small B}}\fatsemi\sqb{\texttt{\small S}}^e)(\textsf{pre} (X)\{\bot\})\subseteq J$ by transitivity.}\\

\end{calculus}
\item[(3)] for any $\sqsubseteq$-increasing chain $\pair{X^\delta}{\delta\in\mathbb{O}}$ of elements $X^\delta\sqsubseteq\Lfp{\sqsubseteq}F^\bot$, $\forall\beta<\lambda\mathrel{.}\textsf{pre} (X^\beta)\{\bot\}\subseteq J$ implies $\textsf{pre} (\bigcup_{\beta<\lambda}X^\beta)\{\bot\}\subseteq J$. The case $\bot\not\in Q$ being  trivially true, let us assume that
$\textsf{pre} (X^\beta)\{\bot\}\subseteq J$ that is, $
\{\sigma\mid \pair{\sigma}{\bot}\in X^{\beta}\}\subseteq J$ for all $\beta<\lambda$, $\lambda$ limit ordinal. Then we have
\begin{calculus}[$\Leftrightarrow$~]<3pt>
\formula{\textsf{pre} (\bigcup_{\beta<\lambda}X^\beta)\{\bot\}\subseteq J}\\
$\Leftrightarrow$
\formulaexplanation{\{\sigma\mid \pair{\sigma}{\bot}\in \bigcup_{\beta<\lambda}X^\beta\}\subseteq J}{def.\ (\ref{eq:def:psott-pre-pret}) of \textsf{\upshape pre}}\\
$\Leftrightarrow$
\formulaexplanation{\bigcup_{\beta<\lambda}\{\sigma\mid \pair{\sigma}{\bot}\in X^\beta\}\subseteq J}{def.\ union $\cup$}\\
$\Leftrightarrow$
\formulaexplanation{\forall\beta<\lambda\mathrel{.}\{\sigma\mid \pair{\sigma}{\bot}\in X^\beta\}\subseteq J}{def.\ inclusion $\subseteq$}\\
$\Leftrightarrow$
\formulaexplanation{\textsf{true}}{by hypothesis}
\end{calculus}
\item[(4)] $J\subseteq P$. 
\end{itemize}
So the four hypotheses of the dual of Th\@. \ref{th:OverapproximationLeastFixpointImage} for proving that $P^{\bot}_{\ell}\subseteq(\textsf{\upshape pre}(\Gfp{\subseteq}{F^\bot})\{\bot\}$ boil down to
\begin{eqntabular}{c}
\exists J\in\wp(\Sigma_\bot)\mathrel{.}
\textsf{pre}(\sqb{\texttt{\small B}}\fatsemi\sqb{\texttt{\small S}}^e)(J)\subseteq J\wedge P^{\bot}_{\ell}\subseteq J
\renumber{(\ref{cal:Underapproximation-gfp-pre})}
\end{eqntabular}
so we can go on with our formal calculation. We left it at
\begin{calculus}[$\Leftrightarrow$~]<3pt>
$\Leftrightarrow$
\formula{\{\triple{P\cup P^{\bot}_{\ell}}{Q}{T}\mid
P\subseteq\Lfp{\subseteq}\LAMBDA{X}(\textsf{\upshape pre}\sqb{\neg\texttt{\small B}}Q_{\angelic})
\cup
(\textsf{\upshape pre}(\sqb{\texttt{\small B}}\fatsemi\sqb{\texttt{\small S}}^b)Q_{\angelic})
\cup
(\textsf{\upshape pre}(\sqb{\texttt{\small B}}\fatsemi\sqb{\texttt{\small S}}^\bot)\{\bot\mid \bot\in Q\})\cup (\textsf{\upshape pre}(\sqb{\texttt{\small B}}\fatsemi\sqb{\texttt{\small S}}^e)X)
\wedge
\si \bot\in Q \alors P^{\bot}_{\ell}\subseteq(\textsf{\upshape pre}(\Gfp{\subseteq}{F^\bot})\{\bot\}\sinon P^{\bot}_{\ell}=\emptyset\fsi)\}}\\
$\Leftrightarrow$
\formulaexplanation{\{\triple{P\cup P^{\bot}_{\ell}}{Q}{T}\mid
\exists\pair{I^n}{n\in\mathbb{Q}}\mathrel{.}
 I^0=\emptyset\wedge
\forall n\in\mathbb{N}\mathrel{.}I^n\subseteq I^{n+1}\subseteq 
(\textsf{\upshape pre}\sqb{\neg\texttt{\small B}}Q_{\angelic})
\cup
(\textsf{\upshape pre}(\sqb{\texttt{\small B}}\fatsemi\sqb{\texttt{\small S}}^b)Q_{\angelic})
\cup(\textsf{\upshape pre}(\sqb{\texttt{\small B}}\fatsemi\sqb{\texttt{\small S}}^\bot)\{\bot\mid \bot\in Q\})
\cup (\textsf{\upshape pre}(\sqb{\texttt{\small B}}\fatsemi\sqb{\texttt{\small S}}^e)I^n)\wedge\exists\ell\in\mathbb{N}\mathrel{.} P\subseteq I^\ell
\wedge
\exists J\in\wp(\Sigma_\bot)\mathrel{.}
\si \bot\in Q\alors
\textsf{pre}(\sqb{\texttt{\small B}}\fatsemi\sqb{\texttt{\small S}}^e)(J)\subseteq J\wedge P^{\bot}_{\ell}\subseteq J
\sinon P^{\bot}_{\ell}=\emptyset\fsi\}}{by (\ref{cal:Fixpoint-Underapproximation-pre}) and  (\ref{cal:Underapproximation-gfp-pre})}\\
$\Leftrightarrow$
\formulaexplanation{\{\triple{P\cup P^{\bot}_{\ell}}{Q}{T}\mid
\exists\pair{I^n}{n\in\mathbb{Q}}\mathrel{.}
 I^0=\emptyset\wedge
\forall n\in\mathbb{N}\mathrel{.}I^n\subseteq I^{n+1}\subseteq 
(\mathcal{B}\sqb{\neg\texttt{\small B}}\cap Q_{\angelic})
\cup
(\mathcal{B}\sqb{\texttt{\small B}}\cap\textsf{\upshape pre}(\sqb{\texttt{\small S}}^b)Q_{\angelic})
\cup(\mathcal{B}\sqb{\texttt{\small B}}\cap \textsf{\upshape pre}\sqb{\texttt{\small S}}^\bot\{\bot\mid \bot\in Q\})\cup (\sqb{\texttt{\small B}}\cap \textsf{\upshape pre}\sqb{\texttt{\small S}}^e(I^n))\wedge\exists\ell\in\mathbb{N}\mathrel{.} P\subseteq I^\ell
\wedge
\exists J\in\wp(\Sigma_\bot)\mathrel{.}
\si \bot\in Q\alors
\mathcal{B}\sqb{\texttt{\small B}}\cap\textsf{pre}(\sqb{\texttt{\small S}}^e)(J)\subseteq J\wedge P^{\bot}_{\ell}\subseteq J
\sinon P^{\bot}_{\ell}=\emptyset\fsi\}}{Boolean precondition (\ref{eq:Boolean-precondition})}\\[-0.5ex]

$\Leftrightarrow$

\formula{\{\triple{P\cup P^{\bot}_{\ell}}{Q}{T}\mid
\exists\pair{I^n}{n\in\mathbb{Q}}\mathrel{.}
 I^0=\emptyset\wedge
\forall n\in\mathbb{N}\mathrel{.}I^n\subseteq I^{n+1}\subseteq 
(\mathcal{B}\sqb{\neg\texttt{\small B}}\cap Q_{\angelic})
\cup
(\mathcal{B}\sqb{\texttt{\small B}}\cap\textsf{\upshape pre}(\sqb{\texttt{\small S}}^b)Q_{\angelic})\cup
(\mathcal{B}\sqb{\texttt{\small B}}\cap \textsf{\upshape pre}\sqb{\texttt{\small S}}^\bot\{\bot\mid \bot\in Q\})\cup (\sqb{\texttt{\small B}}\cap \textsf{\upshape pre}\sqb{\texttt{\small S}}^e(I^n))\wedge\exists\ell\in\mathbb{N}\mathrel{.} P\subseteq I^\ell
\wedge
\exists J\in\wp(\Sigma)\mathrel{.}
\si \bot\in Q\alors
\mathcal{B}\sqb{\texttt{\small B}}\cap\textsf{pre}(\sqb{\texttt{\small S}}^e)(J)\subseteq J\wedge P^{\bot}_{\ell}\subseteq J
\sinon P^{\bot}_{\ell}=\emptyset\fsi\}}\\
\rightexplanation{since $\textsf{pre}(\sqb{\texttt{\small S}}^e)(J)=\textsf{pre}(\sqb{\texttt{\small S}}^e)(J\setminus\{\bot\})$ by (\ref{eq:natural-finite}) and (\ref{eq:natural-transformer-finite})}\\

$\Leftrightarrow$

\formula{\{\triple{P\cup P^{\bot}_{\ell}}{Q}{T}\mid
\exists R^b\mathrel{.}R^b\subseteq\textsf{\upshape pre}(\sqb{\texttt{\small S}}^b)Q_{\angelic}
\wedge
\exists\pair{I^n}{n\in\mathbb{Q}}\mathrel{.}
 I^0=\emptyset\wedge
\forall n\in\mathbb{N}\mathrel{.}I^n\subseteq I^{n+1}\subseteq 
(\mathcal{B}\sqb{\neg\texttt{\small B}}\cap Q_{\angelic})
\cup
(\mathcal{B}\sqb{\texttt{\small B}}\cap R^b)\cup
(\mathcal{B}\sqb{\texttt{\small B}}\cap \textsf{\upshape pre}\sqb{\texttt{\small S}}^\bot\{\bot\mid \bot\in Q\})\cup (\sqb{\texttt{\small B}}\cap \textsf{\upshape pre}\sqb{\texttt{\small S}}^e(I^n))\wedge\exists\ell\in\mathbb{N}\mathrel{.} P\subseteq I^\ell
\wedge
\exists J\in\wp(\Sigma)\mathrel{.}
\si \bot\in Q\alors
\mathcal{B}\sqb{\texttt{\small B}}\cap\textsf{pre}(\sqb{\texttt{\small S}}^e)(J)\subseteq J\wedge P^{\bot}_{\ell}\subseteq J
\sinon P^{\bot}_{\ell}=\emptyset\fsi\}}\\

\explanation{by under approximation of $\textsf{\upshape pre}(\sqb{\texttt{\small S}}^b)Q_{\angelic}$ by $R^b$\\[2pt]
($\Rightarrow$)\quad Take $R^b=\textsf{\upshape pre}(\sqb{\texttt{\small S}}^b)Q_{\angelic}$;\\[2pt]
($\Leftarrow$)\quad since $R^b\subseteq\textsf{\upshape pre}(\sqb{\texttt{\small S}}^b)Q_{\angelic}$ and $X\subseteq R^b$ implies $X\subseteq\textsf{\upshape pre}(\sqb{\texttt{\small S}}^b)Q_{\angelic}$}\\

$\Leftrightarrow$
\formula{\{\triple{P\cup P^{\bot}_{\ell}}{Q}{T}\mid
\exists R^b\mathrel{.}R^b\subseteq\textsf{\upshape pre}(\sqb{\texttt{\small S}}^b)Q_{\angelic}
\wedge
\exists R^\bot\mathrel{.}R^\bot\subseteq\textsf{\upshape pre}\sqb{\texttt{\small S}}^\bot\{\bot\mid \bot\in Q\}
\wedge
\exists\pair{I^n}{n\in\mathbb{Q}}\mathrel{.}
 I^0=\emptyset\wedge
\forall n\in\mathbb{N}\mathrel{.}I^n\subseteq I^{n+1}\subseteq 
(\mathcal{B}\sqb{\neg\texttt{\small B}}\cap Q_{\angelic})\cup(\mathcal{B}\sqb{\texttt{\small B}}\cap R^b)\cup(\mathcal{B}\sqb{\texttt{\small B}}\cap R^\bot)\cup (\mathcal{B}\sqb{\texttt{\small B}}\cap \textsf{\upshape pre}\sqb{\texttt{\small S}}^e(I^n))\wedge\exists\ell\in\mathbb{N}\mathrel{.} P\subseteq I^\ell
\wedge
\exists J\in\wp(\Sigma_\bot)\mathrel{.}
\si \bot\in Q\alors
\mathcal{B}\sqb{\texttt{\small B}}\cap\textsf{pre}(\sqb{\texttt{\small S}}^e)(J)\subseteq J\wedge P^{\bot}_{\ell}\subseteq J
\sinon P^{\bot}_{\ell}=\emptyset\fsi\}}\\
\explanation{by under approximation of $\textsf{\upshape pre}\sqb{\texttt{\small S}}^\bot\{\bot\mid \bot\in Q\}$ by $R^\bot$\\[2pt]
($\Rightarrow$)\quad Take $R^\bot=\textsf{\upshape pre}\sqb{\texttt{\small S}}^\bot\{\bot\mid \bot\in Q\}$;\\[2pt]
($\Leftarrow$)\quad since $R^\bot\subseteq\textsf{\upshape pre}\sqb{\texttt{\small S}}^\bot\{\bot\mid \bot\in Q\}$ and $X\subseteq R^\bot$ implies $X\subseteq\textsf{\upshape pre}\sqb{\texttt{\small S}}^\bot\{\bot\mid \bot\in Q\}$}\\

$\Leftrightarrow$
\formula{\{\triple{P\cup P^{\bot}_{\ell}}{Q}{T}\mid
\exists R^b\mathrel{.}R^b\subseteq\textsf{\upshape pre}(\sqb{\texttt{\small S}}^b)Q_{\angelic}
\wedge
\exists R^\bot\mathrel{.}R^\bot\subseteq\textsf{\upshape pre}\sqb{\texttt{\small S}}_\bot\{\bot\mid \bot\in Q\}
\wedge
\exists\pair{I^n}{n\in\mathbb{Q}}\mathrel{.}
 I^0=\emptyset\wedge
\forall n\in\mathbb{N}\mathrel{.}I^n\subseteq I^{n+1}\subseteq 
(\mathcal{B}\sqb{\neg\texttt{\small B}}\cap Q_{\angelic})
\cup(\mathcal{B}\sqb{\texttt{\small B}}\cap R^b)
\cup(\mathcal{B}\sqb{\texttt{\small B}}\cap R^\bot)\cup ((\mathcal{B}\sqb{\texttt{\small B}}\cap \textsf{\upshape pre}\sqb{\texttt{\small S}}^e(I^n))\wedge\exists\ell\in\mathbb{N}\mathrel{.} P\subseteq I^\ell
\wedge
\exists J\in\wp(\Sigma_\bot)\mathrel{.}
\si \bot\in Q\alors
\mathcal{B}\sqb{\texttt{\small B}}\cap\textsf{pre}(\sqb{\texttt{\small S}}^e)(J)\subseteq J\wedge P^{\bot}_{\ell}\subseteq J
\sinon P^{\bot}_{\ell}=\emptyset\fsi\}}\\
\explanation{def.\ (\ref{eq:def:semantics}) of the relational semantics and (\ref{eq:def:psott-pre-pret}) of \textsf{\upshape pre} so that $\textsf{pre}\sqb{\texttt{\small S}}^\bot\emptyset =\textsf{pre}\sqb{\texttt{\small S}}_\bot\emptyset$ and $\textsf{pre}\sqb{\texttt{\small S}}^\bot\{\bot\} =\textsf{pre}\sqb{\texttt{\small S}}_\bot\{\bot\}$
}\\

$\Leftrightarrow$
\formula{\{\triple{P\cup P^{\bot}_{\ell}}{Q}{T}\mid
\exists R^b\mathrel{.}R^b\subseteq\textsf{\upshape pre}(\sqb{\texttt{\small S}}^b)Q_{\angelic}
\wedge
\exists R^\bot\mathrel{.}\triple{R^\bot}{\{\bot\mid \bot\in Q\}}{\emptyset}\in
\alpha_{\textsf{\upshape pre}}(\sqb{\texttt{\small S}}_{\bot})
\wedge
\exists\pair{I^n}{n\in\mathbb{Q}}\mathrel{.}
 I^0=\emptyset\wedge
\forall n\in\mathbb{N}\mathrel{.}I^n\subseteq I^{n+1}\subseteq 
(\mathcal{B}\sqb{\neg\texttt{\small B}}\cap Q_{\angelic})
\cup(\sqb{\texttt{\small B}}\cap R^b)
\cup(\sqb{\texttt{\small B}}\cap R^\bot)\cup (\sqb{\texttt{\small B}}\cap \textsf{\upshape pre}\sqb{\texttt{\small S}}^e(I^n))\wedge\exists\ell\in\mathbb{N}\mathrel{.} P\subseteq I^\ell
\wedge
\exists J\in\wp(\Sigma_\bot)\mathrel{.}
\si \bot\in Q\alors
\mathcal{B}\sqb{\texttt{\small B}}\cap\textsf{pre}(\sqb{\texttt{\small S}}^e)(J)\subseteq J\wedge P^{\bot}_{\ell}\subseteq J
\sinon P^{\bot}_{\ell}=\emptyset\fsi\}}\\
\explanation{def.\ (\ref{eq:def:alpha:pre}) of $\alpha_{\textsf{\upshape pre}}$}\\

$\Leftrightarrow$
\formula{\{\triple{P\cup P^{\bot}_{\ell}}{Q}{T}\mid
\exists R^b\mathrel{.}R^b\subseteq\textsf{\upshape pre}(\sqb{\texttt{\small S}}^b)Q_{\angelic}
\wedge
\exists R^\bot\mathrel{.}\triple{R^\bot}{\{\bot\mid \bot\in Q\}}{\emptyset}\in
\alpha_{\textsf{\upshape pre}}(\sqb{\texttt{\small S}}_{\bot})
\wedge
\exists\pair{I^n}{n\in\mathbb{Q}}\mathrel{.}
 I^0=\emptyset\wedge
\forall n\in\mathbb{N}\mathrel{.}
\exists R^e_n\mathrel{.}\triple{R^e_n}{I^n}{\emptyset}\in\alpha_{\textsf{\upshape pre}}(\sqb{\texttt{\small S}}_{\bot})\wedge
I^n\subseteq I^{n+1}\subseteq 
(\mathcal{B}\sqb{\neg\texttt{\small B}}\cap Q_{\angelic})\cup(\sqb{\texttt{\small B}}\cap R^b)\cup(\sqb{\texttt{\small B}}\cap R^\bot)\cup (\sqb{\texttt{\small B}}\cap R^e_n)\wedge\exists\ell\in\mathbb{N}\mathrel{.} P\subseteq I^\ell
\wedge
\exists J\in\wp(\Sigma_\bot)\mathrel{.}
\si \bot\in Q\alors
\mathcal{B}\sqb{\texttt{\small B}}\cap\textsf{pre}(\sqb{\texttt{\small S}}^e)(J)\subseteq J\wedge P^{\bot}_{\ell}\subseteq J
\sinon P^{\bot}_{\ell}=\emptyset\fsi\}}\\
\explanation{by under approximation $R^e_n$ of $\textsf{\upshape pre}\sqb{\texttt{\small S}}^e(I^n)$}\\

$\Leftrightarrow$
\formula{\{\triple{P\cup P^{\bot}_{\ell}}{Q}{T}\mid
\exists R^b\mathrel{.}R^b\subseteq\textsf{\upshape pre}(\sqb{\texttt{\small S}}^b)Q_{\angelic}
\wedge
\exists R^\bot\mathrel{.}\triple{R^\bot}{\{\bot\mid \bot\in Q\}}{\emptyset}\in
\alpha_{\textsf{\upshape pre}}(\sqb{\texttt{\small S}}_{\bot})
\wedge
\exists\pair{I^n}{n\in\mathbb{Q}}\mathrel{.}
 I^0=\emptyset\wedge
\forall n\in\mathbb{N}\mathrel{.}
\exists R^e_n\mathrel{.}\triple{R^e_n}{I^n}{\emptyset}\in\alpha_{\textsf{\upshape pre}}(\sqb{\texttt{\small S}}_{\bot})\wedge
I^n\subseteq I^{n+1}\subseteq 
(\mathcal{B}\sqb{\neg\texttt{\small B}}\cap Q_{\angelic})
\cup(\sqb{\texttt{\small B}}\cap R^b)\cup(\sqb{\texttt{\small B}}\cap R^\bot)\cup (\sqb{\texttt{\small B}}\cap R^e_n)\wedge\exists\ell\in\mathbb{N}\mathrel{.} P\subseteq I^\ell
\wedge
\exists J\in\wp(\Sigma_\bot)\mathrel{.}
\exists R^{\bot}_{\ell}\mathrel{.}\triple{R^{\bot}_{\ell}}{J}{\emptyset}\in\alpha_{\textsf{\upshape pre}}(\sqb{\texttt{\small S}})\wedge
\si \bot\in Q\alors
\mathcal{B}\sqb{\texttt{\small B}}\cap R^{\bot}_{\ell}\subseteq J\wedge P^{\bot}_{\ell}\subseteq J
\sinon P^{\bot}_{\ell}=\emptyset\fsi\}}\\
\explanation{by under approximation $R^{\bot}_{\ell}$ of $\textsf{pre}(\sqb{\texttt{\small S}}^e)(J)$}
\end{calculus}
Following Sect\@. \ref{sec:SemanticsLogics}, the theory $\alpha_{\textsf{\upshape pre}}(\sqb{\texttt{\small while(B) S}}_{\bot})$ can be equivalently  defined by the following deductive system.
\begin{eqntabular}[fl]{@{}c@{}}
\frac{\begin{array}{@{}c@{}}
I^0=\emptyset\\[1pt]
\{R^b\}\,\texttt{\small S}\,\cev{\{}ok:\emptyset, br:Q_{\angelic}\}\\[3pt]
\{R^\bot\}\,\texttt{\small S}\,\cev{\{}ok:\{\bot\mid \bot\in Q\}, br:\emptyset\}\\[3pt]
\forall n\in\mathbb{N}\mathrel{.}
\begin{array}[t]{@{}l@{}}
\{R^e_n\}\,\sqb{\texttt{\small S}}\,\cev{\{}ok:I^n, br:
\emptyset\}\wedge{}\\[0.5ex]
 I^n\subseteq I^{n+1}\subseteq \,
(\mathcal{B}\sqb{\neg\texttt{\small B}}\cap Q_{\angelic})\cup(\mathcal{B}\sqb{\texttt{\small B}}\cap  R^b)\cup(\mathcal{B}\sqb{\texttt{\small B}}\cap  R^\bot)\cup (\mathcal{B}\sqb{\texttt{\small B}}\cap R^e_n)\quad\exists\ell\in\mathbb{N}\mathrel{.} P\subseteq I^\ell
\end{array}\\[15pt]
\si \bot\in Q \alors \{R^{\bot}_{\ell}\}\sqb{\texttt{\small S}}\,\cev{\{}ok:J, br:\emptyset\}\wedge \mathcal{B}\sqb{\texttt{\small B}}\cap R^{\bot}_{\ell}\subseteq J\wedge P^{\bot}_{\ell}\subseteq J\sinon  P^{\bot}_{\ell}=\emptyset\fsi\\[0.75ex]
\end{array}
}
{
\mbox{\LARGE\strut}\{P\cup P^{\bot}_{\ell}\}\,{\texttt{\small while(B) S}}\,\cev{\{}ok:Q, br:\emptyset\}
}\renumber{\raisebox{-2ex}[0pt][0pt]{(\ref{eq:deductive:system:pre:while})}}
\end{eqntabular}
\end{toappendix}
\begin{toappendix}
\section{Related Work}\label{sec:RelatedWork}
The paper provides ample citations in the text and mainly owes to the extensive work on Hoare logic \cite{DBLP:journals/cacm/Hoare69,DBLP:journals/toplas/Apt81,DBLP:journals/tcs/Apt84,DBLP:journals/fac/AptO19,DBLP:books/mc/21/AptO21}, abstract interpretation \cite{Cousot-PAI-2021} (for the design of semantics by abstraction \cite{DBLP:journals/tcs/Cousot02} and the specification of program properties by Galois connections \cite{DBLP:conf/popl/CousotC14}),
 fixpoint induction \cite{Park69-MI5,DBLP:conf/lopstr/Cousot19} to handle invariants and termination, a previous study of proof methods \cite{CousotCousot82-TNPC} (now expressed with Galois connections and fixpoint induction), and the nonconformist idea of Derek Dreyer, Ralf Jung, and Peter O'Hearn \cite{DBLP:journals/pacmpl/OHearn20} originating the interest in incorrectness e.g\@.
\cite{DBLP:journals/jlap/PoskittP23,
DBLP:journals/iandc/FengL23,
DBLP:journals/pacmpl/YanJY22,
DBLP:journals/pacmpl/RaadBDO22,
DBLP:journals/pacmpl/ZilbersteinDS23,
DBLP:journals/jacm/BruniGGR23,
DBLP:conf/sefm/VriesK11,
DBLP:journals/pacmpl/LeRVBDO22,
DBLP:conf/fossacs/AscariBG22,
DBLP:conf/sas/Vanegue22,
DBLP:conf/RelMiCS/MollerOH21,
DBLP:journals/pacmpl/ZhangK22,
DBLP:conf/ecoop/MaksimovicCLSG23,
DBLP:conf/tap/NausVSR23,
DBLP:conf/cav/RaadBDDOV20,
DBLP:journals/corr/abs-2310-18156,
DBLP:conf/concur/RaadVBO23}.

Incorrectness has also be studied in the context of logic \cite{DBLP:conf/popl/Shapiro82,DBLP:conf/aadebug/Ferrand93,DBLP:conf/mi/Lloyd95} and constraint programming \cite{DBLP:conf/agp/BerreT96} as well as mathematical logic. \cite{Bergmann-Incorrectness77} discusses [in]correctness in the presence of undefined. The definition of incorrectness requires a
referent \cite{SvobodaPeregrin-incorrect-16}, which for programming languages is their semantics.

For simplicity, we have considered antecedent and consequent to be sets of states. Using logics instead is a further abstraction with no best abstraction (e.g\@. non-compact infinite disjunction in first-order logic). This abstraction introduces incompleteness inherited by transformational logics which by themselves are complete (under expressivity conditions of the underlying logic \cite{DBLP:journals/siamcomp/Cook78,DBLP:journals/siamcomp/Cook81,DBLP:journals/eatcs/BlassG00a}, which amounts to consider the interpretation of logic formulas as sets of states).

The abstraction $\textsf{\textup{post}}({\supseteq},{\subseteq})$ comes from the specification of reverse Hoare logic by $Q\subseteq\textsf{\textup{post}}\sqb{\texttt{\small S}}P$ in \cite{DBLP:conf/sefm/VriesK11} and then a.o\@. \cite{DBLP:journals/pacmpl/OHearn20,DBLP:journals/pacmpl/ZhangK22,DBLP:journals/jlap/PoskittP23}. 
\cite{DBLP:conf/focs/Pratt76,CousotCousot82-TNPC,DBLP:journals/pacmpl/ZhangK22,DBLP:journals/pacmpl/YanJY22} also consider complement duality (\ref{eq:def:post:GC}) 
\cite{DBLP:conf/ecoop/MaksimovicCLSG23}  incorporate reasoning about non-terminating specifications. \cite[Section 5]{DBLP:journals/afp/Dardinier23a} expresses different (hyper) logics in a common framework. \cite{DBLP:journals/entcs/Schmidt07} derives a logic from an abstraction.
In the proof theoretic approach, \cite{DBLP:journals/jacm/HarperHP93}
designs deductive systems by encoding in the common
Edinburgh Logical Framework (LF), e.g\@.\ Hoare logic 
\cite[Sect\@. 6.1]{DBLP:journals/jar/AvronHMP92}.
\end{toappendix}
\nocite{DBLP:journals/cacm/Hoare69,DBLP:journals/toplas/Apt81,DBLP:journals/tcs/Apt84,DBLP:journals/fac/AptO19,DBLP:books/mc/21/AptO21,DBLP:conf/lpar/BubelGHS23,DBLP:conf/scc/Pnueli79,DeneckeErneWismath-GC-03,DBLP:conf/fossacs/AscariBG22,Cousot-PAI-2021,CousotCousot82-TNPC,DBLP:conf/vmcai/CousotCFL13,CousotCousot-PJM-82-1-1979,vonNeumann-1923-ordinals,Monk-Set-Theory,DBLP:journals/jacm/HarperHP93,DBLP:journals/jar/AvronHMP92,DBLP:journals/entcs/Schmidt07,DBLP:journals/iandc/FengL23,DBLP:journals/pacmpl/YanJY22,DBLP:journals/pacmpl/RaadBDO22,DBLP:journals/toplas/Apt81,DBLP:journals/tcs/Apt84,DBLP:journals/fac/AptO19,DBLP:books/mc/21/AptO21,DBLP:conf/scc/Pnueli79,DeneckeErneWismath-GC-03,DBLP:conf/fossacs/AscariBG22,Cousot-PAI-2021,CousotCousot82-TNPC,DBLP:conf/vmcai/CousotCFL13,CousotCousot-PJM-82-1-1979,vonNeumann-1923-ordinals,Monk-Set-Theory,DBLP:journals/jacm/HarperHP93,DBLP:journals/jar/AvronHMP92,DBLP:journals/entcs/Schmidt07,DBLP:journals/iandc/FengL23,DBLP:journals/pacmpl/YanJY22,DBLP:journals/pacmpl/RaadBDO22,DBLP:conf/lpar/BubelGHS23,DBLP:journals/afp/Dardinier23a,
DBLP:journals/eatcs/BlassG00a,
DBLP:conf/sas/Vanegue22,
DBLP:conf/tap/NausVSR23,
DBLP:conf/RelMiCS/MollerOH21,
DBLP:journals/pacmpl/ZhangAG22,
DBLP:journals/pacmpl/ZhangK22,
DBLP:conf/ecoop/MaksimovicCLSG23,
DBLP:journals/jlap/PoskittP23,
DBLP:journals/entcs/Schmidt07,
DBLP:journals/jacm/HarperHP93,
DBLP:journals/corr/abs-2310-18156,
DBLP:journals/afp/Murray20,
DBLP:journals/jar/AvronHMP92,
DBLP:conf/cav/RaadBDDOV20,
DBLP:conf/concur/RaadVBO23,
DBLP:conf/popl/Shapiro82,
DBLP:conf/aadebug/Ferrand93,
DBLP:conf/mi/Lloyd95,
DBLP:conf/agp/BerreT96,
Bergmann-Incorrectness77,
DBLP:conf/tap/NausVSR23,
SvobodaPeregrin-incorrect-16}
\section{Conclusion}
\ifshort Related work was moved to the appendix Sect\@. \ref{sec:RelatedWork} \proofinapx. \fi
We have shown that the theory of abstract interpretation can be used to design program transformational logics, including (non)termination, by defining their theory as an abstraction of the programming language fixpoint natural relational semantics and then their proof system (useful to support mechanization) by fixpoint induction and Aczel correspondence between set-theoretic fixpoint definitions and deductive systems  \cite{Aczel:1977:inductive-definitions}. The approach applies to all other abstractions of the collecting semantics into a relation, not necessarily into a logic. For future work,  this same principled approach can be used to design hyper logics \cite{DBLP:journals/afp/Dardinier23a}, including dependency logics \cite{DBLP:conf/sas/Cousot19a}, to include meta information in predicates with an instrumented semantics (e.g\@. \cite{DBLP:journals/pacmpl/ZilbersteinDS23,DBLP:journals/pacmpl/ZhangK22,DBLP:conf/sas/Vanegue22}), and to extend
temporal logics like \cite{DBLP:conf/scc/Pnueli79} to programming languages by structural induction and local invariants \cite{DBLP:conf/lpar/BubelGHS23}.
\ifshort 
\section*{Data Availability Statement}
The auxiliary material of this article is available at \href{https://doi.org/10.1145/3632849}{https://doi.org/10.1145/3632849} and contains both the article and its appendix in a
single file with clickable hyper references.
\fi
\begin{acks}
I thank the participants to the \href{https://www.dagstuhl.de/seminars/seminar-calendar/seminar-details/23281}{Dagstuhl Seminar on ``Theoretical Advances and Emerging Applications in Abstract Interpretation''
09--14 July 2023} and Jeffery Wang for discussions. I thank the reviewers for appreciation and criticisms, corrections, and useful suggestions, Francesco Ranzato for improvement proposals, and Charles de Haro for numerous corrections. 
\end{acks}

%
\bibliographystyle{ACM-Reference-Format}
\bibliography{\jobname}
\end{document}

\newpage

\section*{Hoare logic rules}

\noindent\hypertarget{th:slides:3}{\textsc{Theorem 3 (Hoare rules for conditional iteration)}}.
\begin{eqntabular*}[fl]{@{\qquad\qquad}c}
\color{blue}\frac {\displaystyle\llstrut P\subseteq I,\ \{I\cap\mathcal{B}\sqb{\texttt{\small B}}\}\,\texttt{\small S}\,\{I\},\ (I\cap\mathcal{\neg B}\sqb{\texttt{\small B}})\subseteq Q}{\displaystyle\ulstrut\{P\}\,\texttt{\small while (B) S}\,\{Q\}}
\end{eqntabular*}

\begin{proof}

We write $\color{blue}\{P\}\,\texttt{\small S}\,\{Q\} \triangleq\pair{P}{Q}\in\mathcal{T_{\textrm{HL}}}(\texttt{\small S})$;

By structural induction (\texttt{\small S} being a strict component of \texttt{\small while (B) S}), the rule for $\color{blue}\{P\}\,\texttt{\small S}\,\{Q\}$ have already been defined; 

By Aczel method, the (constant) fixpoint $\color{blue}\Lfp{\subseteq}\LAMBDA{X}S$ is defined by
$\color{blue}\{\frac{\emptyset}{c}\mid c\in S\}$; 

So for \texttt{\small while (B) S} we have an axiom $\color{blue}\frac {\emptyset}{\displaystyle\ulstrut\{P\}\,\texttt{\small while (B) S}\,\{Q\}}$ with side condition $\color{blue}P\subseteq I,\ \{I\cap\mathcal{B}\sqb{\texttt{\small B}}\}\,\texttt{\small S}\,\{I\},\ (I\cap\mathcal{\neg B}\sqb{\texttt{\small B}})\subseteq Q$;

Traditionally, the side condition is considered a premiss, to get $\color{blue}\frac {\displaystyle\llstrut P\subseteq I,\ \{I\cap\mathcal{B}\sqb{\texttt{\small B}}\}\,\texttt{\small S}\,\{I\},\ (I\cap\mathcal{\neg B}\sqb{\texttt{\small B}})\subseteq Q}{\displaystyle\ulstrut\{P\}\,\texttt{\small while (B) S}\,\{Q\}}$
\let\qed\relax\end{proof}

$\color{blue}\frac {\displaystyle\llstrut P\subseteq I,\quad \{I\cap\mathcal{B}\sqb{\texttt{\small B}}\}\,\texttt{\small S}\,\{I\}}{\displaystyle\ulstrut\{P\}\,\texttt{\small while (B) S}\,\{I\cap\mathcal{\neg B}\sqb{\texttt{\small B}}\}}
\qquad
\frac {\displaystyle\llstrut\{P\}\,\texttt{\small S}\,\{Q\},\quad Q\subseteq Q'}{\displaystyle\ulstrut\{P\}\,\texttt{\small S}\,\{Q'\}}
$

\newpage
\section*{Reverse Hoare aka incorrectness logic theory}

\noindent\hypertarget{th:slides:4}{\textsc{Theorem 4 (theory of IL)}}. 
\bgroup\color{blue}\belowdisplayskip-17pt
\begin{eqntabular*}{rcl}
\mathcal{T_{\textrm{IL}}}(\texttt{\small W})
&\triangleq&\textsf{post}({\subseteq}.{\supseteq})\comp\mathcal{T}(\texttt{\small W})\\
&=&\{\pair{P}{Q}\mid \exists \pair{J^n}{n\in\mathbb{N}}\mathrel{.}
 J^0=P\wedge \pair{J^n\cap\mathcal{B}\sqb{\texttt{\small B}}}{J^{n+1}}\in\mathcal{T}_{\textrm{IL}}(\texttt{\small S})\wedge Q \subseteq(\bigcup_{n\in\mathbb{N}}J^n)\cap\mathcal{B}\sqb{\neg\texttt{\small B}} \}
\end{eqntabular*}
\egroup

\medskip

\begin{proof}
\begin{calculus}
\formula{\mathcal{T}_{\textrm{IL}}(\texttt{\small W})}\\
=
\formulaexplanation{\textsf{post}({\subseteq}.{\supseteq})\comp\mathcal{T}(\texttt{\small W})}{def\@. $\mathcal{T}_{\textrm{IL}}$}\\
=
\formulaexplanation{\{\pair{P}{Q}\mid Q \subseteq\textsf{\upshape post}\sqb{\texttt{\small W}}P \}}{$\subseteq$-order dual of \hyperlink{lem:slides:4}{lem\@. 4}}\\
=
\formulaexplanation{\{\pair{P}{Q}\mid Q \subseteq\textsf{\upshape post}\sqb{\neg\texttt{\small B}}(\Lfp{\subseteq}\bar{\bar{F}}^e_{P}) \}}{\hyperlink{th:slides:1}{Th\@. 1} where $\bar{\bar{F}}^e_P(X)\triangleq P \cup \textsf{\upshape post}(\sqb{\texttt{\small B}}\fatsemi\sqb{\texttt{\small S}}^e)X$}\\
= 
\formula{\{\pair{P}{Q}\mid \exists I\mathrel{.}Q \subseteq\textsf{\upshape post}\sqb{\neg\texttt{\small B}}(I) \wedge I\subseteq \Lfp{\subseteq}\bar{\bar{F}}^e_{P}\}}{}\\[-0.5ex]
\explanation{($\subseteq$)\quad Take $I= \Lfp{\subseteq}\bar{\bar{F}}^e_{P}$ and reflexivity;\\
($\supseteq$)\quad By Galois connection (\ref{eq:def:post:GC}), $\textsf{\upshape post}\sqb{\neg\texttt{\small B}}$ is increasing so $Q \subseteq\textsf{\upshape post}\sqb{\neg\texttt{\small B}}(I)\subseteq \textsf{\upshape post}\sqb{\neg\texttt{\small B}}(\Lfp{\subseteq}\bar{\bar{F}}^e_{P})$ and transitivity}\\
=
\formula{\{\pair{P}{Q}\mid \exists I\mathrel{.}Q \subseteq\textsf{\upshape post}\sqb{\neg\texttt{\small B}}(I) \wedge \exists \pair{J^n}{n<\omega}\mathrel{.}J^0=\emptyset\wedge J^{n+1}\subseteq \bar{\bar{F}}^e_{P}(J^n)\wedge I\subseteq\bigcup_{n<\omega}J^n\}}\\[-0.5ex]\rightexplanation{fixpoint underapproximation Th\@. II.3.6}\\
=
\formula{\{\pair{P}{Q}\mid \exists \pair{J^n}{n<\omega}\mathrel{.}
 J^0=\emptyset\wedge J^{n+1}\subseteq \bar{\bar{F}}^e_{P}(J^n)\wedge Q \subseteq\textsf{\upshape post}\sqb{\neg\texttt{\small B}}(\bigcup_{n<\omega}J^n) \}}\\[-0.5ex]
 \explanation{($\subseteq$)\quad By Galois connection (\ref{eq:def:post:GC}), $\textsf{\upshape post}\sqb{\neg\texttt{\small B}}$ is increasing so $Q \subseteq\textsf{\upshape post}\sqb{\neg\texttt{\small B}}(I)\subseteq\textsf{\upshape post}\sqb{\neg\texttt{\small B}}(\bigcup_{n<\omega}J^n)$ and transitivity;\\
 ($\supseteq$)\quad take $I=\bigcup_{n<\omega}J^n$}\\
=
\formula{\{\pair{P}{Q}\mid \exists \pair{J^n}{n<\omega}\mathrel{.}
 J^0=\emptyset\wedge J^{n+1}\subseteq (P \cup \textsf{\upshape post}(\sqb{\texttt{\small B}}\fatsemi\sqb{\texttt{\small S}}^e)(J^n))\wedge Q \subseteq\textsf{\upshape post}\sqb{\neg\texttt{\small B}}(\bigcup_{n<\omega}J^n) \}}\\[-0.5ex]\rightexplanation{def\@. $\bar{\bar{F}}^e_{P}$}\\
 =
\formula{\{\pair{P}{Q}\mid \exists \pair{J^n}{1\leqslant n<\omega}\mathrel{.}
 J^1=P\wedge J^{n+1}\subseteq \textsf{\upshape post}(\sqb{\texttt{\small B}}\fatsemi\sqb{\texttt{\small S}}^e)(J^n)\wedge Q \subseteq\textsf{\upshape post}\sqb{\neg\texttt{\small B}}(\bigcup_{1\leqslant n<\omega}J^n) \}}\\
 \rightexplanation{getting rid of $ J^0=\emptyset$}\\
 =
 \formula{\{\pair{P}{Q}\mid \exists \pair{J^n}{n\in\mathbb{N}}\mathrel{.}
 J^0=P\wedge J^{n+1}\subseteq \textsf{\upshape post}(\sqb{\texttt{\small B}}\fatsemi\sqb{\texttt{\small S}}^e)(J^n)\wedge Q \subseteq\textsf{\upshape post}\sqb{\neg\texttt{\small B}}(\bigcup_{n\in\mathbb{N}}J^n) \}}\\\rightexplanation{changing $n+1$ to $n$}\\
 =
 \formula{\{\pair{P}{Q}\mid \exists \pair{J^n}{n\in\mathbb{N}}\mathrel{.}
 J^0=P\wedge J^{n+1}\subseteq \textsf{\upshape post}\sqb{\texttt{\small S}}^e(J^n\cap\mathcal{B}\sqb{\texttt{\small B}})\wedge Q \subseteq(\bigcup_{n\in\mathbb{N}}J^n)\cap\mathcal{B}\sqb{\neg\texttt{\small B}} \}}\\[-0.5ex]\rightexplanation{\hyperlink{lem:slides:2}{test lemma 2}}\\
  =
 \formulaexplanation{\{\pair{P}{Q}\mid \exists \pair{J^n}{n\in\mathbb{N}}\mathrel{.}
 J^0=P\wedge \pair{J^n\cap\mathcal{B}\sqb{\texttt{\small B}}}{ J^{n+1}}\in\{\pair{P'}{Q'}\mid Q'\subseteq \textsf{\upshape post}\sqb{\texttt{\small S}}^e)P)\}\wedge Q \subseteq(\bigcup_{n\in\mathbb{N}}J^n)\cap\mathcal{B}\sqb{\neg\texttt{\small B}} \}}{def\@. $\in$}\\
 =
 \formulaexplanation{\{\pair{P}{Q}\mid \exists \pair{J^n}{n\in\mathbb{N}}\mathrel{.}
 J^0=P\wedge \pair{J^n\cap\mathcal{B}\sqb{\texttt{\small B}}}{J^{n+1}}\in\mathcal{T}_{\textrm{IL}}(\texttt{\small S})\wedge Q \subseteq(\bigcup_{n\in\mathbb{N}}J^n)\cap\mathcal{B}\sqb{\neg\texttt{\small B}} \}}{def\@. $ \mathcal{T}_{\textrm{IL}}$} \end{calculus}
\end{proof}

\newpage
\section*{IL rules}

\noindent\hypertarget{th:slides:5}{\textsc{Theorem 5 (IL rules for conditional iteration)}}. 
\begin{eqntabular*}{c}
\color{blue}\frac {\displaystyle\llstrut J^0=P,\ [J^n\cap\mathcal{B}\sqb{\texttt{\small B}}]\,\texttt{\small S}\,[J^{n+1}],\ Q \subseteq(\bigcup_{n\in\mathbb{N}}J^n)\cap\mathcal{B}\sqb{\neg\texttt{\small B}}}{\displaystyle\ulstrut[P]\,\texttt{\small while (B) S}\,[Q]}
\end{eqntabular*}

\begin{proof}

We write $\color{blue}[P]\,\texttt{\small S}\,[Q] \triangleq\pair{P}{Q}\in\mathcal{T_{\textrm{IL}}}(\texttt{\small S})$;

By structural induction (\texttt{\small S} being a strict component of \texttt{\small while (B) S}), the rule for $\color{blue}[P]\,\texttt{\small S}\,[Q]$ have already been defined; 

By Aczel method, the (constant) fixpoint $\color{blue}\Lfp{\subseteq}\LAMBDA{X}S$ is defined by
$\color{blue}\{\frac{\emptyset}{c}\mid c\in S\}$; 

So for \texttt{\small while (B) S} we have an axiom $\color{blue}\frac {\emptyset}{\displaystyle\ulstrut\{P\}\,\texttt{\small while (B) S}\,\{Q\}}$ with side condition $\color{blue}J^0=P,\ [J^n\cap\mathcal{B}\sqb{\texttt{\small B}}]\,\texttt{\small S}\,[J^{n+1}],\ Q \subseteq(\bigcup_{n\in\mathbb{N}}J^n)\cap\mathcal{B}\sqb{\neg\texttt{\small B}}$;

Traditionally, the side condition is considered a premiss, to get $$\color{blue}\frac {\displaystyle\llstrut J^0=P,\ [J^n\cap\mathcal{B}\sqb{\texttt{\small B}}]\,\texttt{\small S}\,[J^{n+1}],\ Q \subseteq(\bigcup_{n\in\mathbb{N}}J^n)\cap\mathcal{B}\sqb{\neg\texttt{\small B}}}{\displaystyle\ulstrut[P]\,\texttt{\small while (B) S}\,[Q]}$$
\let\qed\relax\end{proof}

\newpage
\noindent\hypertarget{lem:slides:6}{\textsc{Lemma 6}}. 
\color{blue}\begin{eqntabular*}{rclcl}
\textsf{post}(R)P\cap Q\neq \emptyset
&\Leftrightarrow&
\exists\sigma\in P\mathrel{.}\exists\sigma'\in Q\mathrel{.}\pair{\sigma}{\sigma'}\in R
&\Leftrightarrow&
P\cap \textsf{pre}(R) Q\neq\emptyset
\end{eqntabular*}\color{black}
\vspace*{-1.75em}
\begin{proof}
\begin{calculus}[$\Leftrightarrow$~]
\formula{\textsf{post}(R)P\cap Q\neq \emptyset}\\
$\Leftrightarrow$
\formulaexplanation{\{\sigma'\mid\exists\sigma\in P\mathrel{.}\pair{\sigma}{\sigma'}\in R\} \cap Q\neq \emptyset}{def\@. \textsf{post}}\\
$\Leftrightarrow$
\formulaexplanation{\exists\sigma'\mathrel{.}\sigma'\in\{\sigma'\mid\exists\sigma\in P\mathrel{.}\pair{\sigma}{\sigma'}\in R\} \wedge\sigma'\in Q}{def\@. $\cap$ and $\emptyset$}\\
$\Leftrightarrow$
\formulaexplanation{\exists\sigma'\mathrel{.}\exists\sigma\in P\mathrel{.}\pair{\sigma}{\sigma'}\in R \wedge\sigma'\in Q}{def\@. $\in$}\\
$\Leftrightarrow$
\formulaexplanation{\exists\sigma\in P\mathrel{.}\exists\sigma'\in Q\mathrel{.}\pair{\sigma}{\sigma'}\in R}{commutativity}\\
$\Leftrightarrow$
\formulaexplanation{\exists\sigma\mathrel{.}\sigma\in P\wedge\sigma\in\{\sigma\mid\exists\sigma'\in Q\mathrel{.}\pair{\sigma}{\sigma'}\in R\}}{def\@. $\in$}\\
$\Leftrightarrow$
\formulaexplanation{P\cap\{\sigma\mid\exists\sigma'\in Q\mathrel{.}\pair{\sigma}{\sigma'}\in R\}\neq\emptyset}{def\@. $\cap$ and $\emptyset$}\\
$\Leftrightarrow$
\lastformulaexplanation{P\cap \textsf{pre}(R) Q\neq\emptyset}{\hyperlink{lem:slides:6}{\textsc{lem\@. 6}}}{\mbox{\qed}}
\end{calculus}\let\qed\relax
\end{proof}

\bgroup
\noindent\hypertarget{lem:slides:7}{\textsc{Lemma 7}}. 
\color{blue}
\color{blue}
\begin{eqntabular}[fl]{rcl}
\neg(\{P\}\,\texttt{\small S}\{Q\})
&\Leftrightarrow&
\textsf{post}(R)P\cap \neg Q\neq \emptyset\nonumber\\[-0.5ex]
&\Leftrightarrow&
\exists \sigma\in P\mathrel{.}\exists\sigma'\not\in Q\mathrel{.}\pair{\sigma}{\sigma'}\in\sqb{\texttt{\small S}}\nonumber\\[-0.5ex]
&\Leftrightarrow&
P\cap \mathsf{pre}\sqb{\texttt{\small S}}\neg Q\neq\emptyset\nonumber\\[0.5ex]
&\smash{\stackrel{\displaystyle\;\not\Rightarrow}{\displaystyle\Leftarrow}}
&
\color{blue}[P]\texttt{\small S}[\neg Q]\renumber{{\color{black}(IL is sufficient but not necessary for incorrectness)}}
\end{eqntabular}\color{black}%
\makeatletter
\makeatother
\begin{proof}
\begin{calculus}[$\Leftrightarrow$~]
\formula{\neg(\{P\}\,\texttt{\small S}\{Q\})}\\
$\Leftrightarrow$
\formulaexplanation{\neg(\textsf{post}\sqb{\texttt{\small S}}P\subseteq Q)}{\hyperlink{lem:slides:4}{\textsc{lem\@. 4}}}\\
$\Leftrightarrow$
\formulaexplanation{\textsf{post}\sqb{\texttt{\small S}}P\cap\neg Q\neq\emptyset}{De Morgan}\\
$\Leftrightarrow$
\formula{\exists \sigma\in P\mathrel{.}\exists\sigma'\not\in Q\mathrel{.}\pair{\sigma}{\sigma'}\in\sqb{\texttt{\small S}}}\\
$\Leftrightarrow$
\formulaexplanation{P\cap \mathsf{pre}\sqb{\texttt{\small S}}\neg Q\neq\emptyset}{def\@. $\mathsf{pre}$}\\[1em]
\formulaexplanation{[P]\texttt{\small S}[\neg Q]}{reverse Hoare aka incorrectness logic}\\
$\Leftrightarrow$
\formulaexplanation{\neg Q\subseteq \textsf{post}\sqb{\texttt{\small S}}P}{def\@. triple}\\
$\Leftrightarrow$
\formulaexplanation{\neg Q\subseteq \{\sigma'\mid\exists\sigma\in P\mathrel{.}\pair{\sigma}{\sigma'}\in\sqb{\texttt{\small S}}\}}{def\@.  \textsf{post}}\\
$\Leftrightarrow$
\formulaexplanation{\forall \sigma'\not\in Q\mathrel{.}\exists\sigma\in P\mathrel{.}\pair{\sigma}{\sigma'}\in\sqb{\texttt{\small S}}}{def\@. $\subseteq$ and $\neg$}\\[-0.5ex]
$\stackrel{\displaystyle\not\Leftarrow}{\displaystyle\;\Rightarrow}$
\formula{\exists \sigma\in P\mathrel{.}\exists\sigma'\mathrel{.}\pair{\sigma}{\sigma'}\in\sqb{\texttt{\small S}}\wedge\sigma'\not\in Q}\\
\lastexplanation{($\Rightarrow$)\quad Assume $\neg Q\neq\emptyset$ so pick $\sigma_0\in\neg Q$. Then, by hypothesis,
$\exists \sigma_1\in P\mathrel{.}\pair{\sigma_0}{\sigma_1}\in\sqb{\texttt{\small S}}$ proving $\exists \sigma\in P\mathrel{.}\exists\sigma'\mathrel{.}\pair{\sigma}{\sigma'}\in\sqb{\texttt{\small S}}\wedge \sigma'\not\in Q$ with $\sigma=\sigma_0$ and $\sigma'=\sigma_1$;\\[0.5ex]
(${\displaystyle\not\Leftarrow}$) If $\neg Q=\emptyset$ i.e\@. $Q=\Sigma$ then $\forall \sigma'\not\in Q\mathrel{.}\exists\sigma\in P\mathrel{.}\pair{\sigma}{\sigma'}\in\sqb{\texttt{\small S}}$ is vacuously true while $\exists\sigma'\mathrel{.}\sigma'\not\in Q$ hence $\exists \sigma\in P\mathrel{.}\exists\sigma'\mathrel{.}\pair{\sigma}{\sigma'}\in\sqb{\texttt{\small S}}\wedge\sigma'\not\in Q$ is false}{\mbox{\qed}}
\end{calculus}\let\qed\relax
\end{proof}
For example, $\neg(\{\textsf{true}\}\,\texttt{\small x = 1}\,\{\texttt{\small x}=2\})$ has
neither $[\textsf{true}]\,\texttt{\small x = 1}\,[\texttt{\small x}=2])$ nor $[\textsf{true}]\,\texttt{\small x = 1}\,[\texttt{\small x}\neq 2])$
\egroup

\newpage
\section*{Hoare incorrectness logic theory}
\noindent\hypertarget{th:slides:6}{\textsc{Theorem 6 (theory of $\overline{\textrm{HL}}$)}}. 
\bgroup\color{blue}
\begin{eqntabular}[fl]{@{\quad}rcl}
\mathcal{T}_{\overline{\textrm{HL}}}(\texttt{\small W})
&\triangleq&\textsf{post}({\subseteq},{\supseteq})\comp\alpha^{\neg}\comp\mathcal{T}_{\textrm{HL}}(\texttt{\small W})\colsep{=}\alpha^{\neg}\comp\mathcal{T}_{\textrm{HL}}(\texttt{\small W})
\renumber{{\color{black}\texttt{\small W} = \texttt{\small while (B) S}}}\\
&=&
\{\pair{P}{Q}\mid \begin{array}[t]{@{}l@{}}
\exists n\geqslant 1\mathrel{.}
\exists\pair{\sigma_i\in I}{i\in\interval{1}{n}}\mathrel{.}\sigma_1\in P\wedge{}\\
\forall  i\in\interval[open right]{1}{n}\mathrel{.}\pair{\mathcal{B}\sqb{\texttt{\small B}}\cap\{\sigma_i\}}{\{\sigma_{i+1}\}}\in\mathcal{T}_{\overline{\textrm{HL}}}(\texttt{\small S})
\wedge
\sigma_n\not\in\mathcal{B}\sqb{\texttt{\small B}} \wedge\sigma_n\not\in Q\}
\end{array}\nonumber
\end{eqntabular}
\egroup
\noindent Observe that $\textsf{post}({\subseteq},{\supseteq})$ can be dispensed with. This shows that the consequence rule is useless for Hoare incorrectness logic.
\begin{proof}
\begin{calculus}
\formulaexplanation{\color{blue}\mathcal{T}_{\overline{\textrm{HL}}}(\texttt{\small W}) \colsep{=}\textsf{post}({\subseteq},{\supseteq})\comp\alpha^{\neg}\comp\mathcal{T}_{\textrm{HL}}(\texttt{\small W})}{def\@. $\mathcal{T}_{\overline{\textrm{HL}}}$}\\
=
\formula{\textsf{post}(({\subseteq},{\supseteq})(\neg\{\pair{P}{Q}\mid\textsf{\upshape post}\sqb{\texttt{\small W}}P \subseteq Q\})}\\[-0.5ex]
\rightexplanation{\hyperlink{lem:slides:4}{{lemma 4 (strongest postcondition over approximation)}} and def\@. (\ref{eq-complement-GC}) of $\alpha^{\neg}$}\\
=
\formulaexplanation{\textsf{post}({\subseteq},{\supseteq})(\{\pair{P}{Q}\mid\neg(\textsf{\upshape post}\sqb{\texttt{\small W}}P \subseteq Q)\})}{def\@. $\neg$}\\
=
\formulaexplanation{\textsf{post}({\subseteq},{\supseteq})(\{\pair{P}{Q}\mid\textsf{\upshape post}\sqb{\texttt{\small W}}P \cap \neg Q\neq\emptyset\})}{def\@. $\subseteq$ and $\neg$}\\
=
\formulaexplanation{\{\pair{P'}{Q'}\mid\exists\pair{P}{Q}\in\{\pair{P}{Q}\mid\textsf{\upshape post}\sqb{\texttt{\small W}}P \cap \neg Q\neq\emptyset\}\mathrel{.}
\pair{P}{Q}\mathrel{{\subseteq},{\supseteq}}\pair{P'}{Q'}\}}{def\@. \textsf{post}}\\
=
\formulaexplanation{\{\pair{P'}{Q'}\mid\exists\pair{P}{Q}\mathrel{.}\textsf{\upshape post}\sqb{\texttt{\small W}}P \cap \neg Q\neq\emptyset\wedge
\pair{P}{Q}\mathrel{{\subseteq},{\supseteq}}\pair{P'}{Q'}\}}{def\@. $\in$}\\
=
\formulaexplanation{\{\pair{P'}{Q'}\mid\exists\pair{P}{Q}\mathrel{.}\textsf{\upshape post}\sqb{\texttt{\small W}}P \cap \neg Q\neq\emptyset\wedge
P\subseteq P'\wedge Q\supseteq Q'\}}{component wise def\@. of $\mathrel{{\subseteq},{\supseteq}}$}\\
=
\formula{\{\pair{P'}{Q'}\mid\exists{Q}\mathrel{.}\textsf{\upshape post}\sqb{\texttt{\small W}}P' \cap \neg Q\neq\emptyset\wedge
Q\supseteq Q'\}}\\[-0.5ex]
\explanation{($\subseteq$)\quad if $P\subseteq P'$ then $\textsf{\upshape post}\sqb{\texttt{\small W}}P \subseteq \textsf{\upshape post}\sqb{\texttt{\small W}}P'$ by
(\ref{eq:def:post:GC}) so that $\textsf{\upshape post}\sqb{\texttt{\small W}}P \cap \neg Q\neq\emptyset$ implies $\textsf{\upshape post}\sqb{\texttt{\small W}}P' \cap \neg Q\neq\emptyset$;
\\
($\supseteq$)\quad conversely, if $\exists{Q}\mathrel{.}\textsf{\upshape post}\sqb{\texttt{\small W}}P'$, then $\exists P\mathrel{.}\textsf{\upshape post}\sqb{\texttt{\small W}}P \cap \neg Q\neq\emptyset\wedge
P\subseteq P'$ by choosing $P=P'$.
}\\
=
\formulaexplanation{\{\pair{P'}{Q'}\mid\textsf{\upshape post}\sqb{\texttt{\small W}}P' \cap \neg Q'\neq\emptyset\}\colsep{=} {\{\pair{P}{Q}\mid\neg(\textsf{\upshape post}\sqb{\texttt{\small W}}P \subseteq Q)\}\color{blue}\colsep{=}\alpha^{\neg}\comp\mathcal{T}_{\textrm{HL}}(\texttt{\small W})}}{Q.E.D.}\\[-0.5ex]
\explanation{($\subseteq$)\quad if $Q\supseteq Q'$ then $\neg Q'\supseteq \neg Q$ so
$\textsf{\upshape post}\sqb{\texttt{\small W}}P' \cap \neg Q\neq\emptyset$ implies $\textsf{\upshape post}\sqb{\texttt{\small W}}P' \cap \neg Q'\neq\emptyset$;\\
($\supseteq$)\quad conversely $\textsf{\upshape post}\sqb{\texttt{\small W}}P' \cap \neg Q'\neq\emptyset$ implies $\exists{Q}\mathrel{.}\textsf{\upshape post}\sqb{\texttt{\small W}}P' \cap \neg Q\neq\emptyset\wedge
Q\supseteq Q'$ by choosing $Q=Q'$.
}\\[1ex]
=
\formula{\{\pair{P}{Q}\mid
\textsf{\upshape post}\sqb{\neg\texttt{\small B}}(\Lfp{\subseteq}\bar{\bar{F}}^e_P) \cap \neg Q\neq\emptyset\}}\\[-0.5ex]
\rightexplanation{\hyperlink{th:slides:1}{{theorem 1 (iteration strongest postcondition)}}, where $\bar{\bar{F}}^e_P(X)\triangleq P \cup \textsf{\upshape post}(\sqb{\texttt{\small B}}\fatsemi\sqb{\texttt{\small S}}^e)X$ }\\
=
\formulaexplanation{\{\pair{P}{Q}\mid \Lfp{\subseteq}\bar{\bar{F}}^e_P \cap \textsf{\upshape pre}\sqb{\neg\texttt{\small B}}(\neg Q)\neq\emptyset\}}{(\ref{eq:post-t-pre-t}.d)}\\
=
\formulaexplanation{
\{\pair{P}{Q}\mid\exists I\in\wp(\Sigma)\mathrel{.}\bar{\bar{F}}^e_P(I)\subseteq I\wedge
\exists\pair{W}{\leqslant}\in\mathfrak{Wf}\mathrel{.}\exists\nu\in I\rightarrow W\mathrel{.}
\exists\pair{\sigma_i\in I}{i\in\interval{1}{\infty}}\mathrel{.}
\sigma_1\in \bar{\bar{F}}^e_P(\emptyset)\wedge
\forall  i\in\interval{1}{\infty}\mathrel{.}\sigma_{i+1}\in\bar{\bar{F}}^e_P(\{\sigma_i\})\wedge
\forall  i\in\interval{1}{\infty}\mathrel{.} (\sigma_i\neq \sigma_{i+1})\Rightarrow(\nu(\sigma_i)>\nu(\sigma_{i+1})\wedge
\forall  i\in\interval{1}{\infty}\mathrel{.} (\nu(\sigma_i)\not>\nu(\sigma_{i+1})\Rightarrow\{\sigma_i\}\cap \textsf{\upshape pre}\sqb{\neg\texttt{\small B}}(\neg Q)\neq 0
\}}{Th\@. \ref{th:abstract-least-fixpoint-non-emptiness}}\\[1ex]
=
\formula{
\{\pair{P}{Q}\mid\exists I\in\wp(\Sigma)\mathrel{.}
P\subseteq I\wedge \textsf{\upshape post}(\sqb{\texttt{\small B}}\fatsemi\sqb{\texttt{\small S}}^e)I\subseteq I\wedge
\exists\pair{W}{\leqslant}\in\mathfrak{Wf}\mathrel{.}\exists\nu\in I\rightarrow W\mathrel{.}
\exists\pair{\sigma_i\in I}{i\in\interval{1}{\infty}}\mathrel{.}
\sigma_1\in P\wedge
\forall  i\in\interval{1}{\infty}\mathrel{.}(\sigma_{i+1}\in P\vee \{\sigma_{i+1}\}\subseteq\textsf{\upshape post}(\sqb{\texttt{\small B}}\fatsemi\sqb{\texttt{\small S}}^e)\{\sigma_i\})\wedge
\forall  i\in\interval{1}{\infty}\mathrel{.} (\sigma_i\neq \sigma_{i+1})\Rightarrow(\nu(\sigma_i)>\nu(\sigma_{i+1})\wedge
\forall  i\in\interval{1}{\infty}\mathrel{.} (\nu(\sigma_i)\not>\nu(\sigma_{i+1})\Rightarrow\sigma_i\in\textsf{\upshape pre}\sqb{\neg\texttt{\small B}}(\neg Q)
\}}\\[0.5ex]
\rightexplanation{def\@. $\bar{\bar{F}}^e_P(X)\triangleq P \cup \textsf{\upshape post}(\sqb{\texttt{\small B}}\fatsemi\sqb{\texttt{\small S}}^e)X$,  $\subseteq$, and \textsf{\upshape post}, which is $\emptyset$-strict}\\
=
\formula{
\{\pair{P}{Q}\mid\exists I\in\wp(\Sigma)\mathrel{.}
P\subseteq I\wedge \textsf{\upshape post}(\sqb{\texttt{\small B}}\fatsemi\sqb{\texttt{\small S}}^e)I\subseteq I\wedge
\exists\pair{W}{\leqslant}\in\mathfrak{Wf}\mathrel{.}\exists\nu\in I\rightarrow W\mathrel{.}
\exists\pair{\sigma_i\in I}{i\in\interval{1}{\infty}}\mathrel{.}
\sigma_1\in P\wedge
\forall  i\in\interval{1}{\infty}\mathrel{.}\{\sigma_{i+1}\}\subseteq\textsf{\upshape post}(\sqb{\texttt{\small B}}\fatsemi\sqb{\texttt{\small S}}^e)\{\sigma_i\}\wedge
\forall  i\in\interval{1}{\infty}\mathrel{.} (\sigma_i\neq \sigma_{i+1})\Rightarrow(\nu(\sigma_i)>\nu(\sigma_{i+1})\wedge
\forall  i\in\interval{1}{\infty}\mathrel{.} (\nu(\sigma_i)\not>\nu(\sigma_{i+1})\Rightarrow\sigma_i\in\textsf{\upshape pre}\sqb{\neg\texttt{\small B}}(\neg Q)
\}}\\[0.5ex]
\rightexplanation{since if $\sigma_{i+1}\in P$, we can equivalently consider the sequence $\pair{\sigma_j\in I}{j\in\interval{i+1}{\infty}}$}\\
=
\formula{
\{\pair{P}{Q}\mid\exists I\in\wp(\Sigma)\mathrel{.}
P\subseteq I\wedge \textsf{\upshape post}(\sqb{\texttt{\small B}}\fatsemi\sqb{\texttt{\small S}}^e)I\subseteq I\wedge
\exists n\geqslant 1\mathrel{.}
\exists\pair{\sigma_i\in I}{i\in\interval{1}{n}}\mathrel{.}
\sigma_1\in P\wedge
\forall  i\in\interval[open right]{1}{n}\mathrel{.}\{\sigma_{i+1}\}\subseteq\textsf{\upshape post}(\sqb{\texttt{\small B}}\fatsemi\sqb{\texttt{\small S}}^e)\{\sigma_i\}\wedge
\sigma_n\in\textsf{\upshape pre}\sqb{\neg\texttt{\small B}}(\neg Q)
\}}\\[0.5ex]
\explanation{($\subseteq$)\quad By $\pair{W}{\leqslant}\in\mathfrak{Wf}$, $\nu\in I\rightarrow W$, $\forall  i\in\interval{1}{\infty}\mathrel{.} (\sigma_i\neq \sigma_{i+1})\Rightarrow(\nu(\sigma_i)>\nu(\sigma_{i+1})$, the sequence is ultimately stationary at some rank $n$. For then on, $\sigma_{i+1}=\sigma_{i}$, $i\geqslant n$ and so $\nu(\sigma_i)=\nu(\sigma_{i+1})$.
Therefore 
 $\forall  i\in\interval{1}{\infty}\mathrel{.} (\nu(\sigma_i)\not>\nu(\sigma_{i+1})\Rightarrow\sigma_i\not\in{Q}$ implies that $\sigma_n\in\textsf{\upshape pre}\sqb{\neg\texttt{\small B}}(\neg Q)$;\\
($\supseteq$)\quad Conversely, from $\pair{\sigma_i\in I}{i\in\interval{1}{n}}$ we can define $W=\{\sigma_i\mid i\in\interval{1}{n}\} \cup \{-\infty\}$ with
$-\infty<\sigma_i<\sigma_{i+1}$ and $\nu(x)=\si x\in\{\sigma_i\mid i\in\interval{1}{n}\alors x\sinon-\infty\fsi$ and the sequence $\pair{\sigma_j\in I}{j\in\interval{1}{\infty}}$ repeats $\sigma_n$ ad infimum for $j\geqslant n$.}\\[0.5ex]
=
\formulaexplanation{
\{\pair{P}{Q}\mid\exists I\in\wp(\Sigma)\mathrel{.}
P\subseteq I\wedge \textsf{\upshape post}(\sqb{\texttt{\small B}}\fatsemi\sqb{\texttt{\small S}}^e)I\subseteq I\wedge
\exists n\geqslant 1\mathrel{.}
\exists\pair{\sigma_i\in I}{i\in\interval{1}{n}}\mathrel{.}
\sigma_1\in P\wedge
\forall  i\in\interval[open right]{1}{n}\mathrel{.}\{\sigma_{i+1}\}\subseteq\textsf{\upshape post}(\sqb{\texttt{\small B}}\fatsemi\sqb{\texttt{\small S}}^e)\{\sigma_i\}\wedge
\sigma_n\not\in\mathcal{B}\sqb{\texttt{\small B}} \wedge\sigma_n\not\in Q\}}{def\@. \textsf{\upshape pre}}\\[0.5ex]
=
\formulaexplanation{
\{\pair{P}{Q}\mid 
\exists n\geqslant 1\mathrel{.}
\exists\pair{\sigma_i\in I}{i\in\interval{1}{n}}\mathrel{.}
\sigma_1\in P\wedge
\forall  i\in\interval[open right]{1}{n}\mathrel{.}\{\sigma_{i+1}\}\subseteq\textsf{\upshape post}(\sqb{\texttt{\small B}}\fatsemi\sqb{\texttt{\small S}}^e)\{\sigma_i\}\wedge
\sigma_n\not\in\mathcal{B}\sqb{\texttt{\small B}} \wedge\sigma_n\not\in Q\}}{$I$ is not used and can always be chosen to be $\Sigma$}\\[0.5ex]
=
\formulaexplanation{
\{\pair{P}{Q}\mid 
\exists n\geqslant 1\mathrel{.}
\exists\pair{\sigma_i\in I}{i\in\interval{1}{n}}\mathrel{.}
\sigma_1\in P\wedge
\forall  i\in\interval[open right]{1}{n}\mathrel{.}\textsf{\upshape post}(\sqb{\texttt{\small B}}\fatsemi\sqb{\texttt{\small S}}^e)\{\sigma_i\}\cap \{\sigma_{i+1}\}\neq\emptyset\wedge
\sigma_n\not\in\mathcal{B}\sqb{\texttt{\small B}} \wedge\sigma_n\not\in Q\}}{since $x\in X\Leftrightarrow X\cap\{x\}\neq\emptyset$}\\[0.5ex]
=
\formulaexplanation{
\{\pair{P}{Q}\mid 
\exists n\geqslant 1\mathrel{.}
\exists\pair{\sigma_i\in I}{i\in\interval{1}{n}}\mathrel{.}
\sigma_1\in P\wedge
\forall  i\in\interval[open right]{1}{n}\mathrel{.}\textsf{\upshape post}(\sqb{\texttt{\small B}}\fatsemi\sqb{\texttt{\small S}}^e)\{\sigma_i\}\cap \neg(\neg\{\sigma_{i+1}\})\neq\emptyset\wedge
\sigma_n\not\in\mathcal{B}\sqb{\texttt{\small B}} \wedge\sigma_n\not\in Q\}}{def\@. $\neg X=\Sigma\setminus X$}\\[0.5ex]
=
\formulaexplanation{
\{\pair{P}{Q}\mid 
\exists n\geqslant 1\mathrel{.}
\exists\pair{\sigma_i\in I}{i\in\interval{1}{n}}\mathrel{.}
\sigma_1\in P\wedge
\forall  i\in\interval[open right]{1}{n}\mathrel{.}\neg(\textsf{\upshape post}(\sqb{\texttt{\small B}}\fatsemi\sqb{\texttt{\small S}}^e)\{\sigma_i\}\subseteq(\neg\{\sigma_{i+1}\}))\wedge
\sigma_n\not\in\mathcal{B}\sqb{\texttt{\small B}} \wedge\sigma_n\not\in Q\}}{$\neg(X\subseteq Y)\Leftrightarrow(X\cap\neg Y\neq\emptyset$}\\[0.5ex]
=
\formulaexplanation{
\{\pair{P}{Q}\mid 
\exists n\geqslant 1\mathrel{.}
\exists\pair{\sigma_i\in I}{i\in\interval{1}{n}}\mathrel{.}
\sigma_1\in P\wedge
\forall  i\in\interval[open right]{1}{n}\mathrel{.}\neg(\textsf{\upshape post}(\sqb{\texttt{\small S}}^e)(\mathcal{B}\sqb{\texttt{\small B}}\cap\{\sigma_i\})\subseteq(\neg\{\sigma_{i+1}\}))\wedge
\sigma_n\not\in\mathcal{B}\sqb{\texttt{\small B}} \wedge\sigma_n\not\in Q\}}{def\@. $\textsf{\upshape post}$, $\sqb{\texttt{\small B}}$, and $\fatsemi$}\\[0.5ex]

=
\formulaexplanation{
\{\pair{P}{Q}\mid 
\exists n\geqslant 1\mathrel{.}
\exists\pair{\sigma_i\in I}{i\in\interval{1}{n}}\mathrel{.}
\sigma_1\in P\wedge
\forall  i\in\interval[open right]{1}{n}\mathrel{.}\pair{\mathcal{B}\sqb{\texttt{\small B}}\cap\{\sigma_i\}}{\{\sigma_{i+1}\}}\in\{\pair{P}{Q}\mid\neg(\textsf{\upshape post}(\sqb{\texttt{\small S}}^e)P\subseteq\neg Q)\}
\wedge
\sigma_n\not\in\mathcal{B}\sqb{\texttt{\small B}} \wedge\sigma_n\not\in Q\}}{def\@. $\in$}\\[0.5ex]

=
\lastformulaexplanation{
\{\pair{P}{Q}\mid 
\exists n\geqslant 1\mathrel{.}
\exists\pair{\sigma_i\in I}{i\in\interval{1}{n}}\mathrel{.}
\sigma_1\in P\wedge
\forall  i\in\interval[open right]{1}{n}\mathrel{.}\pair{\mathcal{B}\sqb{\texttt{\small B}}\cap\{\sigma_i\}}{\{\sigma_{i+1}\}}\in\mathcal{T}_{\overline{\textrm{HL}}}(\texttt{\small S})
\wedge
\sigma_n\not\in\mathcal{B}\sqb{\texttt{\small B}} \wedge\sigma_n\not\in Q\}}{def\@. $\mathcal{T}_{\overline{\textrm{HL}}}(\texttt{\small S})$}{\mbox{\qed}}
\end{calculus}\let\qed\relax
\end{proof}

\section*{$\overline{\mbox{\normalfont\bfseries HL}}$ rules}

\noindent\hypertarget{th:slides:7}{\textsc{Theorem 7 ($\overline{\mbox{\normalfont HL}}$ rules for conditional iteration)}}. 
\begin{eqntabular*}{c}
\color{blue}\frac {\displaystyle\exists\pair{\sigma_i\in I}{i\in\interval{1}{n}}\mathrel{.}\sigma_1\in P\wedge{}
\forall  i\in\interval[open right]{1}{n}\mathrel{.}
\llparenthesis\,\mathcal{B}\sqb{\texttt{\small B}}\cap\{\sigma_i\}\,\rrparenthesis\,\texttt{\small S}\,\llparenthesis\,\{\sigma_{i+1}\}\,\rrparenthesis
\wedge
\sigma_n\not\in\mathcal{B}\sqb{\texttt{\small B}} \wedge\sigma_n\not\in Q}{\displaystyle\llparenthesis\,P\,\rrparenthesis\,\texttt{\small while (B) S}\,\llparenthesis\, Q\,\rrparenthesis}
\end{eqntabular*}

\begin{proof}

We write $\color{blue}\llparenthesis\,P\,\rrparenthesis\,\texttt{\small S}\,\llparenthesis\, Q\,\rrparenthesis \triangleq\pair{P}{Q}\in{\overline{\textrm{HL}}}(\texttt{\small S})$;

By structural induction (\texttt{\small S} being a strict component of \texttt{\small while (B) S}), the rule for $\color{blue}\llparenthesis\,P\,\rrparenthesis\,\texttt{\small S}\,\llparenthesis\, Q\,\rrparenthesis$ have already been defined; 

By Aczel method, the (constant) fixpoint $\color{blue}\Lfp{\subseteq}\LAMBDA{X}S$ is defined by
$\color{blue}\{\frac{\emptyset}{c}\mid c\in S\}$; 

So for \texttt{\small while (B) S} we have an axiom $\color{blue}\frac {\displaystyle\emptyset}{\displaystyle\llparenthesis\,P\,\rrparenthesis\,\texttt{\small while (B) S}\,\llparenthesis\, Q\,\rrparenthesis}$ with side condition $\color{blue}\exists\pair{\sigma_i\in I}{i\in\interval{1}{n}}\mathrel{.}\sigma_1\in P\wedge{}
\forall  i\in\interval[open right]{1}{n}\mathrel{.}
\llparenthesis\,\mathcal{B}\sqb{\texttt{\small B}}\cap\{\sigma_i\}\,\rrparenthesis\,\texttt{\small S}\,\llparenthesis\,\{\sigma_{i+1}\}\,\rrparenthesis
\wedge
\sigma_n\not\in\mathcal{B}\sqb{\texttt{\small B}} \wedge\sigma_n\not\in Q$;

\vskip1mm

Traditionally, the side condition is considered a premiss, to get $$\color{blue}\frac {\displaystyle\exists\pair{\sigma_i\in I}{i\in\interval{1}{n}}\mathrel{.}\sigma_1\in P\wedge{}
\forall  i\in\interval[open right]{1}{n}\mathrel{.}
\llparenthesis\,\mathcal{B}\sqb{\texttt{\small B}}\cap\{\sigma_i\}\,\rrparenthesis\,\texttt{\small S}\,\llparenthesis\,\{\sigma_{i+1}\}\,\rrparenthesis
\wedge
\sigma_n\not\in\mathcal{B}\sqb{\texttt{\small B}} \wedge\sigma_n\not\in Q}{\displaystyle\llparenthesis\,P\,\rrparenthesis\,\texttt{\small while (B) S}\,\llparenthesis\, Q\,\rrparenthesis}$$
\let\qed\relax\end{proof}
This is nothing but debugging formalized as a logic since $\pair{\sigma_i\in I}{i\in\interval{1}{n}}$ is a finite iteration in the loop starting with $P$ true and finishing with $Q$ false, which is obviously a counter example to Hoare triple $\{P\}\,\texttt{\small while (B) S}\,\{Q\}$. Notice that recursively $\llparenthesis\,\mathcal{B}\sqb{\texttt{\small B}}\cap\{\sigma_i\}\,\rrparenthesis\,\texttt{\small S}\,\llparenthesis\,\{\sigma_{i+1}\}\,\rrparenthesis$ enforces the execution of the loop body \texttt{\small S} to terminate.

\end{document}